\documentclass[twoside, openright, 11pt]{memoir}

\usepackage[utf8]{inputenc}
\usepackage[british]{babel}
\usepackage[babel]{microtype}
\usepackage{amsmath}	
\usepackage[bitstream-charter]{mathdesign}
\usepackage{inconsolata}
\usepackage{textcomp}	
\usepackage{mathtools}
\usepackage{tikz}
\usepackage{tikz-cd}	
\usepackage[backend=biber, style=alphabetic, alldates=long]{biblatex}
\usepackage{csquotes}	
\usepackage{hyperref}	
\usepackage{amsthm}	
\usepackage[capitalise]{cleveref}	

\setsecnumdepth{subsection}

\setstocksize{297mm}{210mm}
\settrimmedsize{\stockheight}{\stockwidth}{*}
\settrims{0mm}{0mm}
\setlrmarginsandblock{40mm}{40mm}{*}	
\setulmarginsandblock{46mm}{43mm}{*}	
\checkandfixthelayout

\makepagestyle{bicoltoruslpgp}
\makeheadrule{bicoltoruslpgp}{\textwidth}{\normalrulethickness}
\makeevenhead{bicoltoruslpgp}{\small\thepage}{\small\fscshape\leftmark}{}
\makeevenfoot{bicoltoruslpgp}{}{}{}
\makeoddhead{bicoltoruslpgp}{}{\small\rightmark}{\small\thepage}
\makeoddfoot{bicoltoruslpgp}{}{}{}
\makepsmarks{bicoltoruslpgp}{%
	\nouppercaseheads
	\createmark{chapter}{both}	{nonumber}{}{}
	\createmark{section}{right}	{shownumber}{}{\hskip.7em}
	\createplainmark{toc}{both}{\contentsname}
	\createplainmark{bib}{both}{\bibname}
}

\setpnumwidth{3em}
\setrmarg{4em}

\nocite{*}

\DefineBibliographyStrings{english}{%
	mathesis = {Master's thesis},
}

\hypersetup{
	pdftitle={Bicoloured torus loop groups},
	pdfauthor={Shan H. Shah},
	colorlinks=true,
	citecolor=magenta	
}

\newlength{\tabspace}
\theoremstyle{plain}	
\newtheorem{thm}{Theorem}[section]
\crefname{thm}{Theorem}{Theorems}
\newtheorem{prop}[thm]{Proposition}
\newtheorem{lem}[thm]{Lemma}
\newtheorem{cor}[thm]{Corollary}
\newtheorem*{mainthm}{Main theorem}

\theoremstyle{definition}	
\newtheorem{dfn}[thm]{Definition}
\crefname{dfn}{Definition}{Definitions}
\newtheorem{constr}[thm]{Construction}
\newtheorem{ingreds}[thm]{Ingredients}
\newtheorem{exmp}[thm]{Example}
\crefname{exmp}{Example}{Examples}

\theoremstyle{remark}	
\newtheorem{rmk}[thm]{Remark}
\crefname{rmk}{Remark}{Remarks}

\newtheorem*{refs}{References}

\newcommand{\defn}[1]{\emph{#1}}

\DeclareFontFamily{U}{pseus}{\skewchar\font'60}
\DeclareFontShape{U}{pseus}{m}{n}{%
     <-6> s * [0.95] eusm5%
    <6-8> s * [0.95] eusm7%
    <8-> s * [0.95] eusm10%
}{}
\DeclareFontShape{U}{pseus}{b}{n}{%
     <-6> s * [0.95] eusb5%
    <6-8> s * [0.95] eusb7%
    <8-> s * [0.95] eusb10%
}{}
\DeclareMathAlphabet\EuScript{U}{pseus}{m}{n}
\SetMathAlphabet\EuScript{bold}{U}{pseus}{b}{n}
\renewcommand{\mathcal}{\EuScript}

\renewcommand{\epsilon}{\varepsilon}
\renewcommand{\phi}{\varphi}

\newcommand{\field}[1]{\mathbb{#1}}	

\newcommand{\ZZ}{\field Z}
\newcommand{\QQ}{\field Q}
\newcommand{\RR}{\field R}
\newcommand{\CC}{\field C}

\newcommand{\hilb}[1]{\mathscr{#1}}	
\newcommand{\liealg}[1]{\mathfrak{#1}}	
\newcommand{\net}[1]{\mathcal{#1}}	

\newcommand{\liegp}[1]{\mathrm{#1}}	

\newcommand{\lieU}{\liegp U}
\newcommand{\liePU}{\liegp{PU}}
\newcommand{\liePSU}{\liegp{PSU}}
\newcommand{\lieSL}{\liegp{SL}}

\newcommand{\lieSU}{\liegp{SU}}
\newcommand{\phasegp}{\mathbb{T}}

\newcommand{\cat}[1]{\mathsf{#1}}	
\newcommand{\Ab}{\cat{Ab}}
\newcommand{\INT}{\cat{INT}}
\newcommand{\LieGp}{\cat{Lie}}
\newcommand{\VN}{\cat{VN}}

\newcommand{\abs}[1]{\lvert#1\rvert}
\newcommand{\conj}[1]{\overline{#1}}
\newcommand{\dd}{\mathop{}\!\mathrm{d}}
\newcommand{\norm}[1]{\lVert#1\rVert}

\DeclareMathOperator{\Ad}{Ad}
\DeclareMathOperator{\Aut}{Aut}
\DeclareMathOperator{\avg}{avg}
\DeclareMathOperator{\Bi}{Bi}
\DeclareMathOperator{\ch}{ch}	
\DeclareMathOperator{\Co}{Co}	

\DeclareMathOperator{\Diff}{Diff}
\DeclareMathOperator{\disc}{disc}

\DeclareMathOperator{\Hom}{Hom}
\DeclareMathOperator{\id}{id}
\let\Im\relax	
\DeclareMathOperator{\Im}{Im}	
\DeclareMathOperator{\Ind}{Ind}	
\DeclareMathOperator{\Mob}{M\ddot ob}	
\DeclareMathOperator{\Pth}{Pth}
\DeclareMathOperator{\rank}{rank}
\let\Re\relax	
\DeclareMathOperator{\Re}{Re}	
\DeclareMathOperator{\Rep}{Rep}
\DeclareMathOperator{\Res}{Res}	
\DeclareMathOperator{\Rot}{Rot}
\DeclareMathOperator{\Sect}{Sect}
\DeclareMathOperator{\sgn}{sgn}
\DeclareMathOperator{\supp}{supp}
\DeclareMathOperator{\Sym}{Sym}
\DeclareMathOperator{\TMF}{TMF}
\DeclareMathOperator{\Vir}{Vir}
\DeclareMathOperator{\vN}{vN}

\newcommand{\wh}{\circ}		
\newcommand{\bl}{\bullet}		
\newcommand{\whbl}{{\circ\!\text{\tiny /}\!\!\:\bullet}}
\newcommand{\whblsm}{{\circ\!/\!\!\:\bullet}}

\newcommand{\bound}{\mathscr{B}}
\newcommand{\match}{\text{m}}
\newcommand{\opp}{\text{op}}
\newcommand{\Leech}{\text{24}}
\newcommand{\Eta}{\mathrm H}
\newcommand{\unitnet}{\underline{\CC}}
\newcommand{\centL}{\widetilde L}
\newcommand{\centP}{\widetilde P}
\newcommand{\centV}{\widetilde V}
\newcommand{\centext}{\widetilde{\protect\phantom x}}	
\newcommand{\centiota}{\tilde \iota}
\newcommand{\centker}{\widetilde \ker}

\usetikzlibrary{decorations.markings}
\usetikzlibrary{arrows}

\newcommand{\tikzmath}[2][]{%
	\vcenter{\hbox{\begin{tikzpicture}[#1]#2%
                    \end{tikzpicture}}}%
}

\newcommand{\circpq}{%
	{\tikzmath[scale=0.04]{
		\filldraw[black] (0,0) circle (0.4)
						 (0,-8) circle (0.4);
	}}
}
\newcommand{\circl}{%
	{\tikzmath[scale=0.04, line width=rule_thickness]{
		\draw (0,0) arc (90:270:4);
		\filldraw[black] (0,0) circle (0.4)
						 (0,-8) circle (0.4);
	}}
}
\newcommand{\circlsm}{%
	{\tikzmath[scale=0.027, line width=rule_thickness]{
		\draw (0,0) arc (90:270:4);
		\filldraw[black] (0,0) circle (0.4)
						 (0,-8) circle (0.4);
	}}
}
\newcommand{\circr}{%
	{\tikzmath[scale=0.04, line width=rule_thickness]{
		\draw (0,0) arc (-90:90:4);
		\filldraw[black] (0,0) circle (0.4)
						 (0,8) circle (0.4);
	}}
}
\newcommand{\circrsm}{%
	{\tikzmath[scale=0.027, line width=rule_thickness]{
		\draw (0,0) arc (-90:90:4);
		\filldraw[black] (0,0) circle (0.4)
						 (0,8) circle (0.4);
	}}
}
\newcommand{\circldir}{%
	{\tikzmath[scale=0.04, line width=rule_thickness]{
		\draw[postaction={decorate},
			  decoration={
				  markings,
					  mark=at position 0.5 with {
						  \arrow[xshift=1.5, >={Straight Barb}]{>}
					  }
				  }
			 ] (0,0) arc (90:270:4);
		\filldraw[black] (0,0) circle (0.4)
						 (0,-8) circle (0.4);
	}}
}
\newcommand{\circrdir}{%
	{\tikzmath[scale=0.04, line width=rule_thickness]{
		\draw[postaction={decorate},
			  decoration={
				  markings,
					  mark=at position 0.5 with {
						  \arrow[xshift=1.5, >={Straight Barb}]{>}
					  }
				  }
			 ] (0,0) arc (-90:90:4);
		\filldraw[black] (0,0) circle (0.4)
						 (0,8) circle (0.4);
	}}
}

\tikzcdset{column sep/normal=1.2em, 
			arrow style=math font}	

\begin{document}


\frontmatter

\pagestyle{empty}

{
\centering
\null	
\vfill
{\huge Bicoloured torus loop groups}\bigskip\bigskip

{\Large Shan H. Shah}\par
\vfill
\vfill
\vfill
}
\cleartorecto

\renewcommand{\abstractname}{Preamble}
\begin{abstract}
This document is a modified version of the author's PhD thesis as it was defended on 15 March 2017 at Utrecht University in the Netherlands. The changes made from the latter, officially accepted variant, available from the Utrecht University Library, are minor. The page layout has been adjusted and hyperlinks have been added for the benefit of screen reading. The pages before the table of contents have been remodelled, an abstract has been added and the author's curriculum vitae has been removed. Lastly, errors of typography and in the design of figures have been corrected. No changes in the scientific content have been made and the organisation of the material has not been altered either.
\end{abstract}
\cleartoverso

{
\small
\setlength{\parindent}{0pt}
\textit{Promotor:}\medskip

Prof.dr. E.P. van den Ban\bigskip\medskip

\textit{Copromotor:}\medskip

Dr. A.G. Henriques\bigskip\medskip

\textit{Roles of promotor and copromotor:}\medskip

The primary and secondary PhD supervisors, as these terms are commonly understood in the international research community, are the people listed above as `copromotor' and `promotor', respectively. The reason for this reversal is that Dutch law requires a PhD thesis to be approved by a full professor.\bigskip\medskip

\textit{Assessment committee:}\medskip

Prof.dr. M.N. Crainic, Universiteit Utrecht

Prof.dr. C.L. Douglas, University of Oxford

Prof.dr. K.-H. Neeb, Friedrich-Alexander-Universität Erlangen-Nürnberg

Prof.dr. K.-H. Rehren, Georg-August-Universität Göttingen

Prof.dr. C. Schweigert, Universität Hamburg\bigskip\medskip

\textit{Funding:}\medskip

The research that resulted in this thesis was supported by a Graduate Funding Grant for the project `The classification of abelian Chern--Simons theories', obtained from the `Geometry and Quantum Theory' (\textsc{gqt}) mathematics cluster of the Netherlands Organisation for Scientific Research (\textsc{nwo}).
\vfill
}
\cleartorecto

\renewcommand{\abstractname}{Abstract}
\begin{abstract}
For every finite dimensional Lie group one can consider the group of all smooth loops on it, called its \emph{loop group}. Such loop groups have long been studied for, among other reasons, their relations to conformal field theories and topological quantum field theories. In this thesis we introduce a new generalisation of the loop groups associated to tori, which we name \emph{bicoloured torus loop groups}. An element of such a group consists mainly of two paths, each lying on two (possibly different) fixed tori. The definition of the group additionally imposes a constraint on the endpoints of these paths with respect to each other. We study these bicoloured torus loop groups by demonstrating that they have a theory analogous to that of ordinary torus loop groups. We namely construct certain central extensions of them by the group $\lieU(1)$ and prove that they have many properties in common with specific, known, central extensions of ordinary loop groups. Notably, we are able to construct and classify the irreducible, positive energy representations of these extensions.
\end{abstract}
\cleartorecto


\pagestyle{bicoltoruslpgp}

\microtypesetup{protrusion=false}
\tableofcontents*
\microtypesetup{protrusion=true}

\mainmatter

\chapter{Introduction}
\label{chap:intro}

The material collected and developed in this thesis is motivated by the following question:
\begin{quote}
Given that certain centrally extended loop groups and their positive energy representations give rise to conformal nets, does there exist a generalisation of the theory of loop groups which allows one to analogously construct defects between conformal nets?
\end{quote}
We focus exclusively on a special class of loop groups, namely the \defn{torus loop groups} which give rise to \defn{lattice conformal nets}, and manage to make progress towards answering the above question. Our results can be summarised by the claim
\begin{quote}
Torus loop groups generalise to so called \defn{bicoloured torus loop groups}. The latter enjoy many properties one expects them to have for constructing defects between lattice conformal nets.
\end{quote}

Our first aim in this Introduction will be to explain the various terms used in the above question. We begin with a short review of von Neumann algebras in \cref{sec:vnalgs} because these feature prominently in the definition of conformal nets that follows next in \cref{sec:confnets}. In that section we list some (mathematical) motivations for the theory of conformal nets and basic examples of these objects. We notably give an overview of the construction of certain conformal nets from central extensions of loop groups.

We then turn to introducing the notion of a defect between two conformal nets in \cref{sec:defects} of which we present elementary examples. The scarcity of these examples will immediately spark the question of finding other methods of constructing defects. To answer this question we formulate in \cref{sec:bicol} a loose, hypothetical notion of a bicoloured loop group as a method of producing conformal net defects. Our proposal starts properly in \cref{sec:mainresults}, where we make precise the special case of a definition of a bicoloured torus loop group and summarise the results we obtain in this thesis on this new notion.

We refer to \cref{sec:notsconvs} and \cref{chap:app} for various notations, conventions and definitions used in this Introduction.

\begin{rmk}[Advice to the reader]
The question posed above and the material on von Neumann algebras, conformal nets and defects treated in \crefrange{sec:vnalgs}{sec:defects} has been included only to explain the context of our studies, and it will not make a relevant reappearance until \cref{sec:defectslatticenets}. The bulk of our investigations in \cref{chap:unicol,chap:bicol} does not strictly require knowledge of these matters. The reader may therefore safely skip \crefrange{sec:vnalgs}{sec:defects} without missing technical background needed for the rest of the thesis.
\end{rmk}

\section{Von Neumann algebras}
\label{sec:vnalgs}

Von Neumann algebras were introduced by F.~Murray and J.~von Neumann in a series of papers in the 1930s and 40s with applications to representation theory and physics in mind. Their relevance to this Introduction is their appearance in the definitions of conformal nets and defects we give in \cref{sec:confnets,sec:defects}, respectively, and they will make a brief reappearance in \cref{chap:outlook}. We assume a passing familiarity with the weak, the strong operator and ultraweak topology on the algebra of bounded operators on a Hilbert space from the reader, as treated in for example \parencite[Chapters IV and IX]{conway:funcana}.

Let $\hilb H$ be a Hilbert space and denote by $\bound(\hilb H)$ its algebra of bounded operators. For a subset $S \subseteq \bound(\hilb H)$, its \defn{commutant} $S'$ is defined as the set of bounded operators on $\hilb H$ which commute with all operators of $S$. We list some algebraic properties of taking commutants. It reverses inclusions of subsets, and we have $S \subseteq S''$ and $S''' = S'$. This means that the operation of taking the commutant does not continue indefinitely. It is always true that $S'$ is a unital subalgebra of $\bound(\hilb H)$ and if $S$ is \defn{self-adjoint}, meaning that $S^* = S$, then $S'$ is a unital $*$-subalgebra.

A topological property of forming the commutant is that $S'$, and hence also $S''$ is weakly closed in $\bound(\hilb H)$. The following important result strengthens this fact dramatically when $S$ is a unital $*$-algebra.

\begin{thm}[Von Neumann's Bicommutant Theorem]
The double commutant $A''$ of a unital $*$-subalgebra $A$ of $\bound(\hilb H)$ is equal to both the closure of $A$ in the weak, and in the strong operator topology.
\end{thm} 

That is, the algebraic operation of taking the double commutant of a unital $*$-subalgebra of $\bound(\hilb H)$ can be expressed in two, equivalent topological terms. It implies that the following definition is unambiguous.

\begin{dfn}[(Concrete) von Neumann algebras]
\label{dfn:concvnalgs}
A \defn{(concrete) von Neumann algebra} is a unital $*$-subalgebra $A$ of the algebra $\bound(\hilb H)$ of bounded operators on a Hilbert space $\hilb H$ which equivalently
\begin{itemize}
\item is closed in the weak operator topology on $\bound(\hilb H)$,
\item is closed in the strong operator topology on $\bound(\hilb H)$, or
\item satisfies $A'' = A$.
\end{itemize}
\end{dfn}

One can furthermore use the Bicommutant Theorem to prove

\begin{cor}
The smallest von Neumann algebra $\vN(S)$ containing a subset $S \subseteq \bound(\hilb H)$ equals both $(S \cup S^*)''$ and the closure of the unital $*$-subalgebra generated by $S$ in either the weak or strong operator topology.
\end{cor}

So elements of $\vN(S)$ can be thought of as limits (in the weak or strong operator topology) of polynomials with (non-commuting) variables in $S \cup S^*$.

We give some examples of von Neumann algebras. The most obvious one is of course $\bound(\hilb H)$ itself. Slightly more interesting are

\begin{exmp}[Algebras of essentially bounded functions]
Let $(X, \Sigma, \mu)$ be a $\sigma$-finite measure space and $L^2(X) := L^2(X, \Sigma, \mu)$ the Hilbert space of all measurable functions $f\colon X \to \CC$ which are square integrable, modulo functions that are zero almost everywhere. Define next $L^\infty(X) := L^\infty(X, \Sigma, \mu)$ to be the set of measurable functions $f\colon X \to \CC$ which are essentially bounded, divided out by the same equivalence relation as for $L^2(X)$. With the pointwise multiplication, the $*$-operation $f^* := \overline f$ and the essential supremum norm this is a $C^*$-algebra. There is a unital, isometric $*$-homomorphism
\[
L^\infty(X) \hookrightarrow \bound\bigl(L^2(X)\bigr), \qquad f \mapsto m_f,
\]
where $m_f$ is left multiplication by $f$. It can then be proved that $L^\infty(X)' = L^\infty(X)$ (see \parencite[Theorem 6.6]{conway:funcana}), which implies that $L^\infty(X)$ is a von Neumann subalgebra of $\bound(L^2(X))$.
\end{exmp}

The above are examples of \defn{abelian} von Neumann algebras. It can in fact be shown that \emph{all} abelian von Neumann algebras on a separable Hilbert space are of this form (see \parencite[Theorem 7.8]{conway:funcana}). This is the reason why the general theory of von Neumann algebras is sometimes referred to as `non-commutative measure theory'.

The von Neumann algebras that are relevant to us in this Introduction are far from abelian, though. Their construction will instead be more similar in spirit to that of

\begin{exmp}[Group von Neumann algebras]
Let $Q\colon G \to \lieU(\hilb H)$ be a unitary representation of a topological group $G$. Then $Q(G)'$ and $Q(G)''$ are both von Neumann algebras. When $G$ is locally compact and $Q_L$ and $Q_R$ are the left, respectively, right regular representation with respect to Haar measures on $G$ we call $Q_L(G)''$ and $Q_R(G)''$ the left and right \defn{group von Neumann algebras} of $G$. They have the special feature that they are each others commutant \parencite[Proposition VII.3.1]{takesaki:volII}. An elementary example showing that a group von Neumann algebra reflects properties of $G$ is the fact that if $G$ is countable then it is amenable if and only if the algebra $Q_L(G)''$ is the weak closure of an ascending sequence of finite-dimensional $*$-subalgebras.
\end{exmp}

In order to define homomorphisms of von Neumann algebras it is useful to introduce a fourth, `coordinate-free' definition of the latter, which can be shown to be equivalent to the ones listed in \cref{dfn:concvnalgs}:

\begin{dfn}[(Abstract) von Neumann algebras]
(Taken from \parencite[Definition A.1]{bartels:confnetsI}.)
An \defn{(abstract) von Neumann algebra} is a topological, unital $*$-algebra $A$ (we do not require the multiplication to be continuous) for which there exists an injective unital $*$-homomorphism $A \hookrightarrow \bound(\hilb H)$ for some Hilbert space $\hilb H$ that is a homeomorphism onto the image of $A$ and such that the image of $A$ is closed with respect to the ultraweak topology on $\bound(\hilb H)$.

A \defn{homomorphism} between two (abstract) von Neumann algebras is a unital, continuous $*$-homomorphism.
\end{dfn}

We close this section with a brief discussion of the \defn{standard form} of a von Neumann algebra. Recall that, given a $C^*$-algebra $A$ and a positive state $f$ on it, the so called \textsc{gns}-construction allows one to build a Hilbert space $L^2(A, f)$ and a left action $Q$ of $A$ on $L^2(A, f)$ such that $Q$ is cyclic for some cyclic vector $\Omega$ and $f(a) = \langle Q(a) \Omega, \Omega \rangle$ for all $a \in A$. This construction is unique in the sense that it only depends on $f$.

However, when $A$ is a von Neumann algebra and $f$ a faithful, continuous state, remarkably, the \textsc{gns}-construction is endowed with extra structures sufficient to characterise $L^2(A, f)$ even independently from $f$. We call it the \defn{standard form} of $A$ and denote it by $L^2 A$. We refer to \parencite[46]{bartels:confnetsI} and \parencite{haagerup:stdform} for a listing of these characterising structures and we only point out one of them here, namely that $L^2 A$ is an $A$--$A$-bimodule. In general, $A$ might not have any faithful, continuous states, so that the \textsc{gns}-construction can not be used as a model for the standard form. A different, more widely applicable model can be found in \parencite[Section IX.1]{takesaki:volII}.

\section{Conformal nets}
\label{sec:confnets}

The notion of a conformal net has its origins in the field of \defn{algebraic quantum field theory} (\textsc{aqft})---a topic initiated in the 1960s by R.~Haag and D.~Kastler which seeks to describe and study quantum field theories (\textsc{qft}s) in a mathematically rigourous way through operator algebraic methods. More specifically, conformal nets are one possible mathematical model, based on general axioms formulated by Haag and Kastler, of a class of \textsc{qft}s called \defn{chiral conformal field theories} (chiral \textsc{cft}s). We refer to \parencite{araki:aqft} for general background material on \textsc{aqft} and to \parencite{kong:algstructs} for a physical motivation of the study of conformal nets, and we continue by giving their definition.

Consider the topological group
\[
\lieSU(1,1) := \Biggl\{
\begin{bmatrix*}[r]
\alpha & \beta \\
\overline\beta & \overline\alpha
\end{bmatrix*}
\Biggm\vert
\alpha, \beta \in \CC, \quad \abs\alpha^2 - \abs\beta^2 = 1
\Biggr\}.
\]
It acts by orientation-preserving diffeomorphisms on $S^1$ by setting for $g \in \lieSU(1,1)$ and\footnote{This is the sole exception we make in this thesis on our convention of denoting points on $S^1$ as $\theta$, $p$ or $q$.} $z \in S^1$,
\[
g\cdot z := \frac{\alpha z + \beta}{\overline \beta z + \overline \alpha}.
\]
Clearly, the central subgroup $\{\pm I\} \subseteq \lieSU(1,1)$ then acts trivially and we therefore get a well-defined action of the quotient group $\liePSU(1,1) := \lieSU(1,1)/\{\pm I\}$. We call $\liePSU(1,1)$ the \defn{Möbius group} and abbreviate its notation to $\Mob$. The group $\Rot(S^1)$ of counterclockwise rotations of $S^1$ embeds into $\Mob$ via
\[
\phi_\theta \mapsto
\begin{bmatrix*}[l]
e^{\pi i \theta} & 0 \\
0 & e^{-\pi i \theta} 
\end{bmatrix*},
\]
where $\phi_\theta$ is the rotation along an angle $\theta \in [0,1]$.

All representations considered here are strongly continuous and unitary, and some elementary notions regarding them are presented in \cref{app:posenergyreps}. In this Introduction we will also have use for projective representations.

\begin{dfn}
A \defn{projective representation} of a topological group $G$ on a Hilbert space $\hilb H$ is a continuous homomorphism of $G$ into the quotient group $\liePU(\hilb H) := \lieU(\hilb H)/\phasegp$, where $\phasegp$ is the topological group of complex numbers of modulus~$1$. Here, we give $\liePU(\hilb H)$ the quotient topology inherited from the strong operator topology on $\lieU(\hilb H)$.
\end{dfn}

Denote by $\INT_{S^1}$ the poset of subintervals of $S^1$ and for a Hilbert space $\hilb H$ by $\VN_{\hilb H}$ the poset of von Neumann subalgebras of $\bound(\hilb H)$, both ordered by inclusion.

\begin{dfn}
\label{dfn:concconfnet}
A \defn{(concrete, positive energy) conformal net} consists of a Hilbert space $\hilb H$ called the \defn{vacuum sector}, a unit vector $\Omega \in \hilb H$ called the \defn{vacuum vector}, a map of posets
\begin{equation}
\label{eq:cnet-posetmap}
\net A\colon \INT_{S^1} \to \VN_{\hilb H}, \qquad I \mapsto \net A(I)
\end{equation}
and a representation $U$ of $\Mob$ on $\hilb H$ satisfying the following properties:
\begin{enumerate}
\item \defn{(Locality)} if $I, J \in \INT_{S^1}$ are intervals with disjoint interiors, then the algebras $\net A(I)$ and $\net A(J)$ commute,
\item \defn{(Diffeomorphism covariance)} $U$ extends to a projective representation (which we will still denote by $U$) of $\Diff_+(S^1)$ on $\hilb H$ such that\footnote{Even though $U(\phi)$ for $\phi \in \Diff_+(S^1)$ might not be a well-defined operator on $\hilb H$, the two demands that we ask of $U(\phi)$ in this axiom remain unambiguous.}
\[
U(\phi) \net A(I) U(\phi)^* = \net A\bigl(\phi(I)\bigr)
\]
for all $\phi \in \Diff_+(S^1)$ and $I \in \INT_{S^1}$, and, moreover, if $\phi$ has support in $I$ then $U(\phi)$ commutes with $\net A(I')$,
\item \defn{(Positivity of energy)} the restriction of $U$ to the subgroup $\Rot(S^1)$ of $\Mob$ is of positive energy in the sense of \cref{dfn:posenergyrep},
\item \defn{(Vacuum axiom)} the vector $\Omega$ is invariant under the action $U$ of $\Mob$ and cyclic for the von Neumann algebra $\vN(\cup_{I \subseteq S^1} \net A(I))$.
\end{enumerate}
When $I \subseteq J$ is an inclusion of intervals the corresponding inclusion $\net A(I) \hookrightarrow \net A(J)$ of von Neumann algebras is called an \defn{isotony homomorphism}.

An \defn{isomorphism} from $\net A$ to another conformal net $\net A'$ is an isomorphism of Hilbert spaces $\hilb H \to \hilb H'$ which sends $\Omega$ to $\Omega'$ and intertwines the respective poset maps.
\end{dfn}

When referring to a conformal net we will often omit the data $\hilb H$, $\Omega$ and $U$ and simply write $\net A$, as we already did in the definition of an isomorphism of nets.

\begin{rmk}
We do not include the extension of $U$ to $\Diff_+(S^1)$ as a datum because it is shown in \parencite[Theorem 5.5]{carpi:uniqueness} that, under an additional, mild assumption, the restriction to $\Mob$ and the listed other properties already determine this extension uniquely. Those authors even speculate that that assumption is not necessary.
\end{rmk}

\begin{rmk}
For the definition of an isomorphism of conformal nets we do not require the isomorphism of Hilbert spaces to also intertwine the representations of $\Mob$ because it is a (highly non-trivial) fact that this holds automatically (see \parencite[499]{kawahigashi:classiflocalc1}).
\end{rmk}

\begin{rmk}
The usage of the word `net' is slightly strange because a conformal net is not a net in the point-set topological sense. The poset $\INT_{S^1}$ namely does not form a directed set because two intervals that cover $S^1$ do not have an upper bound in $\INT_{S^1}$. The alternative term \defn{conformal pre-cosheaf} that some authors use might be more suitable.
\end{rmk}

The definition of a conformal net involves by definition a Hilbert space on which the von Neumann algebras are represented. When the algebras can be made to act on a different Hilbert space as well, this is given its own name.

\begin{dfn}
\label{dfn:sector}
Let $\net A$ be a conformal net. A \defn{sector}\footnote{The literature often uses this term for an isomorphism class of what we call sectors. We follow the terminology of \parencite{bartels:confnetsI} instead.} $Q$ of $\net A$ consists of a Hilbert space $\hilb K$ and a family of von Neumann algebra homomorphisms $Q_I\colon \net A(I) \to \bound(\hilb K)$ for all $I \in \INT_{S^1}$ which is compatible with the isotony homomorphisms of $\net A$, meaning that $Q_I|_{\net A(J)} = Q_J$ if $J \subseteq I$. A \defn{morphism} from $Q$ to another sector $Q'$ with underlying Hilbert space $\hilb K'$ is a bounded linear map $\hilb K \to \hilb K'$ which intertwines the homomorphisms $Q_I$ and $Q_I'$ for all $I \in \INT_{S^1}$.
\end{dfn}

\begin{rmk}
One may ask why the definition of a sector does not include an action of $\Mob$ which, just like for the vacuum sector, intertwines covariantly with the algebra homomorphisms $Q_I$ and restricts to a positive energy representation of $\Rot(S^1)$. The answer is that this has been shown to hold automatically by \parencite[Theorem 5]{dantoni:sectorsconfcovariance} and \parencite[Theorem 3.8]{weiner:sectorsposenergy}, with the subtlety added that on a general sector usually only the universal covering group of $\Mob$ acts.
\end{rmk}

\subsection{Examples of conformal nets}
\label{subsec:confnets-exmps}

Let us present some basic examples of families of conformal nets. Our discussions will be deliberately cursory. We start with a family that does not have the focus of this thesis, but whose relatively simple construction will serve to illustrate features common to the constructions of the other, more complicated families.

\begin{exmp}[Heisenberg nets]
\label{exmp:heisnets}
For every finite-dimensional real vector space $F$ with a non-degenerate, symmetric, positive definite bilinear form $\langle \cdot, \cdot \rangle\colon F \times F \to \RR$ on it there exists an associated \defn{Heisenberg} conformal net $\net A_F$. We outline its construction and we refer to \parencite{dong:latticeconfnets} and \parencite{bischoff:models} for details.

One first forms the \defn{loop group} $LF := C^\infty(S^1, F)$ of all smooth maps from $S^1$ to $F$, equipped with the point-wise multiplication. It admits a canonical decomposition $LF \xrightarrow{\sim} F \oplus VF$, where $VF$ is the real vector space of all loops in $LF$ whose average over $S^1$ is zero. We will explain in \cref{subsec:unicol-Vtirreps} how the form $\langle \cdot, \cdot \rangle$ together with the holomorphic structure on the unit disc can be used to turn $VF$ into a complex pre-Hilbert space---hence making $LF$ a topological group. 

Next, the form $\langle \cdot, \cdot \rangle$ is used to construct a central extension $\centL F$ of $LF$ by the group $\phasegp$ through an explicit continuous $2$-cocycle on $LF$. The extension satisfies the so called \defn{disjoint-commutativity} property: if $\centL_I F$ denotes for an interval $I \subseteq S^1$ the pre-image in $\centL F$ of those loops in $LF$ with support in $I$, then $\centL_I F$ and $\centL_J F$ are commuting subgroups of $\centL F$ whenever $I$ and $J$ have disjoint interiors. The left action of $\Diff_+(S^1)$ on $LF$ given by precomposition with inverses of circle diffeomorphisms lifts to $\centL F$ in such a way that if $\phi \in \Diff_+(S^1)$ then $\phi \cdot \centL_I F = \centL_{\phi(I)} F$ and if $\phi$ has support in $I$ then it acts trivially on $\centL_{I'} F$.

In particular, $\Rot(S^1)$ acts on $\centL F$ and it therefore makes sense to discuss its \defn{positive energy} representations. For every $\alpha \in F$ there exists a particular irreducible such representation $W_\alpha$ with underlying Hilbert space $\hilb S_\alpha$ and it turns out that every irreducible, positive energy representation of $\centL F$ is isomorphic to one of this form up to the character by which the central subgroup $\phasegp\subseteq \centL F$ acts. A general positive energy representation is hence a direct sum of these $W_\alpha$'s.

The representation $\hilb S_0$, which we call the \defn{vacuum representation} of $\centL F$ in anticipation, carries even more structure than the action $R$ of $\Rot(S^1)$ that is part of the positive energy property. There exists a projective representation $U$ of $\Diff_+(S^1)$ on $\hilb S_0$ whose interaction with $W_0$ is described by the intertwining relation
\[
U(\phi) W_0\bigl(\centL_I F\bigr) U(\phi)^* = W_0\bigl(\centL_{\phi(I)} F\bigr)
\]
for all $I \in \INT_{S^1}$. It restricts to an honest representation of $\Mob$ fixing a certain unit vector $\Omega \in \hilb S_0$ and extends $R$. (See \parencite[Section 5]{segal:unitary} and \parencite[Subsection 5.3.2]{vromen:circle} for these facts.)

We now create for every interval $I \in \INT_{S^1}$ the von Neumann algebra
\[
\net A_F(I) := \vN\Bigl(W_0\bigl(\centL_I F\bigr)\Bigr)
\]
acting on the Hilbert space $\hilb S_0$. That is, we take the $*$-closed set $W_0(\centL_I F)$ of unitary operators and then form its von Neumann-algebraic completion inside $\bound(\hilb S_0)$. This obviously defines a map of posets $\net A_F$ as in~\eqref{eq:cnet-posetmap}. We claim that this is a conformal net. The locality axiom is deduced from the disjoint-commutative property of $\centL F$. Given the cited results, diffeomorphism covariance, positivity of energy and invariance of $\Omega$ of course also follow. Finally, the cyclicity of $\Omega$ holds because $\Omega$ is cyclic for $\centL F$ and this group is generated by its subgroups $\centL_I F$. This construction of $\net A_F$ only makes use of the representation $W_0$. The role of the other positive energy representations of $\centL F$ is that each of them can be equipped with the structure of a sector of the net $\net A_F$.

Heisenberg nets derive their name from the fact that the restriction of $\centL F$ to $VF$ is a \defn{Heisenberg group} in the sense of \cref{dfn:heisgp}.
\end{exmp}

\begin{exmp}[Lattice nets]
\label{exmp:latticenets}
The contents of this thesis are exclusively motivated by the class of \defn{lattice nets}. The rough idea behind creating one of them is to replace the real vector space $F$ in \cref{exmp:heisnets} by a different real Lie group, namely a torus $T$. Furthermore, the role of the $\RR$-valued form $\langle \cdot, \cdot \rangle$ on $F$ is taken over by the structure of an even, positive definite lattice on the free, finite rank $\ZZ$-module $\ker(\exp)$, where $\exp\colon \liealg t \to T$ is the exponential map on the Lie algebra $\liealg t$ of $T$. (See \cref{app:lattices} for an introduction to lattices.) The construction of the lattice net then proceeds very similar to that of a Heisenberg net, except that the fact that a torus is not simply-connected introduces many new complications. (Essentially, the construction of a lattice net breaks up into two steps: first forging the Heisenberg net $\net A_{\liealg t}$ and, next, enlarging $\net A_{\liealg t}$ appropriately to form the lattice net. We will not emphasise this viewpoint, though.) Note that we might as well take an even lattice $\Lambda$ as the primary datum and then define $T := \Lambda \otimes_\ZZ \phasegp$.

We begin again by forming the loop group $LT := C^\infty(S^1, T)$, which can be made into a topological group. It is not connected: each connected component is labelled by an element of $\Lambda$ which signifies how (many times, when $T$ is $1$-dimensional) a loop winds around $T$. The bi-additive form on $\Lambda$ is then, together with some minor extra data, used to construct a $\phasegp$-central extension $\centL T$ through a $2$-cocycle on $LT$ that is continuous on its identity component and similar to, but more complex than the one on $LF$. Up to non-canonical isomorphism, the extension is independent of the aforementioned extra data. It is disjoint-commutative and admits an action of $\Diff_+(S^1)$ with identical properties as for $\centL F$. A major difference of $\centL T$ with $\centL F$ is that the former possesses only finitely many isomorphism classes of irreducible, positive energy representations after fixing the character by which $\phasegp \subseteq \centL T$ acts. These classes are namely labelled by the \defn{discriminant group} $D_\Lambda := \Lambda^\vee/\Lambda$ of $\Lambda$.

We single out a particular representation denoted by $\Ind W_0$ of $\centL T$ of the type just mentioned, representing the class corresponding to $0 \in D_\Lambda$, and name it the \defn{vacuum representation}. It carries a projective representation of $\Diff_+(S^1)$, intertwining with $\Ind W_0$ and satisying the same properties as mentioned in \cref{exmp:heisnets}. One then constructs a conformal net $\net A_\Lambda$ via the same method as for a Heisenberg net. (We conjecture that this procedure works for an odd lattice as well, but that this will lead to a $\ZZ/2\ZZ$-graded net in the sense of \parencite[Definition 3.7]{douglas:geomstring}.) The other positive energy representations of $\centL T$ carry the structure of a sector of $\net A_\Lambda$ (see \parencite[Proposition 3.15]{dong:latticeconfnets}).

\cref{chap:unicol} of this thesis is devoted to explaining the details on $\centL T$ and its representation theory. For example, the crucial formula for the cocycle defining it is stated in~\eqref{eq:unicol-cocycle}. We refer to \parencite{segal:unitary}, \parencite{pressley:loopgps}, \parencite{buchholz:currentalg} \parencite{staszkiewicz}, \parencite{dong:latticeconfnets} and \parencite{bischoff:models} for further material. (We warn that the statement of \parencite[Proposition (9.5.14)]{pressley:loopgps} is likely wrong for representations that are not irreducible because the direct sum of two projective representations is in general not again a projective representation.) Some of these listed authors do not explicitly state that their constructions are valid for an arbitrary even, positive definite lattice, instead of only for one of $ADE$-type.
\end{exmp}

\begin{exmp}[Affine Kac--Moody nets]
\label{exmp:affinenets}
An \defn{affine Kac--Moody net} $\net A_{G,k}$ is associated to a pair of a compact, connected, simple, simply-connected Lie group $G$ and a choice of a multiple of a certain normalisation of the Killing form on the Lie algebra of $G$ by a positive integer $k$ called a \defn{level}. The literature sometimes refers to this as a \defn{loop group net}, but we believe this to be slightly confusing since a Heisenberg or lattice net deserves that name as well.

The construction of this net starts again with the loop group $LG$, that is, $C^\infty(S^1, G)$. The level $k$ then determines a $\phasegp$-central extension $\centL G$. Unlike in the situations of \cref{exmp:heisnets,exmp:latticenets}, $\centL G$ is not topologically trivialisable as a $\phasegp$-bundle over $LG$. It is therefore not possible to build $\centL G$ by prescribing a globally continuous $2$-cocycle on $(LG)^{\times 2}$---continuity in a neighbourhood around $(1,1) \in (LG)^{\times 2}$ is the best one can hope for. Instead, \parencite[Sections 4.4 and 4.5]{pressley:loopgps}, \parencite[Section III.3]{gabbiani:opalgs} and \parencite[Example 2.6]{waldorf:transgres} present means to build $\centL G$ in a more geometric manner. Again, $\centL G$ satisfies disjoint-commutativity and it carries a $\Diff_+(S^1)$-action with desirable properties regarding supports of loops.

Depending on the group $G$, there exist various methods for constructing the irreducible, positive energy representations of $\centL G$ and their intertwining projective $\Diff_+(S^1)$-actions. A uniform procedure is explained in \parencite[Section 4.\textsc{c}]{bartels:confnetsI} which involves building representations of a dense subalgebra of the complexified Lie algebra of $\centL G$ first and integrating these next using the techniques in \parencite{toledanolaredo:integrating}. Up to isomorphism, $\centL G$ has only finitely many irreducible, positive energy representations after fixing the character by which $\phasegp \subseteq \centL G$ acts (see \parencite[Theorem (9.3.5)]{pressley:loopgps}). (We again warn about the absence of the irreducibility assumption in \parencite[Theorems (9.3.1)(v) and (13.4.2)]{pressley:loopgps}.)

The further construction of the affine Kac--Moody net $\net A_{G,k}$ using a vacuum representation of $\centL G$ is formally similar to that of Heisenberg and lattice nets and we refer for its details to \parencite[Section 4.\textsc c]{bartels:confnetsI}. It is worth mentioning that there exists a precise relationship between affine Kac--Moody and lattice nets: when $G$ is \defn{simply-laced} $\net A_{G,1}$ is isomorphic to a net associated to a lattice of $ADE$-type (see \parencite[Proposition 3.19]{bischoff:models}).

We named affine Kac--Moody nets so because the dense subalgebra of the complexified Lie algebra of $\centL G$ mentioned earlier is a \defn{Kac--Moody algebra} of \defn{affine type}.
\end{exmp}

\begin{rmk}
The positive (definite) assumptions in \cref{exmp:heisnets,exmp:latticenets,exmp:affinenets} are not needed for the constructions of the central extensions of the respective loop groups. They are, however, necessary and sufficient for the natural bilinear forms on their representation spaces to be positive definite, so that the latter can be completed into Hilbert spaces.
\end{rmk}

Loop groups are not the only source of conformal nets:

\begin{exmp}[Virasoro nets]
For certain real numbers $c$ one can define a so called \defn{Virasoro net} $\net A_{\Vir,c}$. The range of admissible values of $c$ is the union of the closed half-line $[1,\infty)$ with a particular infinite discrete set in the interval $[0, 1)$, namely
\[
\biggl\{1 - \frac{6}{m(m+1)} \biggm\vert m = 2,3,4,\ldots\biggr\}.
\]

The general idea of constructing Virasoro nets is to replace the loop groups in \cref{exmp:heisnets,exmp:latticenets,exmp:affinenets} with $\Diff_+(S^1)$. There exists a particular non-trivial $\phasegp$-central extension, named the \defn{Virasoro--Bott group}, of $\Diff_+(S^1)$ which is characterised by the fact that its restriction over the subgroup $\liePSU(1,1)$ of $\Diff_+(S^1)$ is trivial. Because it is topologically trivialisable as a $\phasegp$-bundle over $\Diff_+(S^1)$ it can be defined by a continuous $2$-cocycle on $\Diff_+(S^1)$ called the \defn{Thurston--Bott cocycle} (see \parencite[Section II.2]{khesin:geometryinfdimgps}).

One way to build irreducible, positive energy representations of the Virasoro--Bott group is similar to the method mentioned in \cref{exmp:affinenets}: the integration of certain representations of a dense subalgebra of its complexified Lie algebra using the results of, for example, \parencite{toledanolaredo:integrating}. The peculiar restrictions on the values of $c$ are the result of investigating when these Lie algebra representations are unitary. We refer to \parencite{carpi:virasororepth} and \parencite[Section 3.3]{weiner:thesis} for further information on the construction of Virasoro nets. A proof that $\Diff_+(S^1)$ is generated by its subgroups of diffeomorphisms supported in an interval can be found in \parencite[Proposition 1.2]{loke}. The irreducible representations of $\net A_{\Vir, c}$ have been classified in \parencite[Corollary 3.9]{weiner:sectorsposenergy} and \parencite{weiner:localequivalencediffS1}.

Virasoro nets play a special role in the general theory of conformal nets because it can be shown that every net, thanks to the axiom of diffeomorphism covariance and the \defn{Haag duality} theorem, contains some Virasoro net. In turn, Virasoro nets cannot contain subnets strictly themselves \parencite{carpi:virasorominimal}. This is why (at least for $c < 1$) they are sometimes named \defn{minimal models} in the physics literature.
\end{exmp}

The examples we listed above by no means exhaust all conformal nets. They have in common that they are constructed in terms of `external' data, such as a loop group or the group $\Diff_+(S^1)$ and a choice of level. There exist, however, also plenty of constructions one can perform `internally' to the category of conformal nets, meaning that they take one or more nets as input to produce another. Examples of such constructions are direct sums, tensor products, extensions, orbifolds, mirror extensions and coset constructions. Using these one can construct many more examples of nets. We will not discuss these techniques further.

A final method of producing conformal nets we mention is the recent work \parencite{carpi:voasandconfnets} in which those authors show how to construct nets from different mathematical models of chiral \textsc{cft}s: the unitary \defn{vertex algebras}.

The technique of orbifolding shows that a naive classification of all conformal nets is as infeasible as that of all finite groups. Nevertheless, successes have been booked on the `relative' problem of classifying nets that contain a fixed one when the inclusion satisfies a certain finiteness condition. See for example \parencite{kawahigashi:classiflocalc1} for a tabulation of all nets for which the embedded Virasoro net $\net A_{\Vir,c}$ satisfies $c<1$.

\subsection{Mathematical applications of conformal nets}
\label{subsec:confnets-applic}

While the early literature on conformal nets has a large focus on answering questions motivated by physical considerations, their theory found relevance to areas closer to pure mathematics also. We mention three of these.

\begin{description}
\item[Sporadic groups and Moonshine] Even though the original definition of a vertex algebra in \parencite{borcherds:vertalgs} heavily used notions inspired by the physics literature, R.~Borcherds was at the time not motivated by the question of giving a mathematical formalisation of a chiral \textsc{cft} \parencite{borcherds:vertalgsmotiv}. He instead attempted to axiomatise, and exhibit further structure on certain constructions made by I.~Frenkel, V.~Kac, J.~Lepowsky and A.~Meurman in order to define and study the so called \defn{Monster vertex algebra}. It is a vertex algebra of which the automorphism group is precisely the \defn{Monster group}---the largest of the $26$ sporadic finite simple groups---and it played an important role in Borcherds' eventual proof in \parencite{borcherds:monstmoonshine} of J.H.~Conway and S.P.~Norton's main \defn{Monstrous Moonshine conjecture}.

Because (unitary) vertex algebras and conformal nets are both roughly mathematical models of the same physical notion, it is reasonable to expect that the Monster vertex algebra has a conformal net counterpart. Such a net has indeed been constructed in \parencite{kawahigashi:netsfromframedvoas} by taking an orbifold subnet of the conformal net associated to the Leech lattice followed by a net extension. A different construction has been given in \parencite{carpi:voasandconfnets}, in which its authors also construct a conformal net analogue of G.~Höhn's \defn{Baby Monster vertex algebra} of which the automorphism group is the \defn{Baby Monster}---the second largest of the sporadic finite simple groups.

\item[Modelling topological field theories] The category $\Sect \net A$ of sectors of a conformal net $\net A$ carries far more structure than we indicated in this Introduction.
Most significantly, there is a way of taking the tensor product of two sectors which makes $\Sect \net A$ into a \defn{tensor category}. That this tensor product is very different from the naive one already follows from the facts that it is in general not symmetric, but braided, and that the vacuum sector of $\net A$ is a monoidal unit object.

When $\net A$ satisfies a certain finiteness condition $\Sect \net A$ is of an even more special type: it is a \defn{modular tensor category} (\textsc{mtc}) (see \parencite[Corollary 37]{kawahigashi:modularity} and \parencite[Theorem 3.9]{bartels:confnetsII}). Modular tensor categories are `rare' mathematical objects (in the words of \parencite{maier:equivmodcats}) in the sense that few methods are known for producing them. One motivation for studying them is that each \textsc{mtc} gives rise to a \defn{$3$-dimensional topological quantum field theory} ($3$d \textsc{tqft}) and hence to an invariant of knots and smooth $3$-manifolds. The objects of the \textsc{mtc} then correspond to the bulk line operators of the \textsc{tqft} and the morphisms can be interpreted as local operators sitting at the junction between two bulk line operators.

Lattice nets are known to satisfy this finiteness criterion \parencite[Corollary 3.19]{dong:latticeconfnets} and the same is conjectured to hold for affine Kac--Moody nets, even though this has only been shown in particular cases (see for example \parencite[Theorem 4.18]{bartels:confnetsI} and the references cited therein). The \textsc{tqft}s arising from lattice and affine Kac--Moody nets are expected to equal the \defn{abelian Chern--Simons}, and \defn{(ordinary) Chern--Simons theories}, respectively. (Because these theories are not defined mathematically, this claim contains ample room for interpretation.)

This relation between finite conformal nets and $3$d \textsc{tqft}s is in general many-to-one, though: there might exist non-isomorphic conformal nets of which their categories of sectors are equivalent as \textsc{mtc}s and hence produce the same \textsc{tqft}s. For example, we conjecture following the work of \parencite{hohn:vertalgsgenera} that this occurs for nets associated to lattices that lie in the same \defn{genus}.

\item[Defining elliptic cohomology geometrically] In the last quarter of the 20th century a new family of generalised cohomology theories was discovered in the field of algebraic topology---the so called \defn{elliptic cohomology theories}. There is in general no reason for an arbitrary generalised cohomology theory to have an interpretation in terms of the geometry of the spaces one evaluates it on in the same way that, for example, topological $K$-theory can be defined using vector bundles over the spaces. However, a web of conjectures has gradually been formed suggesting that the elliptic cohomology theories could have a quantum field-theoretic, geometric definition. Based on joint work with A.~Bartels, C.~Douglas and A.~Henriques have started a project investigating whether a certain `universal' elliptic cohomology theory, named $\TMF$, can be modelled via conformal nets (see \parencite{douglas:tmfconfnets} for a survey of some of their results). An early success they and B.~Janssens achieved in unpublished work is a definition in terms of conformal nets of the \defn{String groups}. These groups are related to $\TMF$ in a way that is similar to the relation between the Spin groups and real topological $K$-theory.
\end{description}

\section{Defects between conformal nets}
\label{sec:defects}

Up to this point we have mostly been discussing the theory of conformal nets as it is described in the more traditional literature. We now turn our focus to recent work of A.~Bartels, C.~Douglas and A.~Henriques which revolves around their new notion of a \defn{defect} between conformal nets. Before giving its definition, we first introduce an alternative definition of a conformal net devised by those authors which simultaneously is a `coordinate-free' formulation of \cref{dfn:concconfnet} and mildly generalises it.

Denote by $\INT$ the category of which the objects are intervals\footnote{See \cref{sec:notsconvs} for our definition of an interval.} and the morphisms are smooth embeddings and write $\VN$ for the category of which the objects are (abstract) von Neumann algebras and the morphisms are either von Neumann algebra homomorphisms or antihomomorphisms. If $I \in \INT$ is an interval, $\Diff_+(I)$ stands for the group of its orientation-preserving diffeomorphisms.

\begin{dfn}
\label{dfn:abstrconfnet}
(Taken from \parencite[Definition 1.1]{bartels:confnetsI}.)
An \defn{(abstract) conformal net} is a continuous covariant functor $\net A\colon \INT \to \VN$ sending orientation-preserving and reversing embeddings to injective algebra homomorphisms and antihomomorphisms, respectively, which satifies the following properties:
\begin{enumerate}
\item \defn{(Locality)} if $I \hookrightarrow K$ and $J \hookrightarrow K$ are embeddings of intervals of which the images have disjoint interiors, then the images of the corresponding homomorphisms or antihomomorphisms $\net A(I) \hookrightarrow \net A(K)$ and $\net A(J) \hookrightarrow \net A(K)$ commute inside $\net A(K)$,
\item \defn{(Inner covariance)} if $\phi \in \Diff_+(I)$ is the identity in a neighbourhood of the endpoints of $I$, then there exists a unitary operator $U \in \net A(I)$ such that $\Ad(U) = \net A(\phi)$,
\item \defn{(Existence of a vacuum sector)} assume that $I \subsetneq K$ is a pair of an interval and a properly included subinterval such that $I$ contains a boundary point $p$ of $K$ and write $\bar I$ for the same manifold as $I$, but equipped with the reversed orientation. Then $\net A(I)$ and $\net A(\bar I)$ both act on the left of the standard form $L^2\net A(K)$---the former via the homomorphism $\net A(I) \hookrightarrow \net A(I)$ and the latter through the homomorphism $\net A(\bar I) \xrightarrow{\sim} \net A(I)^\opp \hookrightarrow \net A(K)^\opp$.

We then demand that the left action of the algebraic tensor product $\net A(I) \otimes \net A(\bar I)$ on $L^2\net A(K)$ extends to an action of the algebra $\net A(I \cup_p \bar I)$ along the homomorphism $\net A(I) \otimes \net A(\bar I) \to \net A(I \cup_p \bar I)$. Here, we mean by $I \cup_p \bar I$ any interval obtained by gluing $I$ and $\bar I$ along $p$ with a smooth structure that extends those of $I$ and $\bar I$ and such that the involution which swaps $I$ and $\bar I$ is smooth.
\end{enumerate}
A \defn{morphism} between two abstract conformal nets is defined to be a natural transformation.
\end{dfn}

The meaning of the term \defn{continuous} in the above definition can be found in \parencite{bartels:confnetsI}.

\begin{rmk}
The authors of \parencite{bartels:confnetsI} include two additional axioms in their definition of an abstract conformal net, namely those of \defn{strong additivity} and \defn{splitness}. We did not do so both to simplify the definition for the uninitiated reader and because one would otherwise be excluding some natural, non-pathological examples of conformal nets. For example, we did not include strong additivity since this property is not satisfied by Virasoro nets $\net A_{\Vir,c}$ when $c > 1$ \parencite{buchholz:haagduality}. However, we warn that many structural results in \parencite{bartels:confnetsI} do rely on these two axioms.
\end{rmk}

It is at first sight far from clear in what way the \cref{dfn:abstrconfnet} of an abstract conformal net is a generalisation of that of a concrete, positive energy net in \cref{dfn:concconfnet}. Three features in which they differ immediately stand out:
\begin{itemize}
\item for an abstract net one no longer posits the existence of an ambient Hilbert space on which the von Neumann algebras associated to intervals act. The reason for this is that the \defn{Reeh--Schlieder} and \defn{Bisognano--Wichmann theorem} imply that the axioms of a concrete, positive energy net $\net A$ are sufficiently strong to reconstruct the vacuum sector as the standard form $L^2 \net A(K)$ for \emph{any} interval $K \subseteq S^1$. Conversely, the purpose of the vacuum sector axiom for an abstract net $\net A$ is to ensure that, given an interval $K \subseteq S^1$, the standard form $L^2 \net A(K)$ not only carries actions of $\net A(K)$ and $\net A(K')$ via its $\net A(K)$--$\net A(K)$-bimodule structure, but also of all algebras $\net A(I)$ with $I \subseteq S^1$,
\item an abstract net assigns algebras to all intervals, instead of only to those that are embedded in $S^1$. This is not a stronger demand, though, since, given a concrete, positive energy net, one can assign algebras to any interval $I$ by choosing embeddings of $I$ in $S^1$. Using the diffeomorphism covariance, these algebras can then be glued together into a single algebra which is independent of the choice of embedding,
\item there is no assumption of positivity of energy in the definition of an abstract net, and it is in this direction that such a net is a genuine generalisation of a concrete, positive energy net. The motivation given in \parencite{bartels:confnetsI} for dropping this requirement is to assemble conformal nets into a symmetric monoidal category with duals, and the natural dual of a positive energy net is of negative energy.
\end{itemize}

Even though some axioms were left out in \cref{dfn:concconfnet,dfn:abstrconfnet} for the sake of exposition, we approximately cite

\begin{thm}
(See \parencite[Proposition 4.9]{bartels:confnetsI} for the precise statement.)
Every concrete, positive energy conformal net can be extended to an abstract conformal net.
\end{thm}
%

From this point onwards, whenever we will speak of a \defn{conformal net}, we will mean an abstract conformal net in the sense of \cref{dfn:abstrconfnet}.

Physically, a defect between two field theories or condensed-matter systems can be thought of as a `$2$-sided' boundary condition, designating a submanifold at the junction between the domains of the theories where the fields of one theory experience a discontinuity and transition to those of the other theory. Defects have also been named \defn{domain walls} or \defn{surface operators} in the literature and studies of them (for theories we do not treat in this thesis) can be found in for example \parencite{frohlich:dualitydefects}, \parencite{kapustin:abchernsimons} and \parencite{gukov:surfops}.

The formulation of a defect between two conformal nets introduced by \parencite{bartels:confnetsIII} is a straightforward translation of their \cref{dfn:abstrconfnet} when the category $\INT$ of intervals is enhanced to the category $\INT_{\wh\bl}$ of \defn{bicoloured intervals}. An object of this category is an interval equipped with a covering of two subintervals, one of which is seen as being coloured white, the other one as being coloured black. The subintervals are required to overlap in one point, at which the colour changes. We allow that a bicoloured interval is coloured entirely white or entirely black, meaning that the other subinterval is empty, but we rule out that one of the two subintervals is a singleton. \defn{Morphisms} between bicoloured intervals are smooth embeddings that preserve bicolourings. The full subcategories of entirely white and entirely black intervals are denoted by $\INT_\wh$ and $\INT_\bl$, respectively, and we call a bicoloured interval which is not contained in either of them \defn{genuinely bicoloured}.

After these preliminaries\footnote{The authors of \parencite{bartels:confnetsIII} require one more piece of structure of bicoloured intervals which we left out for brevity, namely a local coordinate around the colour-changing point, and morphisms are required to respect this.} we are ready for

\begin{dfn}
\label{dfn:defect}
(Taken from \parencite[Definition 1.7]{bartels:confnetsIII}.)
If $\net A$ and $\net B$ are two conformal nets, then an \defn{$\net A$--$\net B$-defect} $D$ is a covariant functor $D\colon \INT_{\wh\bl} \to \VN$ sending orientation-preserving and reversing morphisms to algebra homomorphisms and antihomomorphisms, respectively, whose restrictions to $\INT_\wh$ and $\INT_\bl$ are equal to $\net A$ and $\net B$, respectively, and which satisfies the following properties:
\begin{itemize}
\item \defn{(Isotony)} if $I \hookrightarrow K$ is a morphism of genuinely bicoloured intervals, then the corresponding homomorphism or antihomomorphism $D(I) \to D(K)$ is injective,
\item \defn{(Locality)} if $I \hookrightarrow K$ and $J \hookrightarrow K$ are morphisms of bicoloured intervals of which the images have disjoint interiors, then the images of $D(I)$ and $D(J)$ in $D(K)$ commute,
\item \defn{(Existence of a vacuum sector)} let $K$ be a genuinely bicoloured interval and $I \in \INT_\wh \cup \INT_\bl$ a subinterval of $K$ containing a boundary point $p$ of $K$. Write $\bar I$ for the same manifold as $I$, but equipped with the reversed orientation. Then $D(I)$ and $D(\bar I)$ both act on the left of the standard form $L^2 D(K)$---the former via the homomorphism $D(I) \to D(K)$ and the latter through the homomorphism $D(\bar I) \xrightarrow{\sim} D(I)^\opp \to D(K)^\opp$.

We then demand that the left action of the algebraic tensor product $D(I) \otimes D(\bar I)$ on $L^2 D(K)$ extends to an action of the algebra $D(I \cup_p \bar I)$ along the homomorphism $D(I) \otimes D(\bar I) \to D(I \cup_p \bar I)$.\footnote{See \cref{dfn:abstrconfnet} for the definition of the unicoloured interval $I \cup_p \bar I$.}
\end{itemize}
When $I \hookrightarrow K$ is a morphism of bicoloured intervals the corresponding algebra homomorphism $D(I) \to D(K)$ is called an \defn{isotony homomorphism}.

A \defn{morphism} between two $\net A$--$\net B$-defects is defined to be a natural transformation.
\end{dfn}

\begin{rmk}
The reason the above \cref{dfn:defect} does not include requirements of continuity and inner covariance analogous to those in \cref{dfn:abstrconfnet} is that these can be shown to hold automatically (see \parencite[Proposition 1.21]{bartels:confnetsIII} and \parencite[Proposition 1.10]{bartels:confnetsIII}, respectively) when a \defn{strong additivity} axiom of the nets and the defect is assumed. We removed the latter axiom for brevity.
\end{rmk}

We stress that the non-trivial information contained in an $\net A$--$\net B$ defect consists of the algebras assigned to genuinely bicoloured intervals and the isotony homomorphisms associated to inclusions of white, black or genuinely bicoloured into genuinely bicoloured ones. The rest of the data is already determined by the nets $\net A$ and $\net B$.

The notion of a defect between two conformal nets bears a formal similarity to that of a bimodule between two rings. This is not meant to suggest that, just like for a bimodule one can forget the action of one ring and obtain an ordinary module for the other ring, also a net defect possesses an underlying sector structure for both nets. Instead, it is more fruitful to think of defects as a generalisation of morphisms of conformal nets. Indeed, we will see in \cref{exmp:embeddingdefects} that suitable morphisms give rise to defects.

The analogy between defects and ring bimodules runs deeper still in that both admit a form of composition. Recall that when $M$ is a $Q$--$R$-bimodule and $N$ is an $R$--$S$-bimodule for rings $Q$, $R$ and $S$ the tensor product $M \otimes_R N$ is a $Q$--$S$-bimodule. This operation is weakly associative and weakly unital. Taking bimodules as $1$-morphisms and equivariant bimodule maps as $2$-morphisms, rings become in this way a symmetric monoidal bicategory.

Similarly, it is shown in \parencite{bartels:confnetsIII,bartels:confnetsIV} that defects (under certain conditions) can be composed and that a restricted class of conformal nets becomes a particular kind of symmetric monoidal $3$-category when defects are taken as the $1$-morphisms. The $2$-morphisms used in this $3$-category are not the ordinary morphisms between defects as stated in \cref{dfn:defect}, though. They are rather a natural generalisation of the notion of sectors as defined in \cref{dfn:sector}, which explains the increase in the categorical level compared to the situation of rings.

The study of this $3$-category has the potential to reveal a richer structure on the family of conformal nets---one that is not easily visible at the level of objects and ordinary morphisms. Unfortunately, there is very little known about it so far. For example, it is very simple to point out for every two nets $\net A$ and $\net B$ \emph{some} $\net A$--$\net B$-defect $D$. As explained in \parencite[Proposition 1.23]{bartels:confnetsI}, one namely always has an $\net A$--$\unitnet$- and a $\unitnet$--$\net B$-defect and $D$ can then be defined as their composition over the net $\unitnet$. However, this defect is `large' in the sense that it is not \defn{dualisable} (in the higher categorical sense), and no necessary conditions are known to impose on a pair of nets for a dualisable defect to exist between them.

\begin{rmk}[Alternative definitions of defects]
\cref{dfn:defect} is not the only notion of a conformal net defect put forward in the literature: in independent work the authors of \parencite{bischoff:phasebounds} offer an alternative, but more restrictive, formulation. (Their definition applies to $1$-dimensional defects between 2d \emph{full} \textsc{cft}s as well.) At the moment of this writing the precise relation between their theory and that of \parencite{bartels:confnetsIII} is not entirely clear, but we refer to \parencite[Remark 1.28]{bartels:confnetsIII} and \parencite[Remark 3.4]{bischoff:phasebounds} for interesting discussions on this matter.
\end{rmk}

\subsection{Examples of defects}

A few elementary examples of defects between conformal nets were obtained in \parencite{douglas:geomstring} and \parencite[Section 1.\textsc c]{bartels:confnetsIII}. We summarise a selection of those here.

\begin{exmp}[(Twisted) identity defects]
For every conformal net $\net A$ there is an \defn{identity} $\net A$--$\net A$-defect denoted by $1_{\net A}$. It is defined by simply setting $1_{\net A}(I) := \net A(I)$ for every bicoloured interval $I$, where we make sense of the evaluation $\net A(I)$ by ignoring the bicolouring of $I$. Furthermore, the isotony homomorphisms induced by inclusions of white, black and genuinely bicoloured intervals into genuinely bicoloured ones are also defined to be those we get from $\net A$ by forgetting bicolourings. This defect serves as the identity $1$-morphism in the $3$-category of conformal nets.

This example can be tweaked as follows. Let $g$ be an automorphism of $\net A$. Then define an $\net A$--$\net A$-defect $D_g$ again by $D_g(I) := \net A(I)$ for every bicoloured interval $I$, except that we `twist' the isotony homomorphisms on, say, the black side by $g$. More precisely: if $I$ is a black interval, $J$ is genuinely bicoloured and $f\colon I \hookrightarrow J$ is an embedding then the homomorphism $D_g(f)\colon D_g(I) \to D_g(J)$ is set to be $\net A(f) \circ g_I$, where $g_I$ is the automorphism of $\net A(I)$ given by conjugation by $g$. All other types of isotony homomorphisms are inherited from $\net A$, unaffected by $g$.

The construction of such a twisted identity defect can be compared to the situation when we consider a ring $R$ as an $R$--$R$-bimodule. Twisting the right action by an automorphism of $R$ then namely also produces a new $R$--$R$-bimodule.
\end{exmp}

\begin{exmp}[Defects from conformal embeddings]
\label{exmp:embeddingdefects}
Continuing the analogy of defects with ring bimodules, we note that if $R$ is a subring of a ring $S$, then $S$ is an $R$--$S$-, $S$--$R$- and an $R$--$R$-bimodule. The counterpart of this observation for defects is as follows.

A morphism $\tau\colon \net A \to \net B$ between two conformal nets is called a \defn{conformal embedding} in \parencite[Definition 1.45]{bartels:confnetsI} if it respects unitary operators that implement algebra automorphisms induced by interval diffeomorphisms. That is, we demand of $\tau$ that $\Ad(U) = \net A(\phi)$ implies that $\Ad(\tau(U))) = \net B(\phi)$ for every interval $I$, diffeomorphism $\phi$ of $I$ which is the identity near $\partial I$ and unitary operator $U \in \net A(I)$. It is then shown in \parencite[Proposition 1.24]{bartels:confnetsIII} that there exists an $\net A$--$\net B$-defect\footnote{Those authors actually only need the morphism $\tau$ to be a conformal embedding to prove the \defn{strong additivity} property of the resulting defect.} which evaluates genuinely bicoloured intervals like $\net B$ does after forgetting bicolourings and of which the isotony homomorphisms are the obvious ones induced by $\tau$. Similarly, $\tau$ gives rise to a $\net B$--$\net A$- and an $\net A$--$\net A$-defect.
\end{exmp}

\begin{exmp}[Defects from $Q$-systems]
\label{exmp:qsystemdefects}
Let us return to the construction in \cref{exmp:embeddingdefects} of an $\net A$--$\net A$-defect from a conformal embedding $\tau\colon \net A \to \net B$ of nets. A closer inspection of the locality axiom of a defect teaches us that it is not necessary for the net $\net B$ to be local as well---it is sufficient for $\net B$ to be local \defn{relative} to $\net A$, meaning that if $I \hookrightarrow K$ and $J \hookrightarrow K$ are embeddings of intervals of which the images have disjoint interiors, then the images of $\tau(\net A(I))$ and $\net B(J)$ in $\net B(K)$ commute. It is known that, upon imposing a certain finiteness condition on $\tau$, such relatively local extensions $\net B$ of $\net A$ can be classified in terms of certain algebra objects called \defn{$Q$-systems} in the category of sectors of $\net A$. In \parencite[15--17]{bartels:confnetsIII} those authors demonstrate how to produce an $\net A$--$\net A$-defect from a $Q$-system directly, without first needing to build the corresponding net extension.
\end{exmp}

\begin{exmp}[Defects realising invertibility of nets]
The authors of \parencite{douglas:geomstring} call a conformal net $\net A$ \defn{invertible} if there exist another net $\net A^{-1}$ and an invertible defect $D$ from $\net A \otimes \net A^{-1}$ to $\unitnet$, where $\unitnet$ is the trivial conformal net that is constantly $\CC$ and of which all isotony homomorphisms are $\id_\CC$. (Here, the invertibility of a defect should be interpreted in the higher categorical sense.) Together with the result \parencite[Corollary 3.26]{bartels:confnetsI} they proved that $\net A$ is invertible if and only if all its algebras have trivial centre and its representation theory is trivial. This was done by exhibiting for the `if' direction an explicit defect implementing the invertibility.
\end{exmp}

Lastly, just like conformal nets, there is also a direct sum operation for defects (see \parencite[Lemma 1.29]{bartels:confnetsI}), which can be used to produce more defects from a set of given ones.

\section{Bicoloured loop groups}
\label{sec:bicol}

The examples of defects presented in the previous \namecref{sec:defects} are fairly formal in nature. They are built `internally' to the relevant categories, as opposed to the examples of conformal nets presented in \cref{subsec:confnets-exmps} that were constructed in terms of `external' data. In order to gain a greater understanding of defects it is natural to ask for a larger source of richer examples of them.

We now outline a proposal suggesting that the construction of conformal nets from loop groups in \cref{exmp:heisnets,exmp:latticenets,exmp:affinenets} might admit a generalisation that allows one to build defects between these loop group nets. Let us repeat the essential features of those loop groups:
\begin{quote}
If $G$ is a real vector space, a real torus or a compact, connected, simple, simply-connected Lie group, then its associated loop group $LG := C^\infty(S^1, G)$ admits certain disjoint-commutative $\phasegp$-central extensions $\centL G$. Such an extension carries an action of $\Diff_+(S^1)$ which lifts the natural one on $LG$ and is compatible with respect to supports of loops and circle diffeomorphisms. One is able to classify and explicitly construct the irreducible, positive energy representations of $\centL G$.
\end{quote}
We use the notations $\circl$\,, $\circr$, $p$ and $q$ as they are explained in \cref{sec:notsconvs}. That is, $\circl$ and $\circr$ are the closed left and right halves of $S^1$, respectively, and $p$ and $q$ denote the two points $i$ and $-i$ on $S^1$. We shall think of $\circl$ as being coloured white and $\circr$ as having the colour black. An interval on $S^1$ which either does not contain the points $p$ and $q$, or contains exactly one of them, not on its boundary, is then a bicoloured interval.

For every two real Lie groups $G_\wh$ and $G_\bl$, we propose that there should exist a list of \defn{matching conditions}. (In this generality, when $G_\wh$ and $G_\bl$ are not necessarily tori, we will remain deliberately vague on their precise nature.) To every such matching condition $M$ there should be associated a \defn{bicoloured loop group} $L(G_\wh, M, G_\bl)$. The data of an element of this group consists (at least) of two smooth maps $\gamma_\wh\colon \circl \to G_\wh$ and $\gamma_\bl\colon \circr \to G_\bl$. The endpoints of these two paths are not placed arbitrarily, though, and this is where the matching condition comes in: the two points $\gamma_\wh(p)$ and $\gamma_\bl(p)$ are constrained with respect to each other as dictated by $M$, and the same holds for $\gamma_\wh(q)$ and $\gamma_\bl(q)$. This constraint need not necessarily be a property satisfied by the pair $(\gamma_\wh, \gamma_\bl)$---it might also be expressed in terms of a \defn{matching datum} $\gamma_\match$, which is an additional piece of information attached to $(\gamma_\wh, \gamma_\bl)$. In that case we denote an element of $L(G_\wh, M, G_\bl)$ as a triple $(\gamma_\wh, \gamma_\match, \gamma_\bl)$.

We call elements of $L(G_\wh, M, G_\bl)$ \defn{bicoloured loops} and we refer to the elements of an ordinary loop group $LG$ as being \defn{unicoloured} to make the distinction with our new notion.

With an eye towards the axioms in \cref{dfn:defect} for conformal net defects, we postulate that bicoloured loop groups should satisfy the following properties:
\begin{itemize}
\item for every real Lie group $G$ there should exist a particular matching condition $M_G$ from $G$ to itself such that $L(G, M_G, G)$ is isomorphic to the loop group $LG$. This is the first property allowing one to consider bicoloured loop groups as generalisations of unicoloured ones. It is inspired by the demand that, for every loop group from which one can construct a Heisenberg, lattice or loop group conformal net, there should be a bicoloured loop group which reproduces the identity defect from this net to itself,
\item if we write $L_\circlsm G_\wh$ for those loops in $LG_\wh$ which have support in $\circl$ and define $L_\circrsm G_\bl$ similarly, then there should exist two injective group homomorphisms
\[
L_\circlsm G_\wh \hookrightarrow L(G_\wh, M, G_\bl) \hookleftarrow L_\circrsm G_\bl.
\]
This corresponds to the axiom that a defect between two conformal nets restricts to each of the nets on the full subcategories of white and black intervals respectively,
\item elements of $L(G_\wh, M, G_\bl)$ should have a notion of \defn{support} in a bicoloured interval that is embedded in $S^1$ (see \cref{dfn:bicolinterval} for the meaning we will assign to the latter term). The support of a bicoloured loop in the image of the homomorphisms $LG \xrightarrow{\sim} L(G, M_G, G)$, $L_\circlsm G_\wh \hookrightarrow L(G_\wh, M, G_\bl)$ or $L_\circrsm G_\bl \hookrightarrow L(G_\wh, M, G_\bl)$ should coincide with that of its unicoloured pre-image, and
\item the group $L(G_\wh, M, G_\bl)$ should carry an action of the group of orientation preserving circle diffeomorphisms which fix the points $p$ and $q$. Moreover, there should exist an action of $\Rot(S^1)$ or a cover thereof as well.
\end{itemize}

Adding to these properties, we require that a bicoloured loop group admits a certain privileged list of $\phasegp$-central extensions $\centL(G_\wh, M, G_\bl)$ such that
\begin{itemize}
\item when $G$ is a real vector space, a real torus or a compact, connected, simple, simply-connected Lie group and $\centL G$ is a central extension of $LG$ of the type discussed in \cref{exmp:heisnets,exmp:latticenets,exmp:affinenets}, then the central extension $\centL(G, M_G, G)$ that is pulled back from $\centL G$ under the isomorphism $L(G, M_G, G) \cong LG$ is present on the aforementioned list,
\item let $\centL G_\wh$ and $\centL G_\bl$ be centrally extended unicoloured loop groups of the types in \cref{exmp:heisnets,exmp:latticenets,exmp:affinenets}. Then there should exist a central extension $\centL(G_\wh, M, G_\bl)$ which admits lifted injective group homomorphisms into it:
\[
\centL_\circlsm G_\wh \hookrightarrow \centL(G_\wh, M, G_\bl) \hookleftarrow \centL_\circrsm G_\bl,
\]
\item suppose that two elements of the central extension $\centL(G_\wh, M, G_\bl)$ are such that the supports of their images in $L(G_\wh, M, G_\bl)$ are contained in two disjoint bicoloured intervals on $S^1,$ respectively. Then we require that these two elements commute,
\item the action of the group of orientation preserving circle diffeomorphisms which fix the points $p$ and $q$ we expect to exist on $L(G_\wh, M, G_\bl)$ should lift to $\centL(G_\wh, M, G_\bl)$. The same should hold for the action of (some cover of) $\Rot(S^1)$ on $L(G_\wh, M, G_\bl)$, and
\item having the action of (some cover of) $\Rot(S^1)$ on $\centL(G_\wh, M, G_\bl)$ at our disposal, we are able to speak about positive energy representations of the latter group. One should be able to classify and explicitly construct the irreducible such representations up to isomorphism.
\end{itemize}

Having these properties in hand one can attempt to mimic the constructions in \cref{exmp:heisnets,exmp:latticenets,exmp:affinenets} to produce a defect between the nets associated to $\centL G_\wh$ and $\centL G_\bl$. We will elaborate on this further in a special case in \cref{sec:defectslatticenets}.

\begin{rmk}
We give a motivation for requiring $\centL(G_\wh, M, G_\bl)$ to be disjoint-commutative that is unrelated to the axioms of conformal net defects.

For every real Lie group $G$ there exists a procedure, named \defn{transgression}, which produces $\phasegp$-central extensions of $LG$ from certain geometric objects, namely \defn{multiplicative bundle gerbes} (with connection), that are situated over $G$. For example, the central extensions that are discussed in \cref{exmp:affinenets} can be obtained in this way. It is shown in \parencite{waldorf:transgres} that transgression always results in a disjoint-commutative central extension. Hence, we require disjoint-commutativity also in the bicoloured situation as a prerequisite for a possible interpretation of $\centL(G_\wh, M, G_\bl)$ in terms of finite-dimensional, `higher' geometry.
\end{rmk}

\section{Main results}
\label{sec:mainresults}

The goal of this thesis is to test whether our speculative idea of bicoloured loop groups explained in \cref{sec:bicol} can be made sense of when $G_\wh$ and $G_\bl$ are tori $T_\wh := \Lambda_\wh \otimes_\ZZ \phasegp$ and $T_\bl:= \Lambda_\bl \otimes_\ZZ \phasegp$, where $\Lambda_\wh$ and $\Lambda_\bl$ are two even, positive definite lattices of the same rank. One reason for specialising to this case is that the central extensions of unicoloured torus loop groups discussed in \cref{exmp:latticenets} can be specified through explicit $2$-cocycles, as opposed to the geometric methods that are needed in the non-abelian case. A second reason is that also their representations can be built through simple, explicit means without requiring, for example, involved results on integration of Lie algebra representations. Both these reasons therefore lower the barrier of attempting to generalise these constructions. A third reason is that the wealth of examples of lattices and the breadth of their theory makes it plausible that interesting defects between lattice nets can be found as well.

Our proposal for a definition of a matching condition in this situation is as follows. We present it in a slightly simplified form and we give more details in \cref{chap:bicol}.

Let $\Gamma$ be an even, positive definite lattice of the same rank as $\Lambda_\wh$ and $\Lambda_\bl$ and let $\pi_\wh$ and $\pi_\bl$ be two lattice morphisms as in
\[
\Lambda_\wh \xhookleftarrow{\pi_\wh} \Gamma \xhookrightarrow{\pi_\bl} \Lambda_\bl.
\]
If $H$ is the torus $\Gamma \otimes_\ZZ \phasegp$, then $\pi_\wh$ and $\pi_\bl$ induce two surjective torus homomorphisms
\[
\begin{tikzcd}
T_\wh &[1.5em] \ar[l, twoheadrightarrow, "\phasegp\pi_\wh"'] H \ar[r, two heads, "\phasegp\pi_\bl"] &[1.5em] T_\bl
\end{tikzcd}.
\]
We now define the \defn{bicoloured torus loop group} $L(T_\wh, H, T_\bl)$ to be the abelian group of all triples $(\gamma_\wh, \gamma_\match, \gamma_\bl)$ of smooth maps fitting in a commutative diagram
\[
\begin{tikzcd}
\circl \ar[d, "\gamma_\wh"'] &[1.5em] \ar[l, hook'] \circpq \ar[d, "\gamma_\match"] \ar[r, hook] &[1.5em] \circr \ar[d, "\gamma_\bl"] \\
T_\wh & \ar[l, twoheadrightarrow, "\phasegp\pi_\wh"] H \ar[r, two heads, "\phasegp\pi_\bl"'] & T_\bl
\end{tikzcd},
\]
where $\circpq$ is the subset $\{p,q\}$ of $S^1$. That is, $\gamma_\match$ is a matching datum for the pair $(\gamma_\wh, \gamma_\bl)$ dictating the placements of each of the pairs of points $(\gamma_\wh(p), \gamma_\bl(p))$ and $(\gamma_\wh(q), \gamma_\bl(q))$ relative to each other. We call $(\gamma_\wh, \gamma_\match, \gamma_\bl)$ a \defn{bicoloured (torus) loop}.

Our inspiration for this definition has several origins:
\begin{description}
\item[Surface operators between Chern--Simons theories] We briefly commented in \cref{subsec:confnets-applic} on a method by which lattice conformal nets can be used to model a class of $3$d \textsc{tqft}s that are conjecturally the abelian Chern--Simons theories. Each such theory is determined by an even, positive definite lattice as well. The study of defects between abelian Chern--Simons theories was initiated in \parencite{kapustin:abchernsimons}, where these are named \defn{surface operators}, and its authors argue that, at least at the classical level, surface operators between two theories associated to $\Lambda_\wh$ and $\Lambda_\bl$ are classified by Lagrangian subgroups of the pseudo-Riemannian torus $T_\wh \oplus \overline{T}_\bl$. Here, $\overline{T}_\bl$ is the torus $\overline{\Lambda}_\bl \otimes_\ZZ \phasegp$ and $\overline{\Lambda}_\bl$ is the negative definite lattice obtained from $\Lambda_\bl$ by negating its bi-additive form. We then note that in the case when $\Gamma$ is not just contained in the intersection of $\Lambda_\wh$ and $\Lambda_\bl$ but equal to it, the homomorphism
\[
(\phasegp\pi_\wh, \phasegp\pi_\bl)\colon H \to T_\wh \oplus \overline{T}_\bl
\]
is injective with a Lagrangian image. Our definition of a bicoloured torus loop group is an attempt to translate aspects of the study done in \parencite{kapustin:abchernsimons} to the language of loop groups.

\item[A number-theoretical motivation] With the analogy between bimodules and defects in mind, the question emerges naturally whether there exist examples of pairs of conformal nets that are `Morita-equivalent', meaning that there exists an invertible defect between them. We suspect that such a defect induces a braided monoidal equivalence between the categories of sectors of the respective nets (see \parencite[Proposition 8.6]{gordon:coherence}), and hence one must only consider those nets for which these categories are (conjectured to be) braided equivalent. As we stated earlier, we expect this to be true for two lattice nets $\net A_\wh$ and $\net A_\bl$ associated to $\Lambda_\wh$ and $\Lambda_\bl$ when these lattices are in the same \defn{genus}. (See \parencite{hohn:vertalgsgenera} for similar thoughts in the context of vertex algebras.) This is to say that they both have the same signature (which we already assumed to be true) and they become isomorphic after tensoring over $\ZZ$ with the $p$-adic integers for all prime numbers $p$. An equivalent definition is to require that they become isomorphic after summing onto them a single copy of the indefinite hyperbolic plane lattice from \cref{exmp:hyperbolicplane}. This second formulation is indeed used in \parencite{kapustin:abchernsimons} to construct an invertible surface operator between the abelian Chern--Simons theories associated to $\Lambda_\wh$ and $\Lambda_\bl$. However, their method also makes use of the theory assigned to the hyperbolic plane lattice. Because its conformal net counterpart does not exist it remains unclear how their proof can be translated to the language of conformal nets.

We abandon our hope of finding Morita-equivalent lattice nets and ask more generally whether there exists a not-necessarily invertible defect between $\net A_\wh$ and $\net A_\bl$ when $\Lambda_\wh$ and $\Lambda_\bl$ satisfy the weaker demand of becoming isomorphic after tensoring with the $p$-adic \emph{rational} numbers for all $p$. By the \defn{Hasse--Minkowski theorem} this is equivalent to the existence of an isomorphism $\Lambda_\wh \otimes_\ZZ \QQ \cong \Lambda_\bl \otimes_\ZZ \QQ$, which is in turn the same as saying that $\Lambda_\wh$ and $\Lambda_\bl$ share a common sublattice $\Gamma$ of finite index.

\item[Moduli spaces of flat connections over quilted surfaces] A third influence on our definition of $L(T_\wh, H, T_\bl)$ is the work of \parencite{li-bland:quilted} on the study of so called moduli spaces of flat connections over \defn{quilted surfaces}. Such surfaces are divided into regions which are each `coloured' by different structure groups equipped with invariant quadratic forms on their Lie algebras. A flat connection over a quilted surface is then supposed to break down to a co-isotropic relation at the edges where the regions meet. The relevance of that work to loop groups is that the representations of centrally extended loop groups discussed in \cref{exmp:heisnets,exmp:latticenets,exmp:affinenets} can be considered as geometric quantisations of moduli spaces of flat connections over the unit disc. It is hence reasonable to ask whether there exists some bicoloured generalisation of these moduli spaces of which the group of gauge transformations gives rise to a conformal net defect. The group $L(T_\wh, H, T_\bl)$ is an attempt at defining such a group of gauge transformations directly.

The work of \parencite{fuchs:geometricdwtheories} on defects between certain types of 3d \textsc{tqft}s, using notions of \defn{relative} principal bundles over \defn{relative} manifolds, looks to be highly related as well.
\end{description}

Our study of $L(T_\wh, H, T_\bl)$ depends on the auxiliary abelian group $P(H, (\Lambda_\wh - \Lambda_\bl)/\Gamma)$ of all smooth paths $\gamma\colon [0,1] \to H$ such that
\[
\gamma(1) - \gamma(0) \in \frac{\Lambda_\wh - \Lambda_\bl}{\Gamma} \subseteq H.
\]
Here, we consider $\Lambda_\wh$ and $\Lambda_\bl$ as submodules of the Lie algebra of $H$ and we define
\[
\Lambda_\wh - \Lambda_\bl := \{\lambda_\wh - \lambda_\bl \mid \lambda_\wh \in \Lambda_\wh, \lambda_\bl \in\Lambda_\bl\}.
\]
(A more precise description of $\Lambda_\wh - \Lambda_\bl$ can be found at the start of \cref{sec:bicol-struct}.) Many results on the unicoloured torus loop group $LH$ easily generalise to this larger group $P(H, (\Lambda_\wh - \Lambda_\bl)/\Gamma)$. These can be transported next to $L(T_\wh, H, T_\bl)$ via a surjective homomorphism
\begin{equation}
\label{eq:Pth-hom}
\Pth\colon L(T_\wh, H, T_\bl) \twoheadrightarrow P\bigl(H, (\Lambda_\wh - \Lambda_\bl)/\Gamma\bigr)
\end{equation}
with a finite kernel, which is constructed in \cref{subsec:bicol-discontunicolloops}. Among other material, we show in \cref{sec:bicol-struct} that our definition of a bicoloured torus loop group satisfies the first list of demands in \cref{sec:bicol}, namely
\begin{itemize}
\item (\cref{subsec:bicol-unicolisspecialcase}) when $\Lambda_\wh = \Lambda_\bl = \Gamma$ and the morphisms $\pi_\wh$ and $\pi_\bl$ are the identity, the group $L(H, H, H)$ is canonically isomorphic to the unicoloured torus loop group $LH$,
\item (\cref{subsec:bicol-support,subsec:bicol-unicolisotony}) there exist natural injective homomorphisms of abelian groups
\begin{equation}
\label{eq:bicolisotony}
L_\circlsm T_\wh \hookrightarrow L(T_\wh, H, T_\bl) \hookleftarrow L_\circrsm T_\bl,
\end{equation}
and a notion of support for bicoloured loops that is compatible with these injections, and
\item (\cref{subsec:bicol-actionsdiffnS1}) there exists an action of a certain covering group $\Diff_+^{(n)}(S^1)$ on $L(T_\wh, H, T_\bl)$. In particular, the orientation preserving circle diffeomorphisms which fix the points $p$ and $q$ and the elements of $\Rot^{(n)}(S^1)$ act on $L(T_\wh, H, T_\bl)$.
\end{itemize}

In \cref{sec:bicol-centext} we are able to adapt the $2$-cocycles on unicoloured torus loop groups discussed in \cref{exmp:latticenets} to $P(H, (\Lambda_\wh - \Lambda_\bl)/\Gamma)$ in such a way that the pullback of such a cocycle along the homomorphism $\Pth$ in~\eqref{eq:Pth-hom} results in a $\phasegp$-central extension $\centL(T_\wh, H, T_\bl)$ which satisfies the second list of demands in \cref{sec:bicol}. That is, we show in \crefrange{sec:bicol-centext}{sec:bicol-reptheory} that
\begin{itemize}
\item (\cref{subsec:bicol-unicolcentext-isspecialcase}) when $\Lambda_\wh = \Lambda_\bl = \Gamma$ and the morphisms $\pi_\wh$ and $\pi_\bl$ are the identity the isomorphism $LH \cong L(H,H,H)$ lifts to an isomorphism of non-abelian groups $\centL H \cong \centL(H,H,H)$,
\item (\cref{subsec:bicol-isot}) the injections~\eqref{eq:bicolisotony} lift to homomorphisms of non-abelian groups as well,
\item (\cref{sec:bicol-centext-diffnS1action}) the action of $\Diff_+^{(n)}(S^1)$ on $L(T_\wh, H, T_\bl)$ can be extended to one on $\centL(T_\wh, H, T_\bl)$, and
\item (\cref{sec:bicol-reptheory}) we are able to classify and explicitly construct the irreducible, positive energy representations of $\centL(T_\wh, H, T_\bl)$. We find that, up to isomorphism, there exist only finitely many of these.
\end{itemize}

Let us state the last of these points more precisely as our chief result:

\begin{mainthm}[\cref{thm:bicol-classifyirreps}]
Every irreducible, positive energy representation of $\centL(T_\wh, H, T_\bl)$ such that the central subgroup $\phasegp$ acts as $z \mapsto z$ is (unitarily) isomorphic to a certain such representation $\Ind W_{\chi, l}$ for some characters $\chi$ and $l$ of $(\Lambda_\wh \cap \Lambda_\bl)/\Gamma$ and $H$, respectively. The isomorphism classes of such representations are therefore labelled by two parameters: one is an element of the dual group of the finite abelian group $(\Lambda_\wh \cap \Lambda_\bl)/\Gamma$ and the other is an element of the finite abelian group $\Gamma^\vee/(\Lambda_\wh - \Lambda_\bl)$.
\end{mainthm}

We finally outline in \cref{sec:defectslatticenets} how the construction of a defect from such a central extension $\centL(T_\wh, H, T_\bl)$ might proceed.

\section{Organisation of the text}

The rest of this thesis is organised as follows.

We repeat that our discussion of conformal nets and their defects in this Introduction has solely been used as a motivation for the study of bicoloured loop groups. They will hence not play any role in \cref{chap:unicol,chap:bicol}, which treat the theories of uni- and bicoloured torus loop groups, respectively. \cref{chap:unicol} is mainly expository in nature. It is intended to elaborate on some claims made in \cref{exmp:latticenets} and to collect and setup results from the literature in a way that eases the step of generalisation to \cref{chap:bicol}. This latter chapter starts the investigation of bicoloured torus loop groups properly and contains the proofs of the claims made in \cref{sec:mainresults}. It cannot be read without knowledge of \cref{chap:unicol}. Both chapters assume familiarity with the notions and results in \cref{chap:app}. Finally, we make some remarks in \cref{chap:outlook} on possible further steps that can be taken in the study of bicoloured torus loop groups.

\cref{app:lattices,app:centexts,app:posenergyreps} stand alone from each other and can also be read independently from the rest of the thesis. \cref{app:heisgps}, though, on Heisenberg groups, does depend on \cref{app:centexts,app:posenergyreps}.

\section{Notations and conventions}
\label{sec:notsconvs}

\paragraph{Distinguishing $S^1$, $\phasegp$ and $\Rot(S^1)$}
The notation $S^1$ will be used for the unit circle, considered as the smooth manifold embedded in the complex plane. Its role in this thesis is that of the domain of loops (which will be valued in certain Lie groups---tori, specifically). A generic point on it is denoted by $\theta$. We will often make use of some parametrisation of $S^1$ by the interval $[0,1]$. This will always be done at unit speed and counterclockwise, but the starting point of the parametrisation will not necessarily be $1 \in S^1$. Instead, it will be some point denoted by $q$ which is either arbitrary, or a specific one, depending on the context.

When we want to consider $S^1$ as an abelian topological group we will write it as $\phasegp$ instead. It will show up as the group of `phases' other groups we will be studying will be centrally extended by. Elements of $\phasegp$ will be denoted as $z$ or $w$.

The topological group of counterclockwise rotations of the manifold $S^1$ will be named $\Rot(S^1)$. This group (or finite covers of it) will turn out to act on the loop groups and the representations of them that we will be considering. A typical element will be written as $\phi_\theta$, where $\theta \in [0,1]$ is an angle.

\paragraph{Intervals and points}
An \defn{interval} is defined to be an oriented, smooth manifold that is diffeomorphic to $[0,1]$. This means in particular that it is closed, as opposed to the more common convention in the literature on conformal nets of being open. For a subinterval $I \subseteq S^1$, we denote by $I'$ the subinterval that is the closure of $S^1\backslash I$.

The symbols $\circl$ and $\circr$ will be used for the closed left and right half of $S^1$, respectively. When we want to emphasise the orientations they inherit if $S^1$ is oriented counterclockwise, we will add arrows as follows: $\circldir$ and $\circrdir$\,. In \cref{chap:bicol} the two points $i$ and $-i$, whose union $\circpq$~forms the intersection of $\circl$ and $\circr$\,, will be respectively called $p$ and $q$.

\paragraph{Hilbert spaces, actions and representations}
Our convention is that a Hermitian inner product on a complex vector space is linear in its first variable and antilinear in the second one, and all group actions will be left actions.

Whenever we will speak of a \defn{group representation} we will follow the convention stated in \cref{app:posenergyreps}, which is that the group is topological, the underlying vector space is a complex Hilbert space and that the representation is strongly continuous and unitary, unless stated otherwise.

\chapter{Unicoloured torus loop groups}
\label{chap:unicol}

Let $\Lambda$ be an even, positive definite lattice and $T$ the torus $\Lambda \otimes_\ZZ \phasegp$. In this \namecref{chap:unicol} we study a certain $\phasegp$-central extension $\centL T$ of the torus loop group $LT := C^\infty(S^1, T)$ and its positive energy representations. The title of this \namecref{chap:unicol} and its sections refer to $LT$ as being \defn{unicoloured}, but we do not use this adjective in the text of the current \namecref{chap:unicol}. Much of this material is already covered in some form in the literature, but our aim is to provide a level of detail sufficient for the needs of \cref{chap:bicol} where results and notions in this \namecref{chap:unicol} will be used and generalised. References to the literature will be given per section and we make extensive use of material collected in \cref{chap:app}.

The study of the representation theory of the central extension $\centL T$ will require knowledge of certain subgroups and a direct sum decomposition of $LT$. This is treated in \cref{sec:unicol-struct}, where it is not yet necessary to assume the presence of a bi-additive form on $\Lambda$. We treat $LT$ as an abstract group until \cref{sec:unicol-reptheory}, but we will refer to certain of its subsets earlier already as `connected components', in anticipation. Next, the construction of $\centL T$ is described in \cref{sec:unicol-centext} which also includes a proof of its important disjoint-commutativity property.

\cref{sec:unicol-diffs1action,sec:unicol-latisoms} each have dual purposes. They are firstly intended to clarify to what extent $\centL T$ depends on certain extra data needed to perform its construction. One namely requires a choice of a point on $S^1$ and the fact shown in \cref{sec:unicol-diffs1action} that $\Diff_+(S^1)$ acts on $\centL T$ implies that this choice is irrelevant. On the other hand, we also need to pick a certain $\{\pm 1\}$-central extension $\tilde\Lambda$ of $\Lambda$. We prove in \cref{sec:unicol-latisoms} that it is the automorphism group of $\tilde\Lambda$ which acts on $\centL T$ and not that of $\Lambda$. This choice is therefore not immaterial.

A further goal of \cref{sec:unicol-diffs1action} is that we learn that $\centL T$ admits in particular an action of the group $\Rot(S^1)$, so that we may speak of its positive energy representations. The second aim of \cref{sec:unicol-latisoms}, which admittedly lies outside the scope of this thesis, is that it exhibits automorphisms of the conformal net associated to $\Lambda$.

This chapter culminates with the classification and construction of the irreducible, positive energy representations (up to isomorphism) of $\centL T$ in \cref{sec:unicol-reptheory}.

Modest novelties of our exposition are that we lay out in what way $\centL T$ is a topological group, we occasionally explain how the various constructions can be made to work for odd lattices as well, and that we show how the automorphisms of $\tilde\Lambda$ are symmetries of the collection of irreducible, positive energy representations of $\centL T$. We will not return to the latter two topics in \cref{chap:bicol}, though.

\section{Unicoloured torus loop groups and their structure}
\label{sec:unicol-struct}

Let $\Lambda$ be a free $\ZZ$-module of finite rank. Associated to it is a torus $T := \Lambda \otimes_\ZZ \phasegp$ whose Lie algebra $\Lambda \otimes_\ZZ \RR$ we denote by $\liealg t$. In this section we explain the structure of the torus loop group $LT := C^\infty(S^1, T)$ of $T$. We will show that it can be decomposed (non-canonically) as a direct sum of $T$, a certain infinite-dimensional vector space $V\liealg t$ and $\Lambda$.

\begin{refs}
The structure of torus loop groups was found in \parencite[Section 4]{segal:unitary}. We elaborate on it further by emphasising $\Diff_+(S^1)$-equivariance.
\end{refs}

We first observe that, given a choice of a privileged point $q$ on $S^1$, $LT$ has an alternative description in terms of paths in the Lie algebra $\liealg t$, namely as follows:
\begin{equation}
\label{eq:unicol-liealgpaths}
LT \cong \Bigl\{\xi \in C^\infty\bigl([0,1], \liealg t\bigr) \Bigm\vert \text{$\xi(1) - \xi(0) \in \Lambda$, $\xi^{(k)}(1) = \xi^{(k)}(0)$ for all $k\geq 1$}\Bigr\}\Big/\Lambda.
\end{equation}
Here we used the unit speed parametrisation of $S^1$ by $[0,1]$, starting at $q$.

This description allows us to define for a loop $\gamma \in LT$ its \defn{winding element} as
\[
\Delta_\gamma := \xi(1) - \xi(0) \in \Lambda
\]
for any path $\xi \in C^\infty([0,1], \liealg t)$ which represents $\gamma$. This is not only independent of the choice of representative, but also of our choice of point $q$. The winding element is therefore a canonical homomorphism $\Delta\colon LT \to \Lambda$ of abelian groups. It is surjective because if $\lambda \in \Lambda$, the loop $\gamma_\lambda \in LT$, defined as the projection on $T$ of the path $[0,1] \to \liealg t$, $\theta \mapsto \theta\lambda$, has winding element $\lambda$. The connected components of $LT$ are the fibres of $\Delta$.

Write $(LT)_0$ for the identity component of $LT$. It consists of all loops of winding element zero and it fits into a short exact sequence
\[
\begin{tikzcd}
0 \ar[r] & (LT)_0 \ar[r] & LT \ar[r, "\Delta"] &[0.5em] \Lambda \ar[r] & 0.
\end{tikzcd}
\]
(As an aside, any loop of $(LT)_0$ can be lifted to a map $S^1 \to \liealg t$, without needing to break $S^1$ at a point.) A choice of $q$ splits this sequence. With it, we are namely able to define our standard choices of loops $\gamma_\lambda$ with winding element $\lambda$. We then get an isomorphism
\begin{equation}
\label{eq:unicol-splitlattice}
LT \xrightarrow{\sim} (LT)_0 \oplus \Lambda,
\end{equation}
given by $\gamma \mapsto (\gamma - \gamma_{\Delta_\gamma}, \Delta_\gamma)$ with inverse map $(\gamma, \lambda) \mapsto \gamma + \gamma_\lambda$. This isomorphism is not `natural', though, which can be made precise as follows.

The group $\Diff_+(S^1)$ of orientation preserving circle diffeomorphisms acts\footnote{Recall that by our convention all group actions are left actions.} on $LT$ by $\phi\cdot \gamma := \phi^* \gamma$, where $\phi \in \Diff_+(S^1)$, $\gamma \in LT$ and $\phi^* \gamma \in LT$ is the loop defined by
\[
(\phi^* \gamma)(\theta) := \gamma\bigl(\phi^{-1}(\theta)\bigr), \qquad \theta \in S^1.
\]
Because this action preserves the winding element of any loop, it restricts to each connected component of $LT$---in particular to $(LT)_0$. If we let $\Diff_+(S^1)$ act on $\Lambda$ trivially, then the isomorphism~\eqref{eq:unicol-splitlattice} is not $\Diff_+(S^1)$-equivariant because in general $\phi^* \gamma_\lambda \neq \gamma_\lambda$.

Let us now study the structure of $(LT)_0$ further. Define the real vector space
\begin{equation}
\label{eq:unicol-Vt}
V\liealg t := \biggl\{\xi \in C^\infty(S^1, \liealg t) \biggm\vert \int_{S^1} \xi(\theta) \dd\theta = 0 \biggr\}.
\end{equation}
It carries a (left) action of $\Rot(S^1)$, defined similarly as for $LT$. The groups $V\liealg t$ and $(LT)_0$ are related via a canonical isomorphism of abelian groups
\begin{equation}
\label{eq:unicol-decompidcomp}
(LT)_0 \xrightarrow{\sim} T \oplus V\liealg t,
\end{equation}
given by $\gamma \mapsto (\avg \gamma, \xi - \avg \xi)$, where
\[
\avg \gamma := \exp \avg \xi, \qquad \avg \xi := \int_{S^1} \xi(\theta) \dd \theta
\]
for a choice $\xi\colon S^1 \to \liealg t$ of lift of $\gamma$. Its inverse is given by $(x, \xi) \mapsto x + \exp \xi$ for $x \in T$ and $\xi \in V\liealg t$. It is important to note that the isomorphism~\eqref{eq:unicol-decompidcomp} \emph{is} $\Rot(S^1)$-equivariant if we let $\Rot(S^1)$ act on $T \oplus V\liealg t$ by only affecting the $V\liealg t$-summand.

Putting~\eqref{eq:unicol-splitlattice} and~\eqref{eq:unicol-decompidcomp} together, we conclude with a decomposition
\begin{equation}
\label{eq:unicol-decompfullgp}
LT \cong T \oplus V\liealg t \oplus \Lambda
\end{equation}
that depends on a choice of $q \in S^1$. Despite $LT$ being infinite-dimensional its structure therefore resembles that of a class of finite-dimensional groups, namely the compactly generated, locally compact, Hausdorff, abelian ones. Those admit a decomposition into a compact abelian group, a vector space and a free $\ZZ$-module of finite rank as well (see for example \parencite[Theorem 9.8]{hewitt:harmana1}).

\section{Central extensions associated to lattices}
\label{sec:unicol-centext}

In this section we will explain how, to the data of an even lattice $\Lambda$ together with a choice of a particular kind of $\{\pm 1\}$-central extension of it and a $2$-cocycle for this extension, we can construct a certain $\phasegp$-central extension of the loop group $LT$.

\begin{refs}
The formula for the $2$-cocycle~\eqref{eq:unicol-cocycle} on $LT$ and the study of the resulting central extension originates in \parencite{segal:unitary}. See also \parencite[Proposition 4.8.3]{pressley:loopgps} for a variant of the formula, which we do not use in this thesis. However, these sources are not explicit in noting that the construction makes sense just as well for an arbitrary even lattice $\Lambda$ instead of only one of $ADE$-type. To our knowledge, that generalisation was first applied in \parencite[Section 3]{dong:latticeconfnets}.

\cref{thm:unicol-disjcomm} was stated, respectively proven, in \parencite[Proposition 13.1.3]{pressley:loopgps} and \parencite[Proposition 3.4]{dong:latticeconfnets} for variants of the central extensions we consider. The term \defn{disjoint-commutativity} is due to \parencite[Theorem 3.3.1]{waldorf:transgres}.

Our solution to the problem of constructing central extensions associated to odd lattices  in \cref{rmk:unicol-centextoddcase} is inspired by the literature on vertex algebras, namely \parencite[Corollary 5.5]{kac:beginners}. Central extensions of loop groups that are graded by $\ZZ/2\ZZ$ appeared before in \parencite{freed:loopgpsktheoryIII}.
\end{refs}

So let $\Lambda$ be a lattice with bi-additive form $\langle \cdot, \cdot \rangle$ and let $\tilde \Lambda$ be a central extension of the underlying abelian group of $\Lambda$ by the group $\{\pm 1\}$ with commutator map $(\lambda, \mu) \mapsto (-1)^{\langle \lambda, \mu \rangle}$. This determines $\tilde \Lambda$ up to non-unique isomorphism, as discussed in \cref{subsec:centexts-abgps}. The assumption that this is a commutator map forces $\Lambda$ to be even. A second piece of data we will need is a choice of a $2$-cocycle $\epsilon\colon \Lambda \times \Lambda \to \{\pm 1\}$ for $\tilde \Lambda$. Finally, choose a privileged point $q$ on $S^1$ so that we can make use of the description~\eqref{eq:unicol-liealgpaths} of $LT$.

\begin{constr}[Central extensions of $LT$]
We will construct the $\phasegp$-central extension $\centL T$ of $LT$ by letting its underlying set be $LT \times \phasegp$ and writing down an explicit cocycle. Let $\gamma, \rho \in LT$, $z,w \in \phasegp$ and pick lifts $\xi$ and $\eta$ as in~\eqref{eq:unicol-liealgpaths} of $\gamma$ and $\rho$ respectively. We define the multiplication on $\centL T$ by
\begin{equation}
\label{eq:unicol-mult}
(\gamma, z) \cdot (\rho, w) := \bigl(\gamma + \rho, zw \cdot c(\gamma, \rho)\bigr)
\end{equation}
where $c$ is the $2$-cocycle on $LT$ defined by
\begin{equation}
\label{eq:unicol-cocycle}
\begin{gathered}
c(\gamma, \rho) := \epsilon(\Delta_\gamma, \Delta_\rho) e^{2\pi iS(\xi, \eta)} \\
S(\xi,\eta) := \frac12 \int_0^1 \bigl\langle \xi'(\theta), \eta(\theta) \bigr\rangle \dd\theta + \frac12 \bigl\langle \Delta_\gamma, \eta(0) \bigr\rangle.
\end{gathered}
\end{equation}
Here, $\langle \cdot, \cdot \rangle$ stands for the bilinear extension of the form $\langle \cdot, \cdot \rangle$ on $\Lambda$ to $\liealg t$. Notice that $S$ is bi-additive.

We need to check that this multiplication does not depend on the choices of lifts $\xi$ and $\eta$. If we replace $\xi$ by $\xi + \lambda$ for a $\lambda \in \Lambda$, then obviously the value of $S$ does not change. If we perturb $\eta$ by $\lambda$, however, the following term gets added to $S(\xi, \eta)$:
\begin{align*}
S(\xi, \lambda) &= \frac12 \int_0^1 \bigl\langle \xi'(\theta), \lambda \bigr\rangle \dd\theta + \frac12 \bigl\langle \Delta_\gamma, \lambda \bigr\rangle \\
	&= \frac12 \bigl\langle \Delta_\gamma, \lambda \bigr\rangle + \frac12 \bigl\langle \Delta_\gamma, \lambda \bigr\rangle \\
	&= \bigl\langle \Delta_\gamma, \lambda \bigr\rangle.
\end{align*}
Since this is an integer, we have $e^{2\pi iS(\xi, \lambda)} = 1$. This shows that $c$ is a well-defined function on $LT \times LT$. Lastly, it is a $2$-cocycle since $\epsilon$ is and $S$ is bi-additive. It is a normalised $2$-cocycle if and only if $\epsilon$ is.
\end{constr}

\begin{ingreds}
\label{ingreds:unicol}
We summarise the ingredients used in the construction of the central extension $\centL T$ for clarity:
\begin{itemize}
\item an even lattice $(\Lambda, \langle \cdot, \cdot \rangle)$,
\item a choice of a $\{\pm 1\}$-central extension $\tilde \Lambda$ of $\Lambda$ such that it has commutator map $(\lambda, \mu) \mapsto (-1)^{\langle \lambda, \mu \rangle}$,
\item a choice of a $2$-cocycle $\epsilon\colon \Lambda \times \Lambda \to \{\pm 1\}$ for $\tilde \Lambda$ (although not needed for the construction itself, we will always choose $\epsilon$ to be normalised to make calculations easier),
\item a choice of a privileged point $q$ on $S^1$.
\end{itemize}
Perhaps confusingly, the notation $\centL T$ omits references to any of these ingredients. We promise that whenever we will use this notation we will always explain in the surrounding context which choices were made.
\end{ingreds}

%
This constructed central extension satisfies the following important property:

\begin{thm}[Disjoint-commutativity of central extensions]
\label{thm:unicol-disjcomm}
Let $(\gamma,z)$ and $(\rho, w)$ be two elements of $\centL T$ such that the supports of $\gamma$ and $\rho$ are contained in two disjoint intervals on $S^1$ respectively. Then $(\gamma,z)$ and $(\rho, w)$ commute.
\end{thm}
\begin{proof}
This amounts to showing that $c(\gamma, \rho)c(\rho, \gamma)^{-1} = 1$. Let us first write out $c(\gamma, \rho)c(\rho, \gamma)^{-1}$ without making any assumptions about supports. We have
\[
c(\gamma, \rho)c(\rho, \gamma)^{-1} = \epsilon(\Delta_\gamma, \Delta_\rho) \epsilon(\Delta_\rho, \Delta_\gamma)^{-1} e^{2\pi i\bigl(S(\xi, \eta) - S(\eta, \xi)\bigr)},
\]
where $\xi$ and $\eta$ are lifts of $\gamma$ and $\rho$ respectively as in~\eqref{eq:unicol-liealgpaths} and
\begin{equation}
\begin{split}
S(\xi, \eta) - S(\eta, \xi) &= \frac12 \int_0^1 \bigl\langle \xi'(\theta), \eta(\theta) \bigr\rangle \dd\theta - \frac12 \int_0^1 \bigl\langle \eta'(\theta), \xi(\theta) \bigr\rangle \dd\theta + \phantom{{}} \\
	&\phantom{{}={}} \frac12 \bigl\langle \Delta_\gamma, \eta(0) \bigr\rangle - \frac12 \bigl\langle \Delta_\rho, \xi(0) \bigr\rangle.
\end{split}
\label{eq:disjcomm}
\end{equation}

Let us now make the assumptions about the supports as in the statement of the \namecref{thm:unicol-disjcomm} and write $I$ for the interval containg the support of $\gamma$ and $J$ for the one corresponding to $\rho$. We distinguish two cases: either $I$ does not contain the point $q \in S^1$, or $J$ does not.

In the first case, we perform a partial integration on the second integral in~\eqref{eq:disjcomm} to get:
\begin{align*}
\int_0^1 \bigl\langle \eta'(\theta), \xi(\theta) \bigr\rangle \dd\theta &= \Bigl[\bigl\langle \eta(\theta), \xi(\theta) \bigr\rangle\Bigr]_0^1 - \int_0^1 \bigl\langle \eta(\theta), \xi'(\theta) \bigr\rangle \dd\theta \\
	&= \bigl\langle \eta(1), \xi(1) \bigr\rangle - \bigl\langle \eta(0), \xi(0) \bigr\rangle - \int_0^1 \bigl\langle \eta(\theta), \xi'(\theta) \bigr\rangle \dd\theta.
\end{align*}
After substituting $\eta(1) = \Delta_\rho + \eta(0)$ and $\xi(1) = \Delta_\gamma + \xi(0)$, these first two terms become
\[
\bigl\langle \eta(1), \xi(1) \bigr\rangle - \bigl\langle \eta(0), \xi(0) \bigr\rangle = \langle \Delta_\rho, \Delta_\gamma \rangle + \bigl\langle \Delta_\rho, \xi(0) \bigr\rangle + \bigl\langle \eta(0), \Delta_\gamma \bigr\rangle.
\]
Putting all of this back into~\eqref{eq:disjcomm} gives
\begin{equation}
\label{eq:unicol-commmap}
S(\xi, \eta) - S(\eta, \xi) = \int_0^1 \bigl\langle \xi'(\theta), \eta(\theta) \bigr\rangle \dd\theta - \frac12 \langle \Delta_\rho, \Delta_\gamma \rangle - \bigl\langle \Delta_\rho, \xi(0)\bigr\rangle.
\end{equation}
Because outside of $I$, $\xi$ is constant, the above integral is actually only taken over $I$ instead of over all of $[0,1]$. Since $J$ is disjoint from $I$ and $q \notin I$, $\eta$ is constant with value, say, $\lambda$, in $\Lambda$ on $I$. Therefore we may write
\[
S(\xi, \eta) - S(\eta, \xi) = \bigl\langle \xi(\text{end of $I$}) - \xi(\text{start of $I$}), \lambda \bigr\rangle - \frac12 \langle \Delta_\rho, \Delta_\gamma \rangle - \bigl\langle \Delta_\rho, \xi(0)\bigr\rangle.
\]
Our assumption that $q \notin I$ also implies that
\[
\xi(\text{end of $I$}) = \xi(1) \qquad \text{and} \qquad \xi(\text{start of $I$}) = \xi(0),
\]
so $\xi(\text{end of $I$}) - \xi(\text{start of $I$}) = \Delta_\gamma \in \Lambda$, and that $\xi(0) \in \Lambda$. We now see that thanks to the integrality of $\Lambda$,
\[
\bigl\langle \xi(\text{end of $I$}) - \xi(\text{start of $I$}), \lambda \bigr\rangle - \bigl\langle \Delta_\rho, \xi(0)\bigr\rangle \in \ZZ
\]
and so
\[
e^{2\pi i\bigl(S(\xi, \eta) - S(\eta, \xi)\bigr)} = e^{- 2\pi i \cdot \frac12 \langle \Delta_\rho, \Delta_\gamma \rangle} = (-1)^{\langle\Delta_\gamma, \Delta_\rho\rangle}.
\]
Together with the fact that, by definition of the $2$-cocycle $\epsilon$,
\[
\epsilon(\Delta_\gamma, \Delta_\rho) \epsilon(\Delta_\rho, \Delta_\gamma)^{-1} = (-1)^{\langle\Delta_\gamma, \Delta_\rho\rangle},
\]
this shows that $c(\gamma, \rho)c(\rho, \gamma)^{-1} = 1$.

In the second case where $J$ does not contain $q$, we can perform a calculation similar to the one above by partially integrating the first integral in~\eqref{eq:disjcomm} instead.
\end{proof}

\begin{rmk}[Central extensions associated to odd lattices]
\label{rmk:unicol-centextoddcase}
Observe that the choice of a double cover $\tilde \Lambda$ with commutator map $(\lambda, \mu) \mapsto (-1)^{\langle\lambda, \mu\rangle}$ is exactly what makes the disjoint-commutativity work. If $\Lambda$ is odd, then $(\lambda, \mu) \mapsto (-1)^{\langle\lambda, \mu\rangle}$ is not a commutator map. However, $(\lambda, \mu) \mapsto (-1)^{\langle\lambda, \mu\rangle + \langle \lambda, \lambda \rangle \langle \mu, \mu \rangle}$ \emph{is}. So if we let $\tilde \Lambda$ be a double cover with this commutator map instead, pick a cocycle $\epsilon$ for it and replace it in the definition~\eqref{eq:unicol-cocycle} of the cocycle $c$, then the calculation in the proof of \cref{thm:unicol-disjcomm} shows that
\[
(\gamma, z) \cdot (\rho, w) = (-1)^{p(\gamma) p(\rho)} (\rho, w) \cdot (\gamma, z)
\]
if $\gamma$ and $\rho$ have support in disjoint intervals, where $p\colon LT \twoheadrightarrow \{0,1\}$ is the parity homomorphism of abelian groups given by $p(\gamma) := \langle \Delta_\gamma, \Delta_\gamma \rangle \bmod 2\ZZ$. So $\centL T$ is a \defn{graded} central extension of a \defn{graded} abelian group which satisfies a form of disjoint \defn{graded} commutativity. If we write $\centL T(i) \subseteq \centL T$ for $i \in \{0,1\}$ for the pre-images of $p^{-1}(i) \subseteq LT$, then $\centL T(0)$ is a subgroup and
\[
\centL T(i) \cdot \centL T(j) \subseteq \centL T(i+j),
\]
where the indices are read modulo $2$. We call the elements of $\centL T(0)$ \defn{even} and those of $\centL T(1)$ \defn{odd}.
\end{rmk}

\subsection{The groups $(\centL T)_0$ and $\centV \liealg t$}
\label{subsec:unicol-centext-subgps}

In our study of the structure of $LT$ in \cref{sec:unicol-struct} two of its subgroups were singled out: its identity component $(LT)_0$ and a real vector space $V\liealg t$. We spell out what the restriction of the central extension $\centL T$ to these subgroups looks like.

Consider first the restriction $(\centL T)_0$ to $(LT)_0$. It is the identity component of $\centL T$ and its underlying set is $(LT)_0 \times \phasegp$, meaning that $(\gamma, z) \in (\centL T)_0$ if and only if $\gamma$ has winding element zero. Moreover, recall that such a loop $\gamma$ can be lifted to a map $S^1 \to \liealg t$ without needing to break $S^1$ at a point. These facts imply that the cocycle $c$ on $\centL T$ simplifies to
\begin{equation}
\label{eq:unicol-cocycle-idcomp}
c(\gamma, \rho) := e^{2\pi i S(\xi, \eta)}, \qquad S(\xi, \eta) := \frac12 \int_{S^1} \langle \dd \xi, \eta \rangle
\end{equation}
when restricted to $(\centL T)_0$, where $\gamma, \rho \in (LT)_0$ and $\xi, \eta \colon S^1 \to \liealg t$ are choices of respective lifts.

Now define a $\phasegp$-central extension $\centV \liealg t$ of the underlying abelian group of $V\liealg t$ via the $2$-cocycle\footnote{The usage of the same symbol $c$ for the cocycles on both $LT$ and $V\liealg t$ might cause momentary confusion, but this will clear up in a moment.}
\[
c\colon V\liealg t \times V\liealg t \to \phasegp, \qquad c(\xi, \eta) := e^{2\pi i S(\xi, \eta)}.
\]
So this group has as its underlying set $V\liealg t \times \phasegp$ and its multiplication is given by
\[
(\xi, z) \cdot (\eta, w) := \bigl(\xi + \eta, zw\cdot c(\xi, \eta)\bigr),
\]
for $\xi, \eta \in V\liealg t$ and $z,w \in \phasegp$. It is then easily checked, using the fact that $S(\xi - \avg \xi, \eta - \avg \eta) = S(\xi, \eta)$ for all $\xi, \eta \in V\liealg t$, that the isomorphism of abelian groups~\eqref{eq:unicol-decompidcomp} lifts to an isomorphism of groups
\begin{equation}
\label{eq:unicol-decompidcomp-centext}
(\centL T)_0 \xrightarrow{\sim} T \times \centV \liealg t
\end{equation}
which is the identity on the central subgroup $\phasegp$.

We close this section by making an important observation about the group $\centV \liealg t$. Note that the function $S$ in~\eqref{eq:unicol-cocycle-idcomp} is an $\RR$-valued bilinear form on the real vector space of all smooth maps $S^1 \to \liealg t$.

\begin{prop}
\label{thm:unicol-centextVt-heis}
The bilinear form $S$ in~\eqref{eq:unicol-cocycle-idcomp} is skew and, when the lattice $\Lambda$ is definite, non-degenerate upon restriction to $V\liealg t$.
\end{prop}
\begin{proof}
To see that $S$ is skew we return to parametrising $S^1$ at unit speed by $[0,1]$ starting at the privileged point $q \in S^1$ so that we can write
\[
S(\xi, \eta) = \frac12 \int_0^1 \bigl\langle \xi'(\theta), \eta(\theta) \bigr\rangle \dd\theta.
\]
Applying partial integration, we get
\begin{align*}
S(\xi, \eta) &= \frac12 \bigl\langle \xi(1), \eta(1) \bigr\rangle - \frac12 \bigl\langle \xi(0), \eta(0) \bigr\rangle - \frac12 \int_0^1 \bigl\langle \xi(\theta), \eta'(\theta) \bigr\rangle \dd\theta,
\end{align*}
and since $\xi(1) = \xi(0)$ and $\eta(1) = \eta(0)$ there indeed holds $S(\xi, \eta) = -S(\eta, \xi)$.

Turning to the second claim, suppose $\xi \in V\liealg t$ is such that $S(\xi, \eta) = 0$ for all $\eta \in V\liealg t$. Because $V\liealg t$ is closed under the operation of taking derivatives we have $\xi' \in V\liealg t$, so in particular $S(\xi, \xi') = 0$. By continuity of $\xi'$ we must have $\langle \xi'(\theta), \xi'(\theta) \rangle = 0$ for all $\theta \in [0,1]$. The definiteness of the form $\langle \cdot, \cdot \rangle$ on $\liealg t$ implies that $\xi'(\theta) = 0$ for all $\theta \in [0,1]$ and so $\xi$ is constant. Since $\int_{S^1} \xi(\theta) \dd\theta = 0$, it follows that $\xi$ is identically zero.
\end{proof}

This \namecref{thm:unicol-centextVt-heis} implies that when $\Lambda$ is a definite lattice, $\centV \liealg t$ is the \defn{Heisenberg group} associated to the pair $(V\liealg t, -S)$ as in \cref{dfn:heisgp} (without yet satisfying the topological requirements).

\section{Actions of \texorpdfstring{$\Diff_+(S^1)$}{Diff+(S\textonesuperior)} on central extensions}
\label{sec:unicol-diffs1action}
We assume the setup of \cref{sec:unicol-centext}. That is, $\Lambda$, $\tilde\Lambda$, $\epsilon$ and $q$ are as in \cref{ingreds:unicol}, $T$ is the torus $\Lambda \otimes_\ZZ \phasegp$ and, finally, $\centL T$ is the $\phasegp$-central extension of $LT := C^\infty(S^1, T)$ defined by the $2$-cocycle $c$ in~\eqref{eq:unicol-cocycle} on $LT$ constructed from all this data. In this section we show that there is an action of the group $\Diff_+(S^1)$ of orientation preserving diffeomorphisms of the circle on $\centL T$.

\begin{refs}
The observation that, in the case of an even lattice, there exists such a $\Diff_+(S^1)$-action and the formula~\eqref{eq:unicol-diffs1action} which achieves it are stated in \parencite[325]{segal:unitary} (as a right action).
\end{refs}

Recall from \cref{sec:unicol-struct} that the non-centrally extended loop group $LT$ carries an obvious (left) action of $\Diff_+(S^1)$. A naive guess for an action on $\centL T$ would be to set
\begin{equation}
\label{eq:diffs1action-naive}
\phi \cdot (\gamma, z) := (\phi^* \gamma, z), \qquad (\gamma, z) \in \centL T.
\end{equation}
for $\phi \in \Diff_+(S^1)$. While this does define an action on the underlying set, it respects the multiplication on $\centL T$ if and only if $c(\phi^* \gamma, \phi^* \rho) = c(\gamma, \rho)$ for all $\gamma, \rho \in LT$. We will see in~\cref{thm:unicol-diffs1fail} that this does not hold, which means that~\eqref{eq:diffs1action-naive} is not the correct action to consider. This misbehaviour is to be expected if we recall that the definition of $c$ involves the choice of a privileged point $q$ on $S^1$ in order to define lifts $[0,1] \to \liealg t$ of elements of $LT$ which have non-zero winding element.

To find a correct formula for the action we adjust the guess~\eqref{eq:diffs1action-naive} as follows. We propose that $\Diff_+(S^1)$ should act on $\centL T$ in a way that lifts the action on $LT$ and fixes the embedding $\phasegp \hookrightarrow \centL T$. It is easily checked that the most general such action is of the form\footnote{That is, if we consider $\centL T \twoheadrightarrow LT$ as a principal $\phasegp$-bundle, then we are proposing that $\Diff_+(S^1)$ should act by $\phasegp$-bundle automorphisms over the existing action on the base, but not in a `parallel' way. This suggests that our construction of $\centL T$ as having its underlying bundle \emph{the} trivial one, instead of merely a trivialisable one, is somewhat artificial.}
\begin{equation}
\label{eq:diffs1action-ansatz}
\phi \cdot (\gamma, z) := \bigl(\phi^* \gamma, d(\phi, \gamma) \cdot z\bigr)
\end{equation}
for some $d(\phi, \gamma) \in \phasegp$. In order for this to be an action on the underlying set of $\centL T$, that is,
\[
(\psi \circ \phi) \cdot (\gamma, z) = \psi \cdot \bigl(\phi \cdot (\gamma, z)\bigr)
\]
for all $\phi, \psi \in \Diff_+(S^1)$, $d$ should satisfy
\begin{equation}
\label{eq:diffs1action-d1}
d(\psi \circ \phi, \gamma) = d(\psi, \phi^*\gamma)d(\phi, \gamma).
\end{equation}
If~\eqref{eq:diffs1action-ansatz} respects the multiplication on $\centL T$ we see from~\eqref{eq:unicol-mult} that we should additionally have the following compatibility with the cocycle $c$:
\begin{equation}
\label{eq:diffs1action-d2}
d(\phi, \gamma + \rho)c(\gamma, \rho) = d(\phi, \gamma) d(\phi, \rho) c(\phi^*\gamma, \phi^*\rho).
\end{equation}
In other words, the $1$-cochain $d(\phi, \cdot)\colon LT \to \phasegp$ should exhibit the $2$-cocycle $(\gamma, \rho) \mapsto c(\phi^*\gamma, \phi^*\rho)c(\gamma, \rho)^{-1}$ as a $2$-coboundary.

To find such a function $d\colon \Diff_+(S^1) \times LT \to \phasegp$ we will calculate the difference between $c(\phi^* \gamma, \phi^* \rho)$ and $c(\gamma, \rho)$, that is, the failure for $\Diff_+(S^1)$ to preserve $c$.

Whereas in \cref{sec:unicol-struct,sec:unicol-centext} we represented a loop $\gamma \in LT$ as a smooth map $\xi\colon [0,1] \to \liealg t$ such that $\xi(1) - \xi(0) \in \Lambda$, we will now work instead with the quasi-periodic extension $\Xi\colon \RR \to \liealg t$ to all of $\RR$ of the latter. That is, $\Xi$ is defined as $\Xi(\theta) = \xi(\theta)$ for $\theta \in [0,1]$ and $\Xi(\theta + 1) = \Xi(\theta) + \Delta_\gamma$ for all $\theta \in \RR$. The definition~\eqref{eq:unicol-cocycle} of the cocycle $c$ remains unchanged if the maps $\xi$ and $\eta$ are replaced by their quasi-periodic extensions $\Xi$ and $\Eta$ everywhere.

Let $\Diff_+^{(\infty)}(S^1)$ be the universal covering group of $\Diff_+(S^1)$. It fits in the following short exact sequence:
\[
\begin{tikzcd}
0 \ar[r] & \ZZ \ar[r] & \Diff_+^{(\infty)}(S^1) \ar[r] & \Diff_+(S^1) \ar[r] & 1.
\end{tikzcd}
\]
An explicit model for $\Diff_+^{(\infty)}(S^1)$ is given by
\begin{equation}
\begin{split}
\Diff_+^{(\infty)}(S^1) &\cong \bigl\{\Phi \colon \RR \xrightarrow{\sim} \RR \bigm\vert \text{$\Phi$ an orientation preserving diffeomorphism}, \bigr. \\
	&\phantom{{}\cong \bigl\{\Phi \colon \RR \xrightarrow{\sim} \RR \bigm\vert{}} \text{$\Phi(\theta + 1) = \Phi(\theta) + 1$ for all $\theta \in \RR$}\bigr\},
\end{split}
\label{eq:univcoverdiffs1model}
\end{equation}
and then the homomorphism $\ZZ \to \Diff_+^{(\infty)}(S^1)$ sends $1 \in \ZZ$ to the shift diffeomorphism $\theta \mapsto \theta + 1$. Using this model, it is easy to check that $\Diff_+^{(\infty)}(S^1)$ acts on the abelian group of smooth maps $\Xi\colon \RR \to \liealg t$ with the property that $\Xi(\theta + 1) - \Xi(\theta)$ is a constant element in $\Lambda$, via $\Phi \cdot \Xi := \Phi^* \Xi$, where
\[
(\Phi^* \Xi)(\theta) := \Xi\bigl(\Phi^{-1}(\theta)\bigr), \qquad \theta \in \RR.
\]

With these preliminaries in place, we can start the calculation.

\begin{prop}
\label{thm:unicol-diffs1fail}
Let $\gamma, \rho \in LT$ and $\Xi, \Eta\colon \RR \to \liealg t$ be the quasi-periodic extensions to $\RR$ of choices of lifts $\xi, \eta\colon [0,1] \to \liealg t$ of $\gamma$ and $\rho$ respectively. Take $\phi \in \Diff_+(S^1)$ and make a choice of lift $\Phi \in \Diff_+^{(\infty)}(S^1)$ of $\phi$. Then
\[
c(\phi^* \gamma, \phi^* \rho) = \epsilon(\Delta_\gamma, \Delta_\rho) e^{2\pi i S(\Phi^* \Xi, \Phi^* \Eta)}
\]
and
\begin{align*}
S(\Phi^* \Xi, \Phi^* \Eta) &= S(\Xi, \Eta) + \frac12 \Bigl\langle \Delta_\gamma, \Eta\bigl(\Phi^{-1}(0)\bigr) - \Eta(0) \Bigr\rangle + \phantom{{}} \\
	&\phantom{{}= S(\Xi, \Eta) + {}} \frac12 \Bigl\langle \Xi\bigl(\Phi^{-1}(0)\bigr) - \Xi(0), \Delta_\rho \Bigr\rangle.
\end{align*}
\end{prop}
\begin{proof}
The winding element of a torus loop does not change after precomposition with an orientation preserving circle diffeomorphism, so
\[
\epsilon(\Delta_{\phi^*\gamma}, \Delta_{\phi^*\rho}) = \epsilon(\Delta_\gamma, \Delta_\rho).
\]

Using the chain rule for differentiation and again the fact that $\Delta_{\phi^* \gamma} = \Delta_\gamma$, we can write
\begin{equation}
\begin{split}
S(\Phi^* \Xi, \Phi^* \Eta) &= \frac12 \int_0^1 \Bigl\langle \Xi'\bigl(\Phi^{-1}(\theta)\bigr) \cdot (\Phi^{-1})'(\theta), \Eta\bigl(\Phi^{-1}(\theta)\bigr) \Bigr\rangle \dd\theta + \phantom{{}} \\
	&\phantom{{}={}} \frac12 \Bigl\langle \Delta_\gamma, \Eta\bigl(\Phi^{-1}(0)\bigr) \Bigr\rangle.
\end{split}
\label{eq:unicol-diffs1fail-calc1}
\end{equation}
By substitution, the first term in the above can be rewritten as
\[
\frac12 \int_{\Phi^{-1}(0)}^{\Phi^{-1}(1)} \bigl\langle \Xi'(\theta), \Eta(\theta) \bigr\rangle \dd\theta.
\]
In turn, this integral can be broken up as
\begin{align*}
&\frac12 \int_0^1 \bigl\langle \Xi'(\theta), \Eta(\theta) \bigr\rangle \dd\theta - \frac12 \int_0^{\Phi^{-1}(0)} \bigl\langle \Xi'(\theta), \Eta(\theta) \bigr\rangle \dd\theta + \phantom{{}} \\
&\phantom{\frac12 \int_0^1 \bigl\langle \Xi'(\theta), \Eta(\theta) \bigr\rangle \dd\theta - {}} \frac12 \int_1^{\Phi^{-1}(1)} \bigl\langle \Xi'(\theta), \Eta(\theta) \bigr\rangle \dd\theta,
\end{align*}
and we can replace this first term by
\[
S(\Xi, \Eta) - \frac12 \bigl\langle \Delta_\gamma, \Eta(0) \bigr\rangle.
\]
Putting everything back into~\eqref{eq:unicol-diffs1fail-calc1} gives
\begin{equation}
\begin{split}
S(\Phi^* \Xi, \Phi^* \Eta) &= S(\Xi, \Eta) + \frac12 \Bigl\langle \Delta_\gamma, \Eta\bigl(\Phi^{-1}(0)\bigr) - \Eta(0) \Bigr\rangle - \phantom{{}} \\
	&\phantom{{}={}} \frac12 \int_0^{\Phi^{-1}(0)} \bigl\langle \Xi'(\theta), \Eta(\theta) \bigr\rangle \dd\theta + \frac12 \int_1^{\Phi^{-1}(1)} \bigl\langle \Xi'(\theta), \Eta(\theta) \bigr\rangle \dd\theta.
\end{split}
\label{eq:unicol-diffs1fail-calc2}
\end{equation}

We now focus on the second integral in~\eqref{eq:unicol-diffs1fail-calc2}. Its upper limit $\Phi^{-1}(1)$ is equal to $\Phi^{-1}(0) + 1$, so, using substitution first, the fourth term in~\eqref{eq:unicol-diffs1fail-calc2} becomes
\begin{align*}
\MoveEqLeft \frac12 \int_0^{\Phi^{-1}(0)} \bigl\langle \Xi'(\theta + 1), \Eta(\theta + 1) \bigr\rangle \dd\theta \\
	&= \frac12 \int_0^{\Phi^{-1}(0)} \bigl\langle \Xi'(\theta), \Eta(\theta) + \Delta_\rho \bigr\rangle \dd\theta \\
	&= \frac12 \int_0^{\Phi^{-1}(0)} \bigl\langle \Xi'(\theta), \Eta(\theta) \bigr\rangle \dd\theta + \frac12 \int_0^{\Phi^{-1}(0)} \bigl\langle \Xi'(\theta), \Delta_\rho \bigr\rangle \dd\theta \\
	&= \frac12 \int_0^{\Phi^{-1}(0)} \bigl\langle \Xi'(\theta), \Eta(\theta) \bigr\rangle \dd\theta + \frac12 \Bigl\langle \Xi\bigl(\Phi^{-1}(0)\bigr) - \Xi(0), \Delta_\rho \Bigr\rangle.
\end{align*}
Putting this back into~\eqref{eq:unicol-diffs1fail-calc2} now gives the desired formula for $S(\Phi^* \Xi, \Phi^* \Eta)$.
\end{proof}

From the result of this computation we see that~\eqref{eq:diffs1action-d2} will be satisfied if we define
\begin{equation}
\label{eq:unicol-diffs1action}
d(\phi, \gamma) := e^{\pi i \bigl\langle \Xi(\Phi^{-1}(0)) - \Xi(0), \Delta_\gamma \bigr\rangle} \in \phasegp.
\end{equation}
Thanks to the presence of the minus sign, this definition does not depend on the choice of $\Xi$. It also does not depend on the choice of the lift $\Phi$ of $\phi$. A different choice would namely be of the form $\Phi + k$ for some $k \in \ZZ$. Since $(\Phi + k)^{-1} = \Phi^{-1} - k$, we calculate
\begin{align*}
\Bigl\langle \Xi\bigl(\Phi^{-1}(0) - k\bigr) - \Xi(0), \Delta_\gamma \Bigr\rangle &= \Bigl\langle \Xi\bigl(\Phi^{-1}(0)\bigr) - k\Delta_\gamma - \Xi(0), \Delta_\gamma \Bigr\rangle \\
	&= \Bigl\langle \Xi\bigl(\Phi^{-1}(0) \bigr) - \Xi(0), \Delta_\gamma \Bigr\rangle - k \langle \Delta_\gamma, \Delta_\gamma \rangle.
\end{align*}
Because $\Delta_\gamma \in \Lambda$ and $\Lambda$ is even, $k \langle \Delta_\gamma, \Delta_\gamma \rangle \in 2\ZZ$. This shows what we wanted.

It can furthermore be checked that $d$ satisfies the requirement~\eqref{eq:diffs1action-d1}:
\begin{align*}
d(\psi \circ \phi, \gamma) &= e^{\pi i \bigl\langle \Xi(\Phi^{-1}\Psi^{-1}(0)) - \Xi(0), \Delta_\gamma \bigr\rangle} \\
	&= e^{\pi i \bigl\langle \Xi(\Phi^{-1}\Psi^{-1}(0)) - \Xi(\Phi^{-1}(0)) + \Xi(\Phi^{-1}(0)) - \Xi(0), \Delta_\gamma \bigr\rangle} \\
	&= e^{\pi i \bigl\langle (\Xi \circ \Phi^{-1})(\Psi^{-1}(0)) - (\Xi \circ \Phi^{-1})(0), \Delta_\gamma \bigr\rangle} e^{\pi i \bigl\langle \Xi(\Phi^{-1}(0)) - \Xi(0), \Delta_\gamma \bigr\rangle} \\
	&= d(\psi, \phi^*\gamma) d(\phi, \gamma),
\end{align*}
where $\Psi \in \Diff_+^{(\infty)}(S^1)$ is a choice of a lift of $\psi \in \Diff_+(S^1)$. So this $d$ defines via~\eqref{eq:diffs1action-ansatz} an action of $\Diff_+(S^1)$ on $\centL T$.

\begin{rmk}[A $\Diff_+^{(2)}(S^1)$-action in the case of an odd lattice]
\label{rmk:unicol-diff2s1action-odd}
In our above proof that the $\Diff_+^{(\infty)}(S^1)$-action descends to one of $\Diff_+(S^1)$ we explicitly used that $\Lambda$ is even. Recall from \cref{rmk:unicol-centextoddcase} that if $\Lambda$ is odd it also gives rise to a central extension $\centL T$, which is a $\ZZ/2\ZZ$-graded group. The $\Diff_+^{(\infty)}(S^1)$-action then only descends to one of $\Diff_+^{(2)}(S^1)$, a group that fits into a short exact sequence
\[
\begin{tikzcd}
0 \ar[r] & \ZZ/2\ZZ \ar[r] & \Diff_+^{(2)}(S^1) \ar[r] & \Diff_+(S^1) \ar[r] & 1
\end{tikzcd}
\]
and which can be modelled by $\Diff_+^{(\infty)}(S^1)/2\ZZ$. This action obviously respects the $\ZZ/2\ZZ$-grading on $\centL T$.
\end{rmk}

\begin{rmk}
\label{rmk:unicol-diffS1-compatiblewithsupp}
If we denote for an interval $I \subseteq S^1$ by $\centL_I T$ the normal subgroup of $\centL T$ of those elements $(\gamma, z)$ for which $\supp \gamma \subseteq I$, then
\[
\phi \cdot L_I T = L_{\phi(I)} T \qquad \text{and} \qquad \phi \cdot \centL_I T = \centL_{\phi(I)} T.
\]
\end{rmk}

The following \namecref{thm:unicol-diffS1-actslocally} explains the relation between the notion of support for elements of $\Diff_+(S^1)$ with that for elements of $\centL T$.

\begin{prop}
\label{thm:unicol-diffS1-actslocally}
If $\phi \in \Diff_+(S^1)$ has support in an interval $I \subseteq S^1$, then $\phi$ acts trivially on $\centL_{I'} T$, where $I'$ is the subinterval that is the closure of $S^1 \backslash I$.
\end{prop}
\begin{proof}
Let $\gamma \in L_{I'} T$. We first claim that $\phi^* \gamma = \gamma$. Indeed, if $\theta \in I'$, then $\phi^{-1}(\theta) = \theta$ and so $\gamma(\phi^{-1}(\theta)) = \gamma(\theta)$. Suppose on the other hand that $\theta \in I$, then $\phi^{-1}(\theta) \in I$, and so $\gamma(\phi^{-1}(\theta)) = 0_T = \gamma(\theta)$.

Our second claim is that $d(\phi, \gamma) = 1 \in \phasegp$. We distinguish two (not mutually exclusive) cases depending on the privileged point $q$ along which we cut $S^1$: either $q \in I'$, or $q \in I$. If $q \in I'$ holds, then $\phi^{-1}(q) = q$ and so for a lift $\Phi \in \Diff_+^{(\infty)}(S^1)$ of $\phi$ we must have $\Phi^{-1}(0_\RR) \in \ZZ$. This means that for a lift $\Xi\colon \RR \to \liealg t$,
\[
\Xi\bigl(\Phi^{-1}(0_\RR)\bigr) - \Xi(0_\RR) = \Phi^{-1}(0_\RR) \cdot \Delta_\gamma,
\]
implying that $d(\phi, \gamma) = 1$ because $\Phi^{-1}(0_\RR) \cdot \langle \Delta_\gamma, \Delta_\gamma \rangle \in 2\ZZ$ since $\Lambda$ is even.

Suppose now that $q \in I$. Consider the covering map $\RR \twoheadrightarrow [0,1] \twoheadrightarrow S^1$ that sends $0_\RR$ to $q$, which we have been using all the time in this section, and the subintervals of $\RR$ that are the pre-images of $I$ under this map. The assumption that $q \in I$ means that $0_\RR$ lies in such a subinterval of $\RR$. Let us call that one $J$. Because $\phi^{-1}$ has support in $I$ we must have that $\Phi^{-1}(J) = J + k$ for some $k \in \ZZ$. Therefore,
\[
\Xi\bigl(\Phi^{-1}(J)\bigr) = \Xi(J+k) = \Xi(J) + k\cdot \Delta_\gamma.
\]
Since $\gamma$ is identically $0_T$ on $I$ this implies that $\Xi$ is constant (with value an element of $\Lambda \subseteq \liealg t$) on $J$. This shows that
\[
\Xi\bigl(\Phi^{-1}(0_\RR)\bigr) - \Xi(0_\RR) = k \cdot \Delta_\gamma,
\]
and so $d(\phi, \gamma) = 1$ holds once more.
\end{proof}

\subsection{Actions of $\Rot(S^1)$ on $(\centL T)_0$ and $\centV \liealg t$}
\label{subsec:unicol-diffS1action-subgps}
Notice from the formula~\eqref{eq:unicol-diffs1action} that, because $\Delta_\gamma = 0$ when $\gamma \in (LT)_0$, the subtle $\Diff_+(S^1)$-action we constructed on $\centL T$ restricts to the obvious one on $(\centL T)_0$. Remember furthermore the remark made in \cref{sec:unicol-struct} that the vector space $V\liealg t$ carries a $\Rot(S^1)$-action. The cocycle defining the central extension $\centV \liealg t$ from \cref{subsec:unicol-centext-subgps} is $\Rot(S^1)$-invariant, which is easily seen from an application of the chain rule followed by a substitution, and therefore the $\Rot(S^1)$-action on $V\liealg t$ lifts in the obvious way to $\centV \liealg t$. It is now immediate that the group isomorphism~\eqref{eq:unicol-decompidcomp-centext} is $\Rot(S^1)$-equivariant.

\section{Actions of lifts of lattice automorphisms on central extensions}
\label{sec:unicol-latisoms}
Take again the input data for the construction of a torus loop group central extension $\centL T$ in \cref{sec:unicol-centext}, as summarised in the first paragraph of \cref{sec:unicol-diffs1action}, as given. In \cref{sec:unicol-diffs1action} we learned that a torus loop $S^1 \to T$ can be precomposed with an orientation preserving circle diffeomorphism, and that this lifts to an action of $\Diff_+(S^1)$ on $\centL T$. One can ask with what kind of automorphism of the torus loops can be postcomposed to obtain a group automorphism of $\centL T$ also. In this section we show that this can be done for torus automorphisms induced by automorphisms of a certain $\{\pm 1\}$-central extension of the lattice $\Lambda$.

\begin{refs}
Our idea that such an action on $\centL T$ should exist comes from the vertex algebra literature, namely the claim in \parencite[Section 12]{borcherds:monstmoonshine} that in the case when $\Lambda$ is the Leech lattice $\Lambda_\Leech$ the group $\Aut(\tilde\Lambda^\epsilon; \langle \cdot, \cdot \rangle)$ acts on the Leech lattice vertex algebra. This fact is further elaborated upon for general positive definite, even lattices in for example \parencite[Section 2.4]{dong:automslatticevalgs}. The current section is simultaneously intended to clarify the first paragraph of \parencite[Section 4]{dong:latticeconfnets}. The details regarding the group $\Aut(\tilde\Lambda^\epsilon; \langle \cdot, \cdot \rangle)$ are taken from the start of \parencite[Section 6.4]{frenkel:monster}.
\end{refs}

An automorphism $g$ of $\Lambda$ is, as discussed in \cref{app:lattices}, a $\ZZ$-module automorphism of $\Lambda$ which preserves the form $\langle \cdot, \cdot \rangle$ on $\Lambda$. There is an obvious action of $\Aut(\Lambda; \langle \cdot, \cdot \rangle)$ on the non-centrally extended loop group $LT$ given by $g\cdot \gamma := g_*\gamma$, where $\gamma \in LT$, $g_*\gamma \in LT$ is the loop defined by
\[
(g_* \gamma)(\theta) := g_T\bigl(\gamma(\theta)\bigr), \qquad \theta \in S^1,
\]
and $g_T$ is the induced automorphism of $T := \Lambda \otimes_\ZZ \phasegp$ by $g$. Just like in \cref{sec:unicol-diffs1action}, a naive attempt at defining an action on $\centL T$ would be
\[
g\cdot (\gamma, z) := (g_* \gamma, z), \qquad z \in \phasegp.
\]
Again, this respects the group multiplication on $\centL T$ if and only if for the $2$-cocycle $c$ in~\eqref{eq:unicol-cocycle} defining $\centL T$ there holds $c(g_*\gamma, g_* \rho) = c(\gamma, \rho)$ for all $\gamma, \rho \in LT$. Let us check whether this is true.

Recall that $c$ is defined in terms of choices of lifts $\xi,\eta \colon [0,1] \to \liealg t$ of $\gamma$ and $\rho$ respectively. We have
\[
\Delta_{g_*\gamma} = (\RR g \circ \xi)(1) - (\RR g \circ \xi)(0) = \RR g\bigl(\xi(1) - \xi(0)\bigr) = g\Delta_\gamma,
\]
where $\RR g$ denotes the linear extension of $g$ to $\liealg t$. Similarly, $\Delta_{g_*\rho} = g\Delta_\rho$. Therefore,
\begin{gather*}
c(g_*\gamma, g_* \rho) = \epsilon(g\Delta_\gamma, g\Delta_\rho) e^{2\pi i S(\RR g \circ \xi, \RR g \circ \eta)} \\
S(\RR g \circ \xi, \RR g \circ \eta) = \frac12 \int_0^1 \bigl\langle (\RR g \circ \xi)'(\theta), (\RR g \circ \eta)(\theta) \bigr\rangle \dd \theta + \frac12 \bigl\langle g \Delta_\gamma, (g \circ \eta)(0) \bigr\rangle.
\end{gather*}
Since $(\RR g \circ \xi)'(\theta) = \RR g(\xi'(\theta))$ and $\RR g$ preserves the form $\langle \cdot, \cdot \rangle$, we find that
\[
S(\RR g \circ \xi, \RR g \circ \eta) = S(\xi, \eta).
\]
However, it is not necessarily true that $\epsilon(g\Delta_\gamma, g\Delta_\rho) = \epsilon(\Delta_\gamma, \Delta_\rho)$. Explicit counterexamples can be found. We conclude that in general $c(g_*\gamma, g_* \rho) \neq c(\gamma, \rho)$ and that the blame lies on the extra piece of data that is the $2$-cocycle $\epsilon$, which is not invariant under automorphisms of $\Lambda$.

In our next attempt we will try not to let $\Aut(\Lambda; \langle \cdot, \cdot \rangle)$ act on $\centL T$, but a certain bigger group instead.

Let $\tilde\Lambda^\epsilon$ be the $\{\pm 1\}$-central extension of $\Lambda$ associated to $\epsilon$. Its underlying set is $\Lambda \times \{\pm 1\}$ and its group multiplication is defined by
\[
(\lambda, z) \cdot (\mu, w) := \bigl(\lambda + \mu, zw \cdot \epsilon(\lambda, \mu)\bigr),
\]
where $\lambda, \mu \in \Lambda$ and $z,w \in \{\pm 1\}$. An automorphism $\tilde g$ of the group $\tilde\Lambda^\epsilon$ induces an automorphism $g$ of the underlying abelian group of $\Lambda$ by picking for $\lambda \in \Lambda$ an arbitrary lift $\tilde \lambda \in \tilde\Lambda^\epsilon$ and defining $g\lambda$ to be the image in $\Lambda$ of $\tilde g\tilde \lambda$. This is independent of the choice of $\tilde\lambda$. Write $\Aut(\tilde \Lambda^\epsilon; \langle \cdot, \cdot \rangle)$ for the group of those automorphisms $\tilde g$ of $\tilde\Lambda$ such that $g$ lies in $\Aut(\Lambda; \langle \cdot, \cdot \rangle)$, that is, $g$ preserves the form $\langle \cdot, \cdot \rangle$. Notice that such a $\tilde g$ fixes the embedding $\{\pm 1\} \hookrightarrow \tilde\Lambda^\epsilon$, because its image $\{(0,\pm 1)\}$ consists of the only elements of $\tilde\Lambda^\epsilon$ of finite order.

This group of automorphisms of $\tilde\Lambda$ lies in the following short exact sequence of groups:
\[
\begin{tikzcd}
1 \ar[r] & \Hom_\Ab\bigl(\Lambda, \{\pm 1\}\bigr) \ar[r] & \Aut\bigl(\tilde \Lambda^\epsilon; \langle \cdot, \cdot \rangle\bigr) \ar[r] & \Aut\bigl(\Lambda; \langle \cdot, \cdot \rangle\bigr) \ar[r] & 1.
\end{tikzcd}
\]
The second arrow sends $f \in \Hom_\Ab(\Lambda, \{\pm 1\})$ to the automorphism $\tilde\lambda \mapsto \tilde\lambda\cdot f(\lambda)$, where $\tilde\lambda \in \tilde\Lambda^\epsilon$ and $\lambda$ is its image in $\Lambda$. The third arrow sends $\tilde g \in \Aut(\tilde \Lambda^\epsilon; \langle \cdot, \cdot \rangle)$ to the induced automorphism $g \in \Aut(\Lambda, \langle \cdot, \cdot \rangle)$. Note that a choice of a basis for $\Lambda$ gives an isomorphism
\[
\Hom_\Ab\bigl(\Lambda, \{\pm 1\}\bigr) \cong \{\pm 1\}^{\rank \Lambda},
\]
so $\Aut(\tilde \Lambda^\epsilon; \langle \cdot, \cdot \rangle)$ has order $2^{\rank \Lambda}$ times that of $\Aut(\Lambda; \langle \cdot, \cdot \rangle)$.

In \cref{sec:unicol-struct} we learned that, given a choice of a privileged point $q$ on $S^1$ (which we already made to construct $\centL T$), there are standard choices of loops $\gamma_\lambda$ with winding element $\lambda \in \Lambda$. This gave an injective homomorphism of abelian groups $\Lambda \hookrightarrow LT$, $\lambda \mapsto \gamma_\lambda$. It lifts to a homomorphism $\tilde\Lambda^\epsilon \hookrightarrow \centL T$, $(\lambda, z) \mapsto (\gamma_\lambda, z)$ of non-abelian groups. We would like to find an action of $\Aut(\tilde \Lambda^\epsilon; \langle \cdot, \cdot \rangle)$ on $\centL T$ which restricts to the one on $\tilde \Lambda^\epsilon$.

Let us therefore begin by describing the elements of $\Aut(\tilde \Lambda^\epsilon; \langle \cdot, \cdot \rangle)$ a bit more explicitly. If $\tilde g$ is such an automorphism and $(\lambda, z) \in \tilde\Lambda^\epsilon$, then
\[
\tilde g \cdot (\lambda, z) = \bigl(g\lambda, e(\tilde g, \lambda) \cdot z\bigr)
\]
for some $g \in \Aut(\Lambda; \langle \cdot,\cdot \rangle)$ and $e(\tilde g, \lambda) \in \{\pm 1\}$, which does not depend on $z$. If we namely assume that $\epsilon$ is a normalised cocycle, then we can write
\[
(\lambda, z) = \bigl(\lambda, z \cdot \epsilon(\lambda, 0)\bigr) = (\lambda, 1) \cdot (0,z).
\]
So therefore
\begin{align*}
\tilde g (\lambda, z) &= \tilde g (\lambda, 1) \cdot \tilde g (0,z) \\
	&= \bigl(g\lambda, e(\tilde g, \lambda)\bigr) \cdot (0,z) \\
	&= \bigl(g\lambda, e(\tilde g, \lambda) \cdot z \cdot \epsilon(g\lambda, 0)\bigr) \\
	&= \bigl(g\lambda, e(\tilde g,\lambda) \cdot z\bigr).
\end{align*}
Here we used in the second equality that $\tilde g$ fixes the embedding $\{\pm 1\} \hookrightarrow \tilde\Lambda^\epsilon$, and in the fourth equality it was used again that $\epsilon$ is normalised.

This $e$ satisfies the following two identities:
\begin{gather*}
e(\tilde g \circ \tilde h, \lambda) = e(\tilde g, h\lambda) e(\tilde h, \lambda) \\
e(\tilde g, \lambda + \mu) \epsilon(\lambda, \mu) =  e(\tilde g, \lambda) e(\tilde g, \mu) \epsilon(g\lambda, g\mu).
\end{gather*}
The latter should be read as $e(\tilde g, \cdot)\colon \Lambda \to \{\pm 1\}$ being a $1$-cochain which exhibits the $2$-cocycle $(\lambda, \mu) \mapsto \epsilon(g \lambda, g \mu) \epsilon(\lambda, \mu)^{-1}$ as a $2$-coboundary.

Let us now define an action of $\Aut(\tilde\Lambda^\epsilon; \langle \cdot,\cdot \rangle)$ on $\centL T$ by setting
\begin{equation}
\label{eq:unicol-latisomaction}
\tilde g \cdot (\gamma, z) := \bigl(g_* \gamma, e(\tilde g, \Delta_\gamma) \cdot z\bigr)
\end{equation}
for $\tilde g \in \Aut(\tilde\Lambda^\epsilon; \langle \cdot,\cdot \rangle)$ and $(\gamma, z) \in \centL T$. One can check that, thanks to the two identities satisfied by $e$ above, this does respect the multiplication on $\centL T$ and that it defines a group action.

\begin{prop}
\label{thm:unicol-rotandisometry-commute}
The action of $\Diff_+(S^1)$ on $\centL T$ given by~\eqref{eq:diffs1action-ansatz} and~\eqref{eq:unicol-diffs1action} commutes with the action of $\Aut(\tilde\Lambda^\epsilon; \langle \cdot, \cdot \rangle)$ given in~\eqref{eq:unicol-latisomaction}. 
\end{prop}
\begin{proof}
Let $(\gamma, z) \in \centL T$, $\phi \in \Diff_+(S^1)$ and $\tilde g \in \Aut(\tilde\Lambda^\epsilon; \langle \cdot,\cdot \rangle)$. We have to show that the equality
\[
\phi \cdot \bigl(\tilde g \cdot (\gamma, z)\bigr) = \tilde g \cdot \bigl(\phi\cdot (\gamma,z)\bigr)
\]
holds true. The left hand side is
\[
\phi \cdot \bigl(g_* \gamma, e(\tilde g, \Delta_\gamma) \cdot z\bigr) = \bigl(\phi^*g_*\gamma, d(\phi, g_*\gamma) \cdot e(\tilde g, \Delta_\gamma) \cdot z\bigr),
\]
while the right hand side is
\[
\tilde g \cdot \bigl(\phi^*\gamma, d(\phi, \gamma) \cdot z\bigr) = \bigl(g_*\phi^*\gamma, e(\tilde g, \Delta_{\phi^*\gamma}) \cdot d(\phi, \gamma) \cdot z\bigr).
\]

We obviously have that $\phi^*g_*\gamma = g_*\phi^*\gamma$, since pre- and postcompositions commute. There furthermore holds that $\Delta_{\phi^*\gamma} = \Delta_\gamma$, and
\[
d(\phi, g_*\gamma) = e^{\pi i \bigl\langle (\RR g \circ \Xi)(\Phi^{-1}0) - (\RR g \circ \Xi)(0), g\Delta_\gamma\bigr\rangle} 
	= e^{\pi i \bigl\langle \Xi(\Phi^{-1}0) - \Xi(0), \Delta_\gamma\bigr\rangle} 
	= d(\phi, \gamma).
\]
This shows what we wanted.
\end{proof}

\begin{rmk}
It is easy to see that an element $\tilde g \in \Aut(\tilde\Lambda^\epsilon; \langle \cdot, \cdot \rangle)$ restricts to an automorphism of the normal subgroup $\centL_I T$ for every interval $I \subseteq S^1$. We say that $\Aut(\tilde\Lambda^\epsilon; \langle \cdot, \cdot \rangle)$ acts \defn{locally}.
\end{rmk}

\subsection{Actions of related groups on $(\centL T)_0$ and $\centV \liealg t$}
\label{subsec:unicol-isometryactions-subgps}

The failure for the group $\Aut(\Lambda; \langle \cdot, \cdot \rangle)$ to act on $\centL T$ can, instead of being blamed on the unnatural choice of the $2$-cocycle $\epsilon$, alternatively be pinned on the presence of non-zero winding elements of loops. Hence, this group of lattice isomorphisms \emph{does} act on the identity component $(\centL T)_0$.

Write $\Aut(\liealg t; \langle \cdot, \cdot \rangle)$ for the compact group of automorphisms of the Lie algebra $\liealg t$ that preserve the bilinear extension of the form $\langle \cdot, \cdot \rangle$ on $\Lambda$ to $\liealg t$. It acts through post-composition of loops via $\RR$-linear automorphisms on $V\liealg t$. Let $O_*$ be such an automorphism for $O \in \Aut(\liealg t; \langle \cdot, \cdot \rangle)$. Then $O_*$ preserves the skew form $S$ in~\eqref{eq:unicol-cocycle-idcomp} on $V\liealg t$ since $\dd(O \circ \xi) = O \circ \dd \xi$ for all $\xi \in V\liealg t$. Therefore, the action of $\Aut(\liealg t; \langle \cdot, \cdot \rangle)$ lifts in the obvious way to $\centV \liealg t$.

We summarise these findings in the following figure:
\[
\begin{tikzcd}
\centV \liealg t \ar[r, hook] & (\centL T)_0 \ar[r, hook] & \centL T \\
\Aut\bigl(\liealg t; \langle \cdot, \cdot \rangle\bigr) \ar[u, phantom, "\circlearrowleft"]  & \ar[l, hook'] \Aut\bigl(\Lambda; \langle \cdot, \cdot \rangle\bigr) \ar[u, phantom, "\circlearrowleft"] & \ar[l, two heads] \Aut\bigl(\tilde\Lambda^\epsilon; \langle \cdot, \cdot \rangle\bigr) \ar[u, phantom, "\circlearrowleft"] \\[-2em]
\RR g & \ar[l, mapsto] g & \ar[l, mapsto] \tilde g
\end{tikzcd}.
\]

\section{Irreducible, positive energy representations}
\label{sec:unicol-reptheory}

Assume once more the input data for the construction of a torus loop group central extension $\centL T$ in \cref{sec:unicol-centext} summarised in the first paragraph of \cref{sec:unicol-diffs1action}. For the purpose of the discussion that is about to follow we add the extra condition of being positive definite on the lattice $\Lambda$. This will be necessary for certain vector spaces we will consider to be equipped with positive definite inner products. Their Hilbert space completions will then carry group representations.\footnote{Recall our convention that representations are always meant to be strongly continuous and unitary.}

In this section we will construct and classify a certain type of irreducible representations of $\centL T$, namely the \defn{positive energy} ones (see \cref{dfn:posenergyrep}).

\begin{refs}
The construction and classification we perform here is due to \parencite{segal:unitary}. See also \parencite[Section 9.5]{pressley:loopgps} for an alternative exposition. The outcomes of the calculation of the characters of these representations were learned from \parencite[Remark 7.1.2]{frenkel:monster}, \parencite[Section 4.2.1]{staszkiewicz} and \parencite[356]{mason:voas}.
\end{refs}

We saw in \cref{sec:unicol-diffs1action} that the group $\centL T$ carries an action of the group $\Diff_+(S^1)$ and hence also of the subgroup $\Rot(S^1)$. It therefore indeed makes sense to talk about positive energy representations\footnote{We are taking the central extension $\centL T$ for the group $N$ in \cref{dfn:posenergyrep}, not $LT$.} of $\centL T$, once we make $\centL T$ into a topological group. The decomposition~\eqref{eq:unicol-decompidcomp-centext} will be the key to the construction and classification of such representations because it allows us to do so for the identity component $(\centL T)_0$ first. We will later transfer these results to $\centL T$.

\subsection{Irreducible representations of $\centV \liealg t$}
\label{subsec:unicol-Vtirreps}

Recall from \cref{subsec:unicol-centext-subgps} the definition of the subgroup $\centV \liealg t$ of $\centL T$, made in terms of the skew form $S$ in~\eqref{eq:unicol-cocycle-idcomp}. Because we are assuming $\Lambda$ to be positive definite, the observation at the end of \cref{subsec:unicol-centext-subgps} applies. That is, $\centV \liealg t$ is the Heisenberg group associated to the pair $(V\liealg t, -S)$, in the non-topological sense. Our first step towards studying its representation theory will be to define a specific complex structure on $V\liealg t$ that is compatible with the skew form $S$, so as to turn $\centV \liealg t$ into a topological group. This complex structure will then also be used to construct a Weyl representation of $\centV \liealg t$.

\subsubsection{A complex structure on $V\liealg t$}

Because each loop $\xi\colon S^1 \to \liealg t$ of $V\liealg t$ is smooth, it admits a Fourier decomposition
\[
\xi(\theta) = \sum_{\mathclap{k \in \ZZ \backslash \{0\}}} \xi_k e^{2\pi i k \theta},
\]
where the Fourier coefficient $\xi_k$ is a vector of the complexification $\CC\liealg t := \liealg t \otimes_\RR \CC$. (The condition $k \neq 0$ is there because $\int_{S^1} \xi(\theta) \dd \theta = 0$ and so $\xi_0 = 0$.) We therefore see that the complexified vector space $\CC(V\liealg t) := V\liealg t \otimes_\RR \CC$, which can be identified with the space of smooth loops $\xi\colon S^1 \to \CC\liealg t$ such that $\int_{S^1} \xi(\theta) \dd\theta = 0$, can be decomposed as a direct sum of two $\CC$-linear subspaces
\begin{equation}
\label{eq:compstruct-Vtdecomp}
\CC(V\liealg t) = V\liealg t^+ \oplus V\liealg t^-,
\end{equation}
where $V\liealg t^+$ is the subspace consisting of the loops whose negative Fourier coefficients are zero and $V\liealg t^-$ is defined similarly. Equivalently, $V\liealg t^+$ is the subspace of loops which admit an extension to the closed unit disc in the complex plane such that the restriction to the open unit disc is holomorphic, and a similar condition holds for the loops of $V\liealg t^-$ with respect to the other half of the Riemann sphere.

The standard complex conjugation $\conj{\xi \otimes z} := \xi \otimes \conj z$ on $\CC(V\liealg t)$ interchanges $V\liealg t^+$ and $V\liealg t^-$. The image of the canonical $\RR$-linear injection $V\liealg t \hookrightarrow \CC(V\liealg t)$ given by $\xi \mapsto \xi \otimes 1$ therefore equals $\{\xi^+ + \conj{\xi^+} \mid \xi^+ \in V\liealg t^+\}$ and so the compositions of this injection with the two $\CC$-linear projections $\CC(V\liealg t) \twoheadrightarrow V\liealg t^\pm$ are $\RR$-linear isomorphisms.

Now we define a complex structure $J\colon V\liealg t \to V\liealg t$ by setting $V\liealg t^+$ to be the $+i$-eigenspace of the $\CC$-linear extension $\CC J$ of $J$ to $\CC(V\liealg t)$ and $V\liealg t^-$ the $-i$-eigenspace, meaning that $\CC J$ multiplies a vector of $\CC(V\liealg t)$ of the form $\xi_k e^{2\pi i k \theta}$ with $\sgn(k) \cdot i$. In other words, $J$ is determined by
\[
(\CC J)(\xi \otimes 1) := J(\xi) \otimes 1 := i \xi^+ - i \conj{\xi^+},
\]
where $\xi^+ \in V\liealg t^+$ is the image of $\xi$ under $V\liealg t \xrightarrow{\sim} V\liealg t^+$.  It is then clear that $J$ squares to $-\id_{V\liealg t}$. This makes $(V\liealg t)_J \xrightarrow{\sim} V\liealg t^+$ a $\CC$-linear isomorphism by definition.

\subsubsection{Compatibility of the complex structure with the skew form}

Let us extend the form $\langle \cdot, \cdot \rangle$ on $\liealg t$ complex bilinearly to $\CC\liealg t$. Note that this extended form is still symmetric and not conjugate symmetric. In turn, this then defines a complex bilinear extension of $S$ to $\CC(V\liealg t)$. This extension of $S$ is again skew. We denote these extensions to $\CC \liealg t$ and $\CC(V\liealg t)$, respectively, by $\langle \cdot, \cdot \rangle$ and $S$ as well
. If $\xi$ and $\eta$ are vectors in $V\liealg t$, then it can be checked by expanding them into their Fourier series that
\[
S(\xi ,\eta) =
\pi i \sum_{%
	\mathclap{k \in \ZZ \backslash \{0\}}%
} k \langle \xi_k, \eta_{-k} \rangle.
\]
This expression shows that the subspaces $V\liealg t^+$ and $V\liealg t^-$ are isotropic for $S$ and that
\[
S(\xi, \eta) = S(\xi^+, \eta^-) + S(\xi^-, \eta^+),
\]
where $\xi = \xi^+ + \xi^-$ and $\eta = \eta^+ + \eta^-$ are their decompositions along~\eqref{eq:compstruct-Vtdecomp}.

The complex bilinearity and skewness of $S$ can now be used to calculate
\begin{align*}
S(\xi, J\xi) &=  S\bigl(\xi^+, (\CC J)\xi^-\bigr) + S\bigl(\xi^-, (\CC J)\xi^+\bigr) \\
			&=  S(\xi^+, -i\xi^-) + S(\xi^-, i\xi^+) \\
			&=  -iS(\xi^+, \xi^-) + i S(\xi^-, \xi^+) \\
			&= -2i S(\xi^+, \xi^-) \\
			&= 2\pi \sum_{k=1}^\infty k \langle \xi_k, \xi_{-k} \rangle.
\end{align*}
Because $\xi$ takes its values in the real Lie algebra $\liealg t$ we have $\xi_{-k} = \overline{\xi_k}$. Since the form $\langle \cdot, \cdot \rangle$ on $\liealg t$ is positive definite, there holds $\langle \xi_k, \xi_{-k} \rangle \geq 0$. This proves that if $\xi \neq 0$, then $S(\xi, J\xi) > 0$. It is furthermore easily checked, again using the complex bilinearity of $S$, that $S(J\xi, J\eta) = S(\xi, \eta)$. We conclude that $J$ satisfies the two properties listed in \cref{thm:tamedcompstruct}\ref{thmitm:tamedcompstruct-compat}.

We may now proceed with the constructions described in \cref{app:heisgps}. The norm topology from $\langle \cdot, \cdot \rangle_J$ makes $V\liealg t$ a topological real vector space and the compatibility of $S$ with $J$ implies that $S$ is continuous with respect to this topology. Hence, $\centV \liealg t$ is a Heisenberg group. Furthermore, $(V\liealg t)_J$ is a complex pre-Hilbert space equipped with a Hermitian inner product $\langle \cdot, \cdot \rangle_J$. We may therefore form a bosonic Fock space $\hilb S((V\liealg t)_J)$. It carries an irreducible Weyl representation $W_J$ of the Heisenberg group $\centV \liealg t$. The central subgroup $\phasegp$ of $\centV \liealg t$ acts under $W_J$ as $z \mapsto z$. We will abbreviate this representation and its underlying Fock space by $W$ and $\hilb S$, respectively, that is, we suppress the reference to the specific complex structure $J$.

\subsubsection{Positivity of energy}

As already noted in \cref{subsec:unicol-diffS1action-subgps}, the $\Rot(S^1)$-action on $V\liealg t$ by $\RR$-linear operators preserves the skew form $S$. Moreover, the $\CC$-linear extension of the $\Rot(S^1)$-action to $\CC(V\liealg t)$ preserves the decomposition~\eqref{eq:compstruct-Vtdecomp} of $\CC(V\liealg t)$ because it is given on each vector of the form $\xi_k e^{2\pi i k\theta}$ simply by
\begin{equation}
\label{eq:unicol-rotaction-fourierterm}
(\CC \phi_{\theta'}^*)(\xi_k e^{2\pi i k\theta}) = \xi_k e^{2\pi i k(\theta-\theta')},
\end{equation}
where we write $\CC\phi_{\theta'}^*$ for the $\CC$-linear extension of the rotation operator $\phi_{\theta'}$ along the angle $\theta'$. Therefore, the $\Rot(S^1)$-action on $V\liealg t$ commutes with the complex structure $J$.

This $\Rot(S^1)$-action is strongly continuous with respect to the norm topology on $V\liealg t$ induced by $\langle \cdot, \cdot \rangle_J$. 
Hence, we may apply \cref{thm:weylrep-intertwin} to conclude that the action of $\Rot(S^1)$ on $\centV \liealg t$ is strongly continuous, and that there is a representation $R$ of $\Rot(S^1)$ on $\hilb S$ which extends the action on the subspace $(V\liealg t)_J = \Sym^1((V\liealg t)_J) \subseteq \hilb S$ and such that the intertwining property
\begin{equation}
\label{eq:unicol-rotaction-intertwineheis}
R(\phi_\theta) W(\xi, z) R(\phi_\theta)^* = W\bigl(\phi_\theta \cdot (\xi, z)\bigr)
\end{equation}
is satisfied for all $\phi_\theta \in \Rot(S^1)$ and $(\xi, z) \in \centV \liealg t$.

Let us calculate the \defn{character} of $R$ (see \cref{dfn:gradchar}). For this purpose we might as well temporarily redefine $\hilb S$ to be $\hilb S(V\liealg t^+)$ instead of $\hilb S((V\liealg t)_J)$ because the $\CC$-linear isomorphism $(V\liealg t)_J \xrightarrow{\sim} V\liealg t^+$ intertwines the respective $\Rot(S^1)$-actions. Recall that $V\liealg t^+$ is the complex vector space of smooth loops $S^1 \to \CC\liealg t$ that are of the form
\[
\sum_{k=1}^\infty \xi_k e^{2\pi ik \theta}, \qquad \theta \in S^1.
\]
We then learn from the expression~\eqref{eq:unicol-rotaction-fourierterm} for the rotation action on $V\liealg t^+$, together with the fact that $R$ is defined in the proof of \cref{thm:weylrep-intertwin} by extending the rotation action on $V\liealg t^+$ to monomials factor-wise, that the $k$-th energy eigenspace $\hilb S(k)$ is spanned by monomials of the form
\[
\xi_{1r} e^{2\pi i k_1\theta} \xi_{2r} e^{2\pi i k_2\theta} \cdots \xi_{rr} e^{2\pi i k_r\theta},
\]
where $r \geq 0$, $\xi_{ir} \in \CC\liealg t$, $k_i \geq 1$ and $k_1 + \cdots + k_r = k$. Hence, $\hilb S$ contains no vectors of negative energy:

\begin{prop}
The intertwining $\Rot(S^1)$-action $R$ on the Weyl representation $W$ of $\centV \liealg t$ is of positive energy.
\end{prop}

By picking a basis for the Lie algebra $\CC\liealg t$ we see that the dimension of $\hilb S(k)$ equals the coefficient of $q^k$ in the power series
\[
\biggl(\sum_{k=0}^\infty p(k)q^k\biggr)^{\dim \liealg t},
\]
where $p(k)$ is the \defn{Euler partition function}, that is, $p(k)$ for $k \geq 1$ denotes the number of ways to write $k$ as the sum of positive integers, and $p(0) := 1$ by convention. (This coefficient can also be understood as the number of \defn{$(\dim \liealg t)$-coloured partitions} of $k$.) In turn, L.~Euler found that the generating function of the partition function named after him admits an expression as an infinite product:
\[
\sum_{k=0}^\infty p(k)q^k = \prod_{j=1}^\infty (1-q^j)^{-1},
\]
so if we set
\[
\eta(q) := q^{1/24} \prod_{j=1}^\infty (1-q^j),
\]
then we obtain the formula
\[
\ch_R(q) = q^{\dim \liealg t/24} \eta(q)^{-\dim \liealg t}.
\]

Our reason for writing $\ch_R(q)$ in this way is that if the formal variable $q$ is substituted by $e^{2 \pi i z}$ for $z \in \CC$ with $\Im z > 0$, then $\eta(q)$ is known as \defn{Dedekind's eta function}. It is a \defn{holomorphic modular form} for the full group $\lieSL(2,\ZZ)$ of weight $1/2$ (with a highly non-trivial multiplier) (see \parencite[Theorem 3.4]{apostol:modular}). While $\ch_R(q)$ does not satisfy the necessary transformation behaviour to be a modular form as well, the normalised character
\begin{equation}
\label{eq:partfuncfocksp}
Z_W(q) := q^{-\dim \liealg t/24} \ch_R(q) = \eta(q)^{-\dim \liealg t}
\end{equation}
does.

Having confirmed that $W$ is of positive energy, the unicity result \cref{thm:stonevonneumann} now tells us that

\begin{thm}
\label{thm:unicol-classifyVt-irreps}
Every irreducible, positive energy representation of $\centV \liealg t$ such that the central subgroup $\phasegp \subseteq \centV \liealg t$ acts  as $z \mapsto z$ is isomorphic to $W$.
\end{thm}

\subsubsection{The intertwining action of Lie algebra isometries}

We exhibit, on top of the representations of $\centV \liealg t$ and $\Rot(S^1)$, a third piece of structure that the Hilbert space $\hilb S$ possesses.

Recall from \cref{subsec:unicol-isometryactions-subgps} that the group $\Aut(\liealg t; \langle \cdot, \cdot \rangle)$ acts through $\RR$-linear automorphisms on $V\liealg t$ in a way that preserves the skew form $S$ in~\eqref{eq:unicol-cocycle-idcomp}. Furthermore, if $O \in \Aut(\liealg t; \langle \cdot, \cdot \rangle)$ and $O_*$ is the corresponding operator on $V\liealg t$, then the $\CC$-linear extension $\CC(O_*)$ of $O_*$ to $\CC(V\liealg t)$ preserves each of the subspaces $V\liealg t^\pm$ because it is given on each vector of the form $\xi_k e^{2\pi i k \theta}$ by
\[
\CC(O_*)(\xi_k e^{2\pi i k \theta}) = (\CC O)(\xi_k) e^{2\pi i k \theta},
\]
where $\CC O$ is the $\CC$-linear extension of $O$ to $\CC \liealg t$. Therefore, $O_*$ commutes with the complex structure $J$ on $V\liealg t$.

It is strongly continuous with respect to the norm topology on $V\liealg t$ induced by $\langle \cdot, \cdot \rangle_J$. 
Hence the demands of \cref{thm:weylrep-intertwin} are met and $\Aut(\liealg t; \langle \cdot, \cdot \rangle)$ acts strongly continuously on $\centV \liealg t$ (we denote the translate of $(\xi, z) \in \centV \liealg t$ by $O$ as $O \cdot (\xi, z)$) and there is a representation $U$ of $\Aut(\liealg t; \langle \cdot, \cdot \rangle)$ on $\hilb S$ which extends the action on $(V\liealg t)_J$ such that the intertwining property
\begin{equation}
\label{eq:unicol-liealgisomets-intertwin}
U(O) W(\xi, z) U(O)^* = W\bigl(O \cdot (\xi, z)\bigr)
\end{equation}
is satisfied. This action fixes the vacuum vector.

It is obvious from looking at their actions on $\CC(V\liealg t)$ that the actions of $\Rot(S^1)$ and $\Aut(\liealg t; \langle \cdot, \cdot \rangle)$ on $V\liealg t$ commute. The same then holds for their respective extensions $R$ and $U$ to $\hilb S$, given the way they are defined in the proof of \cref{thm:weylrep-intertwin}.

\subsection{Irreducible representations of $(\centL T)_0$}
\label{subsec:unicol-irrepsofidcomp}

In \cref{subsec:unicol-Vtirreps} we made $V\liealg t$ and $\centV\liealg t$ into topological groups. By giving $T \oplus V\liealg t$ and $T \times \centV\liealg t$ the product topologies, $(LT)_0$ and $(\centL T)_0$ become topological groups as well through the decompositions~\eqref{eq:unicol-decompidcomp} and~\eqref{eq:unicol-decompidcomp-centext}, respectively.

The representation $W$ of $\centV \liealg t$ now allows us to easily build irreducible, positive energy representations of $(\centL T)_0$. We namely construct for every character $l \in \Hom_\LieGp(T, \phasegp)$ a representation $W_l$ of $(\centL T)_0$ on the Hilbert space tensor product $\hilb S_l := \CC_l \otimes \hilb S$, where $\CC_l$ denotes a copy of $\CC$, as follows. Let $(\gamma, z) \in (\centL T)_0$ and consider its image
\[
\bigl(\avg \gamma, (\xi - \avg \xi, z)\bigr) \in T \times \centV \liealg t
\]
under the isomorphism~\eqref{eq:unicol-decompidcomp-centext}. Here, $\xi\colon S^1 \to \liealg t$ is any choice of lift of $\gamma$. Then make $\avg \gamma$ act on $\CC_l$ via $l$, and let $(\xi - \avg \xi, z)$ act on $\hilb S$ via $W$. In other words, $W_l$ is the tensor product representation of $l$ and $W$. It is irreducible because $l$ and $W$ are.

\begin{rmk}[Characters of $T$ as elements of the dual lattice of $\Lambda$]
We will often use the isomorphism of $\ZZ$-modules
\[
\Hom_\LieGp\bigl(T, \phasegp\bigr) \xrightarrow{\sim} \Lambda^\vee := \Hom_\Ab(\Lambda, \ZZ).
\]
It is given by first differentiating a character $l\colon T \to \phasegp$ and then restricting to $\Lambda \subseteq \liealg t := \Lambda \otimes_\ZZ \RR$. The inverse isomorphism takes the tensor product of a homomorphism $f\colon \Lambda \to \ZZ$ with $\phasegp$. We will use this identification to give meaning to expressions like $\langle l, l \rangle$, or $\langle l, \alpha \rangle$ if $\alpha \in \liealg t$.
\end{rmk}

Next, we equip such a representation $W_l$ of $(\centL T)_0$ with an intertwining rotation action. Unlike what we did for $\hilb S$, we will not let $\Rot(S^1)$ itself act on $\hilb S_l$, but a certain finite cover of $\Rot(S^1)$ instead.\footnote{See the beginning of \cref{app:posenergycond} for the models and corresponding notations we use for the covering groups of $\Rot(S^1)$.} Let $m$ be the smallest positive integer such that $m \langle l,l \rangle \in 2\ZZ$. We define a representation $R_l$ of $\Rot^{(m)}(S^1)$ on $\hilb S_l$ through the character $[\Phi_\theta] \mapsto e^{\pi i \langle l,l \rangle \theta}$ on $\CC_l$ and the action $R$ of $\Rot(S^1)$ on $\hilb S$ via the covering homomorphism $\Rot^{(m)}(S^1) \twoheadrightarrow \Rot(S^1)$. That is, $R_l$ is the tensor product representation of this character with $R$. Then, because $R$ intertwines in the manner~\eqref{eq:unicol-rotaction-intertwineheis} with $W$, it can be checked that also $R_l$ does so with $W_l$. To prove this one uses that if $\xi\colon S^1 \to \liealg t$ is a lift of $\gamma \in (LT)_0$, then $[\Phi_\theta]^*\xi$ is a lift of $[\Phi_\theta]^*\gamma$.

There is a Hilbert space isomorphism
\[
f_l\colon \hilb S_l \xrightarrow{\sim} \hilb S, \qquad 1 \otimes v \mapsto v,
\]
which obviously intertwines the representations $W_l|_{\centV \liealg t}$ and $W$ of $\centV \liealg t$, but does not do so for $R_l$ and $R$. Instead, for $\Phi_\theta \in \Rot^{(\infty)}(S^1)$ the following square commutes:\footnote{In writing $R[\Phi_\theta]$, the notation $[\Phi_\theta]$ stands for the image of $\Phi_\theta$ in $\Rot(S^1)$, while $[\Phi_\theta]$ in $R_l[\Phi_\theta]$ means its image in $\Rot^{(m)}(S^1)$.}
\[
\begin{tikzcd}
\hilb S_l \ar[r, "f_l", "\sim"'] \ar[d, "R_l{[\Phi_\theta]}"'] &[3em] \hilb S \ar[d, "R{[\Phi_\theta]}"] \\
\hilb S_l \ar[r, "e^{-\pi i \langle l, l \rangle \theta} f_l"', "\sim"] & \hilb S.
\end{tikzcd}
\]
Therefore, $f_l$ restricts for each $a \in (1/m)\ZZ$ on the energy eigenspace $\hilb S(a)$ to an isomorphism
\[
f_l|_{\hilb S_l(a)}\colon \hilb S_l(a) \xrightarrow{\sim} \hilb S\bigl(a - \langle l, l \rangle /2 \bigr)
\]
of finite-dimensional Hilbert spaces. That is, the energy of $\hilb S_l$ is shifted from that of $\hilb S$ by a fractional value. Because $\hilb S(a - \langle l,l \rangle/2) = \{0\}$ for $a < \langle l,l \rangle/2$ and there holds $\langle l, l \rangle \geq 0$ by the positive definiteness of $\Lambda$, this shows

\begin{prop}
\label{thm:unicol-posenergyidcomp}
The intertwining $\Rot^{(m)}(S^1)$-action $R_l$ on the representation $W_l$ of $(\centL T)_0$ is of positive energy.
\end{prop}

Using the isomorphisms $f_l|_{\hilb S_l(a)}$ we can even calculate the character of $R_l$ knowing that of $R$:
\begin{align*}
\ch_{R_l}(q) &= \sum_{\mathclap{a \in (1/m)\ZZ}} \dim \bigl(\hilb S_l(a)\bigr) q^a \\
	&= \sum_{\mathclap{a \in \langle l, l \rangle/2 + \ZZ_{\geq 0}}} \dim \Bigl(\hilb S\bigl(a - \langle l,l \rangle/2\bigr)\Bigr) q^a \\
	&= \sum_{k = 0}^\infty \dim \bigl(\hilb S(k)\bigr) q^{k + \langle l, l \rangle/2} \\
	&= q^{\langle l, l \rangle/2} \ch_R(q) \\
	&= q^{\dim T/24} q^{\langle l, l \rangle/2} \eta(q)^{-\dim T}.
\end{align*}
Just as we did for the triple $(W, \hilb S, R)$ in~\eqref{eq:partfuncfocksp}, this result becomes nicer if we shift the energy downwards a bit:
\[
Z_{W_l}(q) := q^{-\dim T/24} \ch_{R_l}(q) = q^{\langle l, l \rangle/2} \eta(q)^{-\dim T}.
\]

These representations $W_l$ are clearly mutually non-isomorphic since this holds for their restrictions to $T$. What is more, all irreducible representations of $(\centL T)_0$ of the type we are concerned with are of this form. More precisely,

\begin{thm}
\label{thm:unicol-classifyidcompirreps}
Every irreducible, positive energy representation of $(\centL T)_0$ such that the central subgroup $\phasegp$ acts  as $z \mapsto z$ is isomorphic to $W_l$ for some character $l$ of $T$.
\end{thm}
\begin{proof}
Let $Q$ be such a representation on a Hilbert space $\hilb H$ and consider it as a representation of $T \times \centV \liealg t$ via the isomorphism~\eqref{eq:unicol-decompidcomp-centext}. Because $Q|_T$ commutes with $Q$, $Q|_T$ is a character, say, $l$, by Schur's lemma. This implies that $Q|_{\centV \liealg t}$ is irreducible. A closed linear subspace of $\hilb H$ which is stable under $\centV \liealg t$ then namely is stable also under $T \times \centV \liealg t$. Furthermore, $Q|_{\centV \liealg t}$ is again of positive energy. Let $R'$ namely be a positive energy representation of $\Rot^{(n)}(S^1)$, for some $n$, which intertwines in the manner~\eqref{eq:rotS1intertwin} with $Q$, now considered as a representation of $(\centL T)_0$ again. As we already noted in \cref{subsec:unicol-diffS1action-subgps}, the isomorphism~\eqref{eq:unicol-decompidcomp-centext} is $\Rot(S^1)$-equivariant. Hence $R'$ also intertwines with $Q|_{\centV \liealg t}$. By the unicity result \cref{thm:unicol-classifyVt-irreps} for $\centV \liealg t$, this proves that $Q|_{\centV \liealg t}$ is isomorphic to $W$.
\end{proof}

Having discussed these representations of $(\centL T)_0$ and $\Rot^{(m)}(S^1)$ on $\hilb S_l$, we finally demonstrate a third piece of structure on the collection of all these Hilbert spaces.

Recall from \cref{subsec:unicol-Vtirreps} that there is a representation $U$ of the group $\Aut(\liealg t; \langle \cdot, \cdot \rangle)$ of Lie algebra isometries on $\hilb S$ which intertwines in a certain way with $W$. Because we learned in \cref{subsec:unicol-isometryactions-subgps} that the smaller group $\Aut(\Lambda; \langle \cdot, \cdot \rangle)$ of lattice automorphisms acts on $(\centL T)_0$, one might intially expect a representation of $\Aut(\Lambda; \langle \cdot, \cdot \rangle)$ on $\hilb S_l$ to exist which similarly intertwines with $W_l$. However, this turns out not to be the case since $\Aut(\Lambda; \langle \cdot, \cdot \rangle)$ translates the parameter $l$ also, as noted in \cref{subsec:duallattice-discgp}. The situation is instead described by

\begin{prop}
\label{eq:unicol-isometry-idcomp}
Let $l$ be a character of $T$ and $g$ a lattice automorphism of $\Lambda$. Then the Hilbert space isomorphism
\[
U_l(g)\colon \hilb S_l \xrightarrow{\sim} \hilb S_{g \cdot l},
	\qquad
1 \otimes v \mapsto 1 \otimes U(\RR g)(v).
\]
satisfies the intertwining properties
\begin{equation}
\label{eq:unicol-Ul-intertwin}
U_l(g) W_l(\gamma, z) U_l(g)^* = W_{g\cdot l}\bigl(g \cdot (\gamma, z)\bigr)
\end{equation}
for all $(\gamma, z) \in (\centL T)_0$, and
\begin{equation}
\label{eq:unicol-Ul-intertwin-rot}
U_l(g) R_l[\Phi_\theta] U_l(g)^* = R_{g \cdot l}[\Phi_\theta]
\end{equation}
for all $[\Phi_\theta] \in \Rot^{(m)}(S^1)$.
\end{prop}

Note that the smallest positive integer $m$ such that $m\langle l, l \rangle \in 2\ZZ$ is also the smallest positive integer $m$ such that $m\langle g \cdot l, g \cdot l \rangle \in 2\ZZ$ because $\langle g \cdot l, g \cdot l \rangle = \langle l, l \rangle$. Hence, $R_l$ and $R_{g \cdot l}$ are both representations of the same covering group $\Rot^{(m)}(S^1)$. Moreover, they act by the same character $[\Phi_\theta] \mapsto e^{\pi i \langle l, l \rangle\theta} = e^{\pi i \langle g\cdot l, g\cdot l \rangle\theta}$ on the tensor factors $\CC_l$ and $\CC_{g \cdot l}$ of $\hilb S_l$ and $\hilb S_{g \cdot l}$, respectively.

\begin{proof}
On the one hand, there holds
\begin{align*}
U_l(g) W_l(\gamma, z) U_l(g)^* (1 \otimes v)
	&= U_l(g) W_l(\gamma, z) \bigl(1 \otimes U(\RR g)^*(v)\bigr) \\
	&= U_l(g) \bigl(l(\avg \gamma) \otimes W(\xi - \avg \xi, z) U(\RR g)^*(v)\bigr) \\
	&= l(\avg \gamma) \otimes U(\RR g)W(\xi - \avg \xi, z) U(\RR g)^*(v),
\end{align*}
and thanks to the intertwining property~\eqref{eq:unicol-liealgisomets-intertwin} of $U$ with $W$ we can write
\begin{align*}
U(\RR g)W(\xi - \avg \xi, z) U(\RR g)^*
	&= W\bigl((\RR g) \cdot (\xi - \avg \xi, z)\bigr) \\
	&= W\bigl((\RR g)_*\xi - \avg (\RR g)_*\xi, z\bigr).
\end{align*}
On the other hand,
\begin{align*}
W_{g \cdot l}\bigl(g \cdot (\gamma, z)\bigr)(1 \otimes v)
	&= W_{g \cdot l}(g_*\gamma, z)(1 \otimes v) \\
	&= (g \cdot l)(\avg g_*\gamma) \otimes W\bigl((\RR g)_*\xi - \avg (\RR g)_*\xi, z\bigr),
\end{align*}
because if $\xi\colon S^1 \to \liealg t$ is a lift of $\gamma \in (\centL T)_0$, then $(\RR g)_*\xi$ is a lift of $g_*\gamma$. Now observe that $(g \cdot l)(\avg g_*\gamma) = l(\avg \gamma)$. This proves~\eqref{eq:unicol-Ul-intertwin}.

The intertwining property~\eqref{eq:unicol-Ul-intertwin-rot} is now easily seen to be equivalent to $U$ commuting with $R$. This is, in turn, a fact we already noted in \cref{subsec:unicol-Vtirreps}.
\end{proof}

\subsection{Irreducible representations of $\centL T$}
\label{subsec:unicol-irrepsLT}

Having given $(LT)_0$ the structure of a topological group in \cref{subsec:unicol-irrepsofidcomp}, the full group $LT$ now acquires by \cref{thm:topgpfromsubgp} a unique structure of a topological group such that $(LT)_0$ is open in $LT$.

\begin{prop}
\label{thm:unicol-topgpstruct}
There exists a unique structure of a topological group on the central extension $\centL T$ such that $(\centL T)_0$ is open in $\centL T$.
\end{prop}
\begin{proof}
We employ \cref{thm:centexttopgp} with $G := LT$, $A := \phasegp$ and $G_0 := (LT)_0$. What needs to be checked is whether for every fixed loop $\gamma \in LT$ the map $(LT)_0 \to \phasegp$ given by $\rho \mapsto c(\gamma, \rho) c(\gamma+\rho, -\gamma)$ is continuous.

Because of the way the topology on $(LT)_0$ is defined (namely as a product topology) it is convenient to work with elements of $LT$ via the right hand side of~\eqref{eq:unicol-decompfullgp}. That is, we write $\gamma$ as a triple
\[
\bigl([\xi_0], \xi, \Delta_\gamma\bigr) \in T \oplus V\liealg t \oplus \Lambda.
\]
Here, $[\xi_0] \in T \cong \liealg t/\Lambda$ is the equivalence class of an element $\xi_0 \in \liealg t$. Similarly, we will work with a triple $([\eta_0], \eta, 0)$ in the place of $\rho$. Lifts $[0,1] \to \liealg t$ of $\gamma$ and $\rho$ are then given by $\theta \mapsto \xi(\theta) + \xi_0 + \Delta_\gamma\theta$ and $\theta \mapsto \eta(\theta) + \eta_0$, respectively.

Filling in these lifts into $c(\gamma, \rho) c(\gamma+\rho, -\gamma)$ and using that $\Delta_\rho = 0$, we quickly realise that we need to verify whether the real number
\begin{equation}
\label{eq:unicol-contcocycle1}
\begin{multlined}
\frac12 \int_0^1 \bigl\langle \xi'(\theta) + \Delta_\gamma, \eta(\theta) + \eta_0\bigr\rangle \dd\theta + \frac12 \bigl\langle \Delta_\gamma, \eta(0) + \eta_0\bigr\rangle - \phantom{{}}\\
\frac12 \int_0^1 \bigl\langle \xi'(\theta) + \eta'(\theta) + \Delta_\gamma, \xi(\theta) + \xi_0 + \Delta_\gamma \theta\bigr\rangle \dd\theta - \frac12\langle\Delta_\gamma, \Delta_\gamma \rangle
\end{multlined}
\end{equation}
depends continuously on $\eta$ and $\eta_0$. The first term in~\eqref{eq:unicol-contcocycle1} is equal to
\[
\frac12 \int_0^1 \bigl\langle \xi'(\theta), \eta(\theta)\bigr\rangle \dd\theta + \frac12 \langle \Delta_\gamma, \eta_0 \rangle = S(\xi,\eta) + \frac12 \langle \Delta_\gamma, \eta_0 \rangle
\]
and we already noted in \cref{subsec:unicol-Vtirreps} that the restriction of $S$ to $V\liealg t$ is continuous. The third term in~\eqref{eq:unicol-contcocycle1} can be expanded to
\begin{equation}
\label{eq:unicol-contcocycle2}
\begin{multlined}
\frac12 \int_0^1 \bigl\langle \xi'(\theta) + \Delta_\gamma, \xi(\theta) + \xi_0 + \Delta_\gamma\theta\bigr\rangle \dd\theta + \phantom{{}} \\
	\frac12 \int_0^1 \bigl\langle \eta'(\theta), \xi(\theta)\bigr\rangle \dd\theta +
	\frac12\Bigl[\bigl\langle\eta(\theta), \Delta_\gamma\theta\bigr\rangle\Bigr]_0^1
\end{multlined}
\end{equation}
using partial integration. Now note that the third term in~\eqref{eq:unicol-contcocycle2} is equal to $\frac12 \langle \eta(0), \Delta_\gamma\rangle$ and hence cancels against the same term in~\eqref{eq:unicol-contcocycle1}. This concludes the proof.
\end{proof}

With the irreducible, positive energy representations $W_l$ of the identity component $(\centL T)_0$ in hand, we will be able to construct and classify the same class of representations of the full group $\centL T$, starting as follows. Let us take such a $W_l$ for a character $l$ of $T$ and consider the induced representation
\[
\Ind_{(\centL T)_0}^{\centL T} W_l
\]
of $\centL T$, which we will shorten to $\Ind W_l$.

We refer to \cref{subsec:indreps} for generalities regarding induced group representations, and content ourselves here with spelling out some details of this particular one. The underlying Hilbert space of $\Ind W_l$ is the completion of an algebraic direct sum indexed over the left cosets $\sigma$ of $(\centL T)_0$ in $\centL T$:
\begin{equation}
\label{eq:unicol-defIndSl}
\Ind_{(\centL T)_0}^{\centL T} \hilb S_l := \overline{%
	\bigoplus_{%
		\mathclap{%
			\sigma \in \centL T/(\centL T)_0%
		}
	}
	\hilb S_l^\sigma%
}~.
\end{equation}
We will abbreviate it as $\Ind \hilb S_l$. Here, the summand $\hilb S_l^\sigma$ is the Hilbert space
\[
\hilb S_l^\sigma := \sigma \times_{(\centL T)_0} \hilb S_l.
\]
A vector of $\hilb S_l^\sigma$ is an equivalence class $[(\gamma, z), v]$ of pairs $((\gamma, z), v)$ with $(\gamma, z) \in \sigma$ and $v \in \hilb S_l$, and the relations are given by
\[
\bigl((\gamma, z) \cdot (\rho, w), v\bigr) \sim \bigl((\gamma, z), W_l(\rho, w)(v)\bigr)
\]
for $(\rho, w) \in (\centL T)_0$. A general vector of $\Ind \hilb S_l$ is an infinite tuple of vectors
\begin{equation}
\label{eq:unicol-vectorIndSl}
\Bigl(\bigl[(\gamma^\sigma, z^\sigma), v^\sigma\bigr]\Bigr)_{\sigma \in \centL T/(\centL T)_0},
\end{equation}
where $[(\gamma^\sigma, z^\sigma), v^\sigma] \in \hilb S_l^\sigma$, satisfying a square-integrability condition. The action $\Ind W_l$ of $\centL T$ on $\Ind \hilb S_l$ is given by setting for an element $(\gamma, z) \in \centL T$
\[
(\Ind W_l)(\gamma, z) \cdot \Bigl(\bigl[(\gamma^\sigma, z^\sigma), v^\sigma\bigr]\Bigr)_\sigma := \biggl(\biggl[(\gamma,z) \cdot \Bigl(\gamma^{(\gamma, z)^{-1}\sigma}, z^{(\gamma, z)^{-1}\sigma}\Bigr), v^{(\gamma, z)^{-1}\sigma}\biggr]\biggr)_\sigma.
\]

If we again let $m$ be the smallest positive integer such that $m\langle l, l \rangle \in 2\ZZ$, then there is an action, which we denote by $\Ind R_l$, of $\Rot^{(m)}(S^1)$ on $\Ind \hilb S_l$ affecting each summand $\hilb S_l^\sigma$ individually. Namely, if $\Phi_\theta \in \Rot^{(\infty)}(S^1)$ and $[(\gamma, z), v] \in \hilb S_l^\sigma$, then we set
\[
(\Ind R_l)[\Phi_\theta] \cdot \bigl[(\gamma, z), v\bigr] := \bigl[[\Phi_\theta] \cdot (\gamma, z), R_l[\Phi_\theta](v)\bigr].
\]
It can be checked that $\Ind R_l$ is well-defined, linear, unitary and that it satisfies the intertwining property~\eqref{eq:rotS1intertwin} with respect to $\Ind W_l$ because $R_l$ does so with $W_l$. The reason why this $\Rot^{(m)}(S^1)$-action on $\Ind \hilb S_l$ is the correct one to take, given the one we already defined on $\hilb S_l$, is that it is the one which appears when we consider $\Ind W_l$ as a representation induced up from the semidirect product $(\centL T)_0 \rtimes \Rot^{(m)}(S^1)$ to $\centL T \rtimes \Rot^{(m)}(S^1)$.

Note that the pre-image in $\centL T$ of any subgroup of $LT$ is a normal subgroup. In particular, $(\centL T)_0$ is normal in $\centL T$, being the pre-image of the identity component $(LT)_0$ of $LT$. Therefore, the representations conjugate to $W_l$ are again representations of $(\centL T)_0$. We start our study of $\Ind W_l$ by calculating these conjugate representations:

\begin{lem}
\label{thm:unicol-conjrep}
Let $(\gamma, z)$ be an element of $\centL T$ that is not contained in the (normal) subgroup $(\centL T)_0$ and consider the representation $W_l^{(\gamma, z)}$ of $(\centL T)_0$ conjugate to $W_l$, defined by
\[
W_l^{(\gamma, z)}(\rho, w) := W_l\bigl((\gamma, z)^{-1} (\rho, w) (\gamma, z)\bigr)
\]
for $(\rho, w) \in (\centL T)_0$. Then $W_l^{(\gamma, z)}$ is the tensor product representation of $W_l$ and the character
\begin{equation}
\label{eq:unicol-centextchar}
(\centL T)_0 \twoheadrightarrow (LT)_0 \to \phasegp, \qquad (\rho, w) \mapsto \rho \mapsto c(\rho, \gamma) c(\gamma, \rho)^{-1},
\end{equation}
where $\phasegp$ acts on $\CC$ as $z \mapsto z$. In turn, for any lifts $\xi, \eta\colon [0,1] \to \liealg t$ of $\gamma$ and $\rho$, respectively, there holds
\begin{equation}
\label{eq:unicol-conjrep-1}
c(\rho, \gamma) c(\gamma, \rho)^{-1} = e^{2\pi i\bigl(S(\eta, \xi) - S(\xi, \eta)\bigr)},
\end{equation}
where
\begin{equation}
\label{eq:unicol-conjrep-2}
S(\eta, \xi) - S(\xi, \eta) = - \int_0^1 \bigl\langle \xi'(\theta), \eta(\theta) \bigr\rangle \dd\theta.
\end{equation}
\end{lem}
\begin{proof}
Using the definition of the multiplication in $\centL T$ in terms of the cocycle $c$ and the standing assumption that $c$ is normalised, the conjugated element $(\gamma, z)^{-1} (\rho, w) (\gamma, z)$ can be simplified and the first claim follows easily.

Since $\rho \in (LT)_0$, we have $\Delta_\rho = 0 \in \Lambda$ and because the cocycle $\epsilon$ for $\tilde\Lambda$ is assumed to be normalised there holds $\epsilon(\Delta_\gamma, \Delta_\rho) = \epsilon(\Delta_\rho, \Delta_\gamma) = 1$. The equations~\eqref{eq:unicol-conjrep-1} and~\eqref{eq:unicol-conjrep-2} now follow from the expression~\eqref{eq:unicol-commmap} for $S(\eta, \xi) - S(\xi, \eta)$.
\end{proof}

\begin{prop}
\label{thm:unicol-indrepisirrep}
The induced representation $\Ind W_l$ of $\centL T$ is irreducible.
\end{prop}
\begin{proof}
Since $W_l$ is irreducible and $(\centL T)_0$ is a normal subgroup of $\centL T$, it is by Mackey's irreducibility criterion \cref{thm:mackeyirredcrit} sufficient to show that all the conjugate representations $W_l^{(\gamma, z)}$ as in \cref{thm:unicol-conjrep} are not isomorphic to $W_l$. To do this, we will examine the restriction of $W_l^{(\gamma, z)}$ to the subgroup of $(\centL T)_0$ consisting of the elements of the form $(\rho, 1)$ where $\rho$ is a constant loop. Because this subgroup is canonically isomorphic to $T$, we will denote it as such. Say that $\rho = \exp \alpha$ for some $\alpha \in \liealg t$. According to~\eqref{eq:unicol-conjrep-2} we then have
\[
S(\eta, \xi) - S(\xi, \eta) = - \langle \Delta_\gamma, \alpha \rangle.
\]
So~\eqref{eq:unicol-centextchar} has $T$ acting by the character $-\Delta_\gamma$, which implies that $W_l^{(\gamma, z)}$ is letting $T$ act by $l - \Delta_\gamma$. Because $\Delta_\gamma \neq 0$, we have $l - \Delta_\gamma \neq l$ and therefore $W_l^{(\gamma, z)}$ and $W_l$ are not isomorphic.
\end{proof}

So we constructed a countably infinite family of irreducible representations of $\centL T$; one for every element $l$ of $\Hom_\LieGp(T, \phasegp) \cong \Lambda^\vee$. These are far from mutually non-isomorphic, though. It will turn out that they are partitioned into finitely many isomorphism classes. In order to prove this we will first determine the restriction of an induced representation $\Ind W_l$ back to the identity component $(\centL T)_0$.

Recall that $\hilb S_l$ and $\hilb S_{l'}$, where $l$ and $l'$ are characters of $T$, carry different representations of $(\centL T)_0$ but are identical as Hilbert spaces. We may therefore speak of the identity map of Hilbert spaces $\hilb S_l \xrightarrow{\sim} \hilb S_{l'}$, given by $v \mapsto v$. Note furthermore that $(\centL T)_0$ being a normal subgroup of $\centL T$ implies that the restriction of $\Ind W_l$ to $(\centL T)_0$ restricts to each subspace $\hilb S_l^\sigma$ for all $\sigma$. 

\begin{thm}[Restriction of $\Ind W_l$ from $\centL T$ to $(\centL T)_0$]
\label{thm:unicol-indrep-restrtoidcomp}
Fix a character $l$ of $T$, a lattice element $\lambda \in \Lambda$ and let $\sigma$ be the (left) coset of $(\centL T)_0$ in $\centL T$ consisting of all elements $(\gamma, z)$ such that $\gamma$ has winding element $\lambda$. Then the composite unitary map\footnote{The loop $\gamma_\lambda$ was defined in \cref{sec:unicol-struct}. It is the standard choice of representative of $\sigma$ given by the projection on $T$ of the path $[0,1] \to \liealg t$, $\theta \mapsto \theta\lambda$.}
\[
f_l^\sigma\colon \hilb S_l^\sigma \xrightarrow{\sim} \hilb S_l \xrightarrow{\sim} \hilb S_{l-\lambda}, \qquad \bigl[(\gamma_\lambda, 1), v\bigr] \mapsto v \mapsto v
\]
intertwines the representations $(\Res \Ind W_l)|_{\hilb S_l^\sigma}$ and $W_{l-\lambda}$ of $(\centL T)_0$ and the representations $\Ind R_l|_{\hilb S_l^\sigma}$ and $R_{l-\lambda}$ of $\Rot^{(m)}(S^1)$.
\end{thm}

We defined an action of $\Rot^{(m)}(S^1)$ on $\hilb S_l$ and it might not be immediately clear which action is meant, for this same integer $m$, on $\hilb S_{l-\lambda}$. We will explain this in the proof of this \namecref{thm:unicol-indrep-restrtoidcomp}.

\begin{proof}
To show the first claim of the \namecref{thm:unicol-indrep-restrtoidcomp}, we start by using \cref{thm:decomp-inducedrep} to observe that the first map in the composition $f_l^\sigma$ is an isomorphism from the restriction of $\Res \Ind W_l$ to $\hilb S_l^\sigma$ to the conjugate representation $W_l^{(\gamma_\lambda, 1)}$. The latter was calculated partially in \cref{thm:unicol-conjrep}. To obtain a more precise result we substitute the lift $\theta \mapsto \theta\lambda$ of $\gamma_\lambda$ into~\eqref{eq:unicol-conjrep-2}, where $(\rho, w) \in (\centL T)_0$ and $\eta\colon S^1 \to \liealg t$ is a lift of $\rho$. This gives
\begin{align*}
S(\eta, \xi) - S(\xi, \eta) = - \langle \lambda, \avg \eta \rangle
\end{align*}
and therefore,
\[
c(\rho, \gamma_\lambda) c(\gamma_\lambda, \rho)^{-1} = e^{-2\pi i \langle \lambda, \avg \eta \rangle}.
\]
The crucial observation now is that\footnote{On a pedantic note: on the left hand side of this equation we are considering $v$ as a vector in $\hilb S_l$, while on the right hand side we see it as lying in $\hilb S_{l-\lambda}$ again.\label{fn:unicol-rotaction-not}}
\[
e^{2\pi i \langle -\lambda, \avg \eta \rangle} \cdot W_l(\rho, w)(v) = W_{l - \lambda}(\rho, w)(v)
\]
for all $v$. We conclude that the second map in the composition $f_l^\sigma$ is an isomorphism from $W_l^{(\gamma_\lambda, 1)}$ to $W_{l-\lambda}$.

For the second claim of the \namecref{thm:unicol-indrep-restrtoidcomp}, observe that because $\langle l, \lambda \rangle \in \ZZ$ and $\langle \lambda, \lambda \rangle\in 2\ZZ$, the smallest positive integer $m$ such that $m\langle l, l \rangle \in 2\ZZ$ is also the smallest positive integer $m$ such that $m\langle l-\lambda, l-\lambda \rangle \in 2\ZZ$. Therefore, $\Rot^{(m)}(S^1)$ acts on both $\hilb S_l$ and $\hilb S_{l-\lambda}$. However, the two characters $[\Phi_\theta] \mapsto e^{\pi i \langle l, l\rangle \theta}$ and $[\Phi_\theta] \mapsto e^{\pi i \langle l - \lambda, l - \lambda\rangle \theta}$ by which we defined it to act on their respective tensor factors $\CC_l$ and $\CC_{l-\lambda}$ are different: they differ by the character $e^{2\pi i \langle l, \lambda \rangle\theta} e^{-\pi i \langle \lambda, \lambda \rangle \theta}$.

For $[\Phi_\theta] \in \Rot^{(m)}(S^1)$ and $[(\gamma_\lambda, 1), v] \in \hilb S_l^\sigma$ we have
\begin{align}
\label{eq:unicol-indrepcharcalc}
f_l^\sigma \Bigl((\Ind R_l)[\Phi_\theta] \cdot \bigl[(\gamma_\lambda, 1), v\bigr]\Bigr)
	&= f_l^\sigma \Bigl[[\Phi_\theta] \cdot (\gamma_\lambda, 1), R_l[\Phi_\theta](v)\Bigr] \nonumber \\
	&= f_l^\sigma \biggl[\Bigl([\Phi_\theta]^* \gamma_\lambda, d\bigl([\Phi_\theta], \gamma_\lambda\bigr)\Bigr), R_l[\Phi_\theta](v)\biggr].
\end{align}
Because by definition $\gamma_\lambda(\theta') = \exp(\lambda\theta')$ for all $\theta' \in [0,1]$, there holds $[\Phi_\theta]^*\gamma_\lambda = \gamma_\lambda + \exp(-\lambda\theta)$. That is, $[\Phi_\theta]$ shifts $\gamma_\lambda$ by the constant loop $\exp(-\lambda\theta) \in T$. Furthermore, we can calculate from the definition of $d$ in~\eqref{eq:unicol-diffs1action} that $d([\Phi_\theta], \gamma_\lambda) = e^{-\pi i \langle \lambda, \lambda \rangle\theta}$. Hence,
\[
\Bigl([\Phi_\theta]^* \gamma_\lambda, d\bigl([\Phi_\theta], \gamma_\lambda\bigr)\Bigr) = \bigl(\gamma_\lambda + \exp(-\lambda\theta), e^{-\pi i \langle \lambda, \lambda \rangle\theta}\bigr).
\]
In turn, because we can compute that $c(\gamma_\lambda, \exp(-\lambda\theta)) = e^{-2\pi i\langle \lambda, \lambda \rangle \theta}$, the above can be written as a product
\[
\bigl(\gamma_\lambda + \exp(-\lambda\theta), e^{-\pi i \langle \lambda, \lambda \rangle\theta}\bigr) = (\gamma_\lambda, 1) \cdot \bigl(\exp(-\lambda\theta), e^{\pi i \langle \lambda, \lambda \rangle \theta}\bigr).
\]
Filling this back into~\eqref{eq:unicol-indrepcharcalc} gives
\begin{multline*}
f_l^\sigma \Bigl((\Ind R_l)[\Phi_\theta] \cdot \bigl[(\gamma_\lambda, 1), v\bigr]\Bigr) \\
\begin{aligned}
	&= f_l^\sigma \Bigl[(\gamma_\lambda, 1) \cdot \bigl(\exp(-\lambda\theta), e^{\pi i \langle \lambda, \lambda \rangle \theta}\bigr), R_l[\Phi_\theta](v)\Bigr] \\
	&= f_l^\sigma \Bigl[(\gamma_\lambda, 1), W_l\bigl(\exp(-\lambda\theta), e^{\pi i \langle \lambda, \lambda \rangle \theta}\bigr) \cdot R_l[\Phi_\theta](v)\Bigr] \\
	&= W_l\bigl(\exp(-\lambda\theta), e^{\pi i \langle \lambda, \lambda \rangle \theta}\bigr) \cdot R_l[\Phi_\theta](v).
\end{aligned}
\end{multline*}
Now on the one hand, $W_l(\exp(-\lambda\theta), e^{\pi i \langle \lambda, \lambda \rangle \theta})$ acts as multiplication by the scalar
\[
e^{-2\pi i \langle l, \lambda \rangle \theta} e^{\pi i \langle \lambda, \lambda \rangle \theta},
\]
while, on the other hand, our earlier argument in this proof claimed that\footnote{See \cref{fn:unicol-rotaction-not} for a similar remark on notation.}
\[
R_l[\Phi_\theta](v) = e^{2\pi i \langle l, \lambda \rangle \theta} e^{-\pi i \langle \lambda, \lambda \rangle \theta} R_{l-\lambda}[\Phi_\theta](v).
\]
We therefore see that all scalar factors cancel against each other and we conclude that
\[
f_l^\sigma \Bigl((\Ind R_l)[\Phi_\theta] \cdot \bigl[(\gamma_\lambda, 1), v\bigr]\Bigr)
	= R_{l-\lambda}[\Phi_\theta](v). \qedhere
\]
\end{proof}

Summarising, \cref{thm:unicol-indrep-restrtoidcomp}, together with the winding element isomorphism $\centL T/(\centL T)_0 \xrightarrow{\sim} \Lambda$ describes how $\Ind W_l$ combined with the intertwining $\Rot^{(m)}(S^1)$-action breaks up into irreducible subrepresentations after restriction to $(\centL T)_0 \rtimes \Rot^{(m)}(S^1)$. We namely have a unitary isomorphism
\[
\bigoplus_{\lambda \in \Lambda} f_l^{\sigma_\lambda}\colon \Res_{\centL T}^{(\centL T)_0} \Ind_{(\centL T)_0}^{\centL T} W_l \xrightarrow{\sim} \bigoplus_{\lambda \in \Lambda} W_{l-\lambda}
\]
of representations\footnote{As usual, this direct sum of representations has by definition as underlying Hilbert space the completion of an algebraic direct sum of Hilbert spaces.} of $(\centL T)_0 \rtimes \Rot^{(m)}(S^1)$, where $\sigma_\lambda$ is the coset associated to $\lambda$ as in the statement of \cref{thm:unicol-indrep-restrtoidcomp}.

It is now purely formal to prove

\begin{thm}
\label{thm:unicol-indrepsisom}
Two representations $\Ind W_l$ and $\Ind W_{l'}$ of $\centL T$, where $l$ and $l'$ are characters of $T$, are (unitarily) isomorphic if and only if $l' = l-\lambda$ for some $\lambda \in \Lambda$.
\end{thm}
\begin{proof}
For the `if' claim, write $\sigma$ for the coset of $(\centL T)_0$ in $\centL T$ corresponding to $\lambda$ and let us postcompose the unitary $(\centL T)_0$-intertwiner $(f_l^\sigma)^{-1}$ from \cref{thm:unicol-indrep-restrtoidcomp} with the isometric $(\centL T)_0$-intertwining inclusion $\hilb S_l^\sigma \hookrightarrow \Ind \hilb S_l$. By imitating the proof of \cref{thm:indrep-functorial} it can be shown that $\Ind \hilb S_l$ is large enough for this composition to lift to a unique isometric $\centL T$-intertwiner from $\Ind \hilb S_{l-\lambda}$ to $\Ind \hilb S_l$. It is obviously non-zero and must therefore be a unitary isomorphism by Schur's lemma after we learned in \cref{thm:unicol-indrepisirrep} that $\Ind W_{l-\lambda}$ and $\Ind W_l$ are irreducible.

Conversely, assume that for two characters characters $l$ and $l'$ of $T$ there is an isomorphism $f\colon\Ind \hilb S_{l'} \xrightarrow{\sim} \Ind \hilb S_l$. Let $\sigma_0 := (\centL T)_0$---in other words: the coset in $\centL T$ consisting of all elements $(\rho, w)$ for which $\rho$ has winding element $0 \in \Lambda$. Then the restriction $f|_{\hilb S_{l'}^{\sigma_0}}\colon \hilb S_{l'}^{\sigma_0} \hookrightarrow \Ind \hilb S_l$ is a non-zero $(\centL T)_0$-intertwiner. Because we learned from \cref{thm:unicol-indrep-restrtoidcomp} that all the subspaces $\hilb S_l^\sigma$ of $\Ind \hilb S_l$ for different cosets $\sigma$ are mutually non-isomorphic and irreducible as representations of $(\centL T)_0$, by Schur's lemma we must have an isomorphism $f|_{\hilb S_{l'}^{\sigma_0}}\colon \hilb S_{l'}^{\sigma_0} \xrightarrow{\sim} \hilb S_l^\sigma$ for some coset $\sigma$. Say that $\sigma$ consists of the elements $(\gamma, z)$ for which $\gamma$ has winding element some $\lambda \in \Lambda$. What is more, we learned that $\hilb S_{l'}^{\sigma_0}$ and $\hilb S_l^\sigma$ are isomorphic to $\hilb S_{l'}$ and $\hilb S_{l-\lambda}$, respectively. Restriction to the subgroup $T$ of $(\centL T)_0$ now shows that $l' = l-\lambda$.
\end{proof}

Because $\Ind R_l$ preserves each of the subspaces $\hilb S_l^\sigma \subseteq \Ind \hilb S_l$, it makes sense to talk about their energy eigenspaces $\hilb S_l^\sigma(a)$ for $a \in (1/m)\ZZ$. They sum to the energy eigenspaces of $\Ind \hilb S_l$ (see~\eqref{eq:energyeigenspacesplits}):
\begin{equation}
\label{eq:unicol-indrepenergydecomp}
(\Ind \hilb S_l)(a) = \overline{\bigoplus_{\mathclap{\sigma \in \centL T/(\centL T)_0}} {\hilb S_l}^\sigma(a)}.
\end{equation}

We are now able to calculate the character of $\Ind R_l$. Let namely a fixed element $\lambda \in \Lambda$ and its corresponding coset $\sigma \subseteq \centL T$ be as in \cref{thm:unicol-indrep-restrtoidcomp}. Then that \namecref{thm:unicol-indrep-restrtoidcomp} gives us the character of the restriction of $\Ind R_l$ to $\hilb S_l^\sigma$ because we learned about that of $R_{l-\lambda}$ already in \cref{subsec:unicol-irrepsofidcomp}:\footnote{In the first following equation we put a Hilbert space as a subscript in $\ch_{\hilb S_l^\sigma}$, contrary to the usual notation, to emphasise that we are only considering the restriction of $\Ind R_l$ to $\hilb S_l^\sigma$.}
\[
\ch_{\hilb S_l^\sigma}(q) = \ch_{R_{l-\lambda}}(q) = q^{\rank \Lambda/24} q^{\langle l-\lambda, l-\lambda \rangle /2} \eta(q)^{-\rank \Lambda},
\]
since $\rank \Lambda = \dim T$. Next, we allow our fixed $\lambda$ and so also its corresponding coset $\sigma$ to vary in order to calculate the character of $\Ind R_l$ on all of $\Ind \hilb S_l$. Using the decomposition~\eqref{eq:unicol-indrepenergydecomp}, we get
\begin{align*}
\ch_{\Ind R_l}(q) &= \sum_{\lambda \in \Lambda} \ch_{\hilb S_l^\sigma}(q) \\
	&= q^{\rank \Lambda/24} \sum_{\lambda \in \Lambda} q^{\langle l-\lambda, l-\lambda \rangle /2} \eta(q)^{-\rank \Lambda} \\
	&= q^{\rank \Lambda/24} \theta_{l+\Lambda}(q) \cdot \eta(q)^{-\rank \Lambda},
\end{align*}
where $\theta_{l+\Lambda}(q)$ is the \defn{theta series} of the translated lattice $l+\Lambda$ that is the coset of $l$ in $\Lambda^\vee$ (see \cref{subsec:thetaseries}). This also proves

\begin{prop}
The intertwining $\Rot^{(m)}(S^1)$-action $\Ind R_l$ on the representation $\Ind W_l$ of $\centL T$ is of positive energy.
\end{prop}

Finally, it is again appropriate to apply a small energy correction by defining
\[
Z_{\Ind W_l}(q) := q^{-\rank \Lambda/24} \ch_{\Ind R_l}(q) = \theta_{l+\Lambda}(q) \cdot \eta(q)^{-\rank \Lambda}.
\]
Upon substituting $q$ by $e^{2\pi i z}$, for $z \in \CC$ with $\Im z > 0$, this is a meromorphic modular form of weight $0$ for some congruence subgroup of $\lieSL(2,\ZZ)$ (the full group $\lieSL(2,\ZZ)$ if $\Lambda$ is unimodular) \parencite[Exercise 8.12]{mason:vosandmodforms}.

\begin{thm}
\label{thm:unicol-classifyirreps}
Every irreducible, positive energy representation of $\centL T$ such that the central subgroup $\phasegp$ acts as $z \mapsto z$ is (unitarily) isomorphic to $\Ind W_l$ for some character $l$ of $T$.
\end{thm}
\begin{proof}
Let $Q$ be such a representation on a Hilbert space $\hilb H$ and restrict it to $(\centL T)_0$. Combining \cref{thm:gprep-findimisotcomp-containsirrep} and \cref{thm:irrepofcrossprod-irrepgp} tells us that $Q|_{(\centL T)_0}$ contains an irreducible, positive energy subrepresentation of $(\centL T)_0$. By the classification result \cref{thm:unicol-classifyidcompirreps} the latter must be isomorphic to $W_l$ for some character $l$ of $T$. Suppose we can show that the $(\centL T)_0$-intertwining inclusion $i\colon W_l \hookrightarrow Q|_{(\centL T)_0}$ lifts to a unique morphism $\Ind W_l \to Q$ of $\centL T$-representations. Because we learned in \cref{thm:unicol-indrepisirrep} that $\Ind W_l$ is irreducible and we assumed the same for $Q$, this morphism is then an isomorphism by Schur's lemma.

As explained in \cref{subsec:indreps}, the question of whether the lift of $i$ to $\Ind W_l$ exists comes down to asking whether for every vector of the form~\eqref{eq:unicol-vectorIndSl} of $\Ind \hilb S_l$ the series
\[
\sum_{\mathclap{\sigma \in \centL T/(\centL T)_0}} Q\bigl((\gamma^\sigma, z^\sigma)\bigr) i(v^\sigma)
\]
converges in $\hilb H$. Note that the representation $Q|_{(\centL T)_0}$ on the subspace 
\[
Q\bigl((\gamma^\sigma, z^\sigma)\bigr) \bigl(i(\hilb S_l)\bigr)
\]
of $\hilb H$ is isomorphic to the conjugate representation $W_l^{(\gamma^\sigma, z^\sigma)}$ of $(\centL T)_0$. We learned in \cref{thm:unicol-indrep-restrtoidcomp} that for different cosets $\sigma$ these conjugate representations are mutually non-isomorphic and hence orthogonal. The squared norm of the above series is therefore equal to $\sum_\sigma \norm{v^\sigma}_{\hilb S_l}^2$ and this converges by assumption.
\end{proof}

We conclude from \cref{thm:unicol-indrepsisom,thm:unicol-classifyirreps} that $\centL T$ possesses only finitely many irreducible, positive energy representations up to isomorphism if we fix the character by which $\phasegp$ acts. The isomorphism classes are labelled by the elements of the finite abelian \defn{discriminant group} $D_\Lambda := \Lambda^\vee/\Lambda$ of the lattice $\Lambda$. There is exactly one isomorphism class, represented by $\Ind W_0$, if and only if $\Lambda$ is \defn{unimodular}.

Moreover, \cref{thm:posrepcontainsirrep} tells us that arbitrary positive energy representations of $\centL T$ are now understood as well. They are the direct sums of the irreducible ones we constructed in this section.

\begin{rmk}[The representation theory in the case of an odd lattice]
If $\Lambda$ is an odd lattice, the construction of the representations $\Ind W_l$ is identical to the one above in the case that $\Lambda$ is even. Recall, however, from \cref{subsec:unicol-diffS1action-subgps,rmk:unicol-diff2s1action-odd} that if $\Lambda$ is odd, then $\Diff_+(S^1)$ does act on $(\centL T)_0$ because the cocycle defining this central extension is preserved by $\Diff_+(S^1)$, but only $\Diff_+^{(2)}(S^1)$ acts on $\centL T$. In particular, there is only a $\Rot^{(2)}(S^1)$-action on $\centL T$ which intertwines with a $\Rot^{(m)}(S^1)$-action $\Ind R_l$ on the Hilbert space $\Ind \hilb S_l$. Here, $m$ is the smallest integer such that both $m\langle l, l \rangle \in 2\ZZ$ and $m \geq 2$. The calculation of its graded character and its outcome is identical to the case in which $\Lambda$ is even, which therefore shows that this $\Rot^{(m)}(S^1)$-action $\Ind R_l$ is of positive energy. Unlike in the odd case, $\ch_{\Ind R_0}(q)$ is now a series in half-integral powers of $q$.

The grading~\eqref{eq:unicol-defIndSl} on $\Ind \hilb S_l$ over the cosets $\sigma$ of $(\centL T)_0$ in $\centL T$ now refines to a further $\ZZ/2\ZZ$-grading
\[
\Ind \hilb S_l = \overline{\bigoplus_{\text{$\sigma$ even}} \hilb S_l^\sigma} \oplus \overline{\bigoplus_{\text{$\sigma$ odd}} \hilb S_l^\sigma},
\]
where we call a coset $\sigma$ \defn{even} if $\sigma \subseteq \centL T(0)$ and \defn{odd} if $\sigma \subseteq \centL T(1)$. (See \cref{rmk:unicol-centextoddcase} for the definitions of the $\centL T(i)$.) Every coset is either even or odd. We denote the two subspaces in the above by $(\Ind \hilb S_l)_0$ and $(\Ind \hilb S_l)_1$ and call their vectors \defn{even} and \defn{odd}, respectively. By the definition of $\Ind W_l$ it is clear that even elements of $\centL T$ preserve the parity of vectors, while odd ones reverse them.
\end{rmk}

We close this chapter by lifting the structure on the Hilbert spaces $\hilb S_l$ described in \cref{eq:unicol-isometry-idcomp} to the induced spaces $\Ind \hilb S_l$.

The winding element isomorphism $\centL T/(\centL T)_0 \xrightarrow{\sim} \Lambda$ implies that the group $\Aut(\Lambda; \langle \cdot, \cdot \rangle)$ of automorphisms of the lattice $\Lambda$ acts on $\centL T/(\centL T)_0$ also. The translate $g\sigma$ of a coset $\sigma$ by an element $g \in \Aut(\Lambda; \langle \cdot, \cdot \rangle)$ consists, by definition, of all elements $(\gamma, z)$ such that $\gamma$ has winding element $g\lambda$ if $\sigma$ corresponds similarly to $\lambda$.

\begin{prop}
(Compare with \parencite[Lemma 1.5(2)]{shimakura:automsvoageneral}.) Let $l$ be a character of $T$ and $\tilde g \in \Aut(\tilde\Lambda^\epsilon; \langle \cdot, \cdot \rangle)$. Then the function
\[
(\Ind U_l)(\tilde g)\colon \Ind \hilb S_l \xrightarrow{\sim} \Ind \hilb S_{g \cdot l}
\]
given by
\begin{equation}
\label{eq:unicol-latticeisomaction-IndSl}
(\Ind U_l)(\tilde g) \cdot \Bigl(\bigl[(\gamma^\sigma, z^\sigma), v^\sigma\bigr]\Bigr)_\sigma
	:= \Bigl(\Bigl[\tilde g \cdot \bigl(\gamma^{g^{-1}\sigma}, z^{g^{-1}\sigma}\bigr), U_l(g)\bigl(v^{g^{-1}\sigma}\bigr)\Bigr]\Bigr)_\sigma
\end{equation}
is well-defined, linear, unitary and satisfies the intertwining properties
\begin{equation}
\label{eq:unicol-latticeisomaction-IndSl-intertwin}
(\Ind U_l)(\tilde g)(\Ind W_l)(\gamma, z)(\Ind U_l)(\tilde g)^* = (\Ind W_{g \cdot l})\bigl(\tilde g \cdot (\gamma, z)\bigr)
\end{equation}
for all $(\gamma, z) \in \centL T$, and
\begin{equation}
\label{eq:unicol-latticeisomaction-IndSl-intertwin-rot}
(\Ind U_l)(\tilde g) (\Ind R_l)[\Phi_\theta] (\Ind U_l)(\tilde g)^* = (\Ind R_{g \cdot l})[\Phi_\theta]
\end{equation}
for all $[\Phi_\theta] \in \Rot^{(m)}(S^1)$.
\end{prop}

We refer to the proof of \cref{eq:unicol-isometry-idcomp} for an explanation of why $\Ind R_{g \cdot l}$ is a representation of the same covering group $\Rot^{(m)}(S^1)$ as $\Ind R_l$ is.

\begin{proof}
Verifying well-definedness comes down to checking whether, if we choose for every coset $\sigma$ of $(\centL T)_0$ in $\centL T$ an arbitrary element $(\rho^\sigma, w^\sigma) \in (\centL T)_0$, implying that
\[
\Bigl(\bigl[(\gamma^\sigma, z^\sigma), v^\sigma\bigr]\Bigr)_\sigma
=
\Bigl(\bigl[(\gamma^\sigma, z^\sigma) \cdot (\rho^\sigma, w^\sigma)^{-1}, W_l(\rho^\sigma, w^\sigma)(v^\sigma)\bigr]\Bigr)_\sigma,
\]
then the left hand side of~\eqref{eq:unicol-latticeisomaction-IndSl} equals
\begin{equation}
\label{eq:unicol-latticeisomaction-IndSl-1}
(\Ind U_l)(\tilde g) \cdot \Bigl(\bigl[(\gamma^\sigma, z^\sigma) \cdot (\rho^\sigma, w^\sigma)^{-1}, W_l(\rho^\sigma, w^\sigma)(v^\sigma)\bigr]\Bigr)_\sigma.
\end{equation}
If we fill in the definition of $(\Ind U_l)(\tilde g)$ in~\eqref{eq:unicol-latticeisomaction-IndSl-1} this gives
\begin{multline*}
\Bigl(\Bigl[\tilde g \cdot \bigl(\gamma^{g^{-1}\sigma}, z^{g^{-1}\sigma}\bigr) \cdot g \cdot \bigl(\rho^{g^{-1}\sigma}, w^{g^{-1}\sigma}\bigr)^{-1}, U_l(g) W_l\bigl(\rho^{g^{-1}\sigma}, w^{g^{-1}\sigma}\bigr)\bigl(v^{g^{-1}\sigma}\bigr)\Bigr]\Bigr)_\sigma \\
\begin{aligned}
	&= \Bigl(\Bigl[\tilde g \cdot \bigl(\gamma^{g^{-1}\sigma}, z^{g^{-1}\sigma}\bigr), \phantom{{}} \\
	&\phantom{{} = \Bigl(\Bigl[{}} W_{g \cdot l}\bigl(g \cdot \bigl(\rho^{g^{-1}\sigma}, w^{g^{-1}\sigma}\bigr)\bigr)^* U_l(g) W_l\bigl(\rho^{g^{-1}\sigma}, w^{g^{-1}\sigma}\bigr)\bigl(v^{g^{-1}\sigma}\bigr)\Bigr]\Bigr)_\sigma,
\end{aligned}
\end{multline*}
and there indeed holds
\[
W_{g \cdot l}\Bigl(g \cdot \bigl(\rho^{g^{-1}\sigma}, w^{g^{-1}\sigma}\bigr)\Bigr)^* U_l(g) W_l\bigl(\rho^{g^{-1}\sigma}, w^{g^{-1}\sigma}\bigr) = U_l(g)
\]
thanks to the intertwining property~\eqref{eq:unicol-Ul-intertwin} of $U_l$, which settles the question.

Showing linearity and unitarity can similarly be reduced to an application of the intertwining property of $U_l$. We therefore turn to proving the intertwining relation~\eqref{eq:unicol-latticeisomaction-IndSl-intertwin} for $\Ind U_l$.

Applying the left hand side of~\eqref{eq:unicol-latticeisomaction-IndSl-intertwin} to a vector $([(\gamma^\sigma, z^\sigma), v^\sigma])_\sigma \in \Ind \hilb S_l$ gives
\begin{multline*}
(\Ind U_l)(\tilde g)(\Ind W_l)(\gamma, z) \cdot \Bigl(\bigl[\tilde g^{-1} \cdot (\gamma^{g\sigma}, z^{g\sigma}), U_l(g)^*(v^{g\sigma})\bigr]\Bigr)_\sigma \\
\begin{aligned}
	&= (\Ind U_l)(\tilde g) \cdot \Bigl(\Bigl[(\gamma, z) \cdot \tilde g^{-1} \cdot \bigl(\gamma^{g(\gamma, z)^{-1}\sigma}, z^{g(\gamma, z)^{-1}\sigma}\bigr), U_l(g)^*\bigl(v^{g(\gamma, z)^{-1}\sigma}\bigr)\Bigr]\Bigr)_\sigma \\
	&= \Bigl(\Bigl[\tilde g \cdot (\gamma, z) \cdot \tilde g^{-1} \cdot \bigl(\gamma^{g(\gamma, z)^{-1}g^{-1}\sigma}, z^{g(\gamma, z)^{-1}g^{-1}\sigma}\bigr), v^{g(\gamma, z)^{-1}g^{-1}\sigma}\Bigr]\Bigr)_\sigma \\
	&= \Bigl(\Bigl[\bigl(\tilde g \cdot (\gamma, z)\bigr) \cdot \bigl(\gamma^{g(\gamma, z)^{-1}g^{-1}\sigma}, z^{g(\gamma, z)^{-1}g^{-1}\sigma}\bigr), v^{g(\gamma, z)^{-1}g^{-1}\sigma}\Bigr]\Bigr)_\sigma.
\end{aligned}
\end{multline*}
The crucial insight is now that the coset $g(\gamma, z)^{-1}g^{-1}\sigma$ of $(\centL T)_0$ in $\centL T$ can equivalently be written as $(\tilde g \cdot (\gamma, z))^{-1} \sigma$. If $\sigma$ is namely the coset associated to a lattice element $\lambda \in \Lambda$, then $g(\gamma, z)^{-1}g^{-1}\sigma$ is associated to $g(g^{-1}\lambda - \Delta_\gamma)$, which can be simplified to $\lambda - g\Delta_\gamma$. And, indeed, the image of $(\tilde g \cdot (\gamma, z))^{-1}$ in $LT$ has winding element $-g\Delta_\gamma$.

The intertwining property~\eqref{eq:unicol-latticeisomaction-IndSl-intertwin-rot} has an analogous, but easier proof which uses \cref{thm:unicol-rotandisometry-commute} and the fact stated in \cref{eq:unicol-isometry-idcomp} that $U_l$ intertwines $R_l$ and $R_{g \cdot l}$.
\end{proof}

This result implies in particular that $\Aut(\tilde\Lambda^\epsilon; \langle \cdot, \cdot \rangle)$ acts on the Hilbert space $\Ind \hilb S_0$, intertwining with the representation $\Ind W_0$ in the same way as $\Ind R_0$ and commuting with $\Ind R_0$. This has been proved at the beginning of \parencite[Section 4]{dong:latticeconfnets} for a projective action of the larger group $\Diff_+(S^1)$ instead of only for the honest action of $\Rot(S^1)$ that we studied here.

\chapter{Bicoloured torus loop groups}
\label{chap:bicol}

In this \namecref{chap:bicol} we introduce and study our new notion of a bicoloured torus loop group. Since this is a generalisation of a unicoloured torus loop group as studied in \cref{chap:unicol}, the format of this \namecref{chap:bicol} will largely mirror that of the previous one. We assume that the reader is familiar with the notations from \cref{sec:mainresults} for the next two paragraphs.

We start by laying out the structure of a bicoloured torus loop group $L(T_\wh, H, T_\bl)$ in \cref{sec:bicol-struct} and comparing it to various related unicoloured torus loop groups. We do not regard it as a topological group until \cref{sec:bicol-reptheory}, but we will nevertheless speak freely in advance of its connected components. Two differences between that section and \cref{sec:unicol-struct} are that we introduce another group $P(H, (\Lambda_\wh - \Lambda_\bl)/\Gamma)$ as an auxiliary tool to analyse the group of actual interest $L(T_\wh, H, T_\bl)$, and that certain constructions, such as our description of an action of a covering group of $\Diff_+(S^1)$ on $L(T_\wh, H, T_\bl)$, are much more involved than in the unicoloured situation. All this material does not yet require $\Lambda_\wh$, $\Lambda_\bl$ and $\Gamma$ to be endowed with bi-additive forms.

Next, we construct in \cref{sec:bicol-centext} from the data of the lattices $\Lambda_\wh$, $\Lambda_\bl$ and $\Gamma$ a $\phasegp$-central extension $\centL(T_\wh, H, T_\bl)$ of $L(T_\wh, H, T_\bl)$ which we show to be disjoint-commutative. We also again relate it to the central extensions of unicoloured torus loop groups defined in \cref{chap:unicol}. The action of a covering group of $\Diff_+(S^1)$ on $L(T_\wh, H, T_\bl)$ is proved to lift to $\centL(T_\wh, H, T_\bl)$ in \cref{sec:bicol-centext-diffnS1action}, opening the door to the study of its positive energy representations. The irreducible such representations are finally classified and constructed in \cref{sec:bicol-reptheory}.

\section{Bicoloured torus loop groups and their structure}
\label{sec:bicol-struct}

Consider the following setup. Let
\begin{equation}
\label{eq:bicol-span}
\Lambda_\wh \xhookleftarrow{\pi_\wh} \Gamma \xhookrightarrow{\pi_\bl} \Lambda_\bl
\end{equation}
be a \defn{span} of free $\ZZ$-modules, by which we mean that $\Lambda_\wh$, $\Lambda_\bl$ and $\Gamma$ share the same finite rank, and that the $\pi_\whbl$ are injective module homomorphisms.\footnote{A statement containing the subscript $\wh/\bl$ should be read as being valid when this symbol is replaced everywhere in the sentence with either $\wh$ or $\bl$.} Write $T_\whbl := \Lambda_\whbl \otimes_\ZZ \phasegp$ and $H := \Gamma \otimes_\ZZ \phasegp$ for the three associated tori and $\liealg t_\whbl := \Lambda_\whbl \otimes_\ZZ \RR$ and $\liealg h := \Gamma \otimes_\ZZ \RR$ for their respective Lie algebras. The homomorphisms $\pi_\whbl$ induce surjective Lie group homomorphisms $\phasegp\pi_\whbl\colon H \twoheadrightarrow T_\whbl$ and Lie algebra isomorphisms $\RR\pi_\whbl\colon \liealg h \xrightarrow{\sim} \liealg t_\whbl$. Denote the closed left half of $S^1$ by $\circl$~, the closed right half by $\circr$ and their intersection $\{i, -i\}$ as $\circpq$~. From now on we will call the points $i$ and $-i$, $p$ and $q$ respectively.\footnote{This alludes to the fact that the particular cutting of $S^1$ into two arcs we use is not important.}

We record the following notations for later use:\footnote{The reason why we write $\Lambda_\wh - \Lambda_\bl$ instead of $\Lambda_\wh + \Lambda_\bl$ is that this notation suggests the correct way to map $\Lambda_\wh \oplus \Lambda_\bl$ into $\Lambda_\wh - \Lambda_\bl$, as we will see in \cref{thm:bicol-windeltshoms}.}
\begin{align*}
\Lambda_\wh \cap \Lambda_\bl &:= (\RR \pi_\wh)^{-1}(\Lambda_\wh) \cap (\RR \pi_\bl)^{-1}(\Lambda_\bl) \\
\Lambda_\wh - \Lambda_\bl &:= (\RR \pi_\wh)^{-1}(\Lambda_\wh) - (\RR \pi_\bl)^{-1}(\Lambda_\bl) \\
	&:= \bigl\{(\RR\pi_\wh)^{-1}(\lambda_\wh) - (\RR\pi_\bl)^{-1}(\lambda_\bl) \in \liealg h \bigm\vert \lambda_\wh \in \Lambda_\wh, \lambda_\bl \in \Lambda_\bl\bigr\}.
\end{align*}
These are both free $\ZZ$-submodules of $\liealg h$, of the same rank as $\Lambda_\whbl$ and $\Gamma$. There are inclusions
\[
\begin{tikzcd}[column sep=0.6em, row sep=0.5em]
	& & (\RR\pi_\wh)^{-1}(\Lambda_\wh) \ar[dr, phantom, "\subseteq", sloped, start anchor=south east, end anchor=north west] & &[0.6em] \\
\Gamma \ar[r, phantom, "\subseteq"] & \Lambda_\wh \cap \Lambda_\bl \ar[ur, phantom, "\subseteq", sloped, start anchor=north east, end anchor=south west] \ar[dr, phantom, "\subseteq", sloped, start anchor=south east, end anchor=north west] & & \Lambda_\wh - \Lambda_\bl \ar[r, hook] & \Gamma^\vee. \\
	& & (\RR\pi_\bl)^{-1}(\Lambda_\bl) \ar[ur, phantom, "\subseteq", sloped, start anchor=north east, end anchor=south west] & &
\end{tikzcd}
\]
Associated to the span~\eqref{eq:bicol-span}, we define $n$ to be the smallest positive integer such that $n(\Lambda_\wh - \Lambda_\bl) \subseteq \Gamma$.\label{dfn:bicol-integern}

\begin{dfn}
\label{dfn:bicol-gp}
The \defn{bicoloured torus loop group} $L(T_\wh, H, T_\bl)$ associated to the quintuple
\[
(\Lambda_\wh, \Gamma, \Lambda_\bl, \pi_\wh, \pi_\bl)
\]
is the abelian group of all triples $(\gamma_\wh, \gamma_\match, \gamma_\bl)$, where $\gamma_\wh\colon \circl \to T_\wh$ and $\gamma_\bl\colon \circr \to T_\bl$ are smooth maps and $\gamma_\match\colon \circpq \to H$ is a function\footnote{The subscript `m' stands for `matching datum' or `middle'.} such that the following diagram commutes:
\begin{equation}
\label{eq:bicol-commdiag}
\begin{tikzcd}
\circl \ar[d, "\gamma_\wh"'] &[1.5em] \ar[l, hook'] \circpq \ar[d, "\gamma_\match"] \ar[r, hook] &[1.5em] \circr \ar[d, "\gamma_\bl"] \\
T_\wh & \ar[l, twoheadrightarrow, "\phasegp\pi_\wh"] H \ar[r, two heads, "\phasegp\pi_\bl"'] & T_\bl
\end{tikzcd},
\end{equation}
and such that for all $k\geq 1$ the left and right derivatives of $\gamma_\wh$ and $\gamma_\bl$ respectively at $p$ and $q$ agree:
\begin{equation}
\label{eq:bicol-derivatives-cond}
\begin{split}
(\RR\pi_\wh)^{-1}\bigl(\gamma_\wh^{(k)}(p)\bigr) = (\RR\pi_\bl)^{-1}\bigl(\gamma_\bl^{(k)}(p)\bigr),\\
(\RR\pi_\wh)^{-1}\bigl(\gamma_\wh^{(k)}(q)\bigr) = (\RR\pi_\bl)^{-1}\bigl(\gamma_\bl^{(k)}(q)\bigr).
\end{split}
\end{equation}
(In this notation we silently identify tangent spaces of the tori $T_\whbl$ with the Lie algebras $\liealg t_\whbl$ via the translation homomorphisms, so that we can say that $\gamma_\whbl^{(k)}(p), \gamma_\whbl^{(k)}(q) \in \liealg t_\whbl$.) Such triples $(\gamma_\wh, \gamma_\match, \gamma_\bl)$ are called \defn{bicoloured (torus) loops}.
\end{dfn}

In order to reduce clutter, the notation $L(T_\wh, H, T_\bl)$ does not mention the homomorphisms $\pi_\whbl$, although they are part of the data needed to construct the group.

Paralleling the above terminology, we will often refer to the elements of an ordinary torus loop group $LT$ as \defn{unicoloured (torus) loops}.

\subsection{Bicoloured loops as unicoloured loops with a discontinuity}
\label{subsec:bicol-discontunicolloops}
We now offer an alternative point of view on the group $L(T_\wh, H, T_\bl)$. Let $\gamma = (\gamma_\wh, \gamma_\match, \gamma_\bl)$ be a bicoloured loop. There is a unique lift $\hat\gamma_\whbl$ to $H$ of $\gamma_\whbl$ along $\phasegp\pi_\whbl$ such that $\hat\gamma_\whbl(p) = \gamma_\match(p)$. Glue these two lifts together at $p$ by defining a smooth map $\Pth(\gamma)\colon [0,1] \to H$ as
\[
\Pth(\gamma)(\theta) :=
\begin{cases}
\hat\gamma_\wh(\theta) & \text{if $\theta \in [1/2, 1]$,} \\
\hat\gamma_\bl(\theta) & \text{if $\theta \in [0, 1/2]$,}
\end{cases}
\]
where we used the unit speed parametrisations of $\circrdir$ and $\circldir$ by $[0, 1/2]$ and $[1/2, 1]$ respectively. Then, because $\Pth(\gamma)(1)$ and $\Pth(\gamma)(0)$ differ from $\gamma_\match(q)$ by an element of the subgroups $(\RR\pi_\wh)^{-1}(\Lambda_\wh)/\Gamma$ and $(\RR\pi_\bl)^{-1}(\Lambda_\bl)/\Gamma$ of $H$ respectively, $\Pth(\gamma)$ lies in the following group of paths on $H$:
\begin{equation}
\begin{split}
P\bigl(H, (\Lambda_\wh - \Lambda_\bl)/\Gamma\bigr) &:= \biggl\{\gamma \in C^\infty\bigl([0,1], H\bigr) \biggm\vert \gamma(1) - \gamma(0) \in \frac{\Lambda_\wh - \Lambda_\bl}{\Gamma}, \\
	&\phantom{{}:= \biggl\{\gamma \in C^\infty\bigl([0,1], H\bigr) \biggm\vert {}}  \gamma^{(k)}(1) = \gamma^{(k)}(0) \text{ for all $k \geq 1$}\biggr\}.
\end{split}
\label{eq:bicol-discontpaths}
\end{equation}
This gives a homomorphism
\[
\Pth\colon L(T_\wh, H, T_\bl) \to P\bigl(H, (\Lambda_\wh - \Lambda_\bl)/\Gamma\bigr), \qquad \gamma \mapsto \Pth(\gamma)
\]
of abelian groups. It is surjective because if, conversely, a path $\gamma$ lies in~\eqref{eq:bicol-discontpaths}, then $(\gamma_\wh, \gamma_\match, \gamma_\bl)$ is in $L(T_\wh, H, T_\bl)$, where $\gamma_\wh$ is defined as the composition $\phasegp\pi_\wh \circ \gamma|_{[1/2, 1]}$, $\gamma_\bl$ is $\phasegp\pi_\bl \circ \gamma|_{[0, 1/2]}$, $\gamma_\match(p) := \gamma(1/2)$ and $\gamma_\match(q)$ is an arbitrary element constructed as follows. Choose a decomposition $\gamma(1) - \gamma(0) \equiv \lambda_\wh - \lambda_\bl \bmod \Gamma$ for some $\lambda_\whbl \in (\RR \pi_\whbl)^{-1}(\Lambda_\whbl) \subseteq \liealg h$. Then $\gamma(1) - \lambda_\wh \equiv \gamma(0) - \lambda_\bl \bmod \Gamma$, so setting $\gamma_\match(q) := \gamma(1) - \lambda_\wh \bmod \Gamma$ makes $\gamma_\match$ fit into a commutative diagram~\eqref{eq:bicol-commdiag}.

Note that the datum $\gamma_\match(q)$ is not used by $\Pth$ and, moreover, that this is the only information about $\gamma$ that is thrown away. Said more precisely, the kernel of $\Pth$ is the finite abelian group consisting of triples $(\gamma_\wh, \gamma_\match, \gamma_\bl)$ where the $\gamma_\whbl$ are identically $0_{T_\whblsm}$, $\gamma_\match(p) = 0_H$ and $\gamma_\match(q) \in (\Lambda_\wh \cap \Lambda_\bl)/\Gamma$. We conclude that $L(T_\wh, H, T_\bl)$ is part of a short exact sequence of abelian groups
\begin{equation}
\label{eq:bicol-discontpaths-ses}
\begin{tikzcd}
0 \ar[r] & \displaystyle\frac{\Lambda_\wh \cap \Lambda_\bl}{\Gamma} \ar[r] & L(T_\wh, H, T_\bl) \ar[r, "\Pth"] &[1em] P\bigl(H, (\Lambda_\wh - \Lambda_\bl)/\Gamma\bigr) \ar[r] & 0.
\end{tikzcd}
\end{equation}

In other words, bicoloured loops are almost ordinary, unicoloured loops $S^1 \to H$, except for two differences. First, a bicoloured loop has a specific kind of discontinuity at the point $q \in S^1$, namely one which takes values in the finite abelian group $(\Lambda_\wh - \Lambda_\bl)/\Gamma \subseteq H$. Second, there is a small piece of extra data consisting of a `reference point' in $(\Lambda_\wh \cap \Lambda_\bl)/\Gamma \subseteq H$ with respect to which the two loose ends of the loop at the discontinuity jump.

\subsection{Support of a bicoloured loop}
\label{subsec:bicol-support}
Bicoloured loops have the following natural notion of support:

\begin{dfn}
\label{dfn:bicol-support}
The \defn{support} of a bicoloured loop $\gamma = (\gamma_\wh, \gamma_\match, \gamma_\bl) \in L(T_\wh, H, T_\bl)$ is the closed subset of $S^1$ given by
\[
\supp \gamma := \supp \gamma_\wh \cup \supp \gamma_\bl.
\]
\end{dfn}

Note that a path $\gamma$ in $P(H, (\Lambda_\wh - \Lambda_\bl)/\Gamma)$ has a notion of support also, namely the closure of the points on $S^1$ where it is not $0_H$. If such a path $\gamma$ has support in an interval, then any bicoloured loop in its pre-image under $\Pth$ has support there as well in the sense of \cref{dfn:bicol-support}, but the converse does not necessarily hold. It will be crucial in \cref{subsec:bicol-disjcomm} that, even if $\Gamma = \Lambda_\wh \cap \Lambda_\bl$ in which case $\Pth$ is an isomorphism, we use the \cref{dfn:bicol-support} of support for bicoloured loops.

\subsection{Unicoloured loop groups are a special case}
\label{subsec:bicol-unicolisspecialcase}
Consider the case where $\Lambda_\wh = \Lambda_\bl = \Gamma$ and the homomorphisms $\pi_\whbl\colon \Gamma \hookrightarrow \Lambda_\whbl$ are the identity. Then of course $T_\wh = T_\bl = H$ and the homomorphisms $\phasegp\pi_\whbl\colon H \twoheadrightarrow H$ are the identity. We claim that the bicoloured torus loop group $L(H, H, H)$ is then merely a complicated way to describe the unicoloured torus loop group $LH$. Indeed, in this case $(\Lambda_\wh \cap \Lambda_\bl)/\Gamma = 0$, so the homomorphism $\Pth$ in~\eqref{eq:bicol-discontpaths-ses} is an isomorphism, and we furthermore now have $P(H, (\Lambda_\wh - \Lambda_\bl)/\Gamma) \cong LH$.

Let us describe the identification more directly. There is a homomorphism of abelian groups $\Bi\colon LH \to L(H, H, H)$ given by $\gamma \mapsto (\gamma|_\circlsm, \gamma|_{\{p,q\}}, \gamma|_\circrsm)$. It can be thought of as `bicolourising' unicoloured loops in a canonical way. Its inverse $\Pth$ sends a bicoloured loop $\gamma = (\gamma_\wh, \gamma_\match, \gamma_\bl)$ to $\Pth(\gamma)$, where for $\theta \in S^1$
\[
\Pth(\gamma)(\theta) := 
\begin{cases}
\gamma_\wh(\theta) & \text{if $\theta \in \circl$~,} \\
\gamma_\bl(\theta) & \text{if $\theta \in \circr$~.}
\end{cases}
\]
This unicoloured loop $\Pth(\gamma)$ is well-defined at the points $\theta = p$ and $\theta = q$ because of the commutativity of the diagram~\eqref{eq:bicol-commdiag}. Smoothness at those two points is satisfied thanks to the conditions~\eqref{eq:bicol-derivatives-cond} on the (higher) derivatives of $\gamma_\wh$ and $\gamma_\bl$.

\subsection{The inclusion of $LH$}
\label{subsec:bicol-inclLH}
Generalising the observation made in \cref{subsec:bicol-unicolisspecialcase}, we may see unicoloured loops with values in $H$ as certain bicoloured loops in $L(T_\wh, H, T_\bl)$. Indeed, there is a homomorphism $\Bi\colon LH \to L(T_\wh, H, T_\bl)$ given by $\gamma \mapsto (\gamma_\wh, \gamma|_{\{p,q\}}, \gamma_\bl)$, where $\gamma_\wh$ is the composition $\phasegp\pi_\wh \circ \gamma|_\circlsm$ and $\gamma_\bl$ is $\phasegp\pi_\bl\circ \gamma|_\circrsm$. This homomorphism is injective.

Suppose namely that $\Bi(\gamma) = (0_{T_\wh}, 0_H, 0_{T_\bl})$. This means that the image of $\gamma|_\circlsm$ is in the kernel $(\RR\pi_\wh)^{-1}(\Lambda_\wh)/\Gamma \subseteq H$ of $\phasegp\pi_\wh$ and the image of $\gamma|_\circrsm$ is in the kernel $(\RR\pi_\bl)^{-1}(\Lambda_\bl)/\Gamma \subseteq H$ of $\phasegp\pi_\bl$. Since these kernels are finite abelian groups, $\gamma|_\circlsm$ and $\gamma|_\circrsm$ must both be constant. However, we also assumed that $\gamma|_{\{p,q\}}$ vanishes. So $\gamma|_\circlsm$ and $\gamma|_\circrsm$ are identically $0_H$ as well, proving our claim.

The image of $\Bi$ consists of those bicoloured loops $(\gamma_\wh, \gamma_\match, \gamma_\bl)$ for which the unique lifts $\hat\gamma_\whbl$ to $H$ of the $\gamma_\whbl$ that match at $p$ also match with each other, \emph{and} with $\gamma_\match$, at $q$.

\subsection{Isotony with respect to unicoloured loop groups}
\label{subsec:bicol-unicolisotony}
Write $L_\circlsm T_\wh$ for those loops in $LT_\wh$ which have support in $\circl$ and define $L_\circrsm T_\bl$ similarly. Then there are injective homomorphisms of abelian groups
\begin{equation}
\label{eq:bicol-noncentext-isot}
L_\circlsm T_\wh \hookrightarrow L(T_\wh, H, T_\bl) \hookleftarrow L_\circrsm T_\bl,
\end{equation}
where the homomorphism from $L_\circlsm T_\wh$ is given by $\gamma_\wh \mapsto (\gamma_\wh, 0_H, 0_{T_\bl})$ and the one from $L_\circrsm T_\bl$ by $\gamma_\bl \mapsto (0_{T_\wh}, 0_H, \gamma_\bl)$. These triples indeed make the diagram~\eqref{eq:bicol-commdiag} commute, since $\gamma_\whbl(p) = \gamma_\whbl(q) = 0_{T_\whblsm}$.

\subsection{The special case when $\Gamma = \Lambda_\wh \cap \Lambda_\bl$}
Write $H^\cap$ for the torus $(\Lambda_\wh \cap \Lambda_\bl) \otimes_\ZZ \phasegp$ associated to the $\ZZ$-module $\Lambda_\wh \cap \Lambda_\bl$. The two obvious injections $\Lambda_\wh \cap \Lambda_\bl \hookrightarrow \Lambda_\whbl$, which we will denote by $\pi_\whbl^\cap$, allow us to define in particular a bicoloured torus loop group $L(T_\wh, H^\cap, T_\bl)$. By the short exact sequence~\eqref{eq:bicol-discontpaths-ses} it is isomorphic to the group of paths
\[
P\bigl(H^\cap, (\Lambda_\wh - \Lambda_\bl)/(\Lambda_\wh \cap \Lambda_\bl)\bigr).
\]

An equivalent way to formulate this special feature in this case is to observe that the homomorphism
\[
\bigl(\phasegp\pi_\wh^\cap, \phasegp\pi_\bl^\cap\bigr) \colon H^\cap \to T_\wh \oplus T_\bl
\]
is injective. This implies that for maps $\gamma_\wh\colon \circl \to T_\wh$ and $\gamma_\bl\colon \circr \to T_\bl$ the existence of a function $\gamma_\match\colon \circpq \to H^\cap$ making the diagram~\eqref{eq:bicol-commdiag} (with $H^\cap$ in place of $H$) commute determines $\gamma_\match$ uniquely. That is, $\gamma_\match$ is not genuine extra data. Its existence is merely a property of the pair $(\gamma_\wh, \gamma_\bl)$.\footnote{This fact was already used in \cref{subsec:bicol-unicolisspecialcase}, where we considered a special case of the situation we are in now.} The definition of $L(T_\wh, H^\cap, T_\bl)$ therefore simplifies to
\[
L(T_\wh, H^\cap, T_\bl) =
\begin{multlined}[t]
\Bigl\{(\gamma_\wh, \gamma_\bl) \in C^\infty(\circl, T_\wh) \times C^\infty(\circr, T_\bl) \Bigm\vert \\
\bigl(\gamma_\wh(p), \gamma_\bl(p)\bigr), \bigl(\gamma_\wh(q), \gamma_\bl(q)\bigr) \in \bigl(\phasegp\pi_\wh^\cap, \phasegp\pi_\bl^\cap\bigr)(H^\cap)\Bigr\}.
\end{multlined}
\]
(Here we omitted the conditions on the derivatives of $\gamma_\wh$ and $\gamma_\bl$ at $p$ and $q$ for brevity.)

If we return to the general case considered at the beginning of \cref{sec:bicol-struct} when $\Gamma$ is a possibly-properly included submodule of $\Lambda_\wh \cap \Lambda_\bl$, then its associated bicoloured torus loop group $L(T_\wh, H, T_\bl)$ is related to $L(T_\wh, H^\cap, T_\bl)$ via a short exact sequence of abelian groups
\[
\begin{tikzcd}
0 \ar[r] & \displaystyle\biggl(\frac{\Lambda_\wh \cap \Lambda_\bl}{\Gamma}\biggr)^{\oplus 2} \ar[r] & L(T_\wh, H, T_\bl) \ar[r] & L(T_\wh, H^\cap, T_\bl) \ar[r] & 0.
\end{tikzcd}
\]
The second arrow sends $([\nu_p], [\nu_q])$ to $(0, \gamma_\match, 0)$, where $\gamma_\match(p) := [\nu_p]$ and $\gamma_\match(q) := [\nu_q]$. The third arrow is given by $(\gamma_\wh, \gamma_\match, \gamma_\bl) \mapsto (\gamma_\wh, \gamma_\bl)$. This `forgetting' of $\gamma_\match$ can also be seen as postcomposing $\gamma_\match$ with the surjective homomorphism $H \twoheadrightarrow H^\cap$ induced by the inclusion $\Gamma \subseteq \Lambda_\wh \cap \Lambda_\bl$.

\subsection{Actions of covers of $\Diff_+(S^1)$}
\label{subsec:bicol-actionsdiffnS1}

Recall from \cref{sec:unicol-struct} that a unicoloured torus loop group has an obvious (left) action of $\Diff_+(S^1)$ on it. Similarly, $L(T_\wh, H, T_\bl)$ is naturally acted upon by the circle diffeomorphisms which fix the points $p$ and $q$ and preserve $\circl$ and $\circr$~. However, using~\eqref{eq:bicol-discontpaths-ses} we can prove something more. Denote for any integer $m \geq 1$ by $\Diff_+^{(m)}(S^1)$ the $m$-fold covering group of $\Diff_+(S^1)$. It fits in a short exact sequence
\[
\begin{tikzcd}
0 \ar[r] & \ZZ/m\ZZ \ar[r] & \Diff_+^{(m)}(S^1) \ar[r] & \Diff_+(S^1) \ar[r] & 1
\end{tikzcd}
\]
and can be modelled by the quotient group $\Diff_+^{(\infty)}(S^1)/m\ZZ$, where in turn we use the model~\eqref{eq:univcoverdiffs1model} for the universal covering group $\Diff_+^{(\infty)}(S^1)$. 

Remember the definition of the integer $n$ on \cpageref{dfn:bicol-integern}.

\begin{prop}
\label{thm:bicol-diffnS1action}
There is a (left) action of $\Diff_+^{(n)}(S^1)$ on $P(H, (\Lambda_\wh - \Lambda_\bl)/\Gamma)$ which lifts to one on $L(T_\wh, H, T_\bl)$ in a way that fixes the subgroup $(\Lambda_\wh \cap \Lambda_\bl)/\Gamma$ of $L(T_\wh, H, T_\bl)$.
\end{prop}
\begin{proof}
We start by explaining how $\Diff_+^{(n)}(S^1)$ acts on $P(H, (\Lambda_\wh - \Lambda_\bl)/\Gamma)$. If $\Phi \in \Diff_+^{(\infty)}(S^1)$ and a path $\gamma$ lies in $P(H, (\Lambda_\wh - \Lambda_\bl)/\Gamma)$, then the path $\Phi^*\gamma\colon [0,1] \to H$ is again in $P(H, (\Lambda_\wh - \Lambda_\bl)/\Gamma)$. Here, $\Phi^*\gamma$ is defined by first extending $\gamma$ quasi-periodically to all of $\RR$, precomposing this extension (which we will also denote by $\gamma$) with $\Phi^{-1}$ next and finally restricting to $[0,1]$ again. Indeed, there holds
\begin{align*}
(\Phi^* \gamma)(1) - (\Phi^* \gamma)(0)
	&= \gamma\bigl(\Phi^{-1}(1)\bigr) - \gamma\bigl(\Phi^{-1}(0)\bigr) \\
	&= \gamma\bigl(\Phi^{-1}(0) + 1\bigr) - \gamma\bigl(\Phi^{-1}(0)\bigr) \\
	&= \gamma\bigl(\Phi^{-1}(0)\bigr) + \gamma(1) - \gamma(0) - \gamma\bigl(\Phi^{-1}(0)\bigr) \\
	&= \gamma(1) - \gamma(0) \in (\Lambda_\wh - \Lambda_\bl)/\Gamma.
\end{align*}
This defines a left action $\Phi \cdot \gamma := \Phi^* \gamma$ of $\Diff_+^{(\infty)}(S^1)$ on $P(H, (\Lambda_\wh - \Lambda_\bl)/\Gamma)$. The shift diffeomorphism $\theta \mapsto \theta + n$ then acts trivially, though. Therefore, this action of $\Diff_+^{(\infty)}(S^1)$ descends to $\Diff_+^{(\infty)}(S^1)/n\ZZ$, where $\ZZ$ is the central subgroup of $\Diff_+^{(\infty)}(S^1)$ generated by $\theta \mapsto \theta + 1$. This quotient group is a model for $\Diff_+^{(n)}(S^1)$.

Let us now construct the lift of this $\Diff_+^{(n)}(S^1)$-action to $L(T_\wh, H, T_\bl)$. Take $\Phi \in \Diff_+^{(\infty)}(S^1)$ and a bicoloured loop $\gamma = (\gamma_\wh, \gamma_\match, \gamma_\bl) \in L(T_\wh, H, T_\bl)$. We already just defined the translate $\Phi \cdot (\Pth\gamma)$ for the path $\Pth(\gamma)$, so in order to define $\Phi \cdot \gamma$ in the discrete pre-image under $\Pth$ over $\Phi \cdot (\Pth\gamma)$ all we need to do is prescribe $(\Phi \cdot \gamma)_\match(q)$. That is, we partially define $\Phi \cdot \gamma$ by the demand that $\Pth(\Phi \cdot \gamma) = \Phi \cdot (\Pth \gamma)$. We set
\[
(\Phi \cdot \gamma)_\match(q) := \gamma_\match(q) - (\Pth\gamma)(0) + \bigl(\Phi^* (\Pth\gamma)\bigr)(0).
\]
Since by definition of the homomorphism $\Pth$ we already know that
\begin{align*}
\gamma_\match(q) - (\Pth\gamma)(1) &\in (\RR\pi_\wh)^{-1}(\Lambda_\wh)/\Gamma, \\
\gamma_\match(q) - (\Pth\gamma)(0) &\in (\RR\pi_\bl)^{-1}(\Lambda_\bl)/\Gamma,
\end{align*}
it follows that again $(\Phi \cdot \gamma)_\match(q) - (\Phi^* (\Pth\gamma))(0) \in (\RR\pi_\bl)^{-1}(\Lambda_\bl)/\Gamma$ and that
\begin{align*}
(\Phi \cdot \gamma)_\match(q) - (\Phi^* (\Pth\gamma))(1)
	&= \gamma_\match(q) - (\Pth\gamma)(0) - (\Pth\gamma)(1) + (\Pth\gamma)(0) \\
	&= \gamma_\match(q) - (\Pth\gamma)(1) \in (\RR\pi_\wh)^{-1}(\Lambda_\wh)/\Gamma
\end{align*}
since
\[
\bigl(\Phi^* (\Pth\gamma)\bigr)(0) - \bigl(\Phi^* (\Pth\gamma)\bigr)(1) = - (\Pth\gamma)(1) + (\Pth\gamma)(0).
\]
This shows that $\Phi \cdot \gamma \in L(T_\wh, H, T_\bl)$.

Using the fact that $\Pth$ is a homomorphism and $\Phi$ acts as an automorphism on $P(H, (\Lambda_\wh - \Lambda_\bl)/\Gamma)$ it is easily checked that $\Phi$ also acts as an automorphism on $L(T_\wh, H, T_\bl)$. Furthermore, it can be shown that this is compatible with the composition of the elements of $\Diff_+^{(\infty)}(S^1)$. Therefore the $\Diff_+^{(\infty)}(S^1)$-action on $P(H, (\Lambda_\wh - \Lambda_\bl)/\Gamma)$ lifts to one on $L(T_\wh, H, T_\bl)$. It also descends to $\Diff_+^{(n)}(S^1)$ because $(n \cdot \gamma)_\match(q) = \gamma_\match(q)$. By definition of this $\Diff_+^{(n)}(S^1)$-action, the homomorphism $\Pth$ is $\Diff_+^{(n)}(S^1)$-equivariant.
\end{proof}

\begin{rmk}
The action of $\Diff_+^{(n)}(S^1)$ on $P(H, (\Lambda_\wh - \Lambda_\bl)/\Gamma)$ can alternatively be understood by seeing elements of the latter group as a certain type of maps to the torus $H$ from an $n$-fold cover of $S^1$.
\end{rmk}

\begin{rmk}[$\Diff_+^{(n)}(S^1)$-equivariance with respect to $LH$]
\label{rmk:bicol-diffnS1equivariance-wrtunicol}
Note that the isomorphism $\Bi\colon LH \xrightarrow{\sim} L(H, H, H)$ defined in \cref{subsec:bicol-unicolisspecialcase} is equivariant with respect to the $\Diff_+(S^1)$-actions on both groups.

Regarding the inclusion $\Bi\colon LH \hookrightarrow L(T_\wh, H, T_\bl)$ from \cref{subsec:bicol-inclLH}; breaking $S^1$ at the point $q$ allows us to embed $LH$ into $P(H, (\Lambda_\wh - \Lambda_\bl)/\Gamma)$ as well. We can then namely identify $LH$ with
\begin{equation}
\label{eq:bicol-LH-as-gpofpaths}
\Bigl\{\gamma \in C^\infty\bigl([0,1], H\bigr) \Bigm\vert \text{$\gamma(1) = \gamma(0)$, $\gamma^{(k)}(1) = \gamma^{(k)}(0)$ for all $k\geq 1$}\Bigr\}.
\end{equation}
This inclusion homomorphism $\iota$ is equivariant with respect to the standard $\Diff_+(S^1)$-action on $LH$ described in \cref{sec:unicol-struct} and the $\Diff_+^{(n)}(S^1)$-action on $P(H, (\Lambda_\wh - \Lambda_\bl)/\Gamma)$ constructed in the proof of \cref{thm:bicol-diffnS1action}. The following triangle commutes:
\begin{equation}
\label{eq:bicol-LH-commtriangle}
\begin{tikzcd}
L(T_\wh, H, T_\bl) \ar[rr, "\Pth"] & & P\bigl(H, (\Lambda_\wh - \Lambda_\bl)/\Gamma\bigr) \ar[dl, hookleftarrow, "\iota"] \\
	& LH \ar[ul, hookrightarrow, "\Bi"] &
\end{tikzcd}.
\end{equation}
Together with the equivariance of $\Pth$ this makes it clear that
\[
\Pth\bigl([\Phi] \cdot \Bi(\gamma)\bigr) = \Pth\Bigl(\Bi\bigl([\Phi]^* \gamma\bigr)\Bigr),
\]
where $\gamma \in LH$ and $\Phi \in \Diff_+^{(\infty)}(S^1)$. It is then easily checked that also
\[
\bigl([\Phi] \cdot \Bi(\gamma)\bigr)_\match(q) = \Bi\bigl([\Phi]^* \gamma\bigr)_\match(q)
\]
and therefore $[\Phi] \cdot \Bi(\gamma) = \Bi([\Phi]^* \gamma)$. In other words, $\Bi$ is equivariant with respect to the $\Diff_+(S^1)$-action on $LH$ and the $\Diff_+^{(n)}(S^1)$-action on $L(T_\wh, H, T_\bl)$.
\end{rmk}

\begin{rmk}
It might seem unnatural to consider the action of the full group $\Diff_+^{(n)}(S^1)$ on $L(T_\wh, H, T_\bl)$, instead of for example the smaller group of equivalence classes of elements of $\Diff_+^{(\infty)}(S^1)$ that preserve the intervals $[k, k+1/2]$ for all $k \in \ZZ$. The former group is namely not compatible with the notion of support of bicoloured loops in a way similar to \cref{rmk:unicol-diffS1-compatiblewithsupp} about the unicoloured situation. Instead, $\Diff_+^{(n)}(S^1)$ `mixes' the left and right halves of $S^1$ together.

The reason for our interest in $\Diff_+^{(n)}(S^1)$ is that it contains the subgroup $\Rot^{(n)}(S^1)$, and the action of the latter will allow us to speak about \defn{positive energy} representations of (central extensions of) $L(T_\wh, H, T_\bl)$.
\end{rmk}

\subsection{The connected components of $L(T_\wh, H, T_\bl)$}

Recall from~\eqref{eq:unicol-liealgpaths} how a choice of a privileged point on $S^1$ made it possible to alternatively describe the elements of a unicoloured torus loop group as Lie algebra valued paths. It allowed us to define the winding element of a unicoloured loop and, eventually, to understand the structure of the group. The definition of a bicoloured torus loop group already carries two privileged points $p$ and $q$. We will choose the point $q$ and carry out in this and the coming sections similar steps as in the unicoloured case.

The following \namecref{thm:bicol-liealglifts} explains that, given a triple $\gamma = (\gamma_\wh, \gamma_\match, \gamma_\bl)$ in $L(T_\wh, H, T_\bl)$, which are torus valued maps, we may lift them to Lie algebra valued maps which `glue' to a map that is continuous at the point $p$. This is essentially done by first applying $\Pth$ to $\gamma$, then taking a lift $[0,1] \to \liealg h$ of $\Pth(\gamma)$, writing this as a pair of maps $\circl \to \liealg t_\wh$ and $\circr \to \liealg t_\bl$ and finally putting the datum of a lift of $\gamma_\match$ back in. We will write the proof in a `bicoloured fashion', though, without referring to the homomorphism~$\Pth$.

We first observe that the group $P(H, (\Lambda_\wh - \Lambda_\bl)/\Gamma)$ has an alternative description in terms of paths in the Lie algebra $\liealg t$, namely as follows:
\begin{equation}
\begin{split}
P\bigl(H, (\Lambda_\wh - \Lambda_\bl)/\Gamma\bigr) &\cong \Bigl\{\hat\xi \in C^\infty\bigl([0,1], \liealg h\bigr) \Bigm\vert \hat\xi(1) - \hat\xi(0) \in \Lambda_\wh - \Lambda_\bl, \\
	&\phantom{{}\cong \Bigl\{\hat\xi \in C^\infty\bigl([0,1], \liealg h\bigr) \Bigm\vert {}}  \hat\xi^{(k)}(1) = \hat\xi^{(k)}(0) \text{ for all $k \geq 1$}\Bigr\}\Big/\Gamma.
\end{split}
\label{eq:bicolP-liealgpaths}
\end{equation}

\begin{lem}
\label{thm:bicol-liealglifts}
Let $\gamma = (\gamma_\wh, \gamma_\match, \gamma_\bl) \in L(T_\wh, H, T_\bl)$ be a bicoloured loop. Then
\begin{enumerate}
\item\label{thmitm:bicol-liealglifts-exists}
there exists a triple $\xi = (\xi_\wh, \xi_\match, \xi_\bl)$ of smooth maps
\[
\label{eq:bicoltor-univcover}
\xi_\wh\colon \circl \to \liealg t_\wh, \qquad \xi_\match\colon \circpq \to \liealg h \qquad \text{and} \qquad \xi_\bl\colon \circr \to \liealg t_\bl
\]
such that
\begin{enumerate}
\item\label{thmitm:bicol-liealglifts-exists-lift}
$\xi_\whbl$ is a lift of $\gamma_\whbl$ and $\xi_\match$ is a lift of $\gamma_\match$, that is, ${\exp_\whbl} \circ \xi_\whbl = \gamma_\whbl$ and ${\exp_H} \circ \xi_\match = \gamma_\match$, where $\exp_\whbl\colon \liealg t_\whbl \twoheadrightarrow T_\whbl$ and $\exp_H\colon \liealg h \twoheadrightarrow H$,
\item\label{thmitm:bicol-liealglifts-exists-matchp}
$(\RR\pi_\wh)^{-1}(\xi_\wh(p)) = \xi_\match(p) = (\RR\pi_\bl)^{-1}(\xi_\bl(p))$,
\end{enumerate}
\item\label{thmitm:bicol-liealglifts-ambig}
if $\eta = (\eta_\wh, \eta_\match, \eta_\bl)$ is another triple of maps satisfying~\ref{thmitm:bicol-liealglifts-exists-lift} and~\ref{thmitm:bicol-liealglifts-exists-matchp} above, then
\[
\eta = \xi + \bigl(\pi_\wh(\mu_p), (p,q) \to (\mu_p, \mu_q), \pi_\bl(\mu_p)\bigr)
\]
for some $\mu_p, \mu_q \in \Gamma$,
\item\label{thmitm:bicol-liealglifts-liftofP}
the map $\hat\xi\colon [0,1] \to \liealg h$ defined by
\[
\hat\xi(\theta) :=
\begin{cases}
(\RR\pi_\wh)^{-1}\bigl(\xi_\wh(\theta)\bigr) & \text{if $\theta \in [1/2,1]$,} \\
(\RR\pi_\bl)^{-1}\bigl(\xi_\bl(\theta)\bigr) & \text{if $\theta \in [0,1/2]$,}
\end{cases}
\]
where we used the unit speed parametrisations of $\circrdir$ and $\circldir$ by $[0,1/2]$ and $[1/2, 1]$ respectively again, is a lift of $\Pth(\gamma)$ as in~\eqref{eq:bicolP-liealgpaths}.
\end{enumerate}
\end{lem}
\begin{proof}
\ref{thmitm:bicol-liealglifts-exists}: Pick an arbitrary triple $(\xi_\wh, \xi_\match, \xi_\bl)$ which satisfies~\ref{thmitm:bicol-liealglifts-exists-lift}. Then the diagram
\[
\begin{tikzcd}
\circl \ar[d, "\xi_\wh"'] &[1em] \ar[l, hook'] \circpq \ar[d, "\xi_\match"] \ar[r, hook] &[1em] \circr \ar[d, "\xi_\bl"] \\
\liealg t_\wh & \ar[l, "\RR\pi_\wh", "\sim"'] \liealg h \ar[r, "\RR\pi_\bl"', "\sim"] & \liealg t_\bl
\end{tikzcd}
\]
does not necessarily commute. Instead, we have
\begin{gather*}
\bigl(\xi_\wh(p), \xi_\bl(p)\bigr) = (\lambda_\wh, \lambda_\bl) + (\RR\pi_\wh, \RR\pi_\bl)\bigl(\xi_\match(p)\bigr) \\
\bigl(\xi_\wh(q), \xi_\bl(q)\bigr) = (\lambda_\wh', \lambda_\bl') + (\RR\pi_\wh, \RR\pi_\bl)\bigl(\xi_\match(q)\bigr)
\end{gather*}
for some $\lambda_\whbl, \lambda_\whbl' \in \Lambda_\whbl$. It follows that the triple $(\xi_\wh - \lambda_\wh, \xi_\match, \xi_\bl - \lambda_\bl)$ now satisfies both~\ref{thmitm:bicol-liealglifts-exists-lift} and~\ref{thmitm:bicol-liealglifts-exists-matchp}.

\ref{thmitm:bicol-liealglifts-ambig}: Since $\eta$ also satisfies~\ref{thmitm:bicol-liealglifts-exists-lift}, we have
\[
\eta = \xi + \bigl(\lambda_\wh, (p,q) \to (\mu_p, \mu_q), \lambda_\bl\bigr)
\]
for some $\lambda_\whbl \in \Lambda_\whbl$ and $\mu_p, \mu_q \in \Gamma$. In particular, $\eta_\whbl(p) = \xi_\whbl(p) + \lambda_\whbl$. But using property~\ref{thmitm:bicol-liealglifts-exists-matchp} for both $\eta$ and $\xi$,
\[
(\RR\pi_\whbl)^{-1}\bigl(\eta_\whbl(p) - \xi_\whbl(p)\bigr) = \eta_\match(p) - \xi_\match(p) = \mu_p,
\]
so $\lambda_\whbl = \pi_\whbl(\mu_p)$.

\ref{thmitm:bicol-liealglifts-liftofP}: This claim is obvious given the definition of the homomorphism $\Pth$.
\end{proof}

Because we will use the Lie algebra valued lifts constructed in the previous \namecref{thm:bicol-liealglifts} often, we give them their own name:

\begin{dfn}
For a bicoloured loop $\gamma = (\gamma_\wh, \gamma_\match, \gamma_\bl) \in L(T_\wh, H, T_\bl)$ a triple of maps $\xi = (\xi_\wh, \xi_\match, \xi_\bl)$ as in \cref{thm:bicol-liealglifts}\ref{thmitm:bicol-liealglifts-exists} is called a \defn{glued lift} of $\gamma$.
\end{dfn}

We will denote the quotient of the direct sum $\Lambda_\wh \oplus \Lambda_\bl$ by the image of $\Gamma$ under the homomorphism $(\pi_\wh, \pi_\bl)$ as $(\Lambda_\wh \oplus \Lambda_\bl)/\Gamma$.

\begin{prop}
\label{thm:bicol-windeltshoms}
There are surjective homomorphisms
\begin{gather}
\Delta'\colon P\bigl(H, (\Lambda_\wh - \Lambda_\bl)/\Gamma\bigr) \twoheadrightarrow \Lambda_\wh - \Lambda_\bl, \nonumber \\
\Delta\colon L(T_\wh, H, T_\bl) \twoheadrightarrow \frac{\Lambda_\wh \oplus \Lambda_\bl}{\Gamma} \label{eq:bicol-windelt}
\end{gather}
of abelian groups which make the following diagram commute:
\begin{equation}
\label{eq:bicol-windelt-commdiag}
\begin{tikzcd}
0 \ar[r] & \displaystyle\frac{\Lambda_\wh \cap \Lambda_\bl}{\Gamma} \ar[r] \ar[d, "-1", pos=0.35] & L(T_\wh, H, T_\bl) \ar[r, "\Pth"] \ar[d, twoheadrightarrow, "\Delta"] &[2.3em] P\bigl(H, (\Lambda_\wh - \Lambda_\bl)/\Gamma\bigr) \ar[r] \ar[d, twoheadrightarrow, "\Delta'" pos=0.325] & 0 \\
0 \ar[r] & \displaystyle\frac{\Lambda_\wh \cap \Lambda_\bl}{\Gamma} \ar[r] & \displaystyle\frac{\Lambda_\wh \oplus \Lambda_\bl}{\Gamma} \ar[r, "(\RR\pi_\wh)^{-1} - (\RR\pi_\bl)^{-1}"'] & \Lambda_\wh - \Lambda_\bl \ar[r] & 0,
\end{tikzcd}
\end{equation}
where the top row is the short exact sequence~\eqref{eq:bicol-discontpaths-ses} and the bottom row is exact as well.
\end{prop}
\begin{proof}
The homomorphism $\Delta'$ is defined in exactly the same way as in the unicoloured case, namely by picking for a path $\gamma \in P(H, (\Lambda_\wh - \Lambda_\bl)/\Gamma)$ a lift $\hat\xi\colon [0,1] \to \liealg h$ as in~\eqref{eq:bicolP-liealgpaths}. The element $\Delta_\gamma' := \hat\xi(1) - \hat\xi(0) \in \Lambda_\wh - \Lambda_\bl$ is then independent of the choice of $\hat\xi$ and this clearly defines a homomorphism. It is surjective because $\lambda_\wh - \lambda_\bl \in \Lambda_\wh - \Lambda_\bl$ has as pre-image for example the element of $P(H, (\Lambda_\wh - \Lambda_\bl)/\Gamma)$ defined as the projection on $H$ of the Lie algebra-valued path $[0,1] \to \liealg h$, $\theta \mapsto \theta(\lambda_\wh - \lambda_\bl)$.

To define the homomorphism $\Delta$, let $\gamma = (\gamma_\wh, \gamma_\match, \gamma_\bl) \in L(T_\wh, H, T_\bl)$ be a bicoloured loop and pick a glued lift $\xi = (\xi_\wh, \xi_\match, \xi_\bl)$ of it. As explained in the proof of \cref{thm:bicol-liealglifts}\ref{thmitm:bicol-liealglifts-exists}, there are then elements $\lambda_\whbl\in \Lambda_\whbl$ such that
\[
\bigl(\xi_\wh(q), \xi_\bl(q)\bigr) = (\lambda_\wh, \lambda_\bl) + (\RR\pi_\wh, \RR\pi_\bl)\bigl(\xi_\match(q)\bigr).
\]
If $\eta = (\eta_\wh, \eta_\match, \eta_\bl)$ is another such lift, then by \cref{thm:bicol-liealglifts}\ref{thmitm:bicol-liealglifts-ambig},
\[
\bigl(\eta_\wh(q), \eta_\bl(q)\bigr) = (\lambda_\wh, \lambda_\bl) + \bigl(\pi_\wh(\mu), \pi_\bl(\mu)\bigr) + (\RR\pi_\wh, \RR\pi_\bl)\bigl(\eta_\match(q)\bigr)
\]
for some $\mu \in \Gamma$. Therefore, setting $\Delta_\gamma$ to be the equivalence class $[\lambda_\wh, \lambda_\bl]$ of $(\lambda_\wh, \lambda_\bl)$ gives a well-defined map~\eqref{eq:bicol-windelt}. We show that $\Delta$ is surjective. Let $(\lambda_\wh, \lambda_\bl)$ be an element in $\Lambda_\wh \oplus \Lambda_\bl$. Choose $\xi_\wh\colon \circl \to \liealg t_\wh$ to be any map satisfying $\xi_\wh(p) = 0_{\liealg t_\wh}$ and $\xi_\wh(q) = \lambda_\wh$, pick $\xi_\bl\colon \circr \to \liealg t_\bl$ to be any map with $\xi_\bl(p) = 0_{\liealg t_\bl}$ and $\xi_\bl(q) = \lambda_\bl$ and define $\xi_\match\colon \circpq \to \liealg h$ as $\xi_\match(p) = \xi_\match(q) = 0_{\liealg h}$. Then the triple $(\xi_\wh, \xi_\match, \xi_\bl)$ is a glued lift of a bicoloured loop in $L(T_\wh, H, T_\bl)$ and one has
\[
\bigl(\xi_\wh(q), \xi_\bl(q)\bigr) = (\lambda_\wh, \lambda_\bl) + (\RR\pi_\wh, \RR\pi_\bl)(0),
\]
which shows what we wanted.

The second map in the bottom row of~\eqref{eq:bicol-windelt-commdiag} is given by the obvious inclusion if we identify $(\Lambda_\wh \cap \Lambda_\bl)/\Gamma$ with the quotient group
\begin{equation}
\label{eq:bicol-intersec-altdescr}
\frac{\Bigl\{\bigl(\RR\pi_\wh(\nu), \RR\pi_\bl(\nu)\bigr) \Bigm\vert \nu \in \Lambda_\wh \cap \Lambda_\bl\Bigr\}}{\Bigl\{\bigl(\pi_\wh(\mu), \pi_\bl(\mu)\bigr) \Bigm\vert \mu \in \Gamma\Bigr\}},
\end{equation}
and the third map stands for
\[
[\lambda_\wh, \lambda_\bl] \mapsto (\RR\pi_\wh)^{-1}(\lambda_\wh) - (\RR\pi_\bl)^{-1}(\lambda_\bl),
\]
where $[\lambda_\wh, \lambda_\bl]$ is the equivalence class of $(\lambda_\wh, \lambda_\bl)$ in $(\Lambda_\wh \oplus \Lambda_\bl)/\Gamma$. This third map indeed has~\eqref{eq:bicol-intersec-altdescr} as its kernel.

To explain the commutativity of the first square in~\eqref{eq:bicol-windelt-commdiag}, let $[\nu] \in (\Lambda_\wh \cap \Lambda_\bl)/\Gamma$. Then its image in $(\Lambda_\wh \oplus \Lambda_\bl)/\Gamma$ when travelling along the lower left corner of the square is the equivalence class $-[(\RR\pi_\wh, \RR\pi_\bl)(\nu)]$. On the other hand, if we travel along the upper right corner, recall that the image of $[\nu]$ in $L(T_\wh, H, T_\bl)$ is the bicoloured loop $\gamma = (\gamma_\wh, \gamma_\match, \gamma_\bl)$, where the $\gamma_\whbl$ are identically $0_{T_\whblsm}$, $\gamma_\match(p) = 0_H$ and $\gamma_\match(q) = [\nu]$. A possible glued lift $\xi = (\xi_\wh, \xi_\match, \xi_\bl)$ of $\gamma$ has the maps $\xi_\whbl$ identically $0_{\liealg t_\whblsm}$, $\xi_\match(p) = 0_{\liealg h}$ and $\xi_\match(q) = \nu$. This satisfies
\[
\bigl(\xi_\wh(q), \xi_\bl(q)\bigr) = (0,0) = -(\RR\pi_\wh, \RR\pi_\bl)(\nu) + (\RR\pi_\wh, \RR\pi_\bl)(\nu),
\]
as desired, given the way we defined $\Delta$.

The commutativity of the second square in~\eqref{eq:bicol-windelt-commdiag} follows from the claim of \cref{thm:bicol-liealglifts}\ref{thmitm:bicol-liealglifts-liftofP}.
\end{proof}

The homomorphism $\Delta$ serves the same role in our bicoloured situation as the winding element homomorphism $LT \twoheadrightarrow \Lambda$ does in the unicoloured case in the sense that its fibres are exactly the connected components of the group. The exactness of the bottom row in~\eqref{eq:bicol-windelt-commdiag} implies that the group of connected components $(\Lambda_\wh \oplus \Lambda_\bl)/\Gamma$ of $L(T_\wh, H, T_\bl)$ is torsion-free if and only if $\Gamma = \Lambda_\wh \cap \Lambda_\bl$.

\begin{rmk}
It follows from \cref{thm:bicol-liealglifts}\ref{thmitm:bicol-liealglifts-liftofP} that for a bicoloured loop $\gamma \in L(T_\wh, H, T_\bl)$ we have
\[
\Delta_{\Pth\gamma}' = (\RR\pi_\wh)^{-1}\bigl(\xi_\wh(q)\bigr) - (\RR\pi_\bl)^{-1}\bigl(\xi_\bl(q)\bigr)
\]
for any glued lift $(\xi_\wh, \xi_\match, \xi_\bl)$ of $\gamma$.
\end{rmk}

\subsection{The structure of $L(T_\wh, H, T_\bl)$}
\label{subsec:bicol-struct}
Since we learned in \cref{subsec:bicol-inclLH} that $L(T_\wh, H, T_\bl)$ contains a canonical copy of $LH$ via an injective `bicolouring' homomorphism $\Bi$, there is in particular an inclusion of the identity component $(LH)_0$ of the latter group.

\begin{prop}
The subgroup $\Bi((LH)_0)$ of $L(T_\wh, H, T_\bl)$ is the kernel of the homomorphism $\Delta$, and it therefore also equals the identity component of $L(T_\wh, H, T_\bl)$.
\end{prop}
\begin{proof}
Recall first that $(LH)_0$ can be canonically identified with the group of all smooth maps $S^1 \to \liealg h$, modulo $\Gamma$. Now suppose that $\gamma$ is a bicoloured loop in the kernel of $\Delta$. According to the definition of $\Delta$ in the proof of \cref{thm:bicol-windeltshoms} this means that if $\xi = (\xi_\wh, \xi_\match, \xi_\bl)$ is a glued lift of $\gamma$, then there exists a $\mu \in \Gamma$ such that
\[
\bigl(\xi_\wh(q), \xi_\bl(q)\bigr) = \bigl(\pi_\wh(\mu), \pi_\bl(\mu)\bigr) + (\RR\pi_\wh, \RR\pi_\bl)\bigl(\xi_\match(q)\bigr).
\]
Therefore,
\[
(\RR\pi_\wh)^{-1}\bigl(\xi_\wh(q)\bigr) = (\RR\pi_\bl)^{-1}\bigl(\xi_\bl(q)\bigr),
\]
which together with \cref{thm:bicol-liealglifts}\ref{thmitm:bicol-liealglifts-exists-matchp} implies that the two maps $(\RR\pi_\whbl)^{-1} \circ \xi_\whbl$ glue together to a smooth map $S^1 \to \liealg h$. Its value at $q$ differs from that of $\xi_\match$ by $\mu$. By the characterisation of the image of $\Bi$ mentioned at the end of \cref{subsec:bicol-inclLH} we see that $\gamma$ lies in $\Bi((LH)_0)$.
\end{proof}

We are now able to understand the structure of $L(T_\wh, H, T_\bl)$ by showing that the short exact sequence
\begin{equation}
\label{eq:bicol-sesconncomponents}
\begin{tikzcd}
0 \ar[r] & (LH)_0 \ar[r, "\Bi"] &[0.8em] L(T_\wh, H, T_\bl) \ar[r, "\Delta"] &[0.5em] \displaystyle\frac{\Lambda_\wh \oplus \Lambda_\bl}{\Gamma} \ar[r] & 0.
\end{tikzcd}
\end{equation}
admits a splitting, which we construct as follows. Let $(\lambda_\wh, \lambda_\bl) \in \Lambda_\wh \oplus \Lambda_\bl$. Define the $\liealg h$-valued straight line segment
\[
\hat\xi\colon [0,1] \to \liealg h, \qquad \theta \mapsto (\RR\pi_\bl)^{-1}(\lambda_\bl) + \theta\bigl((\RR\pi_\wh)^{-1}(\lambda_\wh) - (\RR\pi_\bl)^{-1}(\lambda_\bl)\bigr)
\]
from $(\RR\pi_\bl)^{-1}(\lambda_\bl)$ to $(\RR\pi_\wh)^{-1}(\lambda_\wh)$, and use it to define three Lie algebra valued maps $\xi_\wh$, $\xi_\match$ and $\xi_\bl$ in turn by restricting to the relevant parts of the interval $[0,1]$:
\begin{align*}
\xi_\wh &:= \RR\pi_\wh \circ \hat\xi|_{[1/2, 1]} \colon \circl \to \liealg t_\wh \\
\xi_\match(p) &:= \hat\xi(1/2) = \frac12\bigl((\RR\pi_\wh)^{-1}(\lambda_\wh) + (\RR\pi_\bl)^{-1}(\lambda_\bl)\bigr) \in \liealg h \\
\xi_\match(q) &:= 0_{\liealg h} \in \liealg h \\
\xi_\bl &:= \RR\pi_\bl \circ \hat\xi|_{[0, 1/2]} \colon \circr \to \liealg t_\bl.
\end{align*}
Because a different choice of representative from the equivalence class $[\lambda_\wh, \lambda_\bl]$ in $(\Lambda_\wh \oplus \Lambda_\bl)/\Gamma$ of $(\lambda_\wh, \lambda_\bl)$ would merely translate $\hat\xi$ by an element of $\Gamma \subseteq \liealg h$, the following torus valued maps $\gamma_\wh$, $\gamma_\match$ and $\gamma_\bl$ only depend on $[\lambda_\wh, \lambda_\bl]$:
\begin{align*}
\gamma_\wh &:= {\exp_\wh} \circ \xi_\wh \colon \circl \to T_\wh \\
\gamma_\match(p) &:= \exp_H\bigl(\xi_\match(p)\bigr) \in H \\
\gamma_\match(q) &:= \exp_H\bigl(\xi_\match(q)\bigr) = 0_H \in H \\
\gamma_\bl &:= {\exp_\bl} \circ \xi_\bl \colon \circr \to T_\bl.
\end{align*}
This triple $\gamma_{[\lambda_\wh, \lambda_\bl]} := (\gamma_\wh, \gamma_\match, \gamma_\bl)$ forms a bicoloured loop. The triple $(\xi_\wh, \xi_\match, \xi_\bl)$ is then a glued lift of $\gamma_{[\lambda_\wh, \lambda_\bl]}$ and
\[
\bigl(\xi_\wh(q), \xi_\bl(q)\bigr) = \Bigl(\RR\pi_\wh\bigl(\hat\xi(1)\bigr), \RR\pi_\bl\bigl(\hat\xi(0)\bigr)\Bigr) = (\lambda_\wh, \lambda_\bl) + (\RR\pi_\wh, \RR\pi_\bl)\bigl(\xi_\match(q)\bigr)
\]
shows that $\Delta_{\gamma_{[\lambda_\wh, \lambda_\bl]}} = [\lambda_\wh, \lambda_\bl]$. In other words, the homomorphism
\begin{equation}
\label{eq:bicol-splitting-conncomps}
\frac{\Lambda_\wh \oplus \Lambda_\bl}{\Gamma} \to L(T_\wh, H, T_\bl), \qquad [\lambda_\wh, \lambda_\bl] \mapsto \gamma_{[\lambda_\wh, \lambda_\bl]}
\end{equation}
that we just constructed is a section of $\Delta$. This in particular implies that~\eqref{eq:bicol-splitting-conncomps} is injective. The splitting~\eqref{eq:bicol-splitting-conncomps} gives an isomorphism
\[
L(T_\wh, H, T_\bl) \xrightarrow{\sim} (LH)_0 \oplus \biggl(\frac{\Lambda_\wh \oplus \Lambda_\bl}{\Gamma}\biggr)
\]
sending a bicoloured loop $\gamma = (\gamma_\wh, \gamma_\match, \gamma_\bl)$ to $(\gamma - \gamma_{\Delta_\gamma}, \Delta_\gamma)$. Its inverse takes a pair $(\gamma, [\lambda_\wh, \lambda_\bl])$ to the bicoloured loop $\gamma + \gamma_{[\lambda_\wh, \lambda_\bl]}$. Recalling the isomorphism~\eqref{eq:unicol-decompidcomp} we conclude with a decomposition
\[
L(T_\wh, H, T_\bl) \cong H \oplus V\liealg h \oplus \biggl(\frac{\Lambda_\wh \oplus \Lambda_\bl}{\Gamma}\biggr),
\]
which explains the structure of $L(T_\wh, H, T_\bl)$.

\begin{rmk}
In the proof of \cref{thm:bicol-windeltshoms} it was shown that the homomorphism $\Delta$ is surjective by sketching how to find some bicoloured loop in the pre-image of every element $[\lambda_\wh, \lambda_\bl]$ of $(\Lambda_\wh \oplus \Lambda_\bl)/\Gamma$. One could ask whether that construction could be made more precise so as to obtain a simpler splitting of $\Delta$ than the more complicated one we built in the current section. However, the former construction \emph{does} depend on the choice of representative $(\lambda_\wh, \lambda_\bl)$ from $[\lambda_\wh, \lambda_\bl]$, so this plan seems unreasonable.
\end{rmk}

\section{Central extensions associated to spans of lattices}
\label{sec:bicol-centext}

The bicoloured torus loop group $L(T_\wh, H, T_\bl)$ and the group $P(H, (\Lambda_\wh - \Lambda_\bl)/\Gamma)$ defined in \cref{sec:bicol-struct} can in particular be defined if the $\ZZ$-modules $\Lambda_\whbl$ and $\Gamma$ come with the extra structure of bi-additive forms that make them even lattices. In this section we give, using these lattice structures, constructions of $\phasegp$-central extensions $\centL(T_\wh, H, T_\bl)$ and $\centP(H, (\Lambda_\wh - \Lambda_\bl)/\Gamma)$ of $L(T_\wh, H, T_\bl)$ and $P(H, (\Lambda_\wh - \Lambda_\bl)/\Gamma)$ that are analogous to the central extensions of unicoloured torus loop groups described in \cref{sec:unicol-centext}. As we will see in \cref{subsec:bicol-unicolcentext-isspecialcase}, they will in fact reduce exactly to those unicoloured extensions in the special case that $\Lambda_\wh = \Lambda_\bl = \Gamma$ and the homomorphisms $\pi_\whbl\colon \Gamma \hookrightarrow \Lambda_\whbl$ are the identity. We will construct $\centP(H, (\Lambda_\wh - \Lambda_\bl)/\Gamma)$ through an explicit $2$-cocycle. The formula for it will be nearly identical to the one for unicoloured torus loop groups.  The central extension of $L(T_\wh, H, T_\bl)$ will then simply be defined as the pullback of $\centP(H, (\Lambda_\wh - \Lambda_\bl)/\Gamma)$ along the homomorphism $\Pth$. That is, it will be trivial on the kernel of $\Pth$.

So let $\Lambda_\whbl$ and $\Gamma$ be as in \cref{sec:bicol-struct} but assume additionally that they are even lattices, denoting their forms by $\langle \cdot, \cdot \rangle_\whbl$ and $\langle \cdot, \cdot \rangle_\Gamma$ respectively, and that the $\pi_\whbl$ are lattice morphisms. We will use the same notation for the forms that are extended bilinearly to the Lie algebras $\liealg t_\whbl$ and $\liealg h$.

\begin{rmk}
Special cases of spans of lattices have been considered before in the literature. If namely $\Gamma = \Lambda_\wh \cap \Lambda_\bl$ and $[\Lambda_\wh : \Gamma] = p = [\Lambda_\bl : \Gamma]$ for some prime number $p$, then $\Lambda_\wh$ and $\Lambda_\bl$ have been called \defn{$p$-neighbours} (see \parencite[Section (28.2)]{kneser:quadratische}).

If one has the freedom to choose the middle even lattice $\Gamma$ beforehand, then it is easy to construct and classify spans of lattices via the technique of \defn{discriminant groups}. The even lattices $\Lambda_\whbl$ namely correspond to the $q$-isotropic subgroups of the finite abelian discriminant group $D_\Gamma := \Gamma^\vee/\Gamma$ of $\Gamma$, as explained in \cref{app:latticegluing}.
\end{rmk}

\begin{rmk}
The restriction of the $\RR$-valued form $\langle \cdot, \cdot \rangle_\Gamma$ on $\liealg h$ to $\Lambda_\wh \cap \Lambda_\bl \subseteq \liealg h$ is $\ZZ$-valued, making $\Lambda_\wh \cap \Lambda_\bl$ an even lattice, while the restriction to $\Lambda_\wh - \Lambda_\bl \subseteq \liealg h$ is in general $\QQ$-valued, making this a rational lattice.
\end{rmk}

Just like in the unicoloured case, the construction of central extensions of $L(T_\wh, H, T_\bl)$ and $P(H, (\Lambda_\wh - \Lambda_\bl)/\Gamma)$ will require a little bit more data than just the lattices $\Lambda_\whbl$ and $\Gamma$ and the two embeddings $\pi_\whbl$ between them. We will namely need a certain central extension of the underlying abelian group of the rational lattice $\Lambda_\wh - \Lambda_\bl$ to ensure that the central extension of $L(T_\wh, H, T_\bl)$ will be disjoint-commutative.

\begin{constr}[A central extension of $\Lambda_\wh - \Lambda_\bl$]
Let us define a function\footnote{More precisely, $b$ takes values only in the finite cyclic subgroup of $\phasegp$ of order \[2\cdot [\Lambda_\wh - \Lambda_\bl : \Lambda_\wh \cap \Lambda_\bl].\]}
\begin{equation}
\label{eq:bicol-latcommmap}
\begin{gathered}
b\colon (\Lambda_\wh - \Lambda_\bl) \times (\Lambda_\wh - \Lambda_\bl) \to \phasegp, \\
b(\lambda, \mu) := e^{2\pi i b_0(\lambda, \mu)}, \qquad
b_0(\lambda, \mu) := \frac12 \langle \mu, \lambda \rangle_\Gamma + \langle \mu, \lambda_\bl \rangle_\Gamma \in \QQ,
\end{gathered}
\end{equation}
where $\lambda, \mu \in \Lambda_\wh - \Lambda_\bl$ and $\lambda_\bl$ is defined via a choice of decomposition $\lambda = \lambda_\wh - \lambda_\bl$, with $\lambda_\whbl \in (\RR\pi_\whbl)^{-1}(\Lambda_\whbl)$. Notice that $b_0(\lambda, \mu) = - \frac12 \langle \mu, \lambda \rangle_\Gamma + \langle \mu, \lambda_\wh \rangle_\Gamma$ also, since we can add $-\langle \mu, \lambda \rangle_\Gamma + \langle \mu, \lambda \rangle_\Gamma$ to $b_0(\lambda, \mu)$, so we can use this as an alternative expression for $b_0$ if we so please.

The definition of $b$ is independent of the choice of decomposition of $\lambda$ used. Suppose namely that also $\lambda = \lambda_\wh' - \lambda_\bl'$. Then $\lambda_\bl' = \lambda_\bl + \nu$ for some $\nu \in \Lambda_\wh \cap \Lambda_\bl$. This implies that
\[
\frac12 \langle \mu, \lambda \rangle_\Gamma + \langle \mu, \lambda_\bl' \rangle_\Gamma = b_0(\lambda, \mu) + \langle \mu, \nu \rangle_\Gamma \equiv b_0(\lambda, \mu) \bmod \ZZ.
\]

It is clear that $b$ is bi-additive and we furthermore have that $b(\lambda, \lambda) = 1$ for all $\lambda \in \Lambda_\wh - \Lambda_\bl$ since
\[
b_0(\lambda, \lambda) = \frac12 \langle \lambda_\wh - \lambda_\bl, \lambda_\wh - \lambda_\bl \rangle_\Gamma + \langle \lambda_\wh - \lambda_\bl, \lambda_\bl \rangle_\Gamma = \frac12 \langle \lambda_\wh, \lambda_\wh \rangle_\Gamma - \frac12 \langle \lambda_\bl, \lambda_\bl \rangle_\Gamma,
\]
and $\langle \lambda_\whbl, \lambda_\whbl \rangle_\Gamma/2 \in \ZZ$ because the lattices $\Lambda_\whbl$ are even. Together with the bi-additivity, this property of $b$ implies by the discussion in \cref{subsec:centexts-abgps} that there exists a $\phasegp$-central extension $(\Lambda_\wh - \Lambda_\bl)\centext$ of $\Lambda_\wh - \Lambda_\bl$ which has $b$ as a commutator map. It is determined up to non-unique isomorphism.
\end{constr}

We repeat that the formula~\eqref{eq:bicol-latcommmap} for the commutator map $b$, which looks somewhat ad hoc, is designed specifically to make \cref{thm:bicol-disjcomm} true.

\begin{constr}[Central extensions of $P(H, (\Lambda_\wh - \Lambda_\bl)/\Gamma)$]
\label{constr:bicolP-centext}
We will construct a $\phasegp$-central extension $\centP(H, (\Lambda_\wh - \Lambda_\bl)/\Gamma)$ of $P(H, (\Lambda_\wh - \Lambda_\bl)/\Gamma)$ by letting its underlying set be $P(H, (\Lambda_\wh - \Lambda_\bl)/\Gamma) \times \phasegp$ and writing down an explicit $2$-cocycle. We fix a choice of a $\phasegp$-central extension $(\Lambda_\wh - \Lambda_\bl)\centext$ of $\Lambda_\wh - \Lambda_\bl$ with the commutator map~\eqref{eq:bicol-latcommmap} and a choice of a $2$-cocycle
\[
\epsilon\colon (\Lambda_\wh - \Lambda_\bl) \times (\Lambda_\wh - \Lambda_\bl) \to \phasegp
\]
for it. Let $\gamma, \rho \in P(H, (\Lambda_\wh - \Lambda_\bl)/\Gamma)$, $z,w \in \phasegp$ and pick lifts $\xi, \eta\colon [0,1] \to \liealg h$ of $\gamma$ and $\rho$ respectively using the isomorphism~\eqref{eq:bicolP-liealgpaths}. We define the multiplication on $\centP(H, (\Lambda_\wh - \Lambda_\bl)/\Gamma)$ by
\[
(\gamma, z) \cdot (\rho, w) := \bigl(\gamma + \rho, zw \cdot c'(\gamma, \rho)\bigr)
\]
where $c'$ is the $2$-cocycle on $P(H, (\Lambda_\wh - \Lambda_\bl)/\Gamma)$ given by
\begin{equation}
\label{eq:bicolP-cocycle}
\begin{gathered}
c'(\gamma, \rho) := \epsilon(\Delta_\gamma', \Delta_\rho') e^{2\pi i S'(\xi, \eta)}, \\
S'(\xi, \eta) := \frac12 \int_0^1 \bigl\langle \xi'(\theta), \eta(\theta) \bigr\rangle_\Gamma \dd\theta + \frac12 \bigl\langle \Delta_\gamma', \eta(0) \bigr\rangle_\Gamma.
\end{gathered}
\end{equation}
Notice that $S'$ is bi-additive.

Since $\xi$ and $\eta$ are well-defined up to an element of $\Gamma$ and $\langle \Delta_\gamma', \mu \rangle_\Gamma \in \ZZ$ for all $\mu \in \Gamma$ the proof of the well-definedness of $c'$ is identical to the one in the unicoloured case.
\end{constr}

\begin{rmk}[Failure of disjoint-commutativity for $\centP(H, (\Lambda_\wh - \Lambda_\bl)/\Gamma)$]
One could try to imitate for the central extension $\centP(H, (\Lambda_\wh - \Lambda_\bl)/\Gamma)$ we just constructed the calculation in the proof of \cref{thm:unicol-disjcomm} which shows disjoint-commutativity for central extensions of unicoloured torus loop groups. A possible adjustment, needed to even make sense of the statement of that \namecref{thm:unicol-disjcomm}, would be to consider subintervals of $[0,1]$ which are allowed to contain the endpoints $\{0,1\}$ and to use the notion of support for elements of $P(H, (\Lambda_\wh - \Lambda_\bl)/\Gamma)$ already mentioned in \cref{subsec:bicol-support}. The intermediate result would then be that
\begin{equation}
\label{eq:bicolP-disjcomm-fail}
e^{2\pi i \bigl(S'(\xi, \eta) - S'(\eta, \xi)\bigr)} = e^{-\pi i \langle \Delta_\rho', \Delta_\gamma' \rangle_\Gamma}.
\end{equation}
On the other hand, we have that
\[
\epsilon(\Delta_\gamma', \Delta_\rho') \epsilon(\Delta_\rho', \Delta_\gamma')^{-1} = b(\Delta_\gamma', \Delta_\rho') = e^{\pi i \langle \Delta_\rho', \Delta_\gamma' \rangle_\Gamma} e^{2\pi i \langle \Delta_\rho', \lambda_\bl \rangle_\Gamma},
\]
where $\Delta_\gamma' = \lambda_\wh - \lambda_\bl$ with $\lambda_\whbl \in (\RR\pi_\whbl)^{-1}(\Lambda_\whbl)$. If $\lambda_\bl \in \Lambda_\wh \cap \Lambda_\bl$, which happens for example if $\supp \gamma$ does not contain $1 \in [0,1]$, then $\langle \Delta_\rho', \lambda_\bl \rangle_\Gamma \in \ZZ$ and so $\epsilon(\Delta_\gamma', \Delta_\rho') \epsilon(\Delta_\rho', \Delta_\gamma')^{-1}$ would absorb~\eqref{eq:bicolP-disjcomm-fail}. Therefore $(\gamma, z)$ and $(\rho, w)$ would indeed commute. However, it need not be true in general that $\lambda_\bl \in \Lambda_\wh \cap \Lambda_\bl$ when $1 \in \supp \gamma$. We conclude that with the definitions we made $\centP(H, (\Lambda_\wh - \Lambda_\bl)/\Gamma)$ is not disjoint-commutative.

We do not know how to adjust the group $\centP(H, (\Lambda_\wh - \Lambda_\bl)/\Gamma)$ in such a way that it becomes disjoint-commutative also for intervals that do contain the point~$1$, except by inheriting via the homomorphism $\Pth$ the notion of support from $L(T_\wh, H, T_\bl)$ as in \cref{dfn:bicol-support}. 
The central extension of $L(T_\wh, H, T_\bl)$ that we will construct in a moment namely \emph{will} turn out to have the desired disjoint-commutativity property.
\end{rmk}

\begin{constr}[Central extensions of $L(T_\wh, H, T_\bl)$]
\label{constr:bicol-centext}
Given the central extension $\centP(H, (\Lambda_\wh - \Lambda_\bl)/\Gamma)$ from \cref{constr:bicolP-centext}, we define a $\phasegp$-central extension $\centL(T_\wh, H, T_\bl)$ of $L(T_\wh, H, T_\bl)$ as the pullback of $\centP(H, (\Lambda_\wh - \Lambda_\bl)/\Gamma)$ along the homomorphism $\Pth$. That is, the $2$-cocycle $c$ defining $\centL(T_\wh, H, T_\bl)$ is set to be
\begin{equation}
\label{eq:bicol-cocycle-viaP}
c(\gamma, \rho) := c'(\Pth\gamma, \Pth\rho)
\end{equation}
for two bicoloured loops $\gamma, \rho \in L(T_\wh, H, T_\bl)$.
\end{constr}

By construction there is a short exact sequence of groups
\[
\begin{tikzcd}
0 \ar[r] & \displaystyle\frac{\Lambda_\wh \cap \Lambda_\bl}{\Gamma} \ar[r] & \centL(T_\wh, H, T_\bl) \ar[r, "\widetilde\Pth"] &[1em] \centP\bigl(H, (\Lambda_\wh - \Lambda_\bl)/\Gamma\bigr) \ar[r] & 1,
\end{tikzcd}
\]
where $\widetilde\Pth$ is the obvious lift of the homomorphism of abelian groups $\Pth$, sending $(\gamma, z) \in \centL(T_\wh, H, T_\bl)$ to $(\Pth\gamma, z)$. The second arrow in this sequence sends an equivalence class $[\nu] \in (\Lambda_\wh \cap \Lambda_\bl)/\Gamma$ to the element
\[
\Bigl(\bigl(0_{T_\wh}, (p, q) \to \bigl(0_H, [\nu]\bigr), 0_{T_\bl}\bigr), 1\Bigr) \in \centL(T_\wh, H, T_\bl).
\]
Having defined $\widetilde \Pth$, we can now state that we have a pullback diagram in the category of groups:
\begin{equation}
\label{eq:bicol-centext-pullback}
\begin{tikzcd}
\centL(T_\wh, H, T_\bl) \ar[r, "\widetilde\Pth", twoheadrightarrow] \ar[d, twoheadrightarrow] \arrow[dr, phantom, "\lrcorner"] &[1em] \centP\bigl(H, (\Lambda_\wh - \Lambda_\bl)/\Gamma\bigr) \ar[d, twoheadrightarrow] \\
L(T_\wh, H, T_\bl) \ar[r, "\Pth"', twoheadrightarrow] & P\bigl(H, (\Lambda_\wh - \Lambda_\bl)/\Gamma\bigr).
\end{tikzcd}
\end{equation}

The cocycle $c$ can also be written in a `bicoloured fashion' without directly using the group $\centP(H, (\Lambda_\wh - \Lambda_\bl)/\Gamma)$, namely by picking glued lifts $\xi = (\xi_\wh, \xi_\match, \xi_\bl)$ and $\eta = (\eta_\wh, \eta_\match, \eta_\bl)$ of $\gamma$ and $\rho$ respectively. If we then split the integral in~\eqref{eq:bicolP-cocycle} in the middle, appeal to \cref{thm:bicol-liealglifts}\ref{thmitm:bicol-liealglifts-liftofP}, which relates $\xi$ and $\eta$ to lifts of $\Pth(\gamma)$ and $\Pth(\rho)$, and finally use the unit speed parametrisations of $\circldir$ and $\circrdir$ by $[0,1/2]$ and $[1/2, 1]$, we get
\begin{equation}
\label{eq:bicol-cocycle}
\begin{gathered}
c(\gamma, \rho) := \epsilon(\Delta_{\Pth\gamma}', \Delta_{\Pth\rho}') e^{2\pi i S(\xi, \eta)}, \\
S(\xi, \eta) := \frac12 \biggl[\int_\circldir \langle \dd\xi_\wh, \eta_\wh\rangle_\wh + \int_\circrdir \langle \dd\xi_\bl, \eta_\bl\rangle_\bl + \Bigl\langle \Delta_{\Pth\gamma}', (\RR \pi_\bl)^{-1}\bigl(\eta_\bl(q)\bigr) \Bigr\rangle_\Gamma\biggr].
\end{gathered}
\end{equation}
Using this description of $c$ instead, its well-definedness can alternatively be checked via \cref{thm:bicol-liealglifts}\ref{thmitm:bicol-liealglifts-ambig} that explains the ambiguity in the choices of the glued lifts $\xi$ and $\eta$.

\begin{ingreds}
\label{ingreds:bicol}
We summarise the ingredients used in the construction of the central extensions $\centP(H, (\Lambda_\wh - \Lambda_\bl)/\Gamma)$ and $\centL(T_\wh, H, T_\bl)$ for clarity:
\begin{itemize}
\item three even lattices $(\Lambda_\wh, \langle \cdot, \cdot \rangle_\wh)$, $(\Lambda_\bl, \langle \cdot, \cdot \rangle_\bl)$ and $(\Gamma, \langle \cdot, \cdot \rangle_\Gamma)$ of the same rank,
\item two lattice morphisms $\pi_\whbl\colon \Gamma \hookrightarrow \Lambda_\whbl$,
\item a choice of a $\phasegp$-central extension $(\Lambda_\wh - \Lambda_\bl)\centext$ of $\Lambda_\wh - \Lambda_\bl$ such that it has commutator map~\eqref{eq:bicol-latcommmap},
\item a choice of a $2$-cocycle $\epsilon$ for $(\Lambda_\wh - \Lambda_\bl)\centext$ (we will always choose $\epsilon$ to be normalised to make calculations easier),
\item (for the central extension $\centL(T_\wh, H, T_\bl)$) a choice of one of the two points $p$ or $q$ on $S^1$ as being privileged.
\end{itemize}
Just like for a unicoloured central extension $\centL T$, also the notations $\centP(H, (\Lambda_\wh - \Lambda_\bl)/\Gamma)$ and $\centL(T_\wh, H, T_\bl)$ do not refer to these ingredients, so it will be important in the sequel to explain this separately when necessary.
\end{ingreds}

\subsection{Unicoloured central extensions are a special case}
\label{subsec:bicol-unicolcentext-isspecialcase}
Recall from \cref{subsec:bicol-unicolisspecialcase} that if $\Lambda_\wh = \Lambda_\bl = \Gamma$ and the morphisms $\pi_\whbl\colon \Gamma \hookrightarrow \Lambda_\whbl$ are the identity, there exists an isomorphism of abelian groups $\Bi\colon LH \xrightarrow{\sim} L(H, H, H)$ given by $\gamma \mapsto (\gamma|_\circlsm, \gamma|_{\{p,q\}}, \gamma|_\circrsm)$. In this case the $\ZZ$-module $\Lambda_\wh - \Lambda_\bl \subseteq \liealg h$ equals $\Gamma$, and the commutator map $b$ on $\Lambda_\wh - \Lambda_\bl = \Gamma$ defined in~\eqref{eq:bicol-latcommmap} is nothing but the commutator map $(\lambda, \mu) \mapsto (-1)^{\langle \lambda, \mu\rangle_\Gamma}$ since $\Gamma$ is an integral lattice. Now note that a choice of a $\{\pm 1\}$-central extension $\tilde \Gamma$ of $\Gamma$ with this as commutator map, a choice of $2$-cocycle $\epsilon\colon \Gamma \times \Gamma \to \{\pm 1\}$ for it and the choice of the point $q$ among $p$ and $q$, as needed in \cref{constr:bicol-centext}, is also exactly the data we needed in \cref{sec:unicol-centext} to construct a $\phasegp$-central extension $\centL H$ of $LH$. We now claim that the obvious lift to the underlying sets
\begin{equation}
\label{eq:bicol-unicolcentext-isspecialcase}
\centL H \xrightarrow{\sim} \centL(H, H, H), \qquad (\gamma, z) \mapsto \bigl(\Bi(\gamma), z\bigr)
\end{equation}
of $\Bi$ is an isomorphism of (non-abelian) groups. It amounts to showing that for all $\gamma, \rho \in LH$,
\begin{equation}
\label{eq:bicol-unicolcentext-isspecialcase-1}
c\bigl(\Bi(\gamma), \Bi(\rho)\bigr) = c_H(\gamma, \rho),
\end{equation}
where we write the cocycle defining $\centL H$ as $c_H$ and still write $c$ for the one defining $\centL(H, H, H)$. We will also denote by $S_H$ the map in the definition of $c_H$ that takes two $\liealg h$-valued maps as input.

The first observation to make is that if $\xi\colon [0,1] \to \liealg h$ is a lift of $\gamma$ as in~\eqref{eq:unicol-liealgpaths}, made by cutting $S^1$ at the privileged point $q$, then we may choose the triple of maps
\[
(\xi_\wh, \xi_\match, \xi_\bl) := \Bigl(\xi|_{[1/2, 1]}, (p, q) \to \bigl(\xi(1/2), \xi(0)\bigr), \xi|_{[0,1/2]}\Bigr)
\]
as a glued lift of $\Bi(\gamma)$. Knowing this, we see that\footnote{Here, $\Delta_\gamma$ refers to the winding element homomorphism for $LH$ as defined in \cref{sec:unicol-struct}.\label{fn:bicol-unicolwindelt}}
\[
\Delta_{\Pth(\Bi\gamma)}' = \xi_\wh(q) - \xi_\bl(q) = \xi(1) - \xi(0) = \Delta_\gamma.
\]
One chooses a glued lift $(\eta_\wh, \eta_\match, \eta_\bl)$ for $\Bi(\rho)$ similarly using a lift $\eta\colon [0,1] \to \liealg h$ of $\rho$. This gives us
\begin{equation}
\label{eq:bicol-unicolcentext-isspecialcase-2}
\epsilon(\Delta_{\Pth(\Bi\gamma)}', \Delta_{\Pth(\Bi\rho)}') = \epsilon(\Delta_\gamma, \Delta_\rho).
\end{equation}

Next, using again that $\eta_\bl(q) = \eta(0)$ by our definition of $\eta_\bl$, we observe that the second line in the definition~\eqref{eq:bicol-cocycle} of the bicoloured cocycle now reads
\begin{multline*}
S\bigl((\xi_\wh, \xi_\match, \xi_\bl), (\eta_\wh, \eta_\match, \eta_\bl)\bigr) \\
\begin{aligned}
	&= \frac12 \biggl[\int_\circldir \langle \dd\xi_\wh, \eta_\wh \rangle_\Gamma + \int_\circrdir \langle \dd\xi_\bl, \eta_\bl \rangle_\Gamma + \bigl\langle \Delta_{\Pth(\Bi\gamma)}', \eta_\bl(q) \bigr\rangle_\Gamma\biggr] \\
	&= \frac12 \biggl[\int_{\frac12}^1 \bigl\langle \xi'(\theta), \eta(\theta) \bigr\rangle_\Gamma \dd\theta + \int_0^{\frac12} \bigl\langle \xi'(\theta), \eta(\theta) \bigr\rangle_\Gamma \dd\theta + \bigl\langle \Delta_\gamma, \eta(0) \bigr\rangle_\Gamma\biggr].
\end{aligned}
\end{multline*}
This is exactly the formula for $S_H(\xi, \eta)$ in~\eqref{eq:unicol-cocycle}, which together with~\eqref{eq:bicol-unicolcentext-isspecialcase-2} proves the claim~\eqref{eq:bicol-unicolcentext-isspecialcase-1}.

\subsection{The inclusion of $\centL H$}
\label{subsec:bicol-inclcentextLH}
We return to the general setup assumed at the beginning of \cref{sec:bicol-centext}. Recall from \cref{subsec:bicol-inclLH} that there is a canonical inclusion of abelian groups $\Bi\colon LH \hookrightarrow L(T_\wh, H, T_\bl)$ given by $\gamma \mapsto (\gamma_\wh, \gamma|_{\{p,q\}}, \gamma_\bl)$, where $\gamma_\wh := \phasegp\pi_\wh \circ \gamma|_\circlsm$ and $\gamma_\bl := \phasegp\pi_\bl\circ \gamma|_\circrsm$. Define a central extension $\centL H$ of $LH$ via the construction in \cref{sec:unicol-centext} using the following input data:
\begin{itemize}
\item the even lattice $(\Gamma, \langle \cdot, \cdot \rangle_\Gamma)$,
\item the restriction $\tilde\Gamma$ of the chosen $\phasegp$-central extension $(\Lambda_\wh - \Lambda_\bl)\centext$ of $\Lambda_\wh - \Lambda_\bl$ to $\Gamma$,
\item the restriction $\epsilon_\Gamma$ to $\Gamma$ of the chosen $2$-cocycle $\epsilon$ on $\Lambda_\wh - \Lambda_\bl$ for $(\Lambda_\wh - \Lambda_\bl)\centext$,
\item the point $q$ as a choice of privileged point on $S^1$.
\end{itemize}
The use of $\tilde\Gamma$ is permitted, meaning that it indeed has commutator map $(\mu, \mu') \mapsto (-1)^{\langle \mu, \mu' \rangle_\Gamma}$. The function $b_0$ in~\eqref{eq:bicol-latcommmap}, when restricted to $\Gamma$, namely becomes $\frac12 \langle \mu', \mu \rangle_\Gamma$ modulo $\ZZ$ since the restriction of the form $\langle \cdot, \cdot \rangle_\Gamma$ on $\Lambda_\wh - \Lambda_\bl$ to $\Gamma$ is integral, and so for $\mu, \mu' \in \Gamma$ we have
\[
b(\mu, \mu') = e^{2\pi i b_0(\mu, \mu')} = (-1)^{\langle \mu, \mu' \rangle_\Gamma}.
\]

We claim that the obvious lift to the underlying sets
\begin{equation}
\label{eq:bicol-liftofinclLHtocentexts}
\widetilde\Bi\colon \centL H \hookrightarrow \centL(T_\wh, H, T_\bl), \qquad (\gamma, z) \mapsto \bigl(\Bi(\gamma), z\bigr)
\end{equation}
of $\Bi$ is a homomorphism of (non-abelian) groups, meaning that for all $\gamma, \rho \in LH$,
\begin{equation}
\label{eq:bicol-inclLH-centext-1}
c\bigl(\Bi(\gamma), \Bi(\rho)\bigr) = c_H(\gamma, \rho).
\end{equation}
(We will use the same notations $c_H$ and $S_H$ as we did in \cref{subsec:bicol-unicolcentext-isspecialcase}.)

Indeed, if $\xi\colon [0,1] \to \liealg h$ is a lift of $\gamma$ as in~\eqref{eq:unicol-liealgpaths}, made by cutting $S^1$ at $q$, then we may choose as a glued lift $(\xi_\wh, \xi_\match, \xi_\bl)$ of $\Bi(\gamma)$ the triple of maps
\begin{gather*}
\xi_\wh := \RR\pi_\wh \circ \xi|_{[1/2,1]} \colon \circl \to \liealg t_\wh \\
\xi_\match(p) := \xi(1/2), \qquad \xi_\match(q) := \xi(0) \\
\xi_\bl := \RR\pi_\bl \circ \xi|_{[0,1/2]} \colon \circr \to \liealg t_\bl.
\end{gather*}
Knowing this, we see that\footnote{See \cref{fn:bicol-unicolwindelt}.}\footnote{This equality of winding elements also follows from the commutative triangle~\eqref{eq:bicol-LH-commtriangle}.}
\[
\Delta_{\Pth(\Bi\gamma)}' = (\RR\pi_\wh)^{-1}\bigl(\xi_\wh(q)\bigr) - (\RR\pi_\bl)^{-1}\bigl(\xi_\bl(q)\bigr) = \xi(1) - \xi(0) = \Delta_\gamma \in \Gamma.
\]
One chooses a glued lift $(\eta_\wh, \eta_\match, \eta_\bl)$ for $\Bi(\rho)$ similarly using a lift $\eta\colon [0,1] \to \liealg h$ of $\rho$. This gives us
\begin{equation}
\label{eq:bicol-inclLH-centext-2}
\epsilon(\Delta_{\Pth(\Bi\gamma)}', \Delta_{\Pth(\Bi\rho)}') = \epsilon_\Gamma(\Delta_\gamma, \Delta_\rho).
\end{equation}
Next, using again that $(\RR\pi_\bl)^{-1}(\eta_\bl(q)) = \eta(0)$, the second line in~\eqref{eq:bicol-cocycle} reads
\begin{multline*}
S\bigl((\xi_\wh, \xi_\match, \xi_\bl), (\eta_\wh, \eta_\match, \eta_\bl)\bigr) \\
\begin{aligned}
	&= \frac12 \biggl[\int_\circldir \langle \dd\xi_\wh, \eta_\wh \rangle_\wh + \int_\circrdir \langle \dd\xi_\bl, \eta_\bl \rangle_\bl + \Bigl\langle \Delta_{\Pth(\Bi\gamma)}', (\RR\pi_\bl)^{-1}\bigl(\eta_\bl(q)\bigr) \Bigr\rangle_\Gamma\biggr] \\
	&= \frac12 \biggl[\int_{\frac12}^1 \bigl\langle \xi'(\theta), \eta(\theta) \bigr\rangle_\Gamma \dd\theta + \int_0^{\frac12} \bigl\langle \xi'(\theta), \eta(\theta) \bigr\rangle_\Gamma \dd\theta + \bigl\langle \Delta_\gamma, \eta(0) \bigr\rangle_\Gamma\biggr].
\end{aligned}
\end{multline*}
This equals the formula for $S_H(\xi, \eta)$ in~\eqref{eq:unicol-cocycle}, which together with~\eqref{eq:bicol-inclLH-centext-2} proves~\eqref{eq:bicol-inclLH-centext-1}.

\begin{rmk}
\label{rmk:bicol-Bi-centext}
The existence of the homomorphisms~\eqref{eq:bicol-unicolcentext-isspecialcase} and~\eqref{eq:bicol-liftofinclLHtocentexts} of central extensions can be understood in a more conceptual, but less explicit way, namely via the pullback diagram~\eqref{eq:bicol-centext-pullback}.

For the lifting~\eqref{eq:bicol-unicolcentext-isspecialcase} of the isomorphism $\Bi$ we can observe that in the situation of \cref{subsec:bicol-unicolcentext-isspecialcase} we have $P(H, (\Lambda_\wh - \Lambda_\bl)/\Gamma) \cong LH$ and $\Pth = (\Bi)^{-1}$, if we identify $\circldir$ and $\circrdir$ with $[1/2,1]$ and $[0,1/2]$ respectively. So because $\Pth$ is an isomorphism, the same holds for its lift $\widetilde \Pth$. That the cocycle $c'$ for $\centP(H, (\Lambda_\wh - \Lambda_\bl)/\Gamma)$ is identical to the one defining $\centL H$ then shows that $\widetilde{\Pth}^{-1}$ is the desired lift of $\Bi$.

For the lifting~\eqref{eq:bicol-liftofinclLHtocentexts} of the inclusion $\Bi$, we can think of the cocycle $c'$ as generalising the one defining $\centL H$ to a bigger group if we consider $LH$ as a subgroup of $P(H, (\Lambda_\wh - \Lambda_\bl)/\Gamma)$ through~\eqref{eq:bicol-LH-as-gpofpaths}. Together with the commutativity of the triangle~\eqref{eq:bicol-LH-commtriangle} we therefore get a commutative diagram
\[
\begin{tikzcd}
\centL H \ar[rr, hookrightarrow, "\centiota"] \ar[d, twoheadrightarrow] &[1em] &[0.5em] \centP\bigl(H, (\Lambda_\wh - \Lambda_\bl)/\Gamma\bigr) \ar[d, twoheadrightarrow] \\
LH \ar[r, hookrightarrow, "\Bi"'] & L(T_\wh, H, T_\bl) \ar[r, twoheadrightarrow, "\Pth"'] & P\bigl(H, (\Lambda_\wh - \Lambda_\bl)/\Gamma\bigr).
\end{tikzcd}
\]
The universal property of the pullback~\eqref{eq:bicol-centext-pullback} now gives us the lift of $\Bi$.
\end{rmk}

\subsection{Isotony with respect to unicoloured central extensions}
\label{subsec:bicol-isot}
Recall the inclusions~\eqref{eq:bicol-noncentext-isot} of the non-centrally extended groups $L_\circlsm T_\wh$ and $L_\circrsm T_\bl$ of unicoloured loops into $L(T_\wh, H, T_\bl)$ given by $\gamma_\wh \mapsto (\gamma_\wh, 0_H, 0_{T_\bl})$ and $\gamma_\bl \mapsto (0_{T_\wh}, 0_H, \gamma_\bl)$ respectively. Define central extensions $\centL T_\whbl$ of $LT_\whbl$ via the construction in \cref{sec:unicol-centext} using the following input data:
\begin{itemize}
\item the even lattice $(\Lambda_\whbl, \langle \cdot, \cdot \rangle_\whbl)$ (or, more precisely, its pre-image in $\liealg h$ under the isometry $\RR\pi_\whbl$),
\item the restriction $\tilde\Lambda_\whbl$ of the chosen $\phasegp$-central extension $(\Lambda_\wh - \Lambda_\bl)\centext$ of $\Lambda_\wh - \Lambda_\bl$ to $(\RR\pi_\whbl)^{-1}(\Lambda_\whbl)$,
\item the restriction $\epsilon_\whbl$ to $(\RR\pi_\whbl)^{-1}(\Lambda_\whbl)$ of the chosen $2$-cocycle $\epsilon$ on $\Lambda_\wh - \Lambda_\bl$ for $(\Lambda_\wh - \Lambda_\bl)\centext$,
\item the point $q$ as a choice of privileged point on $S^1$.
\end{itemize}
Indeed, the commutator map of $\tilde\Lambda_\whbl$ is $(\lambda_\whbl, \mu_\whbl) \mapsto (-1)^{\langle \lambda_\whblsm, \mu_\whblsm \rangle_\whblsm}$. The function $b_0$ when restricted to $\Lambda_\whbl$ namely becomes $\frac12 \langle \mu_\whbl, \lambda_\whbl \rangle_\whbl$ modulo $\ZZ$ since the restriction $\langle \cdot, \cdot \rangle_\whbl$ of the form $\langle \cdot, \cdot \rangle_\Gamma$ on $\Lambda_\wh - \Lambda_\bl$ to $(\RR\pi_\whbl)^{-1}(\Lambda_\whbl)$ is integral.

Write $\centL_\circlsm T_\wh$ for the restriction of $\centL T_\wh$ to $L_\circlsm T_\wh$ and define $\centL_\circrsm T_\bl$ similarly. We claim that the obvious lifts to the underlying sets
\[
\centL_\circlsm T_\wh \hookrightarrow \centL(T_\wh, H, T_\bl) \hookleftarrow \centL_\circrsm T_\bl
\]
of the inclusions of $L_\circlsm T_\wh$ and $L_\circrsm T_\bl$, namely the ones that are the identity on the central subgroups $\phasegp$, are group homomorphisms. That is, for all $\gamma_\wh, \rho_\wh \in L_\circlsm T_\wh$,
\begin{equation}
\label{eq:bicol-unicolisotony-1}
c\bigl((\gamma_\wh, 0_H, 0_{T_\bl}), (\rho_\wh, 0_H, 0_{T_\bl})\bigr) = c_\wh(\gamma_\wh, \rho_\wh),
\end{equation}
where $c_\wh$ is the $2$-cocycle defining $\centL_\circlsm T_\wh$. We will also denote by $S_\wh$ the map in the definition of $c_\wh$ that takes two $\liealg t_\wh$-valued maps as input. Similarly, we claim that
\begin{equation}
\label{eq:bicol-unicolisotony-2}
c\bigl((0_{T_\wh}, 0_H, \gamma_\bl), (0_{T_\wh}, 0_H, \rho_\bl)\bigr) = c_\bl(\gamma_\bl, \rho_\bl)
\end{equation}
for all $\gamma_\bl, \rho_\bl \in L_\circrsm T_\bl$, where $c_\bl$ is the $2$-cocycle defining $\centL_\circrsm T_\bl$.

We will prove~\eqref{eq:bicol-unicolisotony-1}---the proof of~\eqref{eq:bicol-unicolisotony-2} is similar. If $\xi\colon [0,1] \to \liealg t_\wh$ is a lift of $\gamma_\wh$ as in~\eqref{eq:unicol-liealgpaths}, made by cutting $S^1$ at $q$, then we may choose the triple of maps
\[
(\xi_\wh, \xi_\match, \xi_\bl) := \bigl(\xi|_{[1/2, 1]} - \xi(1/2), 0_{\liealg h}, 0_{\liealg t_\bl}\bigr)
\]
as a glued lift of $(\gamma_\wh, 0, 0)$. Indeed, $\xi(1/2) \in \Lambda_\wh$. What is more, $\xi$ is constant on $[0, 1/2]$. Therefore,
\begin{align*}
\Delta_{\Pth(\gamma_\wh, 0, 0)}' &= (\RR\pi_\wh)^{-1}\bigl(\xi_\wh(q)\bigr) - (\RR\pi_\bl)^{-1}\bigl(\xi_\bl(q)\bigr) \\
	&= (\RR\pi_\wh)^{-1}\bigl(\xi(1) - \xi(1/2)\bigr) - 0_{\liealg h} \\
	&= (\RR\pi_\wh)^{-1}\bigl(\xi(1) - \xi(0)\bigr) = (\RR\pi_\wh)^{-1}(\Delta_{\gamma_\wh}).
\end{align*}
One chooses a glued lift $(\eta_\wh, \eta_\match, \eta_\bl)$ for $(\rho_\wh, 0, 0)$ similarly using a lift $\eta\colon [0,1] \to \liealg t_\wh$ of $\rho_\wh$. The result is
\begin{equation}
\label{eq:bicol-unicolisotony-3}
\epsilon(\Delta_{\Pth(\gamma_\wh, 0, 0)}', \Delta_{\Pth(\rho_\wh, 0, 0)}') = \epsilon_\wh\bigl((\RR\pi_\wh)^{-1}(\Delta_{\gamma_\wh}), (\RR\pi_\wh)^{-1}(\Delta_{\rho_\wh})\bigr).
\end{equation}
Filling in these glued lifts in the second line in the definition~\eqref{eq:bicol-cocycle} of the bicoloured cocycle results in
\begin{multline*}
S\bigl((\xi_\wh, \xi_\match, \xi_\bl), (\eta_\wh, \eta_\match, \eta_\bl)\bigr) \\
\begin{aligned}
	&= \frac12 \biggl[\int_\circldir \langle \dd\xi_\wh, \eta_\wh \rangle_\wh + \int_\circrdir \langle \dd\xi_\bl, \eta_\bl \rangle_\bl + \Bigl\langle \Delta_{\Pth(\gamma_\wh, 0, 0)}', (\RR\pi_\bl)^{-1}\bigl(\eta_\bl(q)\bigr) \Bigr\rangle_\Gamma\biggr] \\
	&= \frac12 \int_\circldir \langle \dd\xi_\wh, \eta_\wh \rangle_\wh \\
	&= \frac12 \int_{\frac12}^1 \bigl\langle \xi'(\theta), \eta(\theta) - \eta(1/2) \bigr\rangle_\wh \dd\theta \\
	&= \frac12 \int_{\frac12}^1 \bigl\langle \xi'(\theta), \eta(\theta) \bigr\rangle_\wh \dd\theta - \frac12 \bigl\langle \Delta_{\gamma_\wh}, \eta(0) \bigr\rangle_\wh \\
	&= \frac12 \int_{\frac12}^1 \bigl\langle \xi'(\theta), \eta(\theta) \bigr\rangle_\wh \dd\theta + \frac12 \bigl\langle \Delta_{\gamma_\wh}, \eta(0) \bigr\rangle_\wh - \bigl\langle \Delta_{\gamma_\wh}, \eta(0) \bigr\rangle_\wh.
\end{aligned}
\end{multline*}
Because $\eta(0) \in \Lambda_\wh$ and so $\langle \Delta_{\gamma_\wh}, \eta(0)\rangle \in \ZZ$, this is modulo $\ZZ$ the same formula as for $S_\wh(\xi, \eta)$ in~\eqref{eq:unicol-cocycle}. Together with~\eqref{eq:bicol-unicolisotony-3} this proves the claim~\eqref{eq:bicol-unicolisotony-1}.

\subsection{Disjoint-commutativity of central extensions}
\label{subsec:bicol-disjcomm}
We will prove a disjoint-commutativity property for $\centL(T_\wh, H, T_\bl)$ similar to the one in \cref{thm:unicol-disjcomm} for unicoloured torus loop groups. For this we will need a more restrictive notion of \emph{interval} than we used in the unicoloured case:

\begin{dfn}
\label{dfn:bicolinterval}
(After \parencite[Section 1.\textsc{a}]{bartels:confnetsIII}.\footnote{Our definition does not agree exactly with that of \parencite{bartels:confnetsIII} because those authors additionally require the datum of a local coordinate around the colour-changing point.}) A \defn{bicoloured interval} on $S^1$ is an interval on $S^1$ which does not contain both points $p$ and $q$, and if it does contain one of them it does so in its interior.
\end{dfn}

A bicoloured interval is therefore either contained in the interior of $\circl$ or $\circr$, or it is split into two (non-singleton) subintervals along the point from $\{p, q\}$ it contains.

The precise statement we will show is as follows:

\begin{thm}
\label{thm:bicol-disjcomm}
Let $(\gamma, z)$ and $(\rho, w)$ be two elements of the central extension $\centL(T_\wh, H, T_\bl)$ such that the supports of $\gamma$ and $\rho$ are contained in two disjoint bicoloured intervals on $S^1$, respectively. Then $(\gamma, z)$ and $(\rho, w)$ commute.
\end{thm}

Remember that here we are using the notion of support of \cref{dfn:bicol-support}.

As preparation for the proof of \cref{thm:bicol-disjcomm} we will first simplify the commutator map of $\centL(T_\wh, H, T_\bl)$ without assuming anything about supports of loops:

\begin{prop}
\label{thm:bicol-commmap}
Let $\gamma = (\gamma_\wh, \gamma_\match, \gamma_\bl)$ and $\rho = (\rho_\wh, \rho_\match, \rho_\bl)$ be in $L(T_\wh, H, T_\bl)$ and let the $2$-cocycle $c$ and the function $S$ be as in~\eqref{eq:bicol-cocycle} defined in terms of glued lifts $\xi = (\xi_\wh, \xi_\match, \xi_\bl)$ and $\eta = (\eta_\wh, \eta_\match, \eta_\bl)$ of $\gamma$ and $\rho$, respectively. Then
\begin{equation}
\label{eq:bicol-commmap-pre}
c(\gamma, \rho) c(\rho, \gamma)^{-1} = b(\Delta_{\Pth\gamma}', \Delta_{\Pth\rho}') e^{2\pi i\bigl(S(\xi, \eta) - S(\eta, \xi)\bigr)},
\end{equation}
where
\begin{equation}
\label{eq:bicol-commmap}
\begin{split}
S(\xi, \eta) - S(\eta, \xi) 
	&= \int_\circldir \langle \dd \xi_\wh, \eta_\wh \rangle_\wh +
		\int_\circrdir \langle \dd \xi_\bl, \eta_\bl \rangle_\bl - \phantom{{}} \\
	&\phantom{{}={}} \frac12 \langle \Delta_{\Pth\rho}', \Delta_{\Pth\gamma}' \rangle_\Gamma - 
		\Bigl\langle \Delta_{\Pth\rho}', (\RR \pi_\bl)^{-1}\bigl(\xi_\bl(q)\bigr) \Bigr\rangle_\Gamma.
\end{split}
\end{equation}
\end{prop}

The above \namecref{thm:bicol-commmap} can be proven via the definition~\eqref{eq:bicol-cocycle-viaP} by taking the result~\eqref{eq:unicol-commmap} about unicoloured central extensions and rewriting it in a bicoloured fashion using \cref{thm:bicol-liealglifts}\ref{thmitm:bicol-liealglifts-liftofP}. A different method is to imitate the steps in the proof of \cref{thm:unicol-disjcomm} for the bicoloured expression \eqref{eq:bicol-cocycle} instead, using \cref{thm:bicol-liealglifts}\ref{thmitm:bicol-liealglifts-exists-matchp} once.

With this calculation in hand we are ready for

\begin{proof}[Proof of \cref{thm:bicol-disjcomm}.]
We must show that the commutator~\eqref{eq:bicol-commmap-pre} vanishes. Write $I$ for the bicoloured interval containing the support of $\gamma$ and $J$ for the one corresponding to $\rho$. After possibly enlarging them we may assume that either $I$ contains $p$ and $J$ contains $q$, or the other way around. Suppose we are in the first situation, as for example in the picture
\[
\begin{tikzpicture}[>={Straight Barb}, line width=rule_thickness]	
\draw[postaction={decorate},
	  decoration={
		  markings,
		  mark=at position 0.0 with {
			  \arrow{>}
		  },
		  mark=at position 0.5 with {
			  \arrow{>}
		  }
	  }
	 ] circle [radius=1];
\draw[{Bracket[]}-{Bracket[]}] ([shift=(-160:1.2)]0,0) arc [radius=1.2, start angle=-160, delta angle=140] node [pos=0.5, below] {$J$};
\draw[{Bracket[]}-{Bracket[]}] ([shift=(20:1.2)]0,0) arc [radius=1.2, start angle=20, delta angle=140] node [pos=0.5, above] {$I$};
\fill (0,1) circle (1.2pt) node [below] {$p$};
\fill (0,-1) circle (1.2pt) node [above] {$q$};
\end{tikzpicture}
\]

Consider the first term in~\eqref{eq:bicol-commmap}. Since $\xi_\wh$ is constant with value $\xi_\wh(q) \in \Lambda_\wh$ outside of $I \cap \circl$, this integral is actually only taken over $I \cap \circl$. Since $I$ and $J$ are disjoint, $\eta_\wh$ is constant with value $\eta_\wh(p) \in \Lambda_\wh$ on $I \cap \circl$. We may therefore write
\begin{equation}
\label{eq:bicol-disjcomm1}
\int_\circldir \langle \dd \xi_\wh, \eta_\wh\rangle_\wh = \bigl\langle \xi_\wh(q) - \xi_\wh(p), \eta_\wh(p) \bigr\rangle_\wh.
\end{equation}
Similarly, the second term in~\eqref{eq:bicol-commmap} is
\begin{equation}
\label{eq:bicol-disjcomm2}
\int_\circrdir \langle \dd\xi_\bl, \eta_\bl\rangle_\bl = \bigl\langle \xi_\bl(p) - \xi_\bl(q), \eta_\bl(p) \bigr\rangle_\bl,
\end{equation}
where $\eta_\bl(p) \in \Lambda_\bl$. Next, apply $(\RR \pi_\wh)^{-1}$ to all entries in the right hand side of~\eqref{eq:bicol-disjcomm1} and $(\RR \pi_\bl)^{-1}$ to all entries in the right hand side of~\eqref{eq:bicol-disjcomm2}. We then recall from \cref{thm:bicol-liealglifts}\ref{thmitm:bicol-liealglifts-exists-matchp} that $(\RR\pi_\wh)^{-1}(\eta_\wh(p))$ and $(\RR\pi_\bl)^{-1}(\eta_\bl(p))$ are equal, and moreover, that they are equal to an element $\nu \in \Lambda_\wh \cap \Lambda_\bl$ since $(\RR\pi_\whbl)^{-1}(\eta_\whbl(p)) \in (\RR\pi_\whbl)^{-1}(\Lambda_\whbl)$. So we can now write the sum of the first two terms of~\eqref{eq:bicol-commmap} as
\[
\Bigl\langle (\RR\pi_\wh)^{-1}\bigl(\xi_\wh(q)\bigr) - (\RR\pi_\bl)^{-1}\bigl(\xi_\bl(q)\bigr), \nu \Bigr\rangle_\Gamma = \langle \Delta_{\Pth\gamma}', \nu \rangle_\Gamma.
\]
This lies in $\ZZ$ since $\Delta_{\Pth\gamma}' \in \Lambda_\wh - \Lambda_\bl$ and $\nu \in \Lambda_\wh \cap \Lambda_\bl$. We may therefore ignore these first two terms of~\eqref{eq:bicol-commmap}.

We will now focus on the last two terms of~\eqref{eq:bicol-commmap}. Because of our assumption on the support of $\gamma$, if we want to write $\Delta_{\Pth\gamma}' = \lambda_\wh - \lambda_\bl$ with $\lambda_\whbl \in (\RR\pi_\whbl)^{-1}(\Lambda_\whbl)$, we may choose $\lambda_\whbl := (\RR\pi_\whbl)^{-1}(\xi_\whbl(q))$. Hence we have
\[
b_0(\Delta_{\Pth\gamma}', \Delta_{\Pth\rho}') = \frac12 \langle \Delta_{\Pth\rho}', \Delta_{\Pth\gamma}' \rangle_\Gamma + \Bigl\langle \Delta_{\Pth\rho}', (\RR\pi_\bl)^{-1}\bigl(\xi_\bl(q)\bigr) \Bigr\rangle_\Gamma.
\]
We therefore see that $b(\Delta_{\Pth\gamma}', \Delta_{\Pth\rho}')$ cancels out the contribution of the exponential of the last two terms of~\eqref{eq:bicol-commmap}. We conclude that in this case indeed $c(\gamma, \rho)c(\rho, \gamma)^{-1} = 1$ and that $(\gamma, z)$ and $(\rho, w)$ commute.

Suppose that instead the interval $I$ contains $q$ and $J$ contains $p$. A similar reasoning as in the first situation shows that we can then write the sum of the first two terms of~\eqref{eq:bicol-commmap} as
\begin{align*}
&\Bigl\langle (\RR \pi_\wh)^{-1} \bigl(\xi_\wh(q)\bigr) - (\RR \pi_\wh)^{-1} \bigl(\xi_\wh(p)\bigr), (\RR \pi_\wh)^{-1} \bigl(\eta_\wh(q)\bigr) \Bigr\rangle_\Gamma + \phantom{{}} \\
&\Bigl\langle (\RR \pi_\bl)^{-1} \bigl(\xi_\bl(p)\bigr) - (\RR \pi_\bl)^{-1} \bigl(\xi_\bl(q)\bigr), (\RR \pi_\bl)^{-1} \bigl(\eta_\bl(q)\bigr) \Bigr\rangle_\Gamma.
\end{align*}
If we substitute $(\RR \pi_\wh)^{-1} (\eta_\wh(q))$ by $\Delta_{\Pth\rho}' + (\RR \pi_\bl)^{-1} (\eta_\bl(q))$ in the above and use that $(\RR \pi_\wh)^{-1} (\xi_\wh(p)) = (\RR \pi_\bl)^{-1} (\xi_\bl(p))$, this becomes
\begin{align*}
&\Bigl\langle (\RR \pi_\wh)^{-1} \bigl(\xi_\wh(q)\bigr) - (\RR \pi_\wh)^{-1} \bigl(\xi_\wh(p)\bigr), \Delta_{\Pth\rho}' \Bigr\rangle_\Gamma + \phantom{{}} \\
&\Bigl\langle (\RR \pi_\wh)^{-1} \bigl(\xi_\wh(q)\bigr) - (\RR \pi_\bl)^{-1} \bigl(\xi_\bl(q)\bigr), (\RR \pi_\bl)^{-1} \bigl(\eta_\bl(q)\bigr) \Bigr\rangle_\Gamma.
\end{align*}
Note that what is in the first slot in this second term is $\Delta_{\Pth\gamma}'$.

Having rewritten these first two terms of~\eqref{eq:bicol-commmap}, we fill it back in there so as to get
\begin{align*}
S(\xi, \eta) - S(\eta, \xi)
	&= \frac12 \langle \Delta_{\Pth\rho}', \Delta_{\Pth\gamma}' \rangle_\Gamma - \Bigl\langle (\RR \pi_\wh)^{-1} \bigl(\xi_\wh(p)\bigr), \Delta_{\Pth\rho}' \Bigr\rangle_\Gamma + \phantom{{}} \\
	 &\phantom{{}={}} \Bigl\langle \Delta_{\Pth\gamma}', (\RR \pi_\bl)^{-1} \bigl(\eta_\bl(q)\bigr) \Bigl\rangle_\Gamma.
\end{align*}
Because $(\RR\pi_\wh)^{-1}(\xi_\wh(p))$ and $(\RR\pi_\bl)^{-1}(\xi_\bl(p))$ are equal, and moreover, they are equal to an element $\nu \in \Lambda_\wh \cap \Lambda_\bl$, we see that the second term in the above lies in $\ZZ$. We may therefore ignore it from now on.

It follows from the definition of $b_0$ and the integrality of the lattices $\Lambda_\whbl$ that
\begin{align*}
b_0(\Delta_{\Pth\gamma}', \Delta_{\Pth\rho}')
	&= -\frac12 \langle \Delta_{\Pth\rho}', \Delta_{\Pth\gamma}' \rangle_\Gamma + \langle \Delta_{\Pth\rho}', \lambda_\wh \rangle_\Gamma \\
	&\equiv -\frac12 \langle \Delta_{\Pth\rho}', \Delta_{\Pth\gamma}' \rangle_\Gamma - \langle \Delta_{\Pth\gamma}', \mu_\bl \rangle_\Gamma \bmod \ZZ,
\end{align*}
where have written $\Delta_{\Pth\gamma}' = \lambda_\wh - \lambda_\bl$ and $\Delta_{\Pth\rho}' = \mu_\wh - \mu_\bl$ for some $\lambda_\whbl, \mu_\whbl \in (\RR\pi_\whbl)^{-1}(\Lambda_\whbl)$. Because of our assumptions on the support of $J$, we may choose $\mu_\whbl := (\RR \pi_\whbl)^{-1}(\eta_\whbl(q))$. Therefore, $b(\Delta_{\Pth\gamma}', \Delta_{\Pth\rho}')$ cancels against
\[
e^{2\pi i \bigl(S(\xi, \eta) - S(\eta, \xi)\bigr)}
\]
in~\eqref{eq:bicol-commmap-pre}. We conclude that also in this case $c(\gamma, \rho)c(\rho, \gamma)^{-1} = 1$.
\end{proof}

\section{Actions of covers of \texorpdfstring{$\Diff_+(S^1)$}{Diff+(S\textonesuperior)} on central extensions}
\label{sec:bicol-centext-diffnS1action}

We assume the setup of \cref{sec:bicol-centext}. That is, we take the input data from \cref{ingreds:bicol} as a given and use the notations $T_\whbl$, $H$, $\liealg t_\whbl$ and $\liealg h$ from \cref{sec:bicol-struct} for the tori and Lie algebras associated to the lattices $\Lambda_\whbl$ and $\Gamma$. We constructed from this data central extensions $\centP(H, (\Lambda_\wh - \Lambda_\bl)/\Gamma)$ and $\centL(T_\wh, H, T_\bl)$ of the group of paths $P(H, (\Lambda_\wh - \Lambda_\bl)/\Gamma)$ and the bicoloured torus loop group $L(T_\wh, H, T_\bl)$, respectively.

Recall from \cref{thm:bicol-diffnS1action} that $P(H, (\Lambda_\wh - \Lambda_\bl)/\Gamma)$ and $L(T_\wh, H, T_\bl)$ both carry actions of the group $\Diff_+^{(n)}(S^1)$, where $n$ is the smallest positive integer such that $n(\Lambda_\wh - \Lambda_\bl) \subseteq \Gamma$. In the proof of that \namecref{thm:bicol-diffnS1action} the action on $L(T_\wh, H, T_\bl)$ was built by first constructing it on $P(H, (\Lambda_\wh - \Lambda_\bl)/\Gamma)$ and then lifting it along the homomorphism $\Pth$ (see \cref{subsec:bicol-discontunicolloops} for the definition of $\Pth$). In this section we will similarly first construct a $\Diff_+^{(n)}(S^1)$-action on the central extension $\centP(H, (\Lambda_\wh - \Lambda_\bl)/\Gamma)$. This will be done in almost exactly the same way as for the unicoloured case in \cref{sec:unicol-diffs1action}. Next, we will lift this action to $\centL(T_\wh, H, T_\bl)$.

\subsection{The action of $\Diff_+^{(n)}(S^1)$ on $\centP(H, (\Lambda_\wh - \Lambda_\bl)/\Gamma)$}
\label{subsec:bicol-diffaction-pathgp}

The proof of \cref{thm:unicol-diffs1fail} also calculates the failure for $\Diff_+^{(n)}(S^1)$ to preserve the cocycle $c'$ defining $\centP(H, (\Lambda_\wh - \Lambda_\bl)/\Gamma)$. The result is as follows.

\begin{prop}
\label{thm:bicol-diffns1failP}
Let $\gamma, \rho \in P(H, (\Lambda_\wh - \Lambda_\bl)/\Gamma)$ and $\Xi, \Eta\colon \RR \to \liealg h$ be the quasi-periodic extensions to $\RR$ of choices of lifts $\xi, \eta\colon [0,1] \to \liealg h$ of $\gamma$ and $\rho$ respectively. Take $[\Phi] \in \Diff_+^{(n)}(S^1)$, where $\Phi \in \Diff_+^{(\infty)}(S^1)$ is a choice of representative. Then
\[
c'\bigl([\Phi]^* \gamma, [\Phi]^* \rho\bigr) = \epsilon(\Delta_\gamma', \Delta_\rho') e^{2\pi i S'(\Phi^* \Xi, \Phi^* \Eta)}
\]
and
\begin{align*}
S'(\Phi^* \Xi, \Phi^* \Eta) &= S'(\Xi, \Eta) + \frac12 \Bigl\langle \Delta_\gamma', \Eta\bigl(\Phi^{-1}(0)\bigr) - \Eta(0) \Bigr\rangle_\Gamma + \phantom{{}} \\
	&\phantom{{}= S'(\Xi, \Eta) + {}} \frac12 \Bigl\langle \Xi\bigl(\Phi^{-1}(0)\bigr) - \Xi(0), \Delta_\rho' \Bigr\rangle_\Gamma.
\end{align*}
\end{prop}

With this precise expression for the failure in hand we can now define the action of $\Diff_+^{(n)}(S^1)$ on $\centP(H, (\Lambda_\wh - \Lambda_\bl)/\Gamma)$. Let $[\Phi] \in \Diff_+^{(n)}(S^1)$ and $(\gamma, z) \in \centP(H, (\Lambda_\wh - \Lambda_\bl)/\Gamma)$. Then we set
\begin{equation}
\label{eq:bicol-diffns1-actionP}
[\Phi] \cdot (\gamma, z) := \Bigl([\Phi]^*\gamma, d'\bigl([\Phi], \gamma\bigr) \cdot z\Bigr),
\end{equation}
where $[\Phi]^*\gamma$ refers to the action of $\Diff_+^{(n)}(S^1)$ on the non-centrally extended group $P(H, (\Lambda_\wh - \Lambda_\bl)/\Gamma)$ constructed in the proof of \cref{thm:bicol-diffnS1action} and
\begin{equation}
\label{eq:bicol-diffns1-actionP-d}
d'\bigl([\Phi], \gamma\bigr) := e^{\pi i \bigl\langle \Xi(\Phi^{-1}(0)) - \Xi(0), \Delta_\gamma' \bigr\rangle_\Gamma} \in \phasegp.
\end{equation}

The value of $d'$ does not depend on the choice of lift $\Xi$. We show that it neither depends on the representative $\Phi$ of the equivalence class $[\Phi]$ in a similar, but slightly more subtle way compared to the unicoloured case. A different choice would namely be of the form $\Phi + nk$ for some $k \in \ZZ$, and seeing whether $d'([\Phi + nk], \gamma) = d'([\Phi], \gamma)$ comes down to proving that $nk \langle \Delta_\gamma', \Delta_\gamma' \rangle_\Gamma \in 2\ZZ$. We may assume that $k=1$. Because $\Delta_\gamma' \in \Lambda_\wh - \Lambda_\bl$ we can make a choice of decomposition $\Delta_\gamma' = \lambda_\wh - \lambda_\bl$ for some $\lambda_\whbl \in (\RR\pi_\whbl)^{-1}(\Lambda_\whbl)$. Then expand as follows:
\[
n \langle \Delta_\gamma', \Delta_\gamma' \rangle_\Gamma = n \langle \lambda_\wh, \lambda_\wh \rangle_\Gamma -2n \langle \lambda_\wh, \lambda_\bl \rangle_\Gamma + n \langle \lambda_\bl, \lambda_\bl \rangle_\Gamma.
\]
Given that the lattices $\Lambda_\whbl$ are even, the two outer terms on the right hand side are obviously in $2\ZZ$. But so is the middle term since $n \lambda_\wh \in \Gamma$ which gives $\langle n \lambda_\wh, \lambda_\bl \rangle_\Gamma \in \ZZ$. Hence $d'$ is well-defined.

The result of \cref{thm:bicol-diffns1failP} shows that $d'$ satisfies an equation similar to~\eqref{eq:diffs1action-d2}, with $[\Phi]$ and $c'$ in place of $\phi$ and $c$ respectively. That is,
\[
d'\bigl([\Phi], \cdot\bigr) \colon P\bigl(H, (\Lambda_\wh - \Lambda_\bl)/\Gamma\bigr) \to \phasegp
\]
is a $1$-cochain exhibiting the $2$-cocycle
\[
(\gamma, \rho) \mapsto c'\bigl([\Phi]^* \gamma, [\Phi]^* \rho\bigr) c'(\gamma, \rho)^{-1}
\]
as a $2$-coboundary. Therefore \eqref{eq:bicol-diffns1-actionP} defines an automorphism of $\centP(H, (\Lambda_\wh - \Lambda_\bl)/\Gamma)$. That $d'$ is compatible with the composition in $\Diff_+^{(n)}(S^1)$ is proven in literally the same way as in the unicoloured situation. We conclude that~\eqref{eq:bicol-diffns1-actionP} is a well-defined $\Diff_+^{(n)}(S^1)$-action on $\centP(H, (\Lambda_\wh - \Lambda_\bl)/\Gamma)$.

\subsection{The action of $\Diff_+^{(n)}(S^1)$ on $\centL(T_\wh, H, T_\bl)$}
\label{subsec:bicol-diffaction}

We are now ready to lift the $\Diff_+^{(n)}(S^1)$-action from $\centP(H, (\Lambda_\wh - \Lambda_\bl)/\Gamma)$ to $\centL(T_\wh, H, T_\bl)$. Let $[\Phi] \in \Diff_+^{(n)}(S^1)$ and $(\gamma, z) \in \centL(T_\wh, H, T_\bl)$. Then define
\begin{equation}
\label{eq:bicol-diffns1-action}
[\Phi] \cdot (\gamma, z) := \Bigl([\Phi] \cdot \gamma, d'\bigl([\Phi], \Pth\gamma\bigr) \cdot z\Bigr),
\end{equation}
where $[\Phi] \cdot \gamma$ refers to the action of $\Diff_+^{(n)}(S^1)$ on the non-centrally extended group $L(T_\wh, H, T_\bl)$ constructed in the proof of \cref{thm:bicol-diffnS1action}.

It is easily checked that, using that we already confirmed it for the action~\eqref{eq:bicol-diffns1-actionP} on $\centP(H, (\Lambda_\wh - \Lambda_\bl)/\Gamma)$, also~\eqref{eq:bicol-diffns1-action} defines an action of $\Diff_+^{(n)}(S^1)$ on $\centL(T_\wh, H, T_\bl)$. Alternatively, we can use the fact that $\centL(T_\wh, H, T_\bl)$ is the pullback of $\centP(H, (\Lambda_\wh - \Lambda_\bl)/\Gamma)$ along the $\Diff_+^{(n)}(S^1)$-equivariant homomorphism $\Pth$, and then note that such a pullback operation is functorial. By design of the $\Diff_+^{(n)}(S^1)$-actions, also the homomorphism $\widetilde\Pth$ that we defined in \cref{constr:bicol-centext} is $\Diff_+^{(n)}(S^1)$-equivariant.

\begin{prop}[$\Diff_+^{(n)}(S^1)$-equivariance with respect to $\centL H$]
\label{thm:bicol-inclcentextLH-diffnS1equiv}
The inclusion homomorphism $\widetilde{\Bi}$ in~\eqref{eq:bicol-liftofinclLHtocentexts} is equivariant with respect to the $\Diff_+(S^1)$-action on $\centL H$ defined in \cref{sec:unicol-diffs1action} and the $\Diff_+^{(n)}(S^1)$-action on $\centL(T_\wh, H, T_\bl)$ defined by~\eqref{eq:bicol-diffns1-action}.
\end{prop}
\begin{proof}
We need to prove that for all $\gamma \in LH$, $z \in \phasegp$ and $\Phi \in \Diff_+^{(\infty)}(S^1)$ there holds
\[
[\Phi] \cdot \bigl(\Bi(\gamma), z\bigr) = \Bigl(\Bi\bigl([\Phi]^*\gamma\bigr), d\bigl([\Phi], \gamma\bigr) \cdot z\Bigr),
\]
where $d$ is the $1$-cochain on $LH$ defined by~\eqref{eq:unicol-diffs1action}. The left hand side is by definition
\[
\Bigl([\Phi] \cdot \Bi(\gamma), d'\bigl([\Phi], \Pth(\Bi\gamma)\bigr) \cdot z\Bigr),
\]
and we already noted in \cref{rmk:bicol-diffnS1equivariance-wrtunicol} the equivariance of $\Bi$. Hence, what is left to show is that
\[
d'\bigl([\Phi], \Pth(\Bi\gamma)\bigr) = d\bigl([\Phi], \gamma\bigr).
\]
We see that this holds if we compare the expressions~\eqref{eq:bicol-diffns1-actionP-d} and~\eqref{eq:unicol-diffs1action}. Indeed, a quasi-periodic lift $\Xi\colon \RR \to \liealg h$ of $\Pth(\Bi\gamma)$ is also a lift of $\gamma$ thanks to the commutativity of the triangle~\eqref{eq:bicol-LH-commtriangle}, and we furthermore already learned that $\Delta'_{\Pth(\Bi\gamma)} = \Delta_\gamma$ in \cref{subsec:bicol-inclcentextLH}.
\end{proof}

\section{Irreducible, positive energy representations}
\label{sec:bicol-reptheory}

Assume the setup of \cref{sec:bicol-centext} which allowed us to construct a central extension $\centL(T_\wh, H, T_\bl)$ of the bicoloured torus loop group $L(T_\wh, H, T_\bl)$. In this section we will construct and classify the irreducible, positive energy representations of $\centL(T_\wh, H, T_\bl)$. This will be done in a way entirely similar to our work for unicoloured central extensions $\centL T$. The role of the normal subgroup $(\centL T)_0$ will in the bicoloured situation be played by a certain normal subgroup denoted by $\centker(\Delta' \circ \Pth)$. Let us therefore define this group and try to understand its structure.

We begin by studying the situation before taking central extensions. The short exact sequence~\eqref{eq:bicol-discontpaths-ses} involving the groups $L(T_\wh, H, T_\bl)$ and $P(H, (\Lambda_\wh - \Lambda_\bl)/H)$ does not split in general. However, its subsequence
\begin{equation}
\label{eq:bicol-discontpaths-subses}
\begin{tikzcd}
0 \ar[r] & \displaystyle\frac{\Lambda_\wh \cap \Lambda_\bl}{\Gamma} \ar[r] & (\Pth)^{-1}\bigl(\iota(LH)\bigr) \ar[r, "\Pth"] &[1em] \iota(LH) \ar[r] & 0
\end{tikzcd}
\end{equation}
\emph{does} split. The subgroup $(\Pth)^{-1}(\iota(LH))$ namely consists of those bicoloured loops $(\gamma_\wh, \gamma_\match, \gamma_\bl)$ for which the unique lifts $\hat\gamma_\whbl$ to $H$ of $\gamma_\whbl$ (see \cref{subsec:bicol-discontunicolloops} for their definitions) that match at $p$ also match with each other at $q$, but need not necessarily match with $\gamma_\match(q)$. (It might be helpful to compare this with the characterisation of the image of $\Bi$ in \cref{subsec:bicol-inclLH}.) That is, there holds $\hat\gamma_\wh(p) = \gamma_\match(p) = \hat\gamma_\bl(p)$ and $\hat\gamma_\wh(q) = \hat\gamma_\bl(q)$, but in general $\gamma_\match(q)$ is of the form
\[
\hat\gamma_\wh(q) + [\nu] = \hat\gamma_\bl(q) + [\nu]
\]
for some $[\nu] \in (\Lambda_\wh \cap \Lambda_\bl)/\Gamma \subseteq H$. Therefore, if we recall the definition of the first arrow in~\eqref{eq:bicol-discontpaths-subses} from \cref{subsec:bicol-discontunicolloops}, we see that the homomorphism $(\gamma_\wh, \gamma_\match, \gamma_\bl) \mapsto [\nu]$ is a left splitting of~\eqref{eq:bicol-discontpaths-subses}. Equivalently, a right splitting can be given by making the proof in~\cref{subsec:bicol-discontunicolloops} of the surjectivity of the homomorphism $\Pth$ more precise: we send a path $\gamma \in \iota(LH)$ to the bicoloured loop $(\gamma_\wh, \gamma_\match, \gamma_\bl)$, where
\begin{align*}
\gamma_\wh &:= \phasegp \pi_\wh \circ \gamma|_{[1/2, 1]} \colon \circl \to T_\wh \\
\gamma_\match(p) &:= \gamma(1/2) \in H \\
\gamma_\match(q) &:= \gamma(0) = \gamma(1) \in H \\
\gamma_\bl &:= \phasegp \pi_\bl \circ \gamma|_{[0, 1/2]} \colon \circr \to T_\bl.
\end{align*}

These splittings now allow us to understand the structure of $(\Pth)^{-1}(\iota(LH))$. There is an isomorphism of abelian groups
\begin{equation}
\label{eq:bicol-decomp-subgp}
(\Pth)^{-1}\bigl(\iota(LH)\bigr) \xrightarrow{\sim} \frac{\Lambda_\wh \cap \Lambda_\bl}{\Gamma} \oplus \iota(LH)
\end{equation}
sending a bicoloured loop $\gamma = (\gamma_\wh, \gamma_\match, \gamma_\bl)$ to the pair $([\nu], \Pth \gamma)$. The $\Diff_+^{(n)}(S^1)$-action on the full group $L(T_\wh, H, T_\bl)$ descends to $\Diff_+(S^1)$ when restricted to the subgroup $(\Pth)^{-1}(\iota(LH))$. This is because the shift diffeomorphism $\theta \mapsto \theta + 1$ of $\RR$, acting on $P(H, (\Lambda_\wh - \Lambda_\bl)/\Gamma)$, acts trivially when restricted to $\iota(LH)$. The isomorphism~\eqref{eq:bicol-decomp-subgp} is $\Diff_+(S^1)$-equivariant if we let $\Diff_+(S^1)$ act on the right hand side by only affecting the $\iota(LH)$-summand thanks to the equivariance of the homomorphism $\Pth$.

Now define a central extension $\centL H$ of $LH$ using the input data listed in \cref{subsec:bicol-inclcentextLH}. As explained in \cref{rmk:bicol-Bi-centext}, the homomorphism $\iota$ of abelian groups then lifts to a homomorphism
\[
\centiota\colon \centL H \hookrightarrow \centP\bigl(H, (\Lambda_\wh - \Lambda_\bl)/\Gamma\bigr)
\]
of non-abelian groups which is the identity on the central subgroups $\phasegp$. The isomorphism~\eqref{eq:bicol-decomp-subgp} then obviously lifts to a $\Diff_+(S^1)$-equivariant isomorphism from the restriction of $\centL(T_\wh, H, T_\bl)$ to $(\Pth)^{-1}(\iota(LH))$ towards
\begin{equation}
\label{eq:bicol-pathgp-subgp}
\frac{\Lambda_\wh \cap \Lambda_\bl}{\Gamma} \oplus \centiota(\centL H).
\end{equation}
It sends an element $(\gamma, z)$ to $([\nu], (\Pth \gamma, z))$.

The subgroup $(\Pth)^{-1}(\iota(LH)_0)$ of $(\Pth)^{-1}(\iota(LH))$ deserves special attention for the study of the representation theory of $\centL(T_\wh, H, T_\bl)$ we are about to commence. It is the inverse image under $\Pth$ of the identity component $\iota(LH)_0$ of $P(H, (\Lambda_\wh - \Lambda_\bl)/\Gamma)$, and contains the identity component $\ker \Delta$ of $L(T_\wh, H, T_\bl)$ (strictly, unless $\Gamma = \Lambda_\wh \cap \Lambda_\bl$ so that $\Pth$ is an isomorphism). An equivalent characterisation of $(\Pth)^{-1}(\iota(LH)_0)$ is that it is the kernel of the composite homomorphism
\[
\begin{tikzcd}
	L(T_\wh, H, T_\bl) \ar[r, twoheadrightarrow, "\Pth"] &[1em]
	P\bigl(H, (\Lambda_\wh - \Lambda_\bl)/\Gamma\bigr) \ar[r, twoheadrightarrow, "\Delta'"] &[0.5em]
	\Lambda_\wh - \Lambda_\bl.
\end{tikzcd}
\]
We will henceforth use the notation $\ker(\Delta' \circ \Pth)$ instead.

We write $\centker(\Delta' \circ \Pth)$ for the restriction of $\centL(T_\wh, H, T_\bl)$ to $\ker(\Delta' \circ \Pth)$ and observe that the restriction of the aforementioned isomorphism having codomain~\eqref{eq:bicol-pathgp-subgp} takes the form
\begin{equation}
\label{eq:bicol-decompsubgp-centext}
\centker(\Delta' \circ \Pth) \xrightarrow{\sim} \frac{\Lambda_\wh \cap \Lambda_\bl}{\Gamma} \times \centiota(\centL H)_0.
\end{equation}

\subsection{Irreducible representations of $\centker(\Delta' \circ \Pth)$}
\label{subsec:bicol-irrepsofsubgp}

Recall from \cref{subsec:unicol-Vtirreps,subsec:unicol-irrepsofidcomp} that we are able to equip the abelian group $(LH)_0$ and its central extension $(\centL H)_0$ with structures of topological groups using the bi-additive form on $\Gamma$. Via the isomorphisms~\eqref{eq:bicol-decomp-subgp} and~\eqref{eq:bicol-decompsubgp-centext} the groups $\ker(\Delta' \circ \Pth)$ and $\centker(\Delta' \circ \Pth)$ then acquire topological group structures as well if we give the finite abelian group $(\Lambda_\wh \cap \Lambda_\bl)/\Gamma$ the discrete topology.

Now define for every pair of characters $\chi$ and $l$ of the finite abelian group $(\Lambda_\wh \cap \Lambda_\bl)/\Gamma$ and the torus $H$, respectively, a representation $W_{\chi, l}$ of $\centker(\Delta' \circ \Pth)$ on the Hilbert space tensor product
\[
\hilb S_{\chi, l} := \CC_\chi \otimes_\CC \hilb S_l := \CC_\chi \otimes_\CC \CC_l \otimes_\CC \hilb S,
\]
where $\CC_\chi$ and $\CC_l$ denote copies of $\CC$, as follows. Let $(\gamma, z) \in \centker(\Delta' \circ \Pth)$ and consider its image
\[
\bigl([\nu], (\Pth \gamma, z)\bigr) \in \frac{\Lambda_\wh \cap \Lambda_\bl}{\Gamma} \times \centiota(\centL H)_0
\]
under the isomorphism~\eqref{eq:bicol-decompsubgp-centext}. Then make $[\nu]$ (which can alternatively be written as $\gamma_\match(q) - \Pth(\gamma)(1)$) act on $\CC_\chi$ via $\chi$, and let $(\Pth \gamma, z)$ act on $\hilb S_l$ via the representation $W_l$ of $\centL H$ defined in \cref{subsec:unicol-irrepsofidcomp}. That is, $W_{\chi, l}$ is the tensor product representation of $\chi$ and $W_l$. It is irreducible because $\chi$ and $W_l$ are.

\begin{rmk}
\label{rmk:bicol-irrepsubgp-depl}
To clarify how $W_{\chi,l}$ depends on $l$, we note that
\[
W_{\chi,l}(\gamma, z) = e^{2\pi i \langle l, \avg \hat\xi \rangle_\Gamma} \cdot W_{\chi,0}(\gamma, z),
\]
where $\hat\xi\colon [0,1] \to \liealg h$ is any lift of $\Pth \gamma$ as in~\eqref{eq:bicolP-liealgpaths} and
\[
\avg \hat\xi := \int_0^1 \hat\xi(\theta) \dd\theta.
\]
\end{rmk}

Let $m$ be, as in \cref{subsec:unicol-irrepsofidcomp}, the smallest positive integer such that $m\langle l, l \rangle \in 2\ZZ$. Now define a representation $R_{\chi,l}$ of $\Rot^{(m)}(S^1)$ on $\hilb S_{\chi, l}$ by acting as the identity on the tensor factor $\CC_\chi$ and as $R_l$ on $\hilb S_l$. So $R_{\chi,l}$ can be said to be equal to $R_l$ and the subscript $\chi$ is actually irrelevant---it serves, just like in the notation $\hilb S_{\chi, l}$, as a reminder that $R_{\chi,l}$ is associated to the representation $W_{\chi,l}$ and the latter \emph{does} depend on $\chi$.

\begin{prop}
The $\Rot^{(m)}(S^1)$-action $R_{\chi, l}$ on the Hilbert space $\hilb S_{\chi, l}$ intertwines in the manner~\eqref{eq:rotS1intertwin} with the representation $W_{\chi,l}$ of $\centker(\Delta' \circ \Pth)$.
\end{prop}
\begin{proof}
Let $[\Phi_\theta] \in \Rot^{(m)}(S^1)$, $(\gamma, z) \in \centker(\Delta' \circ \Pth)$ and $1 \otimes v \in \hilb S_{\chi, l}$ a vector with $v \in \hilb S_l$. Then there holds on the one hand
\begin{multline*}
R_{\chi, l}[\Phi_\theta] W_{\chi, l}(\gamma, z) R_{\chi, l}[\Phi_\theta]^*(1 \otimes v) \\
\begin{aligned}
	&= R_{\chi, l}[\Phi_\theta] W_{\chi, l}(\gamma, z) \bigl(1 \otimes R_l[\Phi_\theta]^*(v)\bigr) \\
	&= R_{\chi, l}[\Phi_\theta] \Bigl(\chi\bigl(\gamma_\match(q) - \Pth(\gamma)(1)\bigr) \otimes W_l(\Pth \gamma, z) R_l[\Phi_\theta]^*(v)\Bigr) \\
	&= \chi\bigl(\gamma_\match(q) - \Pth(\gamma)(1)\bigr) \otimes R_l[\Phi_\theta] W_l(\Pth \gamma, z) R_l[\Phi_\theta]^*(v),
\end{aligned}
\end{multline*}
and because of the way that $R_l$ intertwines with $W_l$ we can write
\[
R_l[\Phi_\theta] W_l(\Pth \gamma, z) R_l[\Phi_\theta]^* = W_l\bigl([\Phi_\theta] \cdot (\Pth\gamma, z)\bigr) = W_l\bigl([\Phi_\theta]^*\Pth(\gamma), z\bigr).
\]
On the other hand, since also $[\Phi_\theta] \cdot (\gamma, z) = ([\Phi_\theta]^*\gamma, z)$,
\begin{multline*}
W_{\chi,l}\bigl([\Phi_\theta] \cdot (\gamma, z)\bigr)(1 \otimes v) \\
	= \chi\Bigl(\bigl([\Phi_\theta]^*\gamma\bigr)_\match(q) - \Pth\bigl([\Phi_\theta]^*\gamma\bigr)(1)\Bigr) \otimes W_l\bigl(\Pth [\Phi_\theta]^*\gamma, z\bigr).
\end{multline*}
We have by the definition of the action of $\Diff_+^{(n)}(S^1)$ on $L(T_\wh, H, T_\bl)$ in the proof of \cref{thm:bicol-diffnS1action} that $\Pth ([\Phi_\theta]^*\gamma) = [\Phi_\theta] \cdot (\Pth \gamma)$ and that
\begin{multline*}
\bigl([\Phi_\theta]^*\gamma\bigr)_\match(q) - \Pth\bigl([\Phi_\theta]^*\gamma\bigr)(1) \\
= \gamma_\match(q) - (\Pth \gamma)(0) + \bigl([\Phi_\theta]^*\Pth(\gamma)\bigr)(0) - \bigl([\Phi_\theta]^*\Pth(\gamma)\bigr)(1).
\end{multline*}
Because both $\Pth \gamma$ and $[\Phi_\theta]^*\Pth(\gamma)$ lie in $\iota(LH)$ we have $(\Pth \gamma)(0) = (\Pth \gamma)(1)$ and
\[
\bigl([\Phi_\theta]^*\Pth(\gamma)\bigr)(0) = \bigl([\Phi_\theta]^*\Pth(\gamma)\bigr)(1).
\]
We conclude that indeed
\[
R_{\chi, l}[\Phi_\theta] W_{\chi, l}(\gamma, z) R_{\chi, l}[\Phi_\theta]^*
=
W_{\chi,l}\bigl([\Phi_\theta] \cdot (\gamma, z)\bigr).\qedhere
\]
\end{proof}

If $a \in (1/m)\ZZ$ is an energy level, then the $a$-th energy eigenspace $\hilb S_{\chi, l}(a)$ for $R_{\chi,l}$ is simply $\CC_\chi \otimes \hilb S_l(a)$. Since we already concluded in \cref{thm:unicol-posenergyidcomp} that $\hilb S_l(a)$ is zero for $a<0$ it follows that

\begin{prop}
The intertwining $\Rot^{(m)}(S^1)$-action $R_{\chi, l}$ on the representation $W_{\chi, l}$ of $\centker(\Delta' \circ \Pth)$ is of positive energy.
\end{prop}

The aforementioned observation implies more precisely that the character of $R_{\chi,l}$ equals that of $R_l$:
\begin{equation}
\label{eq:bicol-charsubgprep}
\ch_{R_{\chi,l}}(q) = \ch_{R_l}(q) = q^{\dim H/24} q^{\langle l, l \rangle/2} \eta(q)^{-\dim H}.
\end{equation}
We define the corresponding normalised character to be
\[
Z_{W_{\chi,l}}(q) := q^{-\dim H/24} \ch_{R_{\chi,l}}(q) = q^{\langle l, l \rangle/2} \eta(q)^{-\dim H}.
\]

We have thus constructed a $2$-parameter family of mutually non-isomorphic, irreducible, positive energy representations $W_{\chi,l}$ of $\centker(\Delta' \circ \Pth)$. It exhausts the class of such representations:

\begin{thm}
\label{thm:bicol-classifysubgpirreps}
Every irreducible, positive energy representation of $\centker(\Delta' \circ \Pth)$ such that the central subgroup $\phasegp$ acts as $z \mapsto z$ is isomorphic to $W_{\chi, l}$ for some characters $\chi$ and $l$ of $(\Lambda_\wh \cap \Lambda_\bl)/\Gamma$ and $H$, respectively.
\end{thm}
\begin{proof}
The proof is identical to that of \cref{thm:unicol-classifyidcompirreps}. We first use the $\Rot(S^1)$-equivariant isomorphism~\eqref{eq:bicol-decompsubgp-centext} to see an arbitrary such representation $Q$ as an irreducible, positive energy representation of $\centV \liealg h$. Here, $\centV \liealg h$ is the Heisenberg group sitting inside $(\centL H)_0$ as described in \cref{subsec:unicol-centext-subgps}. An appeal to the unicity result \cref{thm:unicol-classifyVt-irreps} for $\centV \liealg h$ then concludes the argument.
\end{proof}

\subsection{Irreducible representations of $\centL(T_\wh, H, T_\bl)$}
\label{subsec:bicol-irrepsfullgp}

By \cref{thm:topgpfromsubgp} there exists a unique topological group structure on $L(T_\wh, H, T_\bl)$ that is naturally induced by that of its subgroup $\ker(\Delta' \circ \Pth)$ under the requirement that this subgroup is open in $L(T_\wh, H, T_\bl)$. Because for the homomorphism
\[
\Delta'\colon P\bigl(H, (\Lambda_\wh - \Lambda_\bl)/\Gamma\bigr) \to \Lambda_\wh - \Lambda_\bl
\]
there holds $\ker \Delta' = \iota(LH)_0$, we similarly have a unique structure of a topological group on $P(H, (\Lambda_\wh - \Lambda_\bl)/\Gamma)$ induced by that of $(LH)_0$ such that $\iota(LH)_0$ is open in the former group.

\begin{prop}
There exists a unique structure of a topological group on the central extension $\centL(T_\wh, H, T_\bl)$ such that $\centker(\Delta' \circ \Pth)$ is open in $\centL(T_\wh, H, T_\bl)$.
\end{prop}
\begin{proof}
We show this along the same lines as \cref{thm:unicol-topgpstruct}. That is, we need to check whether for every fixed bicoloured loop $\gamma \in L(T_\wh, H, T_\bl)$ the map $\ker(\Delta' \circ \Pth) \to \phasegp$ given by
\[
\rho \mapsto c(\gamma, \rho) c(\gamma+\rho, -\gamma) = c'(\Pth \gamma, \Pth \rho) c'(\Pth \gamma + \Pth\rho, -\Pth\gamma)
\]
is continuous. Note that the homomorphism $\Delta'$ admits a splitting. It is given by defining for $\lambda \in \Lambda_\wh - \Lambda_\bl$ a path $\gamma_\lambda \in P(H, (\Lambda_\wh - \Lambda_\bl)/\Gamma)$ as the projection on $H$ of the Lie algebra-valued path $[0,1] \to \liealg h$, $\theta \mapsto \theta\lambda$. Therefore, there exists a decomposition
\[
P\bigl(H, (\Lambda_\wh - \Lambda_\bl)/\Gamma\bigr)
	\xrightarrow{\sim} \iota(LH)_0 \oplus (\Lambda_\wh - \Lambda_\bl)
	\xrightarrow{\sim} H \oplus V\liealg h \oplus (\Lambda_\wh - \Lambda_\bl)
\]
generalising the isomorphism~\eqref{eq:unicol-decompfullgp} in the unicoloured situation. The remainder of the argument now proceeds as in the proof of \cref{thm:unicol-topgpstruct}.
\end{proof}

The knowledge of the irreducible, positive energy representations $W_{\chi, l}$ of the subgroup $\centker(\Delta' \circ \Pth)$ now allows us to construct and classify the same class of representations of the full group $\centL(T_\wh, H, T_\bl)$, starting as follows. Let us take such a $W_{\chi, l}$ for characters $\chi$ and $l$ of $(\Lambda_\wh \cap \Lambda_\bl)/\Gamma$ and $H$, respectively, and consider the induced representation
\[
\Ind_{\centker(\Delta' \circ \Pth)}^{\centL(T_\wh, H, T_\bl)} W_{\chi, l}
\]
of $\centL(T_\wh, H, T_\bl)$. We will shorten it and its underlying Hilbert space to $\Ind W_{\chi, l}$ and $\Ind \hilb S_{\chi,l}$, respectively. We refer to \cref{subsec:indreps} and the analogous construction for unicoloured torus loop groups in \cref{subsec:unicol-irrepsLT} for details on how this Hilbert space and its action of $\centL(T_\wh, H, T_\bl)$ are defined.

Now take $m$ to be the smallest positive integer such that both $m \geq n$ and $m\langle l, l \rangle \in 2\ZZ$. The first condition implies according to \cref{subsec:bicol-diffaction} that $\Rot^{(m)}(S^1)$ acts on $\centL(T_\wh, H, T_\bl)$, while the second one means by \cref{subsec:bicol-irrepsofsubgp} that $\hilb S_{\chi, l}$ carries a positive energy representation $R_{\chi,l}$ of $\Rot^{(m)}(S^1)$ which intertwines with $W_{\chi, l}$. Then define a representation $\Ind R_{\chi, l}$ on $\Ind \hilb S_{\chi, l}$ in terms of $R_{\chi,l}$ in exactly the same way as done in \cref{subsec:unicol-irrepsLT}. It satisfies the intertwining property~\eqref{eq:rotS1intertwin} with respect to $\Ind W_{\chi, l}$ because $R_{\chi, l}$ does so with respect to $W_{\chi, l}$.

To study $\Ind W_{\chi,l}$ we first calculate the representations of $\centker(\Delta' \circ \Pth)$ that are conjugate to $W_{\chi,l}$:

\begin{lem}
\label{thm:bicol-conjrep}
(Compare with \cref{thm:unicol-conjrep}.)
Let $(\gamma, z)$ be an element of the group $\centL(T_\wh, H, T_\bl)$ that is not contained in the (normal) subgroup $\centker(\Delta' \circ \Pth)$ and consider the representation $W_{\chi, l}^{(\gamma, z)}$ of $\centker(\Delta' \circ \Pth)$ conjugate to $W_{\chi, l}$, defined by
\[
W_{\chi, l}^{(\gamma, z)}(\rho, w) := W_{\chi, l}\bigl((\gamma, z)^{-1} (\rho, w) (\gamma, z)\bigr)
\]
for $(\rho, w) \in \centker(\Delta' \circ \Pth)$. Then $W_{\chi, l}^{(\gamma, z)}$ is the tensor product representation of $W_{\chi, l}$ and the character
\begin{equation}
\label{eq:bicol-centextchar}
\centker(\Delta' \circ \Pth) \twoheadrightarrow \ker(\Delta' \circ \Pth) \to \phasegp, \qquad (\rho, w) \mapsto \rho \mapsto c(\rho, \gamma) c(\gamma, \rho)^{-1},
\end{equation}
where $\phasegp$ acts on $\CC$ as $z \mapsto z$. In turn, for any glued lifts $\xi = (\xi_\wh, \xi_\match, \xi_\bl)$ and $\eta = (\eta_\wh, \eta_\match, \eta_\bl)$ of $\gamma$ and $\rho$, respectively, there holds
\begin{equation}
\label{eq:bicol-conjrep-1}
c(\rho, \gamma) c(\gamma, \rho)^{-1} = e^{2\pi i\bigl(S(\eta, \xi) - S(\xi, \eta)\bigr)},
\end{equation}
where
\begin{equation}
\label{eq:bicol-conjrep-2}
S(\eta, \xi) - S(\xi, \eta) = - \int_\circldir \langle \dd \xi_\wh, \eta_\wh \rangle_\wh - \int_\circrdir \langle \dd \xi_\bl, \eta_\bl \rangle_\bl.
\end{equation}
\end{lem}

The equations~\eqref{eq:bicol-conjrep-1} and~\eqref{eq:bicol-conjrep-2} follow from the expression for the commutator map associated to $c$ we found in \cref{thm:bicol-commmap}.

\begin{prop}
\label{thm:bicol-indrepisirred}
The induced representation $\Ind W_{\chi, l}$ of $\centL(T_\wh, H, T_\bl)$ is irreducible.
\end{prop}
\begin{proof}
It is sufficient to show that all the conjugate representations $W_{\chi, l}^{(\gamma, z)}$ as in \cref{thm:bicol-conjrep} are not isomorphic to $W_{\chi, l}$. To do this, we will examine the restriction of $W_{\chi, l}^{(\gamma, z)}$ to the subgroup of $\centker(\Delta' \circ \Pth)$ consisting of the elements of the form $(\rho, 1)$ where $\rho$ is a constant bicoloured loop, that is, $\rho \in \Bi(H)$. Because this subgroup is canonically isomorphic to $\Bi(H)$, we will denote it as such and write its elements simply as $\rho$.

Let $\rho \in \Bi(H)$, meaning that there is an $\alpha \in \liealg h$ such that the maps $\rho_\whbl$ are constant with values $\exp_\whbl(\RR\pi_\whbl)(\alpha)$ and $\rho_\match(p) = \rho_\match(q) = \exp_H \alpha$. Then
\[
\eta := \bigl((\RR\pi_\wh)(\alpha), \alpha, (\RR\pi_\bl)(\alpha)\bigr)
\]
is a glued lift of $\rho$. Plugging this into~\eqref{eq:bicol-conjrep-2} gives
\begin{align*}
S(\eta, \xi) - S(\xi, \eta)
	&= - \bigl\langle\xi_\wh(q) - \xi_\wh(p), (\RR\pi_\wh)(\alpha)\bigr\rangle_\wh - \bigl\langle\xi_\bl(p) - \xi_\bl(q), (\RR\pi_\bl)(\alpha)\bigr\rangle_\bl \\
	&= - \bigl\langle \Delta_{\Pth \gamma}', \alpha\bigr\rangle_\Gamma.
\end{align*}

So~\eqref{eq:bicol-centextchar} has $\iota(H)$ acting by the character $-\Delta_{\Pth \gamma}' \in \Lambda_\wh - \Lambda_\bl \hookrightarrow \Gamma^\vee$, which implies that $W_{\chi, l}^{(\gamma, z)}$ is letting $\Bi(H)$ act by $l - \Delta_{\Pth\gamma}'$. Because $\Delta_{\Pth\gamma}' \neq 0$, we have $l - \Delta_{\Pth\gamma}' \neq l$ and therefore $W_{\chi, l}^{(\gamma, z)}$ and $W_{\chi, l}$ are not isomorphic.
\end{proof}

The reason we are inducing up representations not from the identity component of $\centL(T_\wh, H, T_\bl)$, but from the larger subgroup $\centker(\Delta' \circ \Pth)$, is precisely to make the step in the above proof true where we note that $\Delta_{\Pth \gamma}' \neq 0$.

Let us examine when these representations $\Ind W_{\chi,l}$ are isomorphic or not. For the next result we will make use of the bicoloured loops $\gamma_{[\lambda_\wh,\lambda_\bl]}$ defined in \cref{subsec:bicol-struct} associated to elements $[\lambda_\wh, \lambda_\bl] \in (\Lambda_\wh \oplus \Lambda_\bl)/\Gamma$. Note that if
\[
\lambda = (\RR\pi_\wh)^{-1}(\lambda_\wh) - (\RR\pi_\bl)^{-1}(\Lambda_\bl)
			\in \Lambda_\wh - \Lambda_\bl
\]
and $\sigma$ is the (left) coset of $\centker(\Delta' \circ \Pth)$ in $\centL(T_\wh, H, T_\bl)$ consisting of all elements $(\gamma, z)$ such that $\Delta_{\Pth \gamma}' = \lambda$, then the element $(\gamma_{[\lambda_\wh,\lambda_\bl]}, 1)$ is a representative of $\sigma$.

We furthermore observe that because $\centker(\Delta' \circ \Pth)$ is a normal subgroup of $\centL(T_\wh, H, T_\bl)$, the restriction of $\Ind W_{\chi, l}$ to it restricts to each subspace $\hilb S_{\chi, l}^\sigma$ for all cosets $\sigma$.

\begin{thm}[Restriction of $\Ind W_{\chi, l}$ from $\centL(T_\wh, H, T_\bl)$ to $\centker(\Delta' \circ \Pth)$]
\label{thm:bicol-indrep-restrtosubgp}
Fix a character $l$ of $H$, an element $\lambda \in \Lambda_\wh - \Lambda_\bl$ and let $\sigma$ be the (left) coset of $\centker(\Delta' \circ \Pth)$ in $\centL(T_\wh, H, T_\bl)$ corresponding to $\lambda$. Pick any pre-image $[\lambda_\wh, \lambda_\bl]$ of $\lambda$ under the homomorphism
\[
(\RR\pi_\wh)^{-1} - (\RR\pi_\bl)^{-1}\colon \frac{\Lambda_\wh \oplus \Lambda_\bl}{\Gamma} \twoheadrightarrow \Lambda_\wh - \Lambda_\bl.
\]
Then the composite unitary map
\[
f_{\chi, l}^\sigma\colon \hilb S_{\chi, l}^\sigma \xrightarrow{\sim} \hilb S_{\chi, l} \xrightarrow{\sim} \hilb S_{\chi, l-\lambda}, \qquad \bigl[(\gamma_{[\lambda_\wh, \lambda_\bl]}, 1), v\bigr] \mapsto v \mapsto v
\]
intertwines the representations $(\Res \Ind W_{\chi, l})|_{\hilb S_{\chi, l}^\sigma}$ and $W_{\chi, l-\lambda}$ of $\centker(\Delta' \circ \Pth)$ and the representations $\Ind R_{\chi,l}|_{\hilb S_{\chi,l}^\sigma}$ and $R_{\chi,l-\lambda}$ of $\Rot^{(m)}(S^1)$.
\end{thm}
\begin{proof}
The first map in the composition $f_{\chi, l}^\sigma$ is an isomorphism from the restriction of $\Res \Ind W_{\chi, l}$ to $\hilb S_{\chi, l}^\sigma$ to the conjugate representation $W_{\chi, l}^{(\gamma_{[\lambda_\wh, \lambda_\bl]}, 1)}$. The latter was calculated partially in \cref{thm:bicol-conjrep}. For a more precise result we substitute a specific glued lift $(\xi_\wh, \xi_\match, \xi_\bl)$ of $\gamma_{[\lambda_\wh, \lambda_\bl]}$ into~\eqref{eq:bicol-conjrep-2}, namely the one defined in \cref{subsec:bicol-struct}. This gives\footnote{See \cref{rmk:bicol-irrepsubgp-depl} for the definition of $\avg \hat\eta$.}
\begin{align*}
S(\eta, \xi) - S(\xi, \eta)
	&= - \biggl\langle \lambda_\wh - (\RR\pi_\wh) \circ (\RR\pi_\bl)^{-1}(\lambda_\bl), \int_\circldir \eta_\wh(\theta) \dd\theta \biggr\rangle_\wh - \phantom{{}} \\
	&\phantom{{}= -} \biggl\langle (\RR\pi_\bl) \circ (\RR\pi_\wh)^{-1}(\lambda_\wh) - \lambda_\bl, \int_\circrdir \eta_\bl(\theta) \dd\theta \biggr\rangle_\bl \\
	&= - \biggl\langle (\RR\pi_\wh)^{-1}(\lambda_\wh) - (\RR\pi_\bl)^{-1}(\lambda_\bl), \phantom{{}} \\
	&\phantom{{}= - \biggl\langle} (\RR\pi_\wh)^{-1}\biggl(\int_\circldir \eta_\wh(\theta) \dd\theta\biggr) + (\RR\pi_\bl)^{-1}\biggl(\int_\circrdir \eta_\bl(\theta) \dd\theta\biggr) \biggr\rangle_\Gamma \\
	&= - \langle \lambda, \avg{\hat\eta} \rangle_\Gamma
\end{align*}
and therefore,
\[
c(\rho, \gamma_{[\lambda_\wh, \lambda_\bl]}) c(\gamma_{[\lambda_\wh, \lambda_\bl]}, \rho)^{-1} = e^{-2\pi i \langle \lambda, \avg{\hat\eta} \rangle_\Gamma}.
\]
By \cref{rmk:bicol-irrepsubgp-depl} we have
\[
e^{2\pi i \langle -\lambda, \avg \hat\eta \rangle_\Gamma} \cdot W_{\chi, l}(\rho, w)(v) = W_{\chi, l - \lambda}(\rho, w)(v)
\]
for all $v$. We conclude that the second map in the composition $f_{\chi, l}^\sigma$ is an isomorphism from $W_{\chi, l}^{(\gamma_{[\lambda_\wh, \lambda_\bl]}, 1)}$ to $W_{\chi, l-\lambda}$.

For the second claim of the \namecref{thm:bicol-indrep-restrtosubgp}, we first assert that $m$, which we defined to be smallest positive integer such that both $m\geq n$ and $m\langle l,l\rangle \in 2\ZZ$, is also the smallest positive integer $m$ such that both $m \geq n$ and $m\langle l-\lambda, l-\lambda\rangle \in 2\ZZ$. We namely have
\[
m\langle l-\lambda, l-\lambda\rangle
=
m\langle l,l \rangle - 2m\langle l,\lambda \rangle + m \langle \lambda, \lambda \rangle.
\]
Because $m\lambda \in \Gamma$, there holds $2m\langle l,\lambda \rangle \in 2\ZZ$. To show that $m \langle \lambda, \lambda \rangle \in 2\ZZ$ we refer to the proof of the well-definedness of $d'$ in \cref{subsec:bicol-diffaction-pathgp}. We conclude that $\Rot^{(m)}(S^1)$ acts on both $\hilb S_{\chi,l}$ and $\hilb S_{\chi,l-\lambda}$, although the restrictions of these actions to their respective tensor factors $\CC_l$ and $\CC_{l-\lambda}$ are different.

In order to now prove that $f_{\chi,l}^\sigma$ intertwines the representations $\Ind R_{\chi,l}|_{\hilb S_{\chi,l}^\sigma}$ and $R_{\chi,l-\lambda}$ by imitating the proof of the analogous statement in \cref{thm:unicol-indrep-restrtoidcomp} we require the following observations and calculations. Given the definition of $\gamma_{[\lambda_\wh, \lambda_\bl]}$, it is clear that
\[
\Pth(\gamma_{[\lambda_\wh, \lambda_\bl]}) = \exp_H \circ \hat \xi.
\]
Evaluating $\Delta'$ on $\Pth(\gamma_{[\lambda_\wh, \lambda_\bl]})$ therefore gives
\[
\hat\xi(1) - \hat\xi(0) = (\RR\pi_\wh)^{-1}(\lambda_\wh) - (\RR\pi_\bl)^{-1}(\lambda_\bl) = \lambda.
\]
(A different proof of this last fact uses that we already learned in \cref{subsec:bicol-struct} that $\Delta_{\gamma_{[\lambda_\wh, \lambda_\bl]}} = [\lambda_\wh, \lambda_\bl]$, together with the commutativity of the diagram~\eqref{eq:bicol-windelt-commdiag}.) It then follows from the formula for $d'$ in~\eqref{eq:bicol-diffns1-actionP-d} that for $[\Phi_\theta] \in \Rot^{(m)}(S^1)$
\[
d'\bigl([\Phi_\theta], \Pth(\gamma_{[\lambda_\wh, \lambda_\bl]})\bigr) = e^{-\pi i \langle \lambda, \lambda \rangle_\Gamma \theta}.
\]
There furthermore holds, by the construction of the $\Diff_+^{(m)}(S^1)$-action on the groups $P(H, (\Lambda_\wh - \Lambda_\bl)/\Gamma)$ and $L(T_\wh, H, T_\bl)$ described in the proof of \cref{thm:bicol-diffnS1action}, that
\[
[\Phi_\theta]^* \Pth(\gamma_{[\lambda_\wh, \lambda_\bl]}) = \Pth(\gamma_{[\lambda_\wh, \lambda_\bl]}) + \exp_H(-\lambda\theta)
\]
and
\[
\bigl([\Phi_\theta] \cdot \gamma_{[\lambda_\wh, \lambda_\bl]}\bigr)_\match(q) = \exp_H(-\lambda\theta),
\]
which implies
\[
[\Phi_\theta] \cdot \gamma_{[\lambda_\wh, \lambda_\bl]} = \gamma_{[\lambda_\wh, \lambda_\bl]} + \Bi\bigl(\exp_H(-\lambda\theta)\bigr).
\]
Collecting the above calculations, we see from the definition of the lifting of the $\Diff_+^{(m)}(S^1)$-action to $\centL(T_\wh, H, T_\bl)$ described in \cref{subsec:bicol-diffaction} that
\begin{align*}
[\Phi_\theta] \cdot (\gamma_{[\lambda_\wh, \lambda_\bl]}, 1)
	&= \Bigl([\Phi_\theta] \cdot \gamma_{[\lambda_\wh, \lambda_\bl]}, d'\bigl([\Phi_\theta], \Pth(\gamma_{[\lambda_\wh, \lambda_\bl]})\bigr)\Bigr) \\
	&= \Bigl(\gamma_{[\lambda_\wh, \lambda_\bl]} + \Bi\bigl(\exp_H(-\lambda\theta)\bigr), e^{-\pi i \langle \lambda, \lambda \rangle_\Gamma \theta}\Bigr).
\end{align*}

Next, we observe that by~\eqref{eq:bicol-cocycle-viaP},
\[
c\Bigl(\gamma_{[\lambda_\wh, \lambda_\bl]}, \Bi\bigl(\exp_H(-\lambda\theta)\bigr)\Bigr) = c'\bigl(\exp_H \circ \hat\xi, \exp_H(-\lambda\theta)\bigr)
\]
and the latter $2$-cocycle value is easily computed to be $e^{-2\pi i \langle \lambda, \lambda \rangle_\Gamma \theta}$. An imitation of the proof of \cref{thm:unicol-indrep-restrtoidcomp} now goes through without a hitch.
\end{proof}

Summarising, \cref{thm:bicol-indrep-restrtosubgp}, together with the isomorphism
\[
\centL(T_\wh, H, T_\bl)\big/\centker(\Delta' \circ \Pth) \xrightarrow{\sim} \Lambda_\wh - \Lambda_\bl
\]
induced by the homomorphisms $\widetilde\Pth$ and $\Delta'$, describe how $\Ind W_{\chi, l}$ combined with the intertwining $\Rot^{(m)}(S^1)$-action breaks up into irreducible subrepresentations after restriction to $\centker(\Delta' \circ \Pth) \rtimes \Rot^{(m)}(S^1)$. We namely have a unitary isomorphism
\[
\bigoplus_{\mathclap{\lambda \in \Lambda_\wh - \Lambda_\bl}} f_{\chi, l}^{\sigma_\lambda}\colon \Res_{\centL(T_\wh, H, T_\bl)}^{\centker(\Delta' \circ \Pth)} \Ind_{\centker(\Delta' \circ \Pth)}^{\centL(T_\wh, H, T_\bl)} W_{\chi, l}
\xrightarrow{\sim}
\bigoplus_{\mathclap{\lambda \in \Lambda_\wh - \Lambda_\bl}} W_{\chi, l-\lambda}
\]
of representations of $\centker(\Delta' \circ \Pth) \rtimes \Rot^{(m)}(S^1)$, where $\sigma_\lambda$ is the coset associated to $\lambda$ as in the statement of \cref{thm:bicol-indrep-restrtosubgp}.

One can now copy the proof of \cref{thm:unicol-indrepsisom} to show the following result. The applications of \cref{thm:unicol-indrep-restrtoidcomp} and \cref{thm:unicol-indrepisirrep} in that proof should be replaced by \cref{thm:bicol-indrep-restrtosubgp} and \cref{thm:bicol-indrepisirred}, respectively, and the role of the winding element homomorphism on a unicoloured torus loop group is taken over by the composition $\Delta' \circ \Pth$.

\begin{thm}
\label{thm:bicol-indrepsisom}
Two representations $\Ind W_{\chi, l}$ and $\Ind W_{\chi', l'}$ of $\centL(T_\wh, H, T_\bl)$, where $\chi$ and $\chi'$ are characters of $(\Lambda_\wh \cap \Lambda_\bl)/\Gamma$ and $l$ and $l'$ are characters of $H$, are (unitarily) isomorphic if and only if both $\chi' = \chi$ and $l' = l-\lambda$ for some $\lambda \in \Lambda_\wh - \Lambda_\bl$.
\end{thm}

We can calculate the character of $\Ind W_{\chi, l}$ in the same way as we did in the unicoloured situation. Let us fix an element $\lambda \in \Lambda_\wh - \Lambda_\bl$ and its corresponding coset $\sigma \subseteq \centL(T_\wh, H, T_\bl)$ as in the statement of \cref{thm:bicol-indrep-restrtosubgp}. Then we find via~\eqref{eq:bicol-charsubgprep} the character of the subspace $\hilb S_{\chi,l}^\sigma$ to be
\[
\ch_{\hilb S_{\chi,l}^\sigma}(q) = \ch_{R_{\chi,l-\lambda}}(q) = q^{\rank \Gamma/24} q^{\langle l-\lambda, l-\lambda\rangle_\Gamma / 2} \eta(q)^{-\rank \Gamma},
\]
since $\rank \Gamma = \dim H$. The character of $\Ind R_{\chi,l}$ is then obtained by summing over all the subspaces $\hilb S_{\chi,l}^\sigma$:
\begin{align*}
\ch_{\Ind R_{\chi,l}}(q)
	&= \sum_{\mathclap{\lambda \in \Lambda_\wh - \Lambda_\bl}} \ch_{\hilb S_{\chi,l}^\sigma}(q) \\
	&= q^{\rank \Gamma/24} \sum_{\mathclap{\lambda \in \Lambda_\wh - \Lambda_\bl}} q^{\langle l-\lambda, l-\lambda\rangle_\Gamma / 2} \cdot \eta(q)^{-\rank \Gamma} \\
	&= q^{\rank \Gamma/24} \theta_{l+(\Lambda_\wh - \Lambda_\bl)}(q) \cdot \eta(q)^{-\rank \Gamma}.
\end{align*}
Note that, even if $l = 0$, the character will in general be a series in non-integral powers of $q$, in contrast to the unicoloured situation. We learn in particular that

\begin{prop}
The intertwining $\Rot^{(m)}(S^1)$-action $\Ind R_{\chi,l}$ on the representation $\Ind W_{\chi, l}$ of $\centL(T_\wh, H, T_\bl)$ is of positive energy.
\end{prop}

We do not know whether $\ch_{\Ind R_{\chi,l}}(q)$ satisfies a type of modular behaviour, but we expect this not to hold in any case unless we apply an energy correction and define
\[
Z_{\Ind W_{\chi,l}}(q) := q^{-\rank \Gamma/24} \ch_{\Ind R_{\chi,l}}(q) = \theta_{l+(\Lambda_\wh - \Lambda_\bl)}(q) \cdot \eta(q)^{-\rank \Gamma}.
\]

\begin{thm}
\label{thm:bicol-classifyirreps}
Every irreducible, positive energy representation of $\centL(T_\wh, H, T_\bl)$ such that the central subgroup $\phasegp$ acts as $z \mapsto z$ is (unitarily) isomorphic to $\Ind W_{\chi, l}$ for some characters $\chi$ and $l$ of $(\Lambda_\wh \cap \Lambda_\bl)/\Gamma$ and $H$, respectively. The isomorphism classes of such representations are therefore labelled by two parameters: one is an element of the dual group of the finite abelian group $(\Lambda_\wh \cap \Lambda_\bl)/\Gamma$ and the other is an element of the finite abelian group $\Gamma^\vee/(\Lambda_\wh - \Lambda_\bl)$.
\end{thm}
\begin{proof}
This is shown along the lines of the proof of \cref{thm:unicol-classifyirreps}. The role of the subgroup $(\centL T)_0$ of $\centL T$ and the knowledge of the representation theory of the former afforded by \cref{thm:unicol-classifyidcompirreps} is taken over by that of the subgroup $\centker(\Delta' \circ \Pth)$ and \cref{thm:bicol-classifysubgpirreps}. We saw in \cref{thm:bicol-indrepisirred} that $\Ind W_{\chi, l}$ is irreducible and it is seen through \cref{thm:bicol-indrep-restrtosubgp} that for different cosets $\sigma$ the subspaces $\hilb S_{\chi,l}^\sigma$ of $\Ind \hilb S_{\chi,l}$ are mutually non-isomorphic representations of $\centker(\Delta' \circ \Pth)$.
\end{proof}

We conclude from \cref{thm:bicol-indrepsisom,thm:bicol-classifyirreps} that, up to the ambiguity of the character by which $\phasegp$ acts, $\centL(T_\wh, H, T_\bl)$ has only finitely many isomorphism classes of irreducible, positive energy representations. There exists precisely one isomorphism class, represented by $\Ind W_{0,0}$, if and only if both inclusions $\Gamma \subseteq \Lambda_\wh \cap \Lambda_\bl$ and $\Lambda_\wh - \Lambda_\bl \hookrightarrow (\Lambda_\wh \cap \Lambda_\bl)^\vee$ are equalities. Since $(\Lambda_\wh \cap \Lambda_\bl)^\vee = \Lambda_\wh^\vee - \Lambda_\bl^\vee$, the latter condition is fulfilled when $\Lambda_\wh$ and $\Lambda_\bl$ are both unimodular.
%

\chapter{Outlook}
\label{chap:outlook}

In this \namecref{chap:outlook} we offer some speculative thoughts on further directions in which to continue the study of bicoloured torus loop groups started in this thesis.

\section{Defects between lattice conformal nets}
\label{sec:defectslatticenets}

Recall from \cref{chap:intro} that our motivation for introducing and studying bicoloured torus loop groups is to find new examples of defects in the sense of \cref{dfn:defect} between lattice conformal nets. That is, if $\Lambda_\whbl$ are two even lattices that give rise to conformal nets $\net A_\whbl$ as explained in \cref{exmp:latticenets}, then we would like to construct (some) $\net A_\wh$--$\net A_\bl$-defects.

The results achieved in \cref{chap:bicol} suggest the following approach, which imitates the construction in \cref{exmp:latticenets}. If
\[
\Lambda_\wh \xhookleftarrow{\pi_\wh} \Gamma \xhookrightarrow{\pi_\bl} \Lambda_\bl
\]
is a span of even lattices, then belonging to it (and some minor extra data) there is a centrally extended bicoloured torus loop group $\centL(T_\wh, H, T_\bl)$. Here, $T_\whbl$ and $H$ are the tori associated to the above three lattices. When these lattices are positive definite this group has, up to isomorphism and the character by which the central subgroup $\phasegp$ acts, finitely many irreducible, positive energy representations denoted by $\Ind W_{\chi,l}$. We single out one of them, $\Ind W_{0,0}$, and dub it the \defn{vacuum representation} of $\centL(T_\wh, H, T_\bl)$. If $I \subseteq S^1$ is a bicoloured interval, where we consider $\circl$ as being white and $\circr$ as black, and $\centL_I(T_\wh, H, T_\bl)$ is the subgroup of those elements $(\gamma, z)$ of $\centL(T_\wh, H, T_\bl)$ for which $\supp \gamma \subseteq I$, then we may form the von Neumann algebra
\begin{equation}
\label{eq:outl-candvnalgs}
D_H(I) := \vN\Bigl((\Ind W_{0,0})\bigl(\centL_I(T_\wh, H, T_\bl)\bigr)\Bigr)
\end{equation}
acting on the underlying Hilbert space $\Ind \hilb S_{0,0}$ of $\Ind W_{0,0}$. Let $\centL T_\whbl$ be the centrally extended unicoloured torus loop groups described in \cref{subsec:bicol-isot} and assume that the lattice nets $\net A_\whbl$ are constructed from these. We then conjecture that there exists an $\net A_\wh$--$\net A_\bl$-defect $D_H$ of which the algebras attached to those bicoloured intervals that are embedded in $S^1$ are defined as in~\eqref{eq:outl-candvnalgs}.

This (candidate) family of von Neumann algebras obviously satisfies the functoriality, isotony and locality axioms of a defect. That also the requirement holds stating that $D_H$ restricts to $\net A_\wh$ and $\net A_\bl$ on $\circl$ and $\circr$, respectively, follows from the fact that the vacuum representations $\Ind W_0$ of $\centL T_\whbl$ inject as direct summands into $\Ind W_{0,0}$ in a way that is equivariant with respect to the group homomorphisms 
\[
\centL_\circlsm T_\wh \hookrightarrow \centL(T_\wh, H, T_\bl) \hookleftarrow \centL_\circrsm T_\bl.
\]

With the belief that the formalism of bicoloured torus loop groups produces lattice net defects comes the question whether such a defect can be alternatively constructed through other, more traditional methods. Recall namely that, even though the central extension $\centL(T_\wh, H, T_\bl)$ only contains the two unicoloured `halves' $\centL_\circlsm T_\wh$ and $\centL_\circrsm T_\bl$ of $\centL T_\wh$ and $\centL T_\bl$, respectively, there exists a $\Diff_+^{(n)}(S^1)$-equivariant inclusion of the full group $\centL H$ into it as well (see \cref{subsec:bicol-inclcentextLH} and \cref{thm:bicol-inclcentextLH-diffnS1equiv}). Let $\net A_\Gamma$ denote the lattice conformal net constructed from $\centL H$. One can then speculate that there is an inclusion $\net A_\Gamma(I) \hookrightarrow D_H(I)$ of the associated von Neumann algebras which is of so called \defn{finite index}. Even stronger: the defect $D_H$ might be a relatively local, finite index extension of $\net A_\Gamma$. As we already mentioned in \cref{exmp:qsystemdefects}, such extensions are known to be classified by data internal to the net $\net A_\Gamma$, namely via $Q$-systems. This would open a different angle to constructing these defects---one that follows operator-algebraic means.

\section{Generalising to different tori at the two defect points}

Adding to our conjecture that a centrally extended bicoloured torus loop group $\centL(T_\wh, H, T_\bl)$ gives rise to a defect $D_H$ through the construction~\eqref{eq:outl-candvnalgs}, we predict that any positive energy representation $Q$ of this group can be equipped with the structure of a $D_H$--$D_H$-sector in the sense of \parencite[Definition 2.2]{bartels:confnetsIII}:
\[
\begin{tikzcd}
\net A_\wh \ar[r, bend left=45, "D_H", ""{name=U, below}] \ar[r, bend right=45, "D_H"', ""{name=D}] &[1.5em] \net A_\bl
\ar[Rightarrow, from=U, to=D, "Q"]
\end{tikzcd}.
\]
This would generalise the fact that representations of this type of a centrally extended \defn{unicoloured} torus loop group correspond to representations of the associated conformal net (see \parencite[Proposition 3.15]{dong:latticeconfnets}), and hence to self-sectors of its identity defect.

The authors of \parencite{bartels:confnetsIII} define the more general notion of a sector between two \emph{different} defects as well. One could try to construct such sectors via an analogous generalisation of bicoloured loop groups as follows. In our \cref{dfn:bicol-gp} of a bicoloured torus loop group we placed identical matching conditions at the two defect points $p$ and $q$, defined in terms of a single, third even lattice $\Gamma$. Suppose that we are given instead the more complicated datum of a quadruple $(\Lambda_\wh, \Gamma_p, \Gamma_q, \Lambda_\bl)$ of even lattices, together with four lattice homomorphisms
\[
\begin{tikzcd}[sep=1em]
& \ar[dl, bend right, hook', "\pi_{p,\wh}"'] \Gamma_p \ar[dr, bend left, hook, "\pi_{p,\bl}"] & \\
\Lambda_\wh & & \Lambda_\bl \\
& \ar[ul, bend left, hook, "\pi_{q,\wh}"] \Gamma_q \ar[ur, bend right, hook', "\pi_{q,\bl}"'] & 
\end{tikzcd}.
\]
Tensoring this circle with $\phasegp$ over $\ZZ$ gives a circle of four tori $T_\whbl$, $H_p$, and $H_q$ with four surjective torus homomorphisms $\phasegp\pi_{p,\whbl}$ and $\phasegp\pi_{q,\whbl}$ between them. Define then the abelian group $L(T_\wh, H_p, H_q, T_\bl)$ to consist of all quadruples $(\gamma_\wh, \gamma_p, \gamma_q, \gamma_\bl)$, where $\gamma_\wh\colon \circl \to T_\wh$ and $\gamma_\bl\colon \circr \to T_\bl$ are smooth maps and $\gamma_p\colon \{p\} \to H_p$ and $\gamma_q\colon \{q\} \to H_q$ are functions, such that the following diagram commutes:
\[
\begin{tikzcd}[sep=1em]
&[1em] & \ar[ddll, bend right, two heads, "\phasegp\pi_{p,\wh}"'] H_p \ar[ddrr, bend left, two heads, "\phasegp\pi_{p,\bl}"] & &[1em] \\[1em]
& & \ar[dl, bend right, hook'] \{p\} \ar[dr, bend left, hook] \ar[u, "\gamma_p"'] & & \\
T_\wh & \ar[l, "\gamma_\wh"'] \circl & & \circr \ar[r, "\gamma_\bl"] & T_\bl \\
& & \ar[ul, bend left, hook] \{q\} \ar[ur, bend right, hook'] \ar[d, "\gamma_q"] & & \\[1em]
& & \ar[uull, bend left, two heads, "\phasegp\pi_{q,\wh}"] H_q \ar[uurr, bend right, two heads, "\phasegp\pi_{q,\bl}"'] & &
\end{tikzcd},
\]
and such that appropriate conditions on the derivatives of $\gamma_\whbl$ at $p$ and $q$ are satisfied.

While this definition is easily made, the only test for its usefulness is of course whether $L(T_\wh, H_p, H_q, T_\bl)$ admits $\phasegp$-central extensions with desirable properties. The positive energy representations of such a central extension could then possibly give rise to sectors between the two defects $D_p$ and $D_q$ constructed from $\centL(T_\wh, H_p, T_\bl)$ and $\centL(T_\wh, H_q, T_\bl)$.

Evidence that such central extensions might sometimes exist is given by the observation that in the special case when $\Lambda_\wh$, $\Lambda_\bl$, $\Gamma_p$ and $\Gamma_q$ are all equal to the same lattice $\Gamma$, and the homomorphisms $\pi_{p,\whbl}$ and $\pi_{q,\wh}$ are chosen to be the identity but $\pi_{q,\bl}$ is a lattice automorphism of $\Gamma$, the group $L(T_\wh, H_p, H_q, T_\bl)$ is isomorphic to the so called \defn{twisted} unicoloured torus loop group
\[
L_{(g)} H := \Bigl\{\gamma \in C^\infty\bigl([0,1], H\bigr) \Bigm\vert \gamma(1) = g_H\bigl(\gamma(0)\bigr)\Bigr\},
\]
where $g_H$ is the automorphism of the torus $H := \Gamma \otimes_\ZZ \phasegp$ induced by $g$. (In this definition of $L_{(g)}H$ we again omit conditions on the derivatives at $0$ and $1$ for $\gamma$.) Such twisted loop groups and their central extensions have been studied before (see \parencite[Section 3.7 and p.~175]{pressley:loopgps}) in the case that $g$ has finite order---an assumption that is indeed fulfilled when $\Gamma$ is positive definite, as explained in \cref{app:lattices}.


\section{Generalising from tori to non-abelian Lie groups}

Except for selected parts of the Introduction, we spent our attention in this thesis exclusively on torus loop groups and lattice conformal nets. One may wonder whether the notion of a bicoloured torus loop group $L(T_\wh, H, T_\bl)$ can be adapted to the situation when the two tori $T_\whbl$ are instead compact, connected, simple Lie groups $G_\whbl$. Of course, the definition of this group itself does make sense upon replacing $H$ with an arbitrary Lie group with homomorphisms towards $G_\whbl$, but the difficulty lies again in finding central extensions with desirable properties.

Motivated by the discussion in \parencite[Section 5]{kapustin:surfops} we suggest that the search for appropriate matching conditions for bicoloured loops in this non-abelian situation could be guided slightly better by the following observation. The authors of \parencite{kawahigashi:classifnonlocalc1} classify the finite index, relatively local extensions of the affine Kac--Moody nets $\net A_{\lieSU(2),k}$ for all levels $k \geq 1$. By the comments made in \cref{exmp:qsystemdefects} we hence obtain a family of defects from $\net A_{\lieSU(2),k}$ to itself. It then seems fruitful to ask whether there exists a family of matching conditions from $\lieSU(2)$ to itself, depending on a level $k$, with which one is able to reproduce these defects. We leave this investigation for future research.

\appendix
\chapter{Background material}
\label{chap:app}

In this \namecref{chap:app} we supply background material on various notions used throughout the thesis.

\section{Lattices}
\label{app:lattices}

This section is a brief introduction to lattices. We present some basic examples and we attempt to give the reader a feeling for their richness by discussing their automorphism groups and techniques on how to construct new examples of lattices from given ones. Their role in this thesis is that a lattice is the major ingredient for building the centrally extended unicoloured torus loop groups studied in \cref{chap:unicol} and, similarly, a span of even lattices is the most important input datum for the central extensions of bicoloured torus loop groups that we construct in \cref{chap:bicol}. We finish by defining theta series of lattices for their appearance when calculating the (graded) characters of the representations of these loop groups in \cref{sec:unicol-reptheory} and \cref{sec:bicol-reptheory}.

\begin{refs}
The material in this section is largely taken from \parencite{ebeling:lattices}, \parencite{conway:splag} and \parencite[§1]{nikulin:bilinforms}. More specific references will be given in the text.
\end{refs}

By a \defn{lattice} we mean a pair of a free $\ZZ$-module $\Lambda$ of finite rank, together with a non-degenerate, symmetric, bi-additive form $\langle \cdot, \cdot \rangle \colon \Lambda \times \Lambda \to \ZZ$. We will often omit the reference to the form and denote the lattice simply by $\Lambda$. We call $\Lambda$ \defn{even} if $\langle \lambda, \lambda \rangle \in 2\ZZ$ for all $\lambda \in \Lambda$, and \defn{odd} otherwise. The property of a lattice being \defn{positive definite}, \defn{negative definite} or \defn{indefinite} is defined in the obvious way. \defn{Morphisms} of lattices are injective $\ZZ$-module homomorphisms which respect the forms.

The automorphism group of $\Lambda$ is denoted by $\Aut(\Lambda; \langle \cdot, \cdot \rangle)$. It is a finite group if $\Lambda$ is definite since it is a discrete group on the one hand, while on the other hand, it is a subgroup of the compact orthogonal group $\Aut(\Lambda \otimes_\ZZ \RR, \langle \cdot, \cdot \rangle)$, where we extended $\langle \cdot, \cdot \rangle$ bilinearly to $\Lambda \otimes_\ZZ \RR$ and used the same notation for it. If $\Lambda$ is indefinite, then $\Aut(\Lambda; \langle \cdot, \cdot \rangle)$ is typically infinite as shown in \parencite[Satz (30.4)]{kneser:quadratische}.

An element $\mu \in \Lambda$ is called a \defn{(long) root} if $\langle \mu, \mu \rangle = 2$. In general, a lattice might not contain any roots. We say that $\Lambda$ is a \defn{root lattice} if it is generated by its roots. Root lattices are automatically even. They have many automorphisms, since for every root $\mu$ the reflection
\[
\lambda \mapsto \lambda - 2\frac{\langle \lambda, \mu\rangle}{\langle \mu, \mu\rangle} \mu = \lambda - \langle \lambda, \mu\rangle \mu, \qquad \lambda \in \Lambda,
\]
in the hyperplane orthogonal to $\mu$ is a lattice automorphism.

\subsection{The dual lattice and the discriminant group}
\label{subsec:duallattice-discgp}

For a lattice $\Lambda$, define the $\ZZ$-module $\Lambda^\vee := \Hom_\Ab(\Lambda, \ZZ)$. Then $\Lambda^\vee$ is also free and finitely generated, of the same rank as $\Lambda$, and there is an injective homomorphism of $\ZZ$-modules $\Lambda \hookrightarrow \Lambda^\vee$ given by $\lambda \mapsto \langle \lambda, \cdot \rangle$. We claim that $\Lambda^\vee$ carries a non-degenerate, symmetric, bi-additive form $\Lambda^\vee \times \Lambda^\vee \to \QQ$, denoted by $\langle \cdot, \cdot \rangle$ as well, which extends the form on $\Lambda$. One way to define it is to first extend the form on $\Lambda$ bilinearly to $\QQ\Lambda$. Next, one notices that the inclusion $\Lambda \hookrightarrow \Lambda^\vee$ induces an isomorphism of vector spaces $\QQ\Lambda \xrightarrow{\sim} \QQ(\Lambda^\vee)$. This can be used to transport the form to $\QQ(\Lambda^\vee)$ and one finally restricts it to $\Lambda^\vee$. We call $\Lambda^\vee$ together with this form the \defn{dual lattice} of $\Lambda$.

We prove that $\langle \lambda, l \rangle := \langle \langle \lambda, \cdot \rangle, l \rangle = l(\lambda) \in \ZZ$ for all $\lambda \in \Lambda$ and $l \in \Lambda^\vee$. Since both sides of this equality are additive in the arguments $\lambda$ and $l$, it is sufficient to show that $\langle \lambda_i, \lambda_j^\vee \rangle = \delta_{ij}$ for all $i,j$, where $\{\lambda_i\}_i$ is a basis for $\Lambda$ and $\{\lambda_j^\vee\}_j$ is the associated dual basis for $\Lambda^\vee$. There is a unique $\lambda^j \in \QQ\Lambda$ such that $\lambda_j^\vee = \langle \lambda^j, \cdot \rangle$, so $\langle \lambda_i, \lambda_j^\vee \rangle = \langle \lambda_i, \lambda^j \rangle = \lambda_j^\vee(\lambda_i)$, which shows what we wanted.

The action of the group $\Aut(\Lambda; \langle \cdot, \cdot \rangle)$ on $\Lambda$ extends to a (left) action on $\Lambda^\vee$ which preserves the $\QQ$-valued form, namely by $g \cdot l := l \circ g^{-1}$ for $g \in \Aut(\Lambda; \langle \cdot, \cdot \rangle)$ and $l \in \Lambda^\vee$.

If the injection $\Lambda \hookrightarrow \Lambda^\vee$ is surjective, we say that $\Lambda$ is \defn{unimodular}.

Concrete examples of positive definite lattices are often constructed as full rank submodules $\Lambda$ of some linear subspace $E$ of $\RR^n$ for some $n$ and the bi-additive form is then defined to be the restriction of the Euclidean inner product $\langle \cdot, \cdot \rangle_E$ on $E$. The dual lattice can then be identified with $\{\alpha \in E \mid \langle \alpha, \lambda \rangle_E \in \ZZ \text{ for all $\lambda \in \Lambda$}\}$. Indeed, this module has an obvious homomorphism $\alpha \mapsto \langle \alpha, \cdot \rangle_E|_\Lambda$ to $\Lambda^\vee$. The inverse is given by first equating for $l \in \Lambda^\vee$ its $\RR$-linear extension $\RR l\colon \Lambda \otimes_\ZZ \RR \to \RR$ with an $\RR$-linear map $E \to \RR$, and next using the isomorphism of the $\RR$-linear dual $E^\vee$ with $E$ that is induced by $\langle \cdot, \cdot \rangle_E$.

Since $\Lambda$ and its dual have the same rank, the quotient group $\Lambda^\vee/\Lambda$ is a finite abelian group which we denote by $D_\Lambda$. Its order is written as $\disc \Lambda$ and named the \defn{discriminant} of $\Lambda$. This group comes with the extra structure of the symmetric bi-additive form
\[
b_\Lambda\colon D_\Lambda \times D_\Lambda \to \QQ/\ZZ, \qquad (l + \Lambda, l' + \Lambda) \mapsto \langle l, l' \rangle + \ZZ,
\]
called the \defn{discriminant bi-additive form} of $\Lambda$. To see that $b$ is non-degenerate we show that if $l \in \Lambda^\vee$ is such that $\langle l, l'\rangle \in \ZZ$ for all $l' \in \Lambda^\vee$, then $l \in \Lambda$. Such an $l$ defines an element $\langle l, \cdot \rangle$ in $(\Lambda^\vee)^\vee := \Hom_\Ab(\Lambda^\vee, \ZZ)$. By the evaluation isomorphism $\Lambda \xrightarrow{\sim} (\Lambda^\vee)^\vee$ of $\ZZ$-modules there is a unique $\lambda \in \Lambda$ such that $\langle l, l' \rangle = l'(\lambda)$ for all $l' \in \Lambda^\vee$. In turn, $l'(\lambda) = \langle \langle \lambda, \cdot \rangle, l' \rangle$, so by the non-degeneracy of the form on $\Lambda^\vee$ we must have $l = \langle \lambda, \cdot \rangle$.

If $\Lambda$ is even, $b$ can be refined to the quadratic form
\[
q_\Lambda\colon D_\Lambda \to \QQ/2\ZZ, \qquad l + \Lambda \mapsto \langle l,l \rangle + 2\ZZ,
\]
called the \defn{discriminant quadratic form} of $\Lambda$. By this, we mean the following:

\begin{dfn}
A \defn{quadratic form} on an abelian group $D$ with values in another abelian group $A$ is a function $q\colon D \to A$ such that
\begin{enumerate}
\item $q(-d) = q(d)$ for all $d \in D$, and
\item the symmetric form $b(d, d') := q(d + d') - q(d) - q(d')$ on $D$ is bi-additive.
\end{enumerate}
We call $b$ the \defn{bi-additive form of $q$} and say that $q$ is \defn{non-degenerate} when $b$ is.
\end{dfn}

The composition of the bi-additive form $b$ of $q_\Lambda$ with the projection homomorphism $\QQ/2\ZZ \twoheadrightarrow \QQ/\ZZ$ equals $2$ times $b_\Lambda$. Since we learned that $b_\Lambda$ is non-degenerate, so is $b$ and therefore by definition also $q_\Lambda$. Together with the form $b_\Lambda$ if $\Lambda$ is odd, or $q_\Lambda$ if $\Lambda$ is even, $D_\Lambda$ is named the \defn{discriminant group} of $\Lambda$ (after \parencite[§1.3]{nikulin:bilinforms}).

\subsection{Gluing of lattices}
\label{app:latticegluing}
An obvious method of constructing new lattices from given ones is as follows. The \defn{direct sum} of two lattices $(\Lambda_1, \langle \cdot, \cdot \rangle)$ and $(\Lambda_2, \langle \cdot, \cdot \rangle)$ has as its underlying $\ZZ$-module the direct sum of the underlying $\ZZ$-modules of the $\Lambda_i$, and its bi-additive form is set to be
\[
\bigl\langle (\lambda_1, \lambda_2), (\lambda_1', \lambda_2') \bigr\rangle := \langle \lambda_1, \lambda_1' \rangle_1 + \langle \lambda_2, \lambda_2' \rangle_2
\]
for $\lambda_i, \lambda_i' \in \Lambda_i$. So the form is defined in such a way that lattice elements in different summands are orthogonal to each other. Of course, this construction generalises to an arbitrary finite number of lattices.

It is interesting and desirable to have a method of constructing a new lattice from a given one which does not increase the rank, as well. We discuss one such technique known as \defn{self-gluing}. It can be thought of as sliding additional copies of the given lattice in between the points of this same lattice at certain prescribed positions. This is explained in \parencite[Chapter 4, Section 3]{conway:splag} for example, but we will follow the exposition of~\parencite[§1.4]{nikulin:bilinforms} in terms of the discriminant group instead.

Let $\Gamma \hookrightarrow \Lambda$ be a morphism of lattices. Then there are homomorphisms
\[
\Gamma \hookrightarrow \Lambda \hookrightarrow \Lambda^\vee \to \Gamma^\vee,
\]
where the last map is precomposition with $\Gamma \hookrightarrow \Lambda$ and which also respects the ($\QQ$-valued) bi-additive forms. It is injective if $\Gamma$ and $\Lambda$ have the same rank. Suppose namely that $l\in \Lambda^\vee$ and that its image in $\Gamma^\vee$ is zero. Then it factors as a homomorphism $\Lambda/\Gamma \to \ZZ$, but $\Lambda/\Gamma$ is a finite abelian group by our assumption. Since $\ZZ$ has no torsion, $\Lambda/\Gamma \to \ZZ$ must vanish and so the same holds for $l$.

For a given lattice $\Gamma$, we will use the term \defn{overlattice} sometimes to refer to a lattice $\Lambda$ of the same rank as $\Gamma$ together with a specified morphism of lattices $\Gamma \hookrightarrow \Lambda$, and at other times to such a lattice $\Lambda$ alone when the context implies a canonical inclusion.

\begin{lem}
\label{thm:discoverlattices}
(See \parencite[Theorem 2.3.3]{griessjr:introlattices}.)
For an overlattice $\Gamma \hookrightarrow \Lambda$, the discriminants are related by
\[
\disc(\Gamma) = \disc(\Lambda) \cdot [\Lambda : \Gamma]^2.
\]
\end{lem}

As we will see in a moment, overlattices give rise to subgroups of the discriminant group of a special kind:

\begin{dfn}
A subgroup $U$ of the discriminant group $D_\Gamma$ of a lattice $\Gamma$ is called \defn{$b$-isotropic} if $b_\Gamma(u, u') = 0 \in \QQ/\ZZ$ for all $u, u' \in U$. When $\Gamma$ is even $U$ is \defn{$q$-isotropic} if $q_\Gamma(u) = 0 \in \QQ/2\ZZ$ for all $u \in U$.
\end{dfn}

A subgroup that is $q$-isotropic is $b$-isotropic as well by the polarisation identity, but the converse does not necessarily hold.

\begin{thm}[Classification of overlattices]
\label{thm:classifyoverlattices}
(Taken from \parencite[Proposition 1.4.1]{nikulin:bilinforms}.) For a given lattice $\Gamma$, the map $\Lambda \mapsto U_\Lambda := \Lambda/\Gamma$ is an isomorphism of posets between the overlattices $\Lambda$ of $\Gamma$ contained in $\Gamma^\vee$, and the $b$-isotropic subgroups of $D_\Gamma$. If $\Gamma$ is even, even such overlattices correspond to the $q$-isotropic subgroups. Unimodular such overlattices $\Lambda$ correspond to the isotropic subgroups $U_\Lambda$ for which $\abs{U_\Lambda}^2 = \disc \Gamma$.

There is a canonical isomorphism of abelian groups $U_\Lambda^\perp/U_\Lambda \cong D_\Lambda$, where the orthogonal complement is taken inside $D_\Gamma$ with respect to $b_\Gamma$. It respects the forms $b_\Gamma|_{U_\Lambda^\perp}$ and $b_\Lambda$ and, if $\Lambda$ is even, the forms $q_\Gamma|_{U_\Lambda^\perp}$ and $q_\Lambda$.
\end{thm}

Since the overlattices of a given lattice $\Gamma$ have canonical inclusions into $\Gamma^\vee$, the above \namecref{thm:classifyoverlattices} shows that they can be classified and constructed through the study of the isotropic subgroups of $D_\Gamma$. When an overlattice $\Lambda$ of a direct sum $\Gamma_1 \oplus \cdots \oplus \Gamma_n$ has been constructed in this way from an isotropic subgroup $U$ of the discriminant group of $\Gamma_1 \oplus \cdots \oplus \Gamma_n$ we say that the $\Gamma_i$ are \defn{components} of $\Lambda$ which have been `glued together' along representatives of $U$ in the dual lattices $\Gamma_i^\vee$.

All even, unimodular, positive definite lattices of rank $24$ have been classified in \parencite{niemeier:quadformen} (there are $24$ up to isomorphism) and one use of the technique of gluing lattices is the construction of all but one of them (the \defn{Leech lattice}) from certain simpler lattices of lower or equal rank (see \parencite[Chapter 3]{ebeling:lattices} and \parencite[Chapter 16]{conway:splag} for expositions).

\subsection{Examples of lattices}
It is high time we present some basic examples of lattices.

\begin{exmp}[The $A_n$-series]
One series of lattices is the so called \defn{$A_n$-series}, defined for $n \geq 1$. The underlying $\ZZ$-module of $A_n$ is the following submodule of $\ZZ^{n+1}$:
\[
A_n := \biggl\{(x_1, \ldots, x_{n+1}) \in \ZZ^{n+1} \biggm\vert \sum_{i=1}^{n+1} x_i = 0 \biggr\}
\]
and it inherits a positive definite form from the Euclidean inner product on $\RR^{n+1}$. One proves that $A_n$ is even by noting that for $(x_1,\ldots,x_{n+1}) \in A_n$,
\[
\sum_{i=1}^{n+1} x_i^2 + \sum_{i\neq j} 2 x_i x_j = \biggl(\sum_{i=1}^{n+1} x_i\biggr)^2 = 0.
\]
Its rank is $n$ because if $\epsilon_1, \ldots, \epsilon_{n+1}$ denotes the standard basis of $\RR^{n+1}$, then one possible basis for $A_n$ is
\begin{equation}
\label{eq:lattice-An-basis}
\{\epsilon_2 - \epsilon_1, \epsilon_3 - \epsilon_2, \ldots, \epsilon_{n+1} - \epsilon_n\}.
\end{equation}
It is a root lattice of which the $n(n+1)$ roots are $\epsilon_i - \epsilon_j$, where $1 \leq i,j \leq n+1$ and $i \neq j$.

For $n\geq 2$, the automorphism group of $A_n$ is isomorphic to a product $\{\pm 1\} \times S_{n+1}$, where $-1$ is the negation of elements and the symmetric group $S_{n+1}$ on $n+1$ letters acts by permuting the $n+1$ coordinates. Equivalently, the transposition that switches the $i$-th and $j$-th coordinates can be seen as a reflection in the hyperplane orthogonal to the root $\epsilon_i - \epsilon_j$. If $n=1$, negation coincides with flipping the two coordinates, so $\Aut(A_1) \cong S_2$.

The dual lattice $A_n^\vee$ is generated by $A_n$ together with the single element
\[
l := \biggl(\frac{n}{n+1}, -\frac{1}{n+1}, \ldots, -\frac{1}{n+1}\biggr) \in \QQ^{n+1}.
\]
Therefore, $D_{A_n}$ is isomorphic to the cyclic group $\ZZ/(n+1)\ZZ$. Each overlattice, respectively even overlattice, of $A_n$ is generated by $A_n$ together with an element $a\cdot l$, where $0 \leq a \leq n$ is an integer such that $a^2 \langle l, l \rangle$ lies in $\ZZ$, respectively in $2\ZZ$.
\end{exmp}

\begin{exmp}[The $D_n$-series]
For $n\geq 3$, another series of positive definite lattices can be defined as
\[
D_n := \biggl\{(x_1,\ldots, x_n) \in \ZZ^n \biggm\vert \sum_{i=1}^n x_i \in 2\ZZ\biggr\}.
\]
It is clear that $D_n$ contains $A_{n-1}$ and that it is generated by it together with the element $\epsilon_1 + \epsilon_2 \in \ZZ^n$. This observation shows that $D_n$ has rank $n$, and gives via~\eqref{eq:lattice-An-basis} a basis for $D_n$. One proves that $D_n$ is even in the same way as one does for $A_n$. It is a root lattice of which the $2n(n-1)$ roots are $\pm(\epsilon_i + \epsilon_j)$ and $\epsilon_i - \epsilon_j$, where $1 \leq i,j \leq n$ and $i \neq j$.

The dual lattice $D_n^\vee$ is generated by $\ZZ^n$ together with the single element $l_1 := \frac12(1, 1, \ldots, 1) \in \QQ^n$. If $n$ is even the discriminant group $D_{D_n}$ is a direct sum of the two groups of order $2$ generated by the equivalence classes $[l_1]$ and $[l_2]$ of $l_1$ and $l_2 := (1, 0, \ldots, 0) \in \ZZ^n$, while in the odd case $[l_2] = 2[l_1]$, so $D_{D_n}$ is generated by $[l_1]$ alone, which then has order $4$. Summarising, $D_{D_n}$ has the following structure:
\[
D_{D_n} \cong
\begin{cases}
\ZZ/2\ZZ \oplus \ZZ/2\ZZ & \text{if $n$ is even,} \\
\ZZ/4\ZZ 				& \text{if $n$ is odd.}
\end{cases}
\]

If $n$ is odd, then $D_{D_n}$ has no non-zero $b$-isotropic subgroups since
\[
b_{D_n}([l_1], [l_1]) = n/4 \bmod \ZZ
\]
and so $D_n$ has no non-trivial overlattices. If $n$ is even, then $D_{D_n}$ always has at least one $b$-isotropic subgroup that is not $q$-isotropic, namely the one generated by $[l_2]$ because $q_{D_n}([l_2]) = 1 \bmod 2\ZZ$. If $n \in 4\ZZ$, then also the two subgroups generated by $[l_1]$ and $[l_1] + [l_2]$ are $b$-isotropic. If moreover $n \in 8\ZZ$, then both are even $q$-isotropic.
\end{exmp}

\begin{exmp}[The $E_6$, $E_7$ and $E_8$ lattices]
If $n \in 8\ZZ$, the even, positive definite overlattice of $D_n$ corresponding to the $q$-isotropic subgroup of the discriminant group $D_{D_n}$ which is generated by $[l_1]$ is denoted by $D_n^+$. In simpler words, $D_n^+ := D_n + \ZZ \cdot l_1$. Since $[l_1]$ has order $2$ in $D_{D_n}$ and $\disc D_n = 4$, the discriminant group of $D_n^+$ is trivial by \cref{thm:discoverlattices}. That is, $D_n^+$ is a unimodular lattice. If $n = 8$, it is again a root lattice because $D_8$ is and $l_1$ is in this case a root. We denote $D_8^+$ as $E_8$. It is known to be the unique even, unimodular, positive definite lattice of rank $8$ up to isomorphism (see \parencite[Proposition 2.5]{ebeling:lattices}).

Take now any root $\mu \in E_8$ and consider the sublattice $\{\lambda \in E_8 \mid \langle \lambda, \mu \rangle = 0\}$ that is the orthogonal complement of $\mu$ in $E_8$. It turns out that all choices of $\mu$ give in this way an isomorphic sublattice. We denote it by $E_7$. It is a root lattice as well, but it is no longer unimodular. Its discriminant group has order $2$ with generator $[l_1]$ (assuming that $\mu$ was chosen orthogonal to $l_1$).

The lattice $E_8$ contains copies of $A_2$. The orthogonal complements in $E_8$ of all of these are isomorphic and denoted by $E_6$. This is again a root lattice and its discriminant group has order $3$.
\end{exmp}

\begin{exmp}[The hyperbolic plane]
\label{exmp:hyperbolicplane}
The underlying $\ZZ$-module of the \defn{hyperbolic plane lattice} is $\ZZ \oplus \ZZ$ and its form is given by
\[
\bigl\langle (\lambda_1, \lambda_2), (\mu_1, \mu_2) \bigr\rangle := \lambda_1 \mu_2 + \lambda_2 \mu_1.
\]
It is even, unimodular and is an example of an indefinite lattice. One namely has both $\langle (1,1), (1,1) \rangle = 2$ and $\langle (1,-1), (1,-1) \rangle = -2$.
\end{exmp}

\begin{exmp}[Leech lattices]
Any even, unimodular, positive definite lattice of rank $24$ without roots is called a \defn{Leech lattice}. There exists only one Leech lattice up to isomorphism, which we denote by $\Lambda_\Leech$, and many different constructions of it have been devised. This characterisation was proven independently in \parencite[Chapter 12]{conway:splag} and \parencite{niemeier:quadformen}.

Even though the Leech lattice has no roots, which means that it does not have corresponding reflection automorphisms, its automorphism group $\Aut(\Lambda_\Leech, \langle \cdot, \cdot \rangle)$ is large and rich in structure as shown in \parencite{conway:leechautomgp}. It is known as the \defn{zeroth Conway group} and denoted by $\Co_0$. It has order
\[
2^{22} 3^9 5^4 7^2 11 \cdot 13 \cdot 23 = 8\, 315\, 553\, 613\, 086\, 720\, 000
\]
and is known not to be a simple group. Its quotient by the central subgroup $\{\pm \id\}$ \emph{is} simple, though, and is called the \defn{first Conway group} $\Co_1$. It is one of the 26 sporadic finite simple groups. Two other groups in that list, $\Co_2$ and $\Co_3$, can be obtained as subgroups of $\Co_0$ that fix certain elements in $\Lambda_\Leech$.
\end{exmp}

\subsection{The theta series of a lattice}
\label{subsec:thetaseries}

For a lattice $\Lambda$, the series
\[
\theta_\Lambda(q) := \sum_{\lambda \in \Lambda} q^{\langle \lambda, \lambda\rangle/2}
\]
in the formal variable $q^{1/2}$ is called the \defn{theta series} of $\Lambda$. Clearly, if $\Lambda$ is even, $\theta_\Lambda(q)$ contains only integral powers of $q$. If $\Lambda$ is definite it can be seen as a generating function for the lengths of the elements of $\Lambda$ because we can rewrite it as
\[
\theta_\Lambda(q) = \sum_{k\in \ZZ} \abs{\Lambda_k} q^{k/2}, \qquad \Lambda_k := \bigl\{\lambda \in \Lambda \bigm\vert \langle \lambda, \lambda \rangle = k\bigr\},
\]
and $\Lambda_k$ is indeed a finite set since the closed ball of radius $\sqrt k$ in $\Lambda \otimes_\ZZ \RR$ is compact. 
%

In this thesis we will also have use for \defn{rational lattices}. A rational lattice is defined identically as a lattice, except that the bi-additive form may be $\QQ$-valued instead of $\ZZ$-valued. We will sometimes distinguish lattices from rational ones by calling the former \defn{integral}. Examples of rational lattics are the dual $\Lambda^\vee$ of an integral lattice $\Lambda$ and the sum of two integral lattices which contain a common integral sublattice. Of course one can define the theta series of a rational lattice also, which will then be a series in fractional powers of $q$.

\section{Central extensions of groups}
\label{app:centexts}

In this thesis central extensions of certain abelian groups are constructed. This section is devoted to defining central extensions and explaining how the ones we are interested in can be defined by a particular type of maps, called \defn{(group) $2$-cocycles}. Group cocycles of a group $G$ are in general defined for an abelian group $A$ and an action of $G$ on $A$, but for our purposes we merely need the case that $G$ acts trivially on $A$. That is, we will omit this action. Even though only the situation when $G$ is abelian matters to us, we will begin by discussing the non-abelian case as well because doing so does not require any extra effort.

\begin{refs}
The material in this section can be found in \parencite[Section 5.1--5.2]{frenkel:monster}.
\end{refs}

We start our discussion by assuming that $G$ and $A$ are abstract groups, with $A$ abelian. We will write the addition in $A$ multiplicatively.

\begin{dfn}
\label{dfn:centext}
A \defn{central extension} of $G$ by $A$ is a group $\tilde G$ together with two homomorphisms $A \hookrightarrow \tilde G$ and $\tilde G \twoheadrightarrow G$ that make it fit into a short exact sequence
\[
\begin{tikzcd}
1 \ar[r] & A \ar[r] & \tilde G \ar[r] & G \ar[r] & 1,
\end{tikzcd}
\]
such that the image of $A$ is a subgroup of the centre of $\tilde G$. We will often denote the central extension simply by $\tilde G$, thereby suppressing the data of the other two homomorphisms in our notation. A \defn{morphism} between two central extensions $\tilde G$ and $\tilde G'$ is a homomorphism $\tilde G \to \tilde G'$ making the following diagram commute:
\[
\begin{tikzcd}[row sep=1em]
		 &   & \tilde G	\ar[dr] \ar[dd]	&		   & \\
1 \ar[r] & A \ar[ur] \ar[dr] & 				& G \ar[r] & 1. \\
		 &   & \tilde G' \ar[ur]	&		   &
\end{tikzcd}
\]
\end{dfn}

It is not hard to show that a morphism of central extensions is necessarily an isomorphism. In other words, the category of central extensions of $G$ by $A$ is a groupoid.

The algebraic study of central extensions starts with the following observation. Pick a set-theoretic section $s\colon G \to \tilde G$ of the homomorphism $\tilde G \twoheadrightarrow G$. Then $s$ is usually not a homomorphism as well. Instead, $s(g) s(g') s(gg')^{-1}$ lies in (the image of) $A$ for all $g, g' \in G$. The function
\[
c\colon G \times G \to A, \qquad (g, g') \mapsto s(g) s(g') s(gg')^{-1}
\]
can be seen as a measurement for the failure of the multiplicativity of $s$. By writing out both sides of the equation $s((g_1 g_2) g_3) = s(g_1 (g_2 g_3))$ in two different ways it can be checked that $c$ satisfies
\begin{equation}
\label{eq:centexts-cocyclerel}
c(g_1, g_2) c(g_1 g_2, g_3) = c(g_1, g_2 g_3) c(g_2, g_3)
\end{equation}
for all $g_1$, $g_2$ and $g_3$ in $G$. Moreover, if $s(1_G) = 1_{\tilde G}$ then
\begin{equation}
\label{eq:centexts-cocycle-norm}
c(g, 1_G) = c(1_G, g) = 1_A
\end{equation}
for all $g \in G$.

We analyse the dependence of $c$ on $s$ by assuming that $f\colon \tilde G \xrightarrow{\sim} \tilde G'$ is an isomorphism to another central extension. (We allow the case when $\tilde G' = \tilde G$, the two homomorphisms to $G$ are the same and $f$ is the identity.) If $s'$ is a section for $\tilde G'$ we can associate a function $c'$ to it in the same way as we did above. Note that $f \circ s$ is a section of $\tilde G'$ also. Now define an auxiliary function $d\colon G \to A$ as $g \mapsto s'(g) (f \circ s)(g)^{-1}$ for all $g \in G$. If $s'(1_G) = (f\circ s)(1_G)$ then $d(1_G) = 1_A$. Finally setting
\begin{equation}
\label{eq:centexts-boundmap}
\delta d\colon G \times G \to A, \qquad (g,g') \mapsto d(g) d(g') d(g g')^{-1}
\end{equation}
for all $g, g'\in G$, one can prove that $c' = (\delta d) \cdot c$. We thus see that $c$ is not associated canonically to $\tilde G$. Nevertheless, we were able to express its ill-definedness in a precise manner.

We formalise these observations by

\begin{dfn}
A \defn{$1$-cochain} (for the pair $(G, A)$) is a function $d\colon G \to A$. We say that a $1$-cochain $d$ is \defn{normalised} if $d(1_G) = 1_A$. A $1$-cochain is called a \defn{$1$-cocycle} if it is a group homomorphism. The trivial $1$-cocycle is called the \defn{$1$-coboundary}. Given a $1$-cochain $d$, we define a function $\delta d$ as~\eqref{eq:centexts-boundmap} for all $g, g'\in G$.

A \defn{$2$-cochain} is a function $c\colon G \times G \to A$. It is called a \defn{$2$-cocycle} if it satisfies~\eqref{eq:centexts-cocyclerel} for all $g_1$, $g_2$ and $g_3$ in $G$. We say that a $2$-cocycle $c$ is \defn{normalised} if~\eqref{eq:centexts-cocycle-norm} holds for all $g \in G$. A $2$-cochain $c$ is called a \defn{$2$-coboundary} if $c = \delta d$ for some $1$-cochain $d$. Two $2$-cocycles are \defn{cohomologous} if they differ by a $2$-coboundary.
\end{dfn}

It is easily checked that a $2$-coboundary is a $2$-cocycle. A $2$-coboundary is normalised if and only if it comes from a normalised $1$-cochain.
%

\begin{rmk}
If a $2$-cocycle $c$ is normalised, then this implies by the cocycle relation that $c(g, g^{-1}) = c(g^{-1}, g)$ for all $g \in G$. Any $2$-cocycle $c$ is cohomologous to a normalised one. It namely follows from the cocycle relation that $c(g,1) = c(1,g) = c(1,1)$ for all $g \in G$. So let $d$ be the constant $1$-cochain $d(g) := c(1,1)$ for all $g \in G$. Then $(\delta d)^{-1} \cdot c$ is normalised.
\end{rmk}

This procedure of associating (cohomology classes of) $2$-cocycles to central extensions can be reversed. Given a normalised $2$-cocycle $c$ the set $\tilde G_c := G \times A$ namely becomes a group under the multiplication
\[
(g, a) \cdot (g', a') := \bigl(g g', aa' \cdot c(g, g')\bigr).
\]
Its unit element is $(1_G, 1_A)$ and the inverse of $(g,a)$ is given by
\[
\bigl(g^{-1}, a^{-1} c(g, g^{-1})^{-1}\bigr).
\]
It is a central extension of $G$ by $A$ when equipped with the obvious inclusion of $A$ and projection to $G$.

This construction is inverse to the one we studied earlier in the sense that if $c$ comes from a section $s$ of a central extension $\tilde G$ such that $s(1_G) = 1_{\tilde G}$, then there is a canonical isomorphism of central extensions $\tilde G_c \xrightarrow{\sim} \tilde G$ given by $(g, a) \mapsto s(g) a$. If a normalised $2$-cocycle $c'$ is cohomologous to $c$, say, $c' = (\delta d) \cdot c$ for some (normalised) $1$-cochain $d$, then $(g, a) \mapsto (g, d(g)^{-1}\cdot a)$ is an isomorphism $\tilde G_c \xrightarrow{\sim} \tilde G_{c'}$ of the corresponding central extensions.



Let us revisit our study by assuming that $G$ and $A$ are topological groups. The definition of a \emph{topological} central extension $\tilde G$ in \cref{dfn:centext} then additionally demands that the two homomorphisms are continuous. The difference with the situation for abstract groups is that there might not exist any globally continuous section $s\colon G \to \tilde G$. In general, $s$ is continuous only on a neighbourhood around $1_G$. A discontinuous section can then result in a discontinuous cocycle $c$ and so the central extension $\tilde G_c$ is not a topological group when we give the underlying set $G \times A$ of $\tilde G_c$ the product topology.

However, it turns out that under certain conditions $\tilde G_c$ is a topological group anyway when $s$ is only locally continuous (see \parencite[Proposition 2.2 and Remark 2.3]{neeb:centexts}). We state a preparatory \namecref{thm:topgpfromsubgp} and then give a particular version of such a result that is sufficient for the purposes in this thesis.

\begin{lem}
\label{thm:topgpfromsubgp}
Let $G$ be an abstract group, $G_0$ a normal topological subgroup and suppose that for every element of $G$ the associated conjugation map on $G_0$ is continuous. Then there exists a unique structure of a topological group on $G$ such that $G_0$ is open in $G$. The topology is given by declaring a subset $U \subseteq G$ to be open when $gU \cap G_0$ is open in $G_0$ for all $g \in G$.
\end{lem}

\begin{cor}
\label{thm:centexttopgp}
Let $G$ and $A$ be topological groups with $A$ abelian, $G_0$ an open normal topological subgroup of $G$ and $c$ a normalised $2$-cocycle on $G$ such that
\begin{enumerate}
\item $c$ is continuous when restricted to $G_0$, and
\item for every $g \in G$ the map $G_0 \to A$ given by $g' \mapsto c(g, g') c(gg', g^{-1})$ is continuous.
\end{enumerate}
Then there exists a unique structure of a topological group on $\tilde G_c$ such that $(\tilde G_0)_c$ is open in $\tilde G_c$. Here, $(\tilde G_0)_c$ is the topological group of which its underlying set $G_0 \times A$ carries the product topology. 
\end{cor}
\begin{proof}
Note first that $(\tilde G_0)_c$ is normal in $\tilde G_c$. Next, let $g \in G$ and $(g',a) \in (\tilde G_0)_c$. Then we calculate that
\begin{align*}
(g, 1_A)(g', a)(g, 1_A)^{-1}
	&= \bigl(gg', a \cdot c(g,g')\bigr) \bigl(g^{-1}, c(g, g^{-1})^{-1}\bigr) \\
	&= \bigl(gg'g^{-1}, a \cdot c(g,g') c(g, g^{-1})^{-1} c(gg', g^{-1})\bigr).
\end{align*}
Hence, conjugation by $(g, 1_A)$ is a continuous map on $(\tilde G_0)_c$. The result now follows from \cref{thm:topgpfromsubgp}.
\end{proof}

\subsection{Central extensions of abelian groups}
\label{subsec:centexts-abgps}

If $\Lambda$ is an abelian group, then every central extension $\tilde \Lambda$ of $\Lambda$ by $A$ has a canonically associated \defn{commutator map} $b\colon \Lambda \times \Lambda \to A$. It is defined by picking a section $s\colon \Lambda \to \tilde\Lambda$ and setting $b(\lambda, \mu) := s(\lambda) s(\mu) s(\lambda)^{-1} s(\mu)^{-1}$ for all $\lambda, \mu \in \Lambda$, which is independent of the choice of $s$. It is bi-additive and satisfies $b(\lambda, \lambda) = 0$ for all $\lambda \in \Lambda$. (These two properties together imply that $b$ is skew-symmetric.) There furthermore holds
\[
b(\lambda, \mu) = \epsilon(\lambda, \mu) \epsilon(\mu, \lambda)^{-1}
\]
for any choice of $2$-cocycle $\epsilon$ for $\tilde\Lambda$, so this can be used as an alternative definition of $b$.

If $\Lambda$ is a free $\ZZ$-module of finite rank, then also conversely, every bi-additive map $b\colon \Lambda \times \Lambda \to A$ which satisfies $b(\lambda, \lambda) = 0$ arises in this way. We can namely pick an ordered basis $\{\lambda_1, \ldots, \lambda_n\}$ of $\Lambda$ and define a function $\epsilon\colon \Lambda \times \Lambda \to A$ by first setting
\[
\epsilon(\lambda_i, \lambda_j) :=
\begin{cases}
b(\lambda_i, \lambda_j) & \text{if $i < j$,} \\
1 			& \text{if $i \geq j$,}
\end{cases}
\]
and extending this bi-additively next. Since $\epsilon$ is bi-additive, it is a $2$-cocycle and therefore determines a central extension $\tilde \Lambda$. It satisfies
\[
\epsilon(\lambda, \mu) \epsilon(\mu, \lambda)^{-1} = b(\lambda, \mu)
\]
for all $\lambda, \mu \in \Lambda$, which can be checked by hand for the basis elements $\lambda_i$ and must then hold for arbitrary elements of $\Lambda$ since both sides of this equation are bi-additive. This shows that $b$ is the commutator map of $\tilde \Lambda$. With the assumption about $\Lambda$ still in place one can furthermore show that $b$ determines $\tilde \Lambda$ up to non-unique isomorphism.

\section{Group representations}
\label{app:posenergyreps}

In \cref{sec:unicol-reptheory,sec:bicol-reptheory} representations are constructed and classified of central extensions of uni- and bicoloured torus loop groups, respectively, and of related groups. These representations are of the following kind:

\begin{dfn}
\label{dfn:rep}
A unitary representation $Q\colon G \to \lieU(\hilb H)$ of a topological group $G$ on a Hilbert space $\hilb H$ is called \defn{strongly continuous} if $Q$ is a continuous map with respect to the strong operator topology on $\lieU(\hilb H)$ or, equivalently, if for all vectors $v \in \hilb H$ the map $G \to \hilb H$, $g \mapsto Q(g)(v)$ is continuous at $g=1$.

A \defn{morphism}, or \defn{intertwiner}, between two strongly continuous, unitary representations of $G$ on Hilbert spaces $\hilb H$ and $\hilb K$ is a bounded linear map $\hilb H \to \hilb K$ which intertwines the respective actions of $G$. Two such representations are called \defn{unitarily isomorphic} if there exists a unitary morphism between them.

The category of strongly continuous, unitary representations of $G$ is denoted by $\Rep G$ and if $Q$ and $Q'$ are in $\Rep G$ we write $\Hom_G(Q, Q')$ for the complex vector space $\Hom_{\Rep G}(Q, Q')$.
\end{dfn}

In the rest of this section \defn{groups} will always be topological and \defn{representations} will always be meant to be strongly continuous and unitary, unless stated otherwise.

Using the fact that the operator-theoretic adjoint of a morphism is again a morphism one can show

\begin{prop}
(See \parencite[361]{kirillov:orbitlectures} or \parencite[Proposition I.2.6]{sugiura:harmana}.)
Two representations of a group are isomorphic if and only if they are unitarily isomorphic.
\end{prop}

\begin{dfn}
Let $Q$ be a representation of a group on a Hilbert space $\hilb H$. A \defn{subrepresentation} of $Q$ is a closed linear subspace $\hilb K$ of $\hilb H$ that is invariant under $Q$, equipped with the restriction of $Q$ to $\hilb K$. We say that $Q$ is \defn{irreducible} if $\hilb H$ is non-zero and its only subrepresentations are $\{0\}$ and $\hilb H$ itself.
\end{dfn}

The following equivalent characterisation of the invariance of a closed subspace is often useful:

\begin{lem}
(See \parencite[Proposition 3.4]{folland:harmana}.) Let $Q$ be a representation of a group on a Hilbert space $\hilb H$. Then a closed linear subspace $\hilb K$ of $\hilb H$ is invariant under $Q$ if and only if the orthogonal projection of $\hilb H$ onto $\hilb K$ commutes with $Q$.
\end{lem}

A basic fact is that if $\hilb K \subseteq \hilb H$ is a subrepresentation then the same holds for its orthogonal complement $\hilb K^\perp$.

Morphisms between irreducible representations are rigid in the following sense, as is well known from the theory of finite-dimensional representations:

\begin{lem}[Schur's lemma]
(See \parencite[Lemma 3.5]{folland:harmana}.) A non-zero representation $Q$ of a group is irreducible if and only if every endomorphism of $Q$ is a scalar multiple of the identity. A morphism between two irreducible representations is zero when the representations are non-isomorphic and a multiple of a unitary isomorphism otherwise.
\end{lem}

We denote the external or internal direct sum of a family of representations $Q_i$ as $\bigoplus_i Q_i$. If $\hilb H_i$ is the underlying Hilbert space of $Q_i$, then that of $\bigoplus_i Q_i$ is the Hilbert space completion
\[
\overline{\bigoplus_i \hilb H_i}
\]
of the algebraic direct sum of the $\hilb H_i$.

A finite-dimensional representation of a group always contains an irreducible subrepresentation because one can pick a non-zero subrepresentation of minimal dimension. This argument cannot be applied to infinite-dimensional representations. Even worse, the left regular representation of the additive group $\RR$ on the Hilbert space $L^2(\RR)$ is an example which has plenty of subrepresentations but no irreducible ones (as explained in for example \parencite[72]{folland:harmana}), showing that this property is not satisfied in general. When we do have this knowledge at our disposal a general argument can be used for

\begin{prop}
\label{thm:gprep-complred}
A non-zero representation of a group such that every non-zero subrepresentation contains an irreducible subrepresentation is the (internal) direct sum of a set of mutually orthogonal, irreducible subrepresentations.
\end{prop}
\begin{proof}
(As for example in \parencite[Proof of Corollary 5.4.2]{kowalski:reptheory}.) Denote the representation and its underlying Hilbert space by $Q$ and $\hilb H$, respectively. Consider the set $\Sigma$ of which each element is a set of mutually orthogonal, irreducible subrepresentations of $Q$. Such sets can be ordered by inclusion, making $\Sigma$ a poset. Our assumption says that $Q$ contains at least one irreducible subrepresentation, so the singleton set that it forms belongs to $\Sigma$ and therefore $\Sigma$ is non-empty. Every non-empty chain in $\Sigma$ has an upper bound in $\Sigma$, namely the union of the sets belonging to that chain. Zorn's lemma now tells us that $\Sigma$ contains a maximal element, say, a set $\{\hilb K_\alpha\}_{\alpha \in A}$ of subrepresentations for some index set $A$. Define $\hilb K$ to be the (internal, Hilbert space) direct sum of the $\hilb K_\alpha$. If $\hilb K = \hilb H$, we are done. If not, then $\hilb K^\perp$ is a non-zero subrepresentation of $Q$ because $Q$ is unitary. Applying the assumption again to exhibit an irreducible subrepresentation of $\hilb K^\perp$ and adding this to the set $\{\hilb K_\alpha\}_{\alpha \in A}$ contradicts its maximality. We conclude that $\hilb K = \hilb H$ after all.
\end{proof}

Of course, the decomposition that we obtain in this way is in no way unique. Moreover, this result and its proof do not tell you anything about how to concretely decompose a given representation.

\subsection{Induction of representations in the case of a discrete coset space}
\label{subsec:indreps}

Throughout this section $G$ will be a group, $H$ a subgroup such that the coset space $G/H$ is discrete and countable and $Q$ a representation\footnote{We inherit the conventions on terminology from \cref{app:posenergyreps}. So groups are meant to be topological and representations are strongly continuous and unitary.} of $H$ on a Hilbert space $\hilb K$. (This implies that $H$ is open, and hence closed in $G$.) We explain a method of \defn{inducing} $Q$ up to a representation of $G$ and we prove some of its basic properties.

The conditions we impose on $G/H$ are sufficient for the purposes of this thesis because for the pairs $(G,H)$ considered in \cref{subsec:unicol-irrepsLT,subsec:bicol-irrepsfullgp} the coset space is isomorphic to the underlying abelian group of a lattice and of a rational lattice, respectively. Furthermore, for those pairs $H$ is normal in $G$. Therefore we will often include that assumption in our results as well, although this is not essential.

We start by constructing from the data $G$, $H$ and $Q$ a new Hilbert space $\Ind_H^G \hilb K$ that is larger than $\hilb K$ as the following (external) Hilbert space direct sum indexed over the left cosets $\sigma$ of $H$ in $G$:
\[
\Ind_H^G \hilb K := \overline{\bigoplus_{\mathclap{\sigma \in G/H}} \hilb K^\sigma}.
\]

Here, $\hilb K^\sigma$ stands for the Hilbert space $\sigma \times_H \hilb K$. Its vectors are equivalence classes of pairs $(x, v) \in \sigma \times \hilb K$ for the relation $(xh, v) \sim (x, Q(h)v)$. We denote such a class as $[x,v]$. Addition of vectors is defined as $[x,v] + [x', v'] := [x, v + Q(x^{-1}x')v']$. The inverse of $[x, v]$ is $[x, -v]$ and the zero vector is $[x, 0]$. Scalar multiplication is defined as $\alpha \cdot [x, v] := [x, \alpha v]$, and, lastly, the inner product is $\langle [x, v], [x', v'] \rangle := \langle v, Q(x^{-1} x') v'\rangle$. Any choice of representative $x \in \sigma$ induces a unitary map $\hilb K^\sigma \xrightarrow{\sim} \hilb K$ given by $[x, v] \mapsto v$. Hence $\hilb K^\sigma$ can be seen as a copy of $\hilb K$ associated to $\sigma$ which is constructed without making any choices.

We will often abbreviate $\Ind_H^G \hilb K$ as $\Ind \hilb K$ when it is clear which groups are under discussion. A general vector of $\Ind \hilb K$ is a tuple of vectors
\[
\bigl([x^\sigma, v^\sigma]\bigr)_{\sigma \in G/H}, \qquad [x^\sigma, v^\sigma] \in \hilb K^\sigma,
\]
such that
\[
\sum_{\mathclap{\sigma \in G/H}} \norm{v^\sigma}_{\hilb K}^2 < \infty.
\]

Next, we define an action $\Ind_H^G Q$ of $G$ on $\Ind \hilb K$ by setting for $g \in G$
\[
(\Ind_H^G Q)(g) \cdot \bigl([x^\sigma, v^\sigma]\bigr)_\sigma := \Bigl(\bigl[gx^{g^{-1}\sigma}, v^{g^{-1}\sigma}\bigr]\Bigr)_\sigma.
\]
It is easily checked that $(\Ind_H^G)(g)$ is well-defined, linear and unitary, and we will often abbreviate this action as $\Ind Q$. Informally, and more geometrically speaking, it can be understood by seeing $G \times_H \hilb K$ as the total space of a complex vector bundle over the space of cosets $G/H$. The action $\Ind Q$ on $\Ind \hilb K$ is then the natural one when the latter is considered as a space of sections of this bundle that are square-integrable with respect to the counting measure on $G/H$.

\begin{prop}
The action $\Ind Q$ of $G$ on $\Ind \hilb K$ is strongly continuous.
\end{prop}
\begin{proof}
Let $([x^\sigma, v^\sigma])_\sigma$ be a vector in $\Ind \hilb K$ and $(g_n)_n$ a sequence of elements in $G$ converging to $1$. We then wish to show that the sequence of vectors
\[
\Bigr((\Ind Q)(g_n) \cdot \bigl([x^\sigma, v^\sigma]\bigr)_\sigma\Bigl)_n
\]
converges to $([x^\sigma, v^\sigma])_\sigma$. For every $n$ there holds
\begin{multline*}
\bigl([x^\sigma, v^\sigma]\bigr)_\sigma - (\Ind Q)(g_n) \cdot \bigl([x^\sigma, v^\sigma]\bigr)_\sigma \\
\begin{aligned}
	&= \bigl([x^\sigma, v^\sigma]\bigr)_\sigma - \Bigl(\bigl[g_n x^{g_n^{-1}\sigma}, v^{g_n^{-1}\sigma}\bigr]\Bigr)_\sigma \\
	&= \Bigl(\Bigl[x^\sigma, v^\sigma - Q\bigl((x^\sigma)^{-1} g_n x^{g_n^{-1}\sigma}\bigr)\bigl(v^{g_n^{-1}\sigma}\bigr)\Bigr]\Bigr)_\sigma,
\end{aligned}
\end{multline*}
and hence
\begin{multline*}
\Bigl\lVert \bigl([x^\sigma, v^\sigma]\bigr)_\sigma - (\Ind Q)(g_n) \cdot \bigl([x^\sigma, v^\sigma]\bigr)_\sigma \Bigr\rVert_{\Ind \hilb K}^2 \\
	= \sum_{\sigma \in G/H} \Bigl\lVert  v^\sigma - Q\bigl((x^\sigma)^{-1} g_n x^{g_n^{-1}\sigma}\bigr)\bigl(v^{g_n^{-1}\sigma}\bigr) \Bigr\rVert_{\hilb K}^2.
\end{multline*}
Now note that, since $G/H$ is discrete, $(g_n)_n$ converging to $1$ implies that there exists some $N \geq 1$ such that for $n \geq N$ we have $g_n^{-1} \sigma = \sigma$ for all cosets $\sigma$ simultaneously. An appeal to the strong continuity of $Q$ then finishes the argument.
\end{proof}

We call $\Ind Q$ the \defn{representation of $G$ induced from $Q$}.

\begin{prop}
\label{thm:indrep-functorial}
Let $Q_1$ and $Q_2$ be two representations of $H$ with underlying Hilbert space $\hilb K_1$ and $\hilb K_2$, respectively. If $f\colon Q_1 \to Q_2$ is an $H$-intertwiner, then the function $\Ind f\colon \Ind \hilb K_1 \to \Ind \hilb K_2$ defined by
\begin{equation}
\label{eq:indrep-functorial}
(\Ind f)\bigl([x^\sigma, v^\sigma]\bigr)_\sigma := \Bigl(\bigl[x^\sigma, f(v^\sigma)\bigr]\Bigr)_\sigma
\end{equation}
for all vectors $([x^\sigma, v^\sigma])_\sigma \in \Ind \hilb K_1$ is a $G$-intertwiner. It is isometric when $f$ is.
\end{prop}
\begin{proof}
That $\Ind f$ is linear and compatible with the actions of $G$ is easily checked. Because
\[
\sum_\sigma \bigl\lVert f(v^\sigma)\bigr\rVert_{\hilb K_2}^2 \leq \norm f \cdot \sum_\sigma \norm{v^\sigma}_{\hilb K_1}^2 < \infty,
\]
the tuple on the right hand side of~\eqref{eq:indrep-functorial} is indeed a vector of $\Ind \hilb K_2$. This inequality also shows that $\norm{\Ind f} = \norm f$.
\end{proof}

\begin{cor}
The assignment $Q \mapsto \Ind Q$ defines a functor $\Ind_H^G\colon \Rep H \to \Rep G$.
\end{cor}

Unfortunately, this induction functor produces in general representations of $G$ that are too large or too small for the functor to be a left or a right adjoint, respectively, to the restriction functor $\Res_G^H\colon \Rep G \to \Rep H$. Regarding the first of these two claims: if $Q'$ is a representation of $G$ there always exists the injective complex linear map
\begin{equation}
\label{eq:indrep-leftadj}
\Hom_G(\Ind Q, Q') \hookrightarrow \Hom_H(Q, \Res Q')
\end{equation}
given by pre-composition with the $H$-intertwining inclusion $\hilb K \hookrightarrow \Ind \hilb K$, $v \mapsto [1,v]$. To show that~\eqref{eq:indrep-leftadj} is surjective means that for an $H$-intertwiner $f\colon Q \to \Res Q'$ we should produce a $G$-intertwiner $\hat f\colon \Ind Q \to Q'$ such that $\hat f[1,v] = f(v)$ for all $v \in \hilb K$. Algebraic considerations force us to define
\[
\hat f[x^\sigma, v^\sigma] := Q'(x^\sigma) f(v^\sigma)
\]
for every coset $\sigma \in G/H$.

That $\hat f$ can have problems of convergence is illustrated by taking $G := \ZZ$, $H = \{0\}$ and letting $Q$ and $Q'$ be the trivial representations of $H$ and $G$, respectively. In this case the underlying Hilbert space $\Ind \hilb K$ of $\Ind Q$ consists of all square-integrable functions $\ZZ \to \hilb K := \CC$. It contains in particular the function $k \mapsto 1/k$. The domain of the $H$-intertwiner $Q \to \Res Q'$ given by $1_\CC \mapsto 1_\CC$ then cannot be enlarged to $\Ind Q$ because the series $\sum_{k \in \ZZ} 1/k$ does not converge in $\hilb K' := \CC$.

To show that $\Ind$ is neither right adjoint to $\Res$ we observe that there always exists the injective complex linear map
\begin{equation}
\label{eq:indrep-rightadj}
\Hom_G(Q', \Ind Q) \hookrightarrow \Hom_H(\Res Q', Q)
\end{equation}
given by post-composition with the $H$-intertwining orthogonal projection $P$ of $\Ind \hilb K$ to $\hilb K$. Showing that~\eqref{eq:indrep-rightadj} is surjective means that for an $H$-intertwiner $f\colon \Res Q' \to Q$ we should produce an $G$-intertwiner $\hat f\colon Q' \to \Ind Q$ such that $P \circ \hat f = f$. We have no choice but to define
\[
\hat f(v) = \sum_{\mathclap{\sigma \in G/H}} \Bigl[x^\sigma, f\bigl(Q'(x^\sigma)^{-1}(v)\bigr)\Bigr]
\]
for every vector $v \in \hilb K'$, where $\{x^\sigma\}_{\sigma \in G/H}$ is any set of representatives of the (left) cosets of $H$ in $G$.

Now choose again $G := \ZZ$, $H = \{0\}$ and $Q$ and $Q'$ to be the trivial representations of $H$ and $G$ respectively. Then the codomain of the $H$-intertwiner $\Res Q' \to Q$ given by $1_\CC \mapsto 1_\CC$ cannot be enlarged to $\Ind Q$ because the vector
\[
\sum_{\mathclap{\sigma \in G/H}} [x^\sigma, 1_\CC]
\]
does not converge in $\Ind \hilb K$.

Nevertheless, these weak forms~\eqref{eq:indrep-leftadj} and~\eqref{eq:indrep-rightadj} of \defn{Frobenius reciprocity} are sufficient to show that the irreducibility of a representation that is induced from a normal subgroup can be tested by calculating conjugate representations. Either one will do, and we will choose to use~\eqref{eq:indrep-leftadj}. We begin with a little preparatory material that is useful for the study of induced representations in general:

\begin{dfn}
Let $g \in G$. The associated representation $Q^g$ of the subgroup $g H g^{-1} \subseteq G$ that is \defn{conjugate} to $Q$ is defined by $Q^g(ghg^{-1}) := Q(h)$ for all $h \in H$ on the Hilbert space $\hilb K$.
\end{dfn}

Obviously, the conjugates of $Q$ are either all reducible or all irreducible. The importance of conjugate representations is their appearance in the restriction back to a normal subgroup of an induced representation:

\begin{lem}
\label{thm:decomp-inducedrep}
Assume that the subgroup $H$ of $G$ is normal. Then for every set of representatives $\{x^\sigma\}_{\sigma \in G/H}$ of the (left) cosets of $H$ in $G$ the conjugate representations $Q^{x^\sigma}$ are representations of $H$ and the unitary maps $\hilb K^\sigma \xrightarrow{\sim} \hilb K$ given by $[x^\sigma, v^\sigma] \mapsto v^\sigma$ induce a unitary isomorphism
\[
\Res_G^H \Ind_H^G Q \cong \bigoplus_{\mathclap{\sigma \in G/H}} Q^{x^\sigma}
\]
of $H$-representations.
\end{lem}
\begin{proof}
Since $H$ is normal in $G$, we have $x^\sigma H (x^\sigma)^{-1} = H$ and so $Q^{x^\sigma}$ is a representation of $H$ given by $Q^{x^\sigma}(h) = Q((x^\sigma)^{-1} h x^\sigma)$ for all $h \in H$. Furthermore, if $h \in H$, then $H$ being normal implies that $h^{-1}\sigma = \sigma$ for every coset $\sigma$. Hence,
\begin{align*}
(\Res \Ind Q)(h)\bigl([x^\sigma, v^\sigma]\bigr)_\sigma &= \bigl([h x^\sigma, v^\sigma]\bigr)_\sigma \\
	&= \Bigl(\bigl[x^\sigma (x^\sigma)^{-1} h x^\sigma, v^\sigma\bigr]\Bigr)_\sigma \\
	&= \biggl(\Bigl[x^\sigma, Q\bigl((x^\sigma)^{-1} h x^\sigma\bigr)(v^\sigma)\Bigr]\biggr)_\sigma.
\end{align*}
This proves what was asked.
\end{proof}

\begin{thm}[Mackey's irreducibility criterion in the case of a normal subgroup]
\label{thm:mackeyirredcrit}
Assume that the subgroup $H$ of $G$ is normal. If $Q$ is irreducible and for some set of representatives $\{x^\sigma\}_{\sigma \in G/H}$ the associated conjugate representations $Q^{x^\sigma}$ for $\sigma \neq H$ are not isomorphic to $Q$, then the induced representation $\Ind Q$ is irreducible as well.
\end{thm}
\begin{proof}
Let us at first not make any assumptions about the representations at hand. It follows from~\eqref{eq:indrep-leftadj} that there is an inclusion
\[
\Hom_G(\Ind Q, \Ind Q) \hookrightarrow \Hom_H(Q, \Res \Ind Q)
\]
of complex vector spaces. Applying \cref{thm:decomp-inducedrep} next to some set of coset representatives $\{x^\sigma\}_{\sigma \in G/H}$, we may expand this as:
\begin{align*}
\Hom_G(\Ind Q, \Ind Q) &\hookrightarrow \Hom_H\Bigl(Q, \bigoplus_{\mathclap{\sigma \in G/H}} Q^{x^\sigma}\Bigr) \\
	&\hookrightarrow \prod_{\mathclap{\sigma \in G/H}} \Hom_H\bigl(Q, Q^{x^\sigma}\bigr) \\
	&= \Hom_H\bigl(Q, Q^{x^H}\bigr) \oplus \prod_{\mathclap{\substack{\sigma \in G/H\\ \sigma \neq H}}} \Hom_H\bigl(Q, Q^{x^\sigma}\bigr).
\end{align*}
Now suppose that $Q$ is irreducible and that $Q^{x^\sigma} \ncong Q$ when $\sigma \neq H$. Then also $Q^{x^\sigma}$ is irreducible for all $x^\sigma$, so by Schur's lemma $\Hom_H(Q, Q^{x^\sigma})$ is zero when $\sigma \neq H$ and $\Hom_H(Q, Q^{x^H})$ is $1$-dimensional. Therefore, $\Hom_G(\Ind Q, \Ind Q)$ is $1$-dimensional and hence $\Ind Q$ is irreducible by Schur's lemma.
\end{proof}

\subsection{The positive energy condition}
\label{app:posenergycond}

The representations of the various groups studied in \cref{sec:unicol-reptheory} and \cref{sec:bicol-reptheory} satisfy one more property besides the ones listed in \cref{dfn:rep}: they are of positive energy. It is this attribute which makes them amenable to classification. We introduce this notion in this section and explain how it interacts with that of irreducibility.

\begin{refs}
The material in this section is largely taken from \parencite[Section 9.2]{pressley:loopgps} and \parencite[Section I.6]{wassermann:opalgsIII}. More specific references will be given in the text.
\end{refs}

Recall from the Introduction that we denote the topological group of counterclockwise rotations of the manifold $S^1$ by $\Rot(S^1)$. We will use the following models for the covering groups of $\Rot(S^1)$:
\[
\Rot^{(\infty)}(S^1) \cong \{\Phi_\theta\colon \RR \xrightarrow{\sim} \RR, \quad \theta' \mapsto \theta' + \theta\}_{\theta \in \RR}
\]
and for $m \geq 1$
\[
\Rot^{(m)}(S^1) \cong \Rot^{(\infty)}(S^1)/m\ZZ,
\]
where $\ZZ$ stands for the subgroup of $\Rot^{(\infty)}(S^1)$ generated by the shift $\theta \mapsto \theta + 1$. We denote elements of $\Rot^{(m)}(S^1)$ as $[\Phi_\theta]$, where $\Phi_\theta \in \Rot^{(\infty)}(S^1)$ and the square brackets stand for its equivalence class in $\Rot^{(m)}(S^1)$. The image $[\Phi_\theta]$ in $\Rot(S^1)$ of $\Phi_\theta \in \Rot^{(\infty)}(S^1)$ is the anti-clockwise rotation by angle $\theta$, so $[\Phi_1] = \id_{S^1}$.

It is a well-known fact that a non-zero representation of a compact group can be written as the (internal) direct sum of a set of mutually orthogonal, irreducible, finite-dimensional subrepresentations. (See for example \parencite[Corollary 5.4.2]{kowalski:reptheory}.) This applies in particular to the group $\Rot^{(m)}(S^1)$ for some $m \geq 1$. Collecting the irreducible representations together into isotypic components, such a representation $R$ on a Hilbert space $\hilb H$ can then be written as the completion
\[
\hilb H = \overline{\bigoplus_{\mathclap{a \in (1/m)\ZZ}} \hilb H(a)},
\]
where $\hilb H(a)$ is the isotypic component
\[
\hilb H(a) := \bigl\{v \in \hilb H \bigm\vert R[\Phi_\theta](v) = e^{-2\pi i a\theta} v \text{ for all $\theta \in [0,1]$}\bigr\}
\]
on which $R$ acts by the $(-a)$-th character of $\Rot^{(m)}(S^1)$.

\begin{dfn}
\label{dfn:posenergyrep}
A representation $R$ of $\Rot^{(m)}(S^1)$ for some $m \geq 1$ on a Hilbert space $\hilb H$ is said to be of \defn{positive energy} if both
\begin{itemize}
\item the isotypic components $\hilb H(a)$ are zero for $a < 0$, and
\item $R$ is of \defn{finite type}, that is, for each $a \in (1/m)\ZZ_{\geq 0}$ the dimension of $\hilb H(a)$ is finite.
\end{itemize}
We then call $\hilb H(a)$ the \defn{$a$-th energy eigenspace} of $\hilb H$.

Let $N$ be a group together with a continuous $\Rot^{(m)}(S^1)$-action on it. A representation $Q$ of $N$ on a Hilbert space $\hilb H$ is said to be of \defn{positive energy} if there exists an extension of $Q$ to a representation of the semidirect product $N \rtimes \Rot^{(m)}(S^1)$ on $\hilb H$ such that its restriction to $\Rot^{(m)}(S^1)$ is of positive energy. A \defn{morphism} between two positive energy $N$-representations is defined to be a morphism of $N$-representations.
\end{dfn}

Notice that the extension to $N \rtimes \Rot^{(m)}(S^1)$ is not part of the data of a positive energy representation which is why, correspondingly, morphisms are not required to intertwine the rotation actions on the respective Hilbert spaces either.

\begin{rmk}
\label{rmk:rots1action-unique}
Write the (left) action of $\Rot^{(m)}(S^1)$ on $N$ as $[\Phi_\theta]^*g := [\Phi_\theta] \cdot g$, where $g \in N$. Recall that $N \rtimes \Rot^{(m)}(S^1)$ has $N \times \Rot^{(m)}(S^1)$ as its underlying topological space and that its multiplication is defined by
\[
\bigl(g, [\Phi_\theta]\bigr) \cdot \bigl(g', [\Phi_{\theta'}]\bigr) := \bigl(g [\Phi_\theta]^*(g'), [\Phi_{\theta + \theta'}]\bigr)
\]
for all $g, g' \in N$ and $[\Phi_\theta], [\Phi_{\theta'}] \in \Rot^{(m)}(S^1)$. It is then easily seen that $Q$ extending to $N \rtimes \Rot^{(m)}(S^1)$ is equivalent to there existing a representation $R$ of $\Rot^{(m)}(S^1)$ on $\hilb H$ satisfying the intertwining property
\begin{equation}
\label{eq:rotS1intertwin}
R[\Phi_\theta]Q(g)R[\Phi_\theta]^* = Q\bigl([\Phi_\theta]^*g\bigr).
\end{equation}
That is, $R[\Phi_\theta]$ is an isomorphism from $Q$ to the `twisted' representation $Q \circ [\Phi_\theta]$.

If $\chi_a$ is the character $[\Phi_\theta] \mapsto e^{2\pi i a \theta}$ of $\Rot^{(m)}(S^1)$, then of course $\chi_a \cdot R$ also satisfies~\eqref{eq:rotS1intertwin} and $\chi_a \cdot R$ is again of positive energy if $a \leq 0$. This shows that the lift of $Q$ to $N \rtimes \Rot^{(m)}(S^1)$ is never unique. However, this is the only indeterminacy if $Q$ is irreducible. Let $R'$ namely be another representation of $\Rot^{(m)}(S^1)$ on $\hilb H$ intertwining like $R$ with $Q$. Then $[\Phi_\theta] \mapsto R[\Phi_\theta]^* R'[\Phi_\theta]$ is a representation of $\Rot^{(m)}(S^1)$ on $\hilb H$ which commutes with $Q$, hence it must be a character by Schur's lemma.
\end{rmk}

\begin{rmk}
Of course, any finite-dimensional representation of any group $N$ on a Hilbert space $\hilb H$ is of positive energy by making $\Rot(S^1)$ act as the identity on both $N$ and $\hilb H$. However, the groups $N$ we study in this thesis already come with natural, non-trivial actions of $\Rot^{(m)}(S^1)$ for some $m$, hence making the demand of positivity of energy for their representations restrictive and interesting.
\end{rmk}

\begin{rmk}
One could obviously relax the assumption of positivity of energy to obtain the more general notion of a representation of which the energy is \defn{bounded from below}. However, for such a representation it would again not be appropriate to require the $\Rot^{(m)}(S^1)$-action $R$ to be part of the data. Therefore we would be free to `shift the energy', that is, to multiply $R$ with a character $\chi_a$, obtaining a positive energy representation for small enough negative $a$. We thus see that this generalisation is vacuous.
\end{rmk}

The notions of irreducibility for the groups $N$ and $N \rtimes \Rot^{(m)}(S^1)$ agree when positivity of energy is assumed:

\begin{lem}
\label{thm:irrepofcrossprod-irrepgp}
Let $N$ be a group together with a continuous $\Rot^{(m)}(S^1)$-action on it. An irreducible representation of $N \rtimes \Rot^{(m)}(S^1)$ such that the restriction to $\Rot^{(m)}(S^1)$ is of positive energy is also irreducible as a representation of $N$.
\end{lem}
The proof of this \namecref{thm:irrepofcrossprod-irrepgp} will actually not use that the representation of $\Rot^{(m)}(S^1)$ is of finite type.
\begin{proof}
(Taken from \parencite[Proposition 9.2.3]{pressley:loopgps}.) Denote the underlying Hilbert space of the representation in the \namecref{thm:irrepofcrossprod-irrepgp} by $\hilb H$ and the representations of $N$ and $ \Rot^{(m)}(S^1)$ on $\hilb H$ by $Q$ and $R$ respectively. Let $P$ be the orthogonal projection of $\hilb H$ onto a subrepresentation of $Q$. It commutes with $Q$. We will prove that $P$ also commutes with $R$, which will then imply that the subrepresentation is invariant under $N \rtimes \Rot^{(m)}(S^1)$ and must therefore be either $\{0\}$ or $\hilb H$ itself.

Consider the $m$-periodic function $\theta \mapsto R[\Phi_\theta] P R[\Phi_\theta]^*$ on the real line with values in the  orthogonal projections on $\hilb H$. We want to show that it is constant because filling in $\theta = 0$ will then imply that $P$ commutes with $R$. Using the intertwining relation of $R$ with $Q$ and the fact that $P$ commutes with $Q$ it is easily checked that this family of projections commutes with $Q$. Define for each $a \in (1/m)\ZZ$ the $a$-th Fourier coefficient of this function as the operator
\[
P_a := \frac1m \int_0^m e^{-2\pi i a \theta} R[\Phi_\theta] P R[\Phi_\theta]^* \dd \theta.
\]

We give a precise meaning to this integral as follows. If $v, w \in \hilb H$ are two fixed vectors, then the two functions $\RR \to \hilb H$ given by $\theta \mapsto P R[\Phi_\theta]^*v$ and $\theta \mapsto R[\Phi_\theta]^*w$ are continuous thanks to the strong continuity of $R$ and the continuity of $P$. The function
\begin{equation}
\label{eq:posenergy-1}
\theta \mapsto \bigl\langle P R[\Phi_\theta]^*v, R[\Phi_\theta]^* w \bigr\rangle = \bigl\langle R[\Phi_\theta] P R[\Phi_\theta]^*v, w \bigr\rangle
\end{equation}
is then continuous as well and so the $a$-th Fourier coefficient
\[
\frac1m \int_0^m e^{-2\pi i a \theta} \bigl\langle R[\Phi_\theta] P R[\Phi_\theta]^*v, w \bigr\rangle \dd\theta
\]
of~\eqref{eq:posenergy-1} exists. Next, we observe that the absolute value of this integral is bounded by
\begin{align*}
\frac1m \int_0^m \Bigl\lvert  e^{-2\pi i a \theta} \bigl\langle R[\Phi_\theta] P R[\Phi_\theta]^*v, w \bigr\rangle \Bigr\rvert \dd\theta 
	&\leq \frac1m \int_0^m \bigr\lVert R[\Phi_\theta] P R[\Phi_\theta]^*v \bigr\rVert \cdot \norm w \dd\theta \\
	&\leq \frac1m \int_0^m \norm v \cdot \norm w \dd\theta = \norm v \cdot \norm w,
\end{align*}
where we used that $R$ is unitary and $P$ is bounded. This estimate shows both that for $v \in \hilb H$ there exists a unique vector $P_a(v) \in \hilb H$ satisfying
\[
\langle P_a v, w \rangle = \frac1m \int_0^m e^{-2\pi i a \theta} \bigl\langle R[\Phi_\theta] P R[\Phi_\theta]^*v, w \bigr\rangle \dd\theta
\]
for all $w \in \hilb H$, and that the operator $P_a$ thus defined is bounded.

Using the fact that $R[\Phi_\theta] P R[\Phi_\theta]^*$ commutes with $Q$ one can show that $P_a$ commutes with $Q$ as well. Furthermore, since $R[\Phi_\theta] P R[\Phi_\theta]^*$ is self-adjoint we have $P_a^* = P_{-a}$:
\begin{align*}
\langle P_a v, w \rangle &= \frac1m \int_0^m e^{-2\pi i a \theta} \bigl\langle v, R[\Phi_\theta] P R[\Phi_\theta]^*w \bigr\rangle \dd\theta \\
	&= \frac1m \int_0^m \overline{e^{2\pi i a \theta} \bigl\langle R[\Phi_\theta] P R[\Phi_\theta]^*w, v \bigr\rangle} \dd\theta \\
	&= \overline{\langle P_{-a} w, v \rangle} = \langle v, P_{-a} w \rangle.
\end{align*}
Pick $b \in (1/m)\ZZ$ and a vector $v \in \hilb H(b)$ in the associated energy eigenspace. Then we calculate the energy of $P_a(v)$ as follows. Let $\theta' \in \RR$ and write
\begin{multline*}
\bigl\langle R[\Phi_{\theta'}] P_a v, w \bigr\rangle \\
\begin{aligned}
	&= \frac1m \int_0^m e^{-2\pi i a \theta} \bigl\langle R[\Phi_{\theta' + \theta}] P R[\Phi_\theta]^*v, w \bigr\rangle \dd\theta \\
	&= e^{2\pi i a \theta'} \frac1m \int_0^m e^{-2\pi i a (\theta' + \theta)} \bigl\langle R[\Phi_{\theta' + \theta}] P R[\Phi_{\theta' + \theta}]^* R[\Phi_{-\theta'}]^* v, w \bigr\rangle \dd\theta \\
	&= e^{2\pi i (a-b) \theta'} \frac1m \int_0^m e^{-2\pi i a (\theta' + \theta)} \bigl\langle R[\Phi_{\theta' + \theta}] P R[\Phi_{\theta' + \theta}]^* v, w \bigr\rangle \dd\theta \\
	&= e^{2\pi i (a-b) \theta'} \frac1m \int_{\theta'}^{\theta' + m} e^{-2\pi i a \theta} \bigl\langle R[\Phi_\theta] P R[\Phi_\theta]^* v, w \bigr\rangle \dd\theta \\
	&= e^{-2\pi i (b-a) \theta'} \langle P_a v, w \rangle.
\end{aligned}
\end{multline*}
We conclude from this that $P_a$ maps $\hilb H(b)$ to $\hilb H(b-a)$.

Let $b \in (1/m)\ZZ$ be the lowest energy level of $R$. Then $P_a$ annihilates $\hilb H(b)$ for $a > 0$. Because $\hilb H$ is irreducible as a representation of $N \rtimes \Rot^{(m)}(S^1)$ it is generated by $\hilb H(b)$ under the action of $N \rtimes \Rot^{(m)}(S^1)$. In other words, $\hilb H$ is the closure of the span of the set of vectors
\[
\bigl\{Q(g) R[\Phi_\theta](v) \bigm\vert \text{$g \in N$, $\Phi_\theta \in \Rot^{(\infty)}(S^1)$ and $v \in \hilb H(b)$}\bigr\}.
\]
However, since $R$ leaves $\hilb H(b)$ invariant $\hilb H$ is also generated by $\hilb H(b)$ under the action of $Q$ alone. The continuity of $P_a$ and it commuting with $Q$ then implies that $P_a$ for $a > 0$ annihilates all of $\hilb H$, that is, $P_a = 0$.

From the relation $P_a^* = P_{-a}$ it now follows that $P_a = 0$ for all $a \neq 0$. Because for each two fixed $v$ and $w$ the scalars $\langle P_a v, w \rangle$ are the Fourier coefficients of the continuous function~\eqref{eq:posenergy-1}, this function is constant by Fejér's theorem. In particular, $\langle R[\Phi_\theta] P R[\Phi_\theta]^*v, w \rangle = \langle Pv, w \rangle$ for all $\theta \in \RR$. We conclude that $R[\Phi_\theta] P R[\Phi_\theta]^* = P$ for all $\theta \in \RR$.
\end{proof}

An alternative proof of the above result is demonstrated in \parencite[Theorem 1.5]{neeb:posenergyreps} which uses the Borchers--Arveson theorem from the theory of von Neumann algebras.

The next criterion guarantees existence of an irreducible subrepresentation.


\begin{prop}
\label{thm:gprep-findimisotcomp-containsirrep}
Let $G$ be a group containing $\Rot^{(m)}(S^1)$. Then a non-zero representation of $G$ contains an irreducible subrepresentation when at least one non-zero isotypic component of the $\Rot^{(m)}(S^1)$-action is finite-dimensional.
\end{prop}

Our only use of the above result in this thesis is when $G$ is of the form $N \rtimes \Rot^{(m)}(S^1)$ for some group $N$ carrying an action of $\Rot^{(m)}(S^1)$.

\begin{proof}
(Compare with \parencite[Proposition I.6(c)]{wassermann:opalgsIII}.) Denote the representation and its underlying Hilbert space by $Q$ and $\hilb H$, respectively. Pick a non-zero, finite-dimensional $\Rot^{(m)}(S^1)$-isotypic component $\hilb H(a)$, where $a \in (1/m)\ZZ$. We first show that for any subrepresentation $\hilb K$ of $Q$ there is a decomposition
\begin{equation}
\label{eq:energyeigenspacesplits}
\hilb H(a) = \bigl(\hilb H(a) \cap \hilb K\bigr) \oplus \bigl(\hilb H(a) \cap \hilb K^\perp\bigr).
\end{equation}
One inclusion is obvious. For the reverse inclusion, let $v \in \hilb H(a)$ and write it as $v = w + w^\perp$ with $w \in \hilb K$ and $w^\perp \in \hilb K^\perp$. We want to prove that $w, w^\perp \in \hilb H(a)$. Take $[\Phi_\theta] \in \Rot^{(m)}(S^1)$. Then on the one hand
\[
Q\bigl([\Phi_\theta]\bigr)(v) = Q\bigl([\Phi_\theta]\bigr)(w) + Q\bigl([\Phi_\theta]\bigr)(w^\perp).
\]
On the other hand,
\[
Q\bigl([\Phi_\theta]\bigr)(v) = e^{2\pi i a \theta} v = e^{2\pi i a \theta} w + e^{2\pi i a \theta} w^\perp,
\]
which proves~\eqref{eq:energyeigenspacesplits} by the unicity of decompositions of vectors along a direct sum.

We may without loss of generality now assume that $\hilb H$ is generated by $\hilb H(a)$ under $Q$. Otherwise we could consider the subrepresentation of $\hilb H$ which \emph{is} generated by $\hilb H(a)$ under $Q$ instead because this satisfies the same hypotheses in the \namecref{thm:gprep-findimisotcomp-containsirrep} as $\hilb H$.

Consider the following set
\[
\bigl\{\hilb K \cap \hilb H(a) \bigm\vert \text{$\hilb K$ is a subrepresentation of $Q$}\bigr\} \backslash \{0\}
\]
of (finite-dimensional) $\Rot^{(m)}(S^1)$-subrepresentations of $\hilb H(a)$. It is non-empty because it contains at least $\hilb H(a) = \hilb H \cap \hilb H(a)$. Let $\hilb K_0$ be a subrepresentation of $Q$ such that $\hilb K_0 \cap \hilb H(a)$ has the smallest dimension of all the elements in the above set. We claim that $\hilb K_0$ is irreducible for $Q$.

Let $\hilb K \subseteq \hilb K_0$ be a subrepresentation of $Q$. Suppose that $\hilb K \cap \hilb H(a) = \hilb K_0 \cap \hilb H(a)$. Then $(\hilb K^\perp \cap \hilb K_0) \cap \hilb H(a) = \{0\}$ and so by~\eqref{eq:energyeigenspacesplits} there holds
\[
\hilb H(a) \subseteq (\hilb K^\perp \cap \hilb K_0)^\perp = \hilb K + \hilb K_0^\perp.
\]
Because $\hilb H$ is generated under $Q$ by $\hilb H(a)$ it follows that $\hilb K + \hilb K_0^\perp = \hilb H$ and so $\hilb K^\perp \cap \hilb K_0 = \{0\}$. Hence we have $\hilb K = \hilb K_0$. Suppose on the contrary that $\hilb K \cap \hilb H(a) \subsetneq \hilb K_0 \cap \hilb H(a)$. Then, by the minimality assumption on $\hilb K_0$, we have $\hilb K \cap \hilb H(a) = \{0\}$. By~\eqref{eq:energyeigenspacesplits} this implies that $\hilb H(a)$ is entirely contained in $\hilb K^\perp$. Using again that $\hilb H$ is generated under $Q$ by $\hilb H(a)$ we conclude that $\hilb K^\perp = \hilb H$ and so $\hilb K = \{0\}$.
\end{proof}

Having proved these results, we show that positive energy representations are similar to finite-dimensional ones in the following regard:

\begin{prop}
\label{thm:posrepcontainsirrep}
(See \parencite[Proposition I.6(d)]{wassermann:opalgsIII}.) A non-zero, positive energy representation of a group is the (internal) direct sum of a set of mutually orthogonal, irreducible, positive energy subrepresentations.
\end{prop}
\begin{proof}
Denote the group, the representation and its underlying Hilbert space by $N$, $Q$ and $\hilb H$, respectively. We are given that for some $m \geq 1$ there exists a positive energy representation $R$ of $\Rot^{(m)}(S^1)$ on $\hilb H$ which intertwines with $Q$.

A non-zero subrepresentation $\hilb K \subseteq \hilb H$ for $N \rtimes \Rot^{(m)}(S^1)$ is again the direct sum of its energy eigenspaces for the restriction of $R$ to $\hilb K$, each of which is of course again finite-dimensional. Therefore at least one of these eigenspaces must be non-zero and so $\hilb K$ contains an irreducible subrepresentation for $N \rtimes \Rot^{(m)}(S^1)$ by \cref{thm:gprep-findimisotcomp-containsirrep}. The representation of $N \rtimes \Rot^{(m)}(S^1)$ on $\hilb H$ thus satisfies the criterion of \cref{thm:gprep-complred}. The summands in the resulting orthogonal decomposition into irreducible subrepresentations for $N \rtimes \Rot^{(m)}(S^1)$ are then irreducible for $Q$ as well by \cref{thm:irrepofcrossprod-irrepgp}.
\end{proof}

We close this section by introducing notation and a name for the generating function of the dimensions of the energy eigenspaces of a positive energy representation.

\begin{dfn}
\label{dfn:gradchar}
(After \parencite[Section 1.10]{frenkel:monster} and \parencite[Definition 14.1.1]{pressley:loopgps}.)
Let $R$ be a positive energy representation of $\Rot^{(m)}(S^1)$ for some $m \geq 1$ on a Hilbert space $\hilb H$. The formal Laurent series
\[
\ch_R(q) := \sum_{\mathclap{a \in (1/m)\ZZ_{\geq 0}}} \dim \bigl(\hilb H(a)\bigr) q^a
\]
in a formal variable $q$ is called the \defn{(graded) character} of $R$.
\end{dfn}

\section{Heisenberg groups}
\label{app:heisgps}

In this section we will define Heisenberg groups and their Weyl representations. They are relevant to this thesis because the central extensions of uni- and bicoloured torus loop groups constructed in \cref{sec:unicol-centext,sec:bicol-centext} turn out to contain Heisenberg groups. This fact is the key to the representation theory of the former.

\begin{refs}
The material in this section is taken from \parencite[Section 9.5]{pressley:loopgps}, \parencite[Chapter II]{parthasarathy:quantstoch} and \parencite[Chapter I]{ismagilov:reps-infdim}. More specific references will be given in the text.
\end{refs}

\begin{dfn}
\label{dfn:heisgp}
Let $V$ be a topological real vector space carrying a continuous, non-degenerate, bilinear skew form $S\colon V \times V \to \RR$. The \defn{Heisenberg group} $\centV$ associated to the pair $(V, S)$ is the $\phasegp$-central extension of the underlying topological abelian group of $V$ defined by the $2$-cocycle
\[
c\colon V \times V \to \phasegp, \qquad c(\xi, \eta) := e^{-2\pi i S(\xi, \eta)}
\]
on $V$.
\end{dfn}

Spelling this out, $\centV$ has as underlying topological space $V \times \phasegp$ and its continuous multiplication and inverse are given by
\[
(\xi, z) \cdot (\eta, w) := \bigl(\xi + \eta, zw \cdot c(\xi, \eta)\bigr), \qquad (\xi, z)^{-1} = \bigl(-\xi, z^{-1} \cdot c(\xi, -\xi)^{-1}\bigr),
\]
respectively, for $\xi, \eta \in V$ and $z, w \in \phasegp$.

Note that $c$ is indeed a normalised $2$-cocycle because $S$ is bi-additive. The centre of $\centV$ consists of (the image of) $\phasegp$ alone. If namely an element $(\xi, z) \in \centV$ commutes with all $(\eta, w) \in \centV$, then $c(\xi, \eta)^2 = 1$ using the skew-symmetry of $S$ and so $2S(\xi, \eta) \in \ZZ$ for all $\eta \in V$. The bilinearity of $S$ then implies that $S(\xi, \eta) = 0$ and therefore $\xi = 0$ because $S$ is non-degenerate. The non-degeneracy and skew-symmetry of $S$ hence ensure that the centre of $\centV$ is as small as possible.

Of course, we did not need the structure of a scalar multiplication on $V$ in order to define the Heisenberg group. The vector space structure is needed to get a hold on the representation theory---a topic we turn to now. We will first construct the underlying Hilbert spaces of the representations we will be studying.

\subsection{Bosonic Fock spaces}
\label{subsec:focksp}

Let $V$ be a real vector space carrying a \defn{complex structure} $J\colon V \to V$, meaning a linear endomorphism (which is automatically an automorphism) such that $J^2 = -\id_V$. Let furthermore $\langle \cdot, \cdot \rangle_J$ be a Hermitian inner product on the corresponding complex vector space $V_J$. Then we associate a new, larger complex Hilbert space to the complex pre-Hilbert space $(V_J, \langle \cdot, \cdot \rangle_J)$ as follows.

First, form the \defn{symmetric algebra} $\Sym^*(V_J)$ on $V_J$. It is an algebraic direct sum
\[
\Sym^*(V_J) := \bigoplus_{k=0}^\infty \Sym^k(V_J),
\]
where the \defn{$k$-th symmetric power} $\Sym^k(V_J)$ of $V_J$ is the complex linear span of monomials in the vectors of $V_J$ of degree $k$, which are considered as commuting variables. By definition, $\Sym^0(V_J) := \CC$ and the $0$-th power of any vector in $V_J$ is $1 \in \Sym^0(V_J)$. We will ignore the algebra structure and consider $\Sym^*(V_J)$ as a complex vector space. A general vector of it is a sum
\begin{equation}
\label{eq:fockspvector}
\sum_{k=0}^\infty v_k,
\end{equation}
where $v_k \in \Sym^k(V_J)$ and only finitely many summands are non-zero.

This vector space $\Sym^*(V_J)$ inherits a Hermitian inner product from $V_J$. If $\xi_1 \xi_2 \cdots \xi_k$ and $\xi_1' \xi_2' \cdots \xi_k'$ are namely two monomials of the same degree $k$ then we set their inner product to be
\begin{equation}
\label{eq:fockspinnerprod}
\langle \xi_1 \xi_2 \cdots \xi_k, \xi_1' \xi_2' \cdots \xi_k' \rangle := \sum_{\mathclap{\sigma \in S_k}} \langle \xi_1, \xi_{\sigma(1)}' \rangle_J \cdots \langle \xi_k, \xi_{\sigma(k)}' \rangle_J,
\end{equation}
where $S_k$ is the symmetric group on $k$ symbols. Two monomials of different degrees are set to be orthogonal. This inner product on monomials is finally extended to general vectors of $\Sym^*(V_J)$ by (Hermitian) linearity.

The positive definiteness can be seen by picking an orthonormal basis for $V_J$. The values of the inner product are then determined on the monomials in these basis elements. If in the expression~\eqref{eq:fockspinnerprod} the vectors $\xi_i$ and $\xi_i'$ are such basis elements, this value is only non-zero when there is a permutation $\sigma$ such that $\xi_{\sigma(i)}' = \xi_i$ for all $i$ and in that case it is clearly non-negative because $\langle \cdot, \cdot \rangle_J$ is positive definite.

We conclude that $\Sym^*(V_J)$ in this way becomes a complex pre-Hilbert space.

\begin{dfn}
The \defn{bosonic Fock space} $\hilb S(V_J)$ associated to the triple
\[
(V, J, \langle \cdot, \cdot \rangle_J)
\]
is the Hilbert space completion of $\Sym^*(V_J)$.
\end{dfn}

A general vector of $\hilb S(V_J)$ is a series of the form~\eqref{eq:fockspvector} where possibly infinitely many summands are non-zero, but
\[
\sum_{k=0}^\infty \norm{v_k}^2 < \infty.
\]
The inner product of two such vectors is given summand-wise.

In order to later define unitary operators on $\hilb S(V_J)$ we single out a special class of its vectors.

\begin{dfn}
Let $\xi \in V_J$. The \defn{coherent vector} associated to $\xi$ is the formal power series
\[
e^\xi := \sum_{k=0}^\infty \frac{\xi^k}{k!}
\]
seen as an element of the infinite product $\prod_{k=0}^\infty \Sym^k(V_J)$.
\end{dfn}

\begin{prop}
\label{thm:cohvecs-props}
Coherent vectors satisfy the following properties:
\begin{enumerate}
\item\label{thmitm:cohvecs-props-conv} $e^\xi \in \hilb S(V_J)$ for all $\xi \in V_J$, that is, the partial sums defining $e^\xi$ converge in $\hilb S(V_J)$,
\item\label{thmitm:cohvecs-props-innerprod} if also $\eta \in V_J$ then $\langle e^\xi, e^\eta \rangle = e^{\langle \xi, \eta \rangle_J}$,
\item\label{thmitm:cohvecs-props-lindep} every finite set of coherent vectors is linearly independent,
\item\label{thmitm:cohvecs-props-dense} the space of finite linear combinations of coherent vectors lies densely in $\hilb S(V_J)$,
\item\label{thmitm:cohvecs-props-cont} the function $V_J \to \hilb S(V_J)$ given by $\xi \mapsto e^\xi$ is continuous with respect to the norm topology on $V_J$ induced by $\langle \cdot, \cdot \rangle_J$.
\end{enumerate}
\end{prop}
\begin{proof}
If $n \geq 0$, then we calculate that, according to the definition of the inner product~\eqref{eq:fockspinnerprod},
\[
\biggl\langle \sum_{k=0}^n \frac{\xi^k}{k!}, \sum_{l=0}^n \frac{\eta^l}{l!} \biggr\rangle = \sum_{k,l=0}^n \frac{\langle \xi^k, \eta^l \rangle}{k! \cdot l!} 
	= \sum_{k=0}^n \frac{k! \cdot \langle \xi, \eta \rangle_J^k}{(k!)^2} 
	= \sum_{k=0}^n \frac{\langle \xi, \eta \rangle_J^k}{k!}.
\]
This proves both~\ref{thmitm:cohvecs-props-conv} and~\ref{thmitm:cohvecs-props-innerprod}.

\ref{thmitm:cohvecs-props-lindep}: Let $e^{\xi_1}, \ldots, e^{\xi_n}$ be a finite set of coherent vectors and suppose that $\sum_{k=1}^n \alpha_k e^{\xi_k} = 0$ for some scalars $\alpha_k \in \CC$. Then define for each (unordered) pair $\{k, l\} \subseteq \{1, \ldots, n\}$ of indices with $k \neq l$ the following subset of vectors of $V_J$:
\[
E_{k,l} := \bigl\{\xi \in V_J \bigm\vert \langle \xi, \xi_k \rangle_J \neq \langle \xi, \xi_k \rangle_J\bigr\}.
\]
One can think of $E_{k,l}$ as the complement of the hyperplane orthogonal to the line $\CC\cdot(\xi_k - \xi_l)$. It is open and lies densely in $V_J$. Since a finite intersection of open and dense subsets is again dense, the intersection
\[
\bigcap_{\mathclap{\substack{1 \leq k, l \leq n\\ k\neq l}}} E_{k,l}
\]
is in particular not empty. In other words: there exists a vector $\xi \in V_J$ such that the $n$ scalars $z_k := \langle \xi, \xi_k \rangle_J$ are all distinct. It is well known that this implies that the $n$ functions $\CC \to \CC$ given by $z \mapsto e^{z \cdot z_k}$ are linearly independent. Because for all $z \in \CC$ there holds
\[
\sum_{k=1}^n \overline{\alpha_k} e^{z \cdot z_k} = \sum_{k=1}^n \overline{\alpha_k} \langle e^{z \cdot \xi}, e^{\xi_k} \rangle = \Bigl\langle e^{z \cdot \xi}, \sum_{k=1}^n \alpha_k e^{\xi_k} \Bigr\rangle = 0
\]
by \ref{thmitm:cohvecs-props-innerprod}, we conclude that $\alpha_k = 0$ for all $k$.

\ref{thmitm:cohvecs-props-dense}: Set $\hilb H \subseteq \hilb S(V_J)$ to be the closure of the span of all the exponential vectors. Let $\xi \in V_J$ and consider the function $\RR \to \hilb H$ given by $\theta \mapsto e^{\theta \cdot \xi}$. It is smooth because the series $e^{\theta \cdot \xi}$ is absolutely convergent. Its $n$-th derivative is valued in $\hilb H$ also and is given by term-wise differentation:
\[
\theta \mapsto \sum_{k=0}^\infty \xi^k \cdot \frac{(\theta \xi)^k}{k!}.
\]
In particular, the monomial
\begin{equation}
\label{eq:cohvecs-dense-power}
\xi^k = \Bigl.\frac{\dd^k}{\dd\theta^k} e^{\theta \cdot \xi}\Bigr\rvert_{\theta=0}
\end{equation}
lies in $\hilb H$ for all $k \geq 0$. Now note that by taking partial derivatives a more general monomial $\xi_1 \cdots \xi_k$ can be expressed in terms of a power of a single vector:
\begin{equation}
\label{eq:cohvecs-dense-genmon}
k! \cdot \xi_1 \cdots \xi_k = \Bigl.\frac{\partial^k}{\partial \theta_1 \cdots \partial \theta_k} (\theta_1 \xi_1 + \cdots + \theta_k \xi_k)^k\Bigr\rvert_{\theta_1 = \cdots = \theta_k = 0}.
\end{equation}
Having just learned that $(\sum_{i=1}^k \theta_i \xi_i)^k \in \hilb H$, we see that $\xi_1 \cdots \xi_k \in \hilb H$ as well. Hence, $\hilb H$ equals $\hilb S(V_J)$.

\ref{thmitm:cohvecs-props-cont}: If $\xi, \eta \in V_J$ then, using \ref{thmitm:cohvecs-props-innerprod},
\[
\norm{e^\xi - e^\eta}^2 = e^{\norm \xi_J^2} + e^{\norm \eta_J^2} - 2 \cdot \Re e^{\langle \xi, \eta\rangle_J}.
\]
Together with the continuity of the inner product $\langle \cdot, \cdot \rangle_J$ this proves what was asked.
\end{proof}

The properties~\ref{thmitm:cohvecs-props-conv}--\ref{thmitm:cohvecs-props-dense} in this \namecref{thm:cohvecs-props} allow one to uniquely specify a unitary operator on $\hilb S(V_J)$ by prescribing its values on the coherent vectors and checking whether it then preserves inner products.

\subsection{The Weyl representations of a Heisenberg group}

Given a real vector space, we saw in \cref{dfn:heisgp} how to associate a Heisenberg group to a skew form, and in \cref{subsec:focksp} how to forge a bosonic Fock space from a complex structure and a Hermitian inner product. In this subsection we will build a representation of the former group on the latter Hilbert space, but in order to do so the skew form and the complex structure will need to be compatible in a certain sense and, moreover, the inner product should be a specific one derived from these two pieces of structure.

\begin{prop}
\label{thm:tamedcompstruct}
Let $V$ be a real vector space carrying a non-degenerate, bilinear skew form $S$ and a complex structure $J\colon V \to V$. Write $V_J$ for the associated complex vector space. Then the following compatibility requirements between $S$ and $J$ are equivalent:
\begin{enumerate}
\item\label{thmitm:tamedcompstruct-compat} $S$ is $J$-invariant, that is, $S(J\xi, J\eta) = S(\xi, \eta)$ for all $\xi, \eta \in V$, and $J$ \defn{tames} $S$, meaning that $S(\xi, J\xi) > 0$ for all non-zero $\xi \in V$,
\item\label{thmitm:tamedcompstruct-realHilb} the bilinear form
\[
g_J\colon V \times V \to \RR, \qquad g_J(\xi,\eta) := S(\xi, J\eta)
\]
is $J$-invariant and makes $V$ a real pre-Hilbert space,
\item\label{thmitm:tamedcompstruct-compHilb} the form
\[
\langle \cdot, \cdot \rangle_J\colon V_J \times V_J \to \CC, \qquad \langle \xi, \eta \rangle_J := 2 \pi \bigl(S(\xi, J\eta) - i S(\xi, \eta)\bigr),
\]
makes $V_J$ a complex pre-Hilbert space, that is, $\langle \cdot, \cdot \rangle_J$ is a Hermitian inner product.\footnote{Recall our convention that a Hermitian inner product is complex linear in its \emph{first} variable.}\footnote{The factor $2\pi$ in the definition of $\langle \cdot, \cdot \rangle_J$ is of course irrelevant for the statement of this \namecref{thm:tamedcompstruct} on its own. It serves as a normalisation to ensure compatibility with our \cref{dfn:heisgp} of a Heisenberg group.}
\end{enumerate}
\end{prop}

\begin{rmk}
If a pair $(S, J)$ satisfies the equivalent properties of \cref{thm:tamedcompstruct}, then so does the pair $(-S, -J)$. One gets a Hermitian inner product on the conjugate complex pre-Hilbert space $V_{-J}$.
\end{rmk}

Let $V$ be a real vector space carrying a non-degenerate, bilinear skew form $S$ and a complex structure $J\colon V \to V$. Write $V_J$ for the associated complex vector space and suppose that $S$ and $J$ satisfy the equivalent compatibility requirements of \cref{thm:tamedcompstruct}. Then $S$ is continuous with respect to the norm topology on $V$ induced by the Hermitian inner product $\langle \cdot, \cdot \rangle_J$ because the latter is continuous and $S$ is up to a scalar its imaginary part:
\[
S(\xi,\eta)
	= \frac{i}{4\pi} \bigl(\langle \xi,\eta \rangle_J - \langle \eta,\xi \rangle_J\bigr) 
	= -\frac{1}{2\pi} \cdot \Im \langle \xi,\eta \rangle_J.
\]
We may therefore construct a Heisenberg group $\centV$ from $S$.

\begin{thm}
Let $\hilb S(V_J)$ be the bosonic Fock space associated to the Hermitian inner product $\langle \cdot, \cdot \rangle_J$ in \cref{thm:tamedcompstruct}\ref{thmitm:tamedcompstruct-compHilb}. Define for $(\xi, z) \in \centV$ and $\kappa \in V_J$ a vector\footnote{On the right hand side of the following equation we consider $\xi$ as a vector of $V_J$.}
\begin{equation}
\label{eq:weylrep}
W_J(\xi, z)(e^\kappa) := z \cdot e^{-\frac12 \langle \xi, \xi \rangle_J - \langle \kappa, \xi \rangle_J} \cdot e^{\kappa + \xi} \in \hilb S(V_J).
\end{equation}
Then
\begin{enumerate}
\item this is an action of the group $\centV$ on the space of finite linear combinations of the coherent vectors in $\hilb S(V_J)$,
\item which preserves inner products.
\end{enumerate}
Therefore, \eqref{eq:weylrep} extends to a unitary representation $W_J$ of $\centV$ on $\hilb S(V_J)$. It is strongly continuous.
\end{thm}
\begin{proof}
After filling in the definitions of the multiplication in $\centV$ and of $W_J$, checking the equation
\[
W_J\bigl((\xi, z) \cdot (\eta, w)\bigr)(e^\kappa) = W_J(\xi, z) \bigl(W_J(\eta, w)(e^\kappa)\bigr)
\]
for all $(\eta, w) \in \centV$ quickly comes down to verifying whether
\[
\frac12 \langle \xi, \eta \rangle_J - \frac12 \langle \eta, \xi \rangle_J = -2\pi i \cdot S(\xi, \eta).
\]
This follows from the definition in \cref{thm:tamedcompstruct}\ref{thmitm:tamedcompstruct-compHilb} of $\langle \cdot, \cdot \rangle_J$ in terms of $S$ and $J$:
\[
\frac12 \bigl(\langle \xi, \eta \rangle_J - \langle \eta, \xi \rangle_J\bigr) = \frac12 \bigl(\langle \xi, \eta \rangle_J - \overline{\langle \xi, \eta \rangle}_J\bigr) \\
	= i \cdot \Im \langle \xi, \eta \rangle_J \\
	= - 2 \pi i \cdot S(\xi, \eta).
\]

Proving that inner products are preserved is a matter of writing out
\begin{align*}
\bigl\langle W_J(\xi, z)(e^\eta), W_J(\xi, z)(e^\kappa) \bigr\rangle &=
	\abs z e^{-\frac12 \langle \xi, \xi \rangle_J - \langle \eta, \xi \rangle_J} \cdot \overline{e^{-\frac12 \langle \xi, \xi \rangle_J - \langle \kappa, \xi \rangle_J}} \cdot \langle e^{\eta + \xi}, e^{\kappa + \xi} \rangle \\
	&= e^{- \langle \xi, \xi \rangle_J} e^{- \langle \eta, \xi \rangle_J} e^{- \langle \xi, \kappa \rangle_J} \cdot e^{\langle \eta + \xi, \kappa + \xi \rangle_J} \\
	&= e^{\langle \eta, \kappa \rangle_J} = \langle e^\eta, e^\kappa \rangle,
\end{align*}
where we used \cref{thm:cohvecs-props}\ref{thmitm:cohvecs-props-innerprod} for the last equality.

That the action~\eqref{eq:weylrep} first extends from coherent vectors to finite linear combinations of those is implied by \cref{thm:cohvecs-props}\ref{thmitm:cohvecs-props-lindep}, and that it then extends further to all of $\hilb S(V_J)$ follows from \cref{thm:cohvecs-props}\ref{thmitm:cohvecs-props-dense}.

To prove strong continuity of this representation we show that for all vectors $v \in \hilb S(V_J)$ the map $\centV \to \hilb S(V_J)$ given by $(\xi, z) \mapsto W_J(\xi, z)(v)$ is continuous. It is sufficient to prove this when $v$ is an coherent vector $e^\kappa$ and in that case this can be seen from the continuity of the inner product $\langle \cdot, \cdot \rangle_J$ together with \cref{thm:cohvecs-props}\ref{thmitm:cohvecs-props-cont}, which tells us that the composite function
\[
\centV \twoheadrightarrow V \xrightarrow{\sim} V_J \to \hilb S(V_J), \qquad (\xi, z) \mapsto \xi \mapsto \kappa + \xi \mapsto e^{\kappa + \xi}
\]
is continuous.
\end{proof}

The factor $e^{-\frac12 \langle \xi, \xi \rangle_J - \langle \kappa, \xi \rangle_J}$ in~\eqref{eq:weylrep} can be understood as a correction to make the action unitary. It has a more conceptual meaning also, though, namely that of a certain \defn{Radon--Nikodym derivative}. This is explained in, for example, \parencite[Chapter 1, §1.3]{ismagilov:reps-infdim}.

\begin{dfn}
The representation~\eqref{eq:weylrep} of the Heisenberg group $\centV$ on the bosonic Fock space $\hilb S(V_J)$ is called the \defn{Weyl representation} (associated to the complex structure $J$ on $V$).
\end{dfn}

\begin{thm}
The Weyl representation $W_J$ is irreducible.
\end{thm}
\begin{proof}
(After \parencite[Theorem 7.1]{ismagilov:reps-infdim}.) Let $T$ be an endomorphism of $W_J$. We wish to show that $T$ is a scalar multiple of the identity operator because that will imply the irreducibility of $W_J$ by Schur's lemma.

Let us first study the vector $T(1)$, where $1 \in \CC =\vcentcolon \Sym^0(V_J) \subseteq \hilb S(V_J)$. Define for a fixed vector $\xi \in V_J$ the function
\begin{equation}
\label{eq:weylrep-irred-func}
\CC \to \CC, \qquad z \mapsto \langle T1, e^{z\xi} \rangle.
\end{equation}
By expanding $e^{z\xi}$ into a power series we see that this function is anti-holomorphic. On the other hand, using the unitarity of $W_J$ and that $1 = e^0$, we can rewrite it as
\begin{align*}
\langle T1, e^{z\xi} \rangle &= \bigl\langle W_J(-z\xi, 1)(T1), W_J(-z\xi, 1)(e^{z\xi})\bigr\rangle \\
	&= \bigl\langle T W_J(-z\xi, 1)(1), e^{\frac12 \abs z^2 \langle \xi, \xi \rangle_J} \cdot 1\bigr\rangle \\
	&= \bigl\langle T(e^{-\frac12 \abs z^2 \langle \xi, \xi \rangle_J} \cdot e^{-z\xi}), e^{\frac12 \abs z^2 \langle \xi, \xi \rangle_J} \cdot 1\bigr\rangle \\
	&= \langle T e^{-z\xi}, 1 \rangle = \langle e^{-z\xi}, T^* 1 \rangle.
\end{align*}
Hence~\eqref{eq:weylrep-irred-func} is also holomorphic. Therefore, it is a constant function. It then follows from the expressions~\eqref{eq:cohvecs-dense-power} and~\eqref{eq:cohvecs-dense-genmon} respectively that $\langle T1, \xi^k \rangle = 0$ and $\langle T1, \xi_1 \cdots \xi_k \rangle = 0$ for all $k \geq 1$ and monomials $\xi_1 \cdots \xi_k \in \Sym^k(V_J)$. This means that $T(1)$ is a scalar multiple of $1$, say, $T(1) = \alpha \in \CC$. Then also $T(e^\xi) = \alpha \cdot e^\xi$ holds for all $\xi \in V_J$ since
\[
e^\xi = e^{\frac12 \langle \xi, \xi \rangle_J} W_J(\xi, 1)(1).
\]
The density of the coherent vectors in $\hilb S(V_J)$ finishes the argument. In other words, we use that $1$ (or, more generally, any other coherent vector) is cyclic for $W_J$.
\end{proof}

The following \namecref{thm:weylrep-functor} explains that the construction of a Weyl representation is in a certain sense functorial.

\begin{prop}
\label{thm:weylrep-functor}
Let $(V_i, S_i, J_i)$ for $i=1,2$ be two triples satisfying the demands above and let $g\colon V_1 \xrightarrow{\sim} V_2$ be an $\RR$-linear isomorphism which preserves the skew forms $S_i$ and intertwines the complex structures $J_i$. Write $g^* \xi_1 \in V_2$ for the image of $\xi_1 \in V_1$. Then
\begin{enumerate}
\item $g$ lifts to a continuous group isomorphism $\centV_1 \xrightarrow{\sim} \centV_2$ between the associated Heisenberg groups via $g \cdot (\xi_1, z_1) := (g^* \xi_1, z_1)$, where $(\xi_1, z_1) \in \centV_1$, and
\item there is a unitary operator $U(g)\colon \hilb S((V_1)_{J_1}) \xrightarrow{\sim} \hilb S((V_2)_{J_2})$ which extends the ($\RR$-linear) isomorphism $g$ between the subspaces $(V_i)_{J_i} = \Sym^1((V_i)_{J_i}) \subseteq \hilb S((V_i)_{J_i})$,
\end{enumerate}
such that $U(g)$ intertwines the Weyl representations $W_{J_i}$ of $\centV_i$ on $\hilb S((V_i)_{J_i})$, that is,
\[
U(g) W_{J_1}(\xi_1, z_1) U(g)^* = W_{J_2}\bigl(g\cdot(\xi_1, z_1)\bigr).
\]
\end{prop}

The proof of the claim above is identical to the one in the `absolute' (as opposed to `relative') situation when $(V_1, S_1, J_1) = (V_2, S_2, J_2)$. We will therefore only give a proof in this latter case and we simultaneously include the hypothesis that we have a group of such automorphisms $g$:

\begin{prop}
\label{thm:weylrep-intertwin}
Let $G$ be a topological group acting by $\RR$-linear automorphisms on $V$ which preserves $S$, commutes with $J$ and is strongly continuous with respect to the norm topology on $V$ induced by $\langle \cdot, \cdot \rangle_J$. Write $g^*\xi$ for the translate of $\xi \in V$ by $g \in G$. Then
\begin{enumerate}
\item\label{thmitm:weylrep-intertwin-heis} $G$ acts strongly continuously on $\centV$ as $g \cdot (\xi, z) := (g^* \xi, z)$, where $(\xi, z) \in \centV$, and
\item\label{thmitm:weylrep-intertwin-hilb} there is a representation $U\colon G \to \lieU(\hilb S(V_J))$ which extends the ($\RR$-linear) action on the subspace $V_J = \Sym^1(V_J) \subseteq \hilb S(V_J)$,
\end{enumerate}
such that the intertwining property
\begin{equation}
\label{eq:weylrep-intertwin}
U(g) W_J(\xi, z) U(g)^* = W_J\bigl(g \cdot (\xi, z)\bigr)
\end{equation}
is satisfied.
\end{prop}
\begin{proof}
Indeed, \ref{thmitm:weylrep-intertwin-heis} is true because $G$ preserves $S$, and so it also does not perturb the cocycle $c$ defining $\centV$.

To prove \ref{thmitm:weylrep-intertwin-hilb}, we first note that since $G$ commutes with $J$ it acts by $\CC$-linear operators on $V_J$. Moreover, it commuting with $J$ and preserving $S$ implies by the definition of the Hermitian inner product $\langle \cdot, \cdot \rangle_J$ in \cref{thm:tamedcompstruct}\ref{thmitm:tamedcompstruct-compHilb} that $G$ preserves $\langle \cdot, \cdot \rangle_J$. That is, $G$ acts by unitary operators on the complex pre-Hilbert space $(V_J, \langle \cdot, \cdot \rangle_J)$.

By applying $G$ factor-wise to monomials in the vectors of $V_J$, we see from the definition of the inner product~\eqref{eq:fockspinnerprod} that we get a unitary $G$-action on $\Sym^*(V_J)$ and hence also on its completed Hilbert space $\hilb S(V_J)$. We denote the latter action by $U$. Equivalently, we may define $U$ first argument-wise on the coherent vectors by setting $U(g)(e^\kappa) := e^{g^*\kappa}$. \cref{thm:cohvecs-props}\ref{thmitm:cohvecs-props-innerprod} says that this respects inner products and therefore the action extends to all of $\hilb S(V_J)$.

Proving the strong continuity of $U$ comes down to checking whether, if $\kappa \in V_J$, then $g_n \to g$ in $G$ implies that $e^{g_n^* \kappa} \to e^{g^* \kappa}$ in $\hilb S(V_J)$. This follows from the strong continuity of the $G$-action on $V_J$ and \cref{thm:cohvecs-props}\ref{thmitm:cohvecs-props-cont}.

To show the intertwining property~\eqref{eq:weylrep-intertwin} it suffices to equate the actions of the operators on both sides on coherent vectors. If $\kappa \in V_J$, then
\begin{align*}
U(g) W_J(\xi, z) U(g)^*(e^\kappa) &= U(g) W_J(\xi, z) e^{(g^{-1})^* \kappa} \\
	&= z \cdot e^{-\frac12 \langle \xi, \xi \rangle_J - \langle (g^{-1})^* \kappa, \xi \rangle_J} \cdot U(g)\bigl(e^{(g^{-1})^*\kappa + \epsilon}\bigr) \\
	&= z \cdot e^{-\frac12 \langle \xi, \xi \rangle_J - \langle (g^{-1})^* \kappa, \xi \rangle_J} \cdot e^{\kappa + g^*\xi}.
\end{align*}
Because $g$ preserves $\langle \cdot, \cdot \rangle_J$ this is equal to
\[
W_J(g^* \xi, z)(e^\kappa) = W_J\bigl(g \cdot (\xi, z)\bigr)(e^\kappa). \qedhere
\]
\end{proof}

\begin{rmk}[The groupoid of Weyl representations]
While the above result easily allows one to exhibit many symmetries of a Weyl representation---sufficiently many for the purposes in this thesis---it is worth noting that these are not the only ones such a representation possesses. It namely turns out that, in order to implement the elements of a group $G$, they need not necessarily commute with $J$, but conjugating $J$ should merely `not distort $J$ too much'. More precisely, for an element $g \in G$ the commutator $[g, J]$ should be a \defn{Hilbert--Schmidt operator} on $V_J$. This condition is necessary as well. The result is then not a representation of $G$ itself on $\hilb S(V_J)$, but of a certain $\phasegp$-central extension instead. We refer for these claims to \parencite[Proposition 9.5.9]{pressley:loopgps}.

We can therefore say that to the space of complex structures on $V$ that satisfy \cref{thm:tamedcompstruct} there is associated a category of which the objects are the corresponding Weyl representations and the morphisms are the $\centV$-intertwiners. Schur's lemma then tells us that this is a groupoid, while the aforementioned claims (or, rather, their generalisations to the `relative' case) declare that it is in general not connected: the connected components are exactly the full subgroupoids associated to each \emph{polarisation class} of complex structures that differ from each other by a Hilbert--Schmidt operator.
\end{rmk}

With these Weyl representations we have constructed one family of irreducible representations of a Heisenberg group. Having the classical Stone--von Neumann theorem for finite-dimensional Heisenberg groups in mind (see for example \parencite[Theorem 1.50]{folland:phasespace}), it is reasonable to ask to what extent a similar uniqueness result holds for the possibly infinite-dimensional Heisenberg groups we are considering here. Of course, the action~\eqref{eq:weylrep} can be tweaked a little by letting the central subgroup $\phasegp$ act by a non-trivial character instead. It turns out that this is the only freedom we have when we additionally assume the positive energy property:

\begin{thm}[The Stone--von Neumann theorem for positive energy representations]
\label{thm:stonevonneumann}
(See \parencite[Proposition 9.5.10]{pressley:loopgps}.) Suppose there exists a complex structure $J$ and, for some $m \geq 1$, a $\Rot^{(m)}(S^1)$-action on $V$ satisfying the demands of \cref{thm:weylrep-intertwin} such that the resulting representation of $\Rot^{(m)}(S^1)$ on $\hilb S(V_J)$ is of positive energy. Then $W_J$ is, up to isomorphism, the unique irreducible, positive energy representation of $\centV$ such that the central subgroup $\phasegp$ acts as $z \mapsto z$.
\end{thm}

According to \cref{thm:posrepcontainsirrep}, the positive energy assumption additionally gives us complete knowledge of representations of $\centV$ that are not necessarily irreducible: they are simply direct sums of the unique irreducible one.

The proof of the above \namecref{thm:stonevonneumann} involves the Lie algebra of $\centV$ and its representation by densely defined, skew-adjoint operators on any representation of $\centV$. These are topics we did not have the chance to treat.

\backmatter

\printbibliography
\chapter{Samenvatting}

Hier zal ik een poging doen de belangrijkste onderwerpen en resultaten uit dit proefschrift op een meer elementaire manier uiteen te zetten. We zullen eerst de unigekleurde toruslusgroepen uit Hoofdstuk~\hyperref[chap:unicol]{2} behandelen, om vervolgens over te gaan naar de theorie van de bigekleurde toruslusgroepen uit Hoofdstuk~\hyperref[chap:bicol]{3}.

\section{Unigekleurde toruslusgroepen}

Om uit te leggen wat een toruslusgroep is beginnen we met het kiezen van een willekeurige, maar vaste torus $T$ waar de toruslusgroep van af zal hangen. (De betekenis van het bijvoeglijk naamwoord `unigekleurde' in de titel van deze sectie zal toegelicht worden in de volgende sectie.) Vervolgens beschouwen we de verzameling van alle gesloten lussen die op $T$ liggen. Wij eisen van elke lus slechts dat deze in zekere zin glad is, maar het is hem toegestaan om zichzelf te doorsnijden of meerdere keren om de torus te wikkelen vóór hij sluit. Deze collectie van lussen is dus zeer groot, maar blijkt over een interessante structuur te beschikken. Merk namelijk eerst op dat $T$ als torus niet alleen een meetkundig object is, maar ook een algebraïsche structuur bezit. De punten op $T$ kunnen namelijk bij elkaar worden opgeteld, van elkaar worden afgetrokken en ook is er voor deze optelling een neutraal punt $0$. Dit maakt $T$ een \defn{groep}. Stel een lus $\gamma$ op $T$ nu voor als een afbeelding van de eenheidscirkel $S^1$ naar $T$ toe.
\begin{figure}[bth]
\centering%
	\begin{tikzpicture}[>={Straight Barb}, line width=rule_thickness]	
    \node[anchor=south west,inner sep=0] (image) at (0,0) {\includegraphics[width=0.7\textwidth]{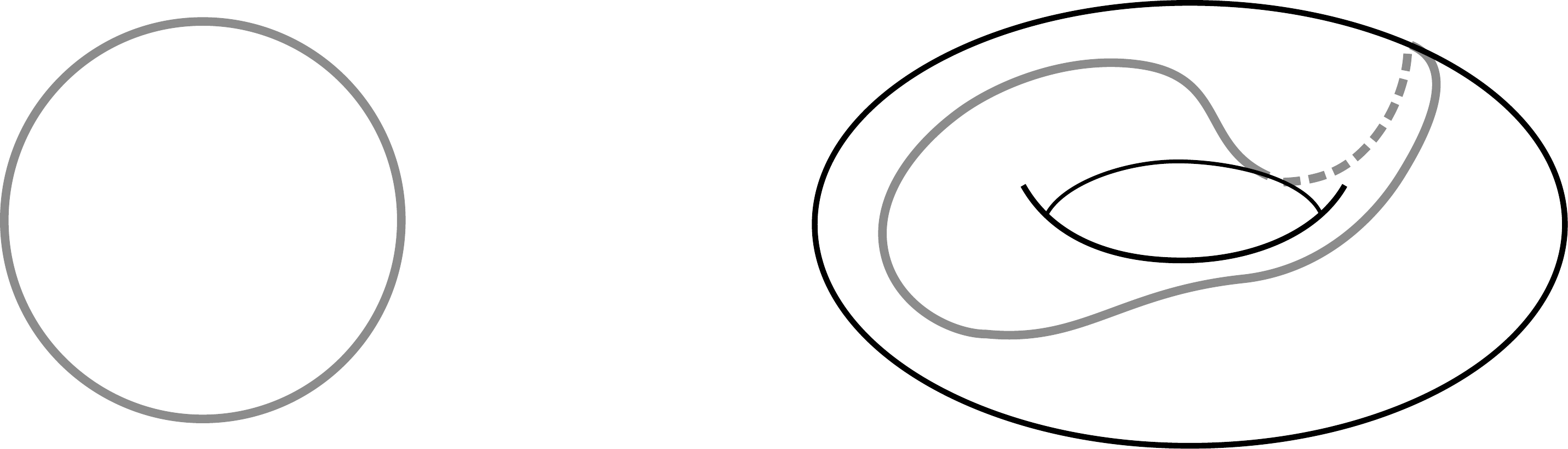}};
    \begin{scope}[x={(image.south east)},y={(image.north west)}]
        \node at (0,0) {$S^1$};
        \node at (1,0) {$T$};
		\draw[white, line width=2.5pt] (0.3, 0.5) to [bend left] node [above] {$\gamma$} (0.55, 0.5);
		\draw[->] (0.3, 0.5) to [bend left] node [above] {$\gamma$} (0.55, 0.5);
    \end{scope}
	\end{tikzpicture}
\end{figure}

(In de illustratie is $T$ als zijnde $2$-dimensionaal weergegeven, maar een hoger-dimensionale torus is in onze discussie zeker ook toegestaan.) Dit wil zeggen dat we de punten op de lus $\gamma$ schrijven als $\gamma(\theta)$, waar $\theta$ een punt op $S^1$ is. Als $\rho$ dan een andere lus op $T$ is kan deze met $\gamma$ puntsgewijs worden opgeteld, wat een nieuwe lus $\gamma + \rho$ oplevert. Hiermee zien we dat de verzameling van alle lussen op $T$ zelf weer een groep is, die we noteren met $LT$ en de \defn{(torus)lusgroep} behorende bij $T$ noemen. Zijn neutrale element voor de optelling is de lus die volledig in het punt $0$ van $T$ geconcentreerd is.

In dit werk hebben niet zo zeer toruslusgroepen zelf onze interesse, als wel bepaalde van hun \defn{centrale uitbreidingen} die geassocieerd zijn aan \defn{roosters}. Een (positief definiet) rooster van dimensie $n$ is een oneindige verzameling regelmatig verdeelde punten in de $n$-dimensionale Euclidische ruimte $\RR^n$ die 1) de oorsprong $0$ bevat, 2) $\RR^n$ opspant, en 3) zodanig gepositioneerd is dat het inprodukt tussen de twee vectoren vanuit $0$ naar iedere twee punten toe een geheel getal is. Onderstaand zijn twee voorbeelden van (delen van) $2$-dimensionale roosters weergegeven.
\begin{figure}[hbt]
\centering
	\begin{minipage}[t]{0.4\textwidth}
	\centering
	\begin{tikzpicture}
	\foreach \i in {0,...,3}
		\foreach \j in {0,...,3} {
			\filldraw (\i,\j) circle (1pt);
		}
	\end{tikzpicture}
	\legend{Een vierkant rooster.}
	\end{minipage}%
	\qquad%
	\begin{minipage}[t]{0.4\textwidth}
	\centering
	\begin{tikzpicture}
	\foreach \i in {0,...,2} {
		\filldraw (\i*1.414, 0) circle (1pt);
	}
	\foreach \i in {0,...,3} {
		\filldraw (-0.707 + \i*1.414, 1.225) circle (1pt);
	}
	\foreach \i in {0,...,2} {
		\filldraw (\i*1.414, 2.449) circle (1pt);
	}
	\foreach \i in {0,...,3} {
		\filldraw (-0.707 + \i*1.414, 3.674) circle (1pt);
	}
	\end{tikzpicture}
	\legend{Een hexagonaal rooster.}
	\end{minipage}
\end{figure}

In hogere dimensies zijn exotischere voorbeelden te vinden en het is vaak vruchtbaar om een rooster als een meetkundig object op zich te beschouwen.

In Sectie~\hyperref[sec:unicol-centext]{2.2} van dit proefschrift wordt verteld hoe, vanuit de toruslusgroep $LT$, met behulp van ieder rooster $\Lambda$ van een speciaal type en van dezelfde dimensie als $T$ een bepaalde grotere groep $\centL T$  gecreëerd kan worden. Een element van $\centL T$ is niet langer slechts een lus, maar een paar $(\gamma, z)$, waar $\gamma$ een lus op $T$ is en dus behoort tot $LT$, terwijl $z$ een complex getal van modulus $1$ is. Merk op dat zulke complexe getallen zelf ook een groep vormen, welke we noteren met $\phasegp$. Voor ieder element van $LT$ bevat $\centL T$ dus zoveel kopieën als er in $\phasegp$ liggen. Tot zover hebben wij $\centL T$ slechts als verzameling beschreven. De rol van het rooster $\Lambda$ ligt in het bepalen van de vermenigvuldiging van de elementen van $\centL T$ met elkaar op zo een manier dat $\centL T$, net zoals $LT$, een groep vormt. We zeggen dat $\centL T$ een \defn{centrale uitbreiding} van $LT$ is (\defn{langs} de groep $\phasegp$).

Een reden om in deze centrale uitbreidingen geïnteresseerd te zijn is het feit dat ze een inzichtelijke \defn{representatietheorie} hebben. Grof gezegd is een \defn{representatie} van een groep een specifieke manier waarop de groep zich kan manifesteren als symmetrieën van een vectorruimte met een inprodukt. Onder een representatie wordt een element van een groep dus voorgesteld als een lineaire transformatie die een gegeven inprodukt behoudt. Het bestuderen van de representaties van een groep kan leiden tot een beter begrip van de groep zelf. Het is handig om hiermee bij de meest elementaire te beginnen, de zogenaamde \defn{irreducibele} representaties, ook omdat deze als bouwstenen kunnen dienen voor algemene representaties.

Voor de centrale uitbreidingen $\centL T$ is deze restrictie helaas toch niet genoeg om inzicht in hun representatietheorie te verkrijgen. Echter blijkt dat wanneer we ook een \defn{positieve energie} conditie eisen deze studie ineens wel behapbaar wordt. In Sectie~\hyperref[sec:unicol-reptheory]{2.5} leggen we namelijk uit dat $\centL T$ slechts eindig veel irreducible, positieve energie representaties bezit en dat deze expliciet te classificeren en construeren zijn.

\section{Bigekleurde toruslusgroepen}

De hier boven uiteengezette theorie van toruslusgroepen, samengevat uit Hoofdstuk~\hyperref[chap:unicol]{2}, is vrij klassiek en al eerder beschreven in de literatuur. In dit proefschrift is een nieuwe generalisatie van toruslusgroepen geïntroduceerd, die we \defn{bigekleurde toruslusgroepen} hebben genoemd, en er is geprobeerd om analoge resultaten over ze te vinden zoals we die kennen over toruslusgroepen.

Voor het definiëren van een bigekleurde toruslusgroep is meer nodig dan een enkele torus. We fixeren in plaats daarvan drie torussen, $T_\wh$, $H$, en $T_\bl$ van dezelfde dimensie, samen met twee surjectieve afbeeldingen, één van $H$ naar $T_\wh$ en één van $H$ naar $T_\bl$. Van deze afbeeldingen wordt geëist dat ze zowel de meetkundige (preciezer: de differentieerbare) structuren van de torussen respecteren, als hun groepsstructuren. We zien $T_\wh$ en $T_\bl$ als respectievelijk een witte en een zwarte torus.

Net zoals een enkele torus heeft ook de zojuist genoemde lijst van data een begrip van `lus', namelijk een \defn{bigekleurde lus}. Deze is gedefinieerd als het geheel van 1) een (glad) pad op $T_\wh$, 2) een (glad) pad op $T_\bl$, en 3) twee aangewezen punten op $H$ zodanig dat ze worden gestuurd naar de eindpunten van de paden op $T_\wh$ en $T_\bl$ onder de afbeeldingen van $H$ naar deze torussen. Het is handig om de eenheidscirkel $S^1$ te zien als zijnde opgeknipt in drie stukken: een linkerhelft, een rechterhelft en hun overlappende twee punten, en vervolgens een bigekleurde lus te beschouwen als een drietal afbeeldingen vanuit deze stukken naar de drie torussen toe.
\begin{figure}[hbt]
\centering%
	\begin{tikzpicture}[>={Straight Barb}, line width=rule_thickness]	
    \node[anchor=south west,inner sep=0] (image) at (0,0) {\includegraphics[width=0.8\textwidth]{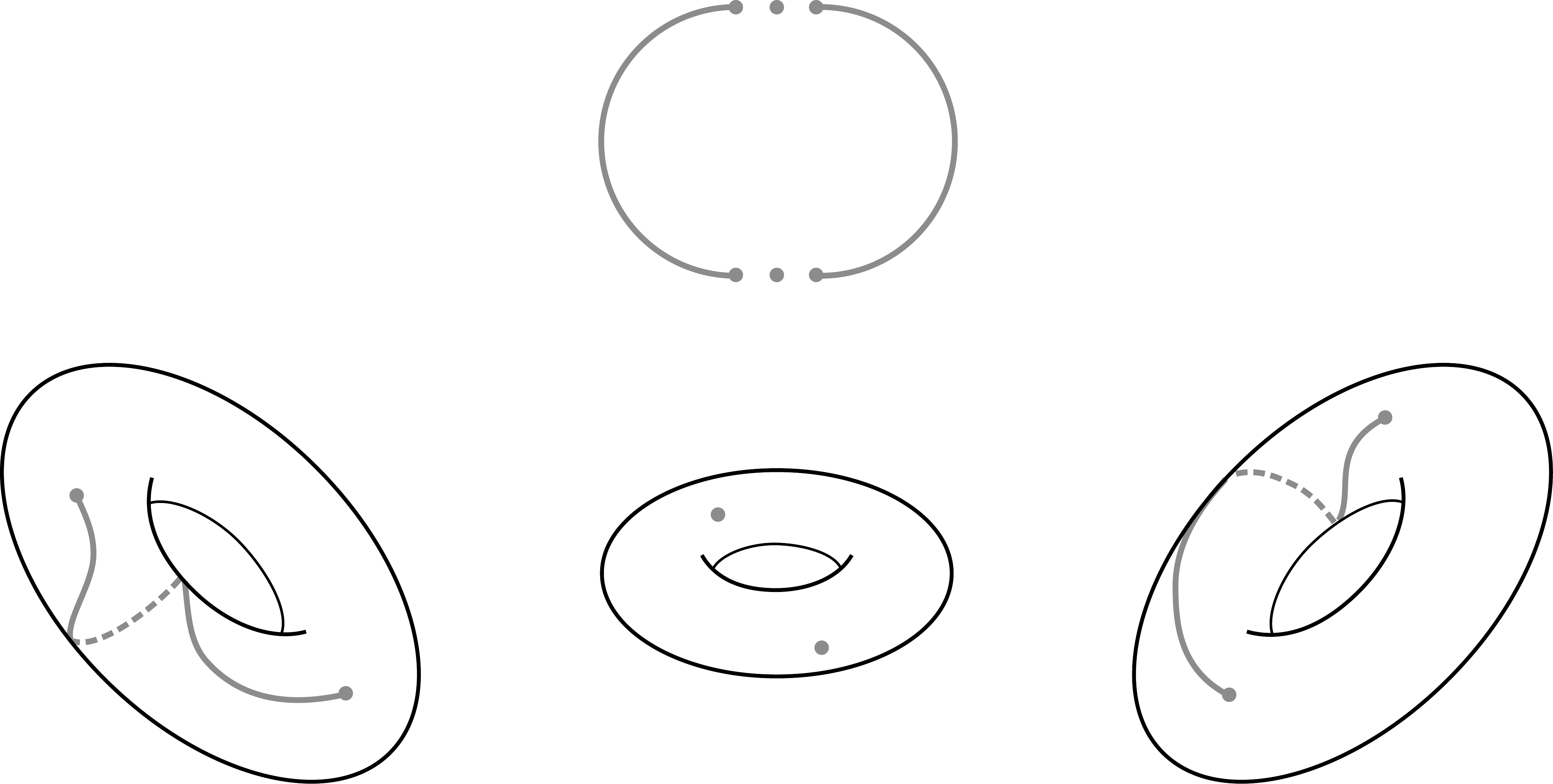}};
    \begin{scope}[x={(image.south east)},y={(image.north west)}]
        \node at (0.04, 0.04) {$T_\wh$};
        \node at (0.96, 0.04) {$T_\bl$};
        \node at (0.5, 0.07) {$H$};
        \node at (0.5, 0.825) {$S^1$};
		\draw[white, line width=2.5pt] (0.5, 0.58) to (0.5, 0.35);
		\draw[->] (0.5, 0.58) to (0.5, 0.35);
		\draw[white, line width=2.5pt] (0.65, 0.75) to [bend left] (0.85, 0.44);
		\draw[->] (0.65, 0.75) to [bend left] (0.85, 0.44);
		\draw[white, line width=2.5pt] (0.35, 0.75) to [bend right] (0.15, 0.44);
		\draw[->] (0.35, 0.75) to [bend right] (0.15, 0.44);
		\draw[->>] (0.64, 0.25) to (0.71, 0.225);
		\draw[->>] (0.36, 0.25) to (0.29, 0.225);
    \end{scope}
	\end{tikzpicture}
\end{figure}

Op dezelfde wijze als in onze eerdere uitleg over gewone lussen op een torus vormt ook de verzameling van bigekleurde lussen een groep onder de puntsgewijze optelling. Deze groep noteren we met $L(T_\wh, H, T_\bl)$ en wordt een \defn{bigekleurde toruslusgroep} genoemd. We zien dat als we de torussen $H$, $T_\wh$ en $T_\bl$ gelijk aan elkaar kiezen en de twee afbeeldingen tussen hen als de identiteit nemen dat $L(T_\wh, H, T_\bl)$ dan niets anders is dan $LH$ (onder enkele extra condities die wij hier niet noemen). Daarom is een bigekleurde toruslusgroep inderdaad een generalisatie van een gewone, en noemen wij dit tweede type achteraf beschouwd \defn{unigekleurd}.

In Sectie~\hyperref[sec:bicol-centext]{3.2} construeren wij $\phasegp$-centrale uitbreidingen $\centL(T_\wh, H, T_\bl)$ van zo een groep $L(T_\wh, H, T_\bl)$ met eigenschappen analoog aan die van unigekleurde groepen. Deze keer vereist dit niet een enkel, maar een drietal roosters $\Lambda_\wh$, $\Gamma$ en $\Lambda_\bl$ van gelijke dimensies, samen met twee injectieve afbeeldingen, één van $\Gamma$ naar $\Lambda_\wh$ en één van $\Gamma$ naar $\Lambda_\bl$. Uiteraard eisen we dat deze afbeeldingen de structuren van de roosters bewaren. Samen met verdere eigenschappen bewezen in Sectie~\hyperref[sec:bicol-centext-diffnS1action]{3.3} toont het bestaan van deze centrale uitbreidingen dat bigekleurde toruslusgroepen een werkelijke generalisatie van de unigekleurde theorie toe laten.

Deze stelling wordt tenslotte kracht bijgezet in Sectie~\hyperref[sec:bicol-reptheory]{3.4} door op een zelfde wijze als in Hoofdstuk~\hyperref[chap:unicol]{2} de irreducibele, positieve energie representaties van zo een groep $\centL(T_\wh, H, T_\bl)$ te classificeren en te construeren. Hier blijkt opnieuw dat er slechts eindig veel van zulke representaties bestaan.

\chapter{Acknowledgements}

Thank you, André, for accepting me as a student and thereby giving me the opportunity to enter the world of research. I am grateful for your trust, patience and the freedom you allowed me. Your efforts of support during trying times and while we were separated geographically and dealt with digital malfunctions were more than a student could ask from an advisor. If I gained any sense of mathematical beauty at all, it owes largely to your lectures, my attempts to view material through your lens and seeing my work slowly evolve by your guiding hand.

The next person to offer thanks to is Erik van den Ban. Erik, while your role as a promotor was initially formal, it expanded greatly in my final year. You became a second sounding board for the ideas developed with André, found crucial, technical errors and helped to repair them. Your experience with administrative matters in the end phase was also much appreciated.

I am grateful to the members of my reading committee, consisting of Marius Crainic, Christopher Douglas, Karl-Hermann Neeb, Karl-Henning Rehren and Christoph Schweigert, for their time and effort spent studying this thesis and their valuable feedback.

Cécile, Jean, Ria and Sylvia, thank you for always keeping your door open for unexpected questions and requests and greeting us with a smile. Without you the workings of the department would immediately grind to a halt. I am also indebted to Gunther Cornelissen for his advice on planning the defence.

I was lucky to share the department with fellow PhD students and postdocs who made work much more enjoyable. Thanks, guys, for the lunch breaks we had and especially Ajinkya, Anshui, Dali, Joost, Pouyan and Sebastian for being great office mates. I would also like to single out all the participants of the \textsc{qft} seminar we organised. Jules, thank you for initiating it, convincing me to take part in it and being a driving force behind it. It gave me much needed practice giving talks and I regret not starting the seminar earlier.

During my visits abroad I had the joy to meet many other students, among whom I want to thank in particular Byung Do, Daniel, Dominic, Lars, Marco, Mark, Massi, Matt and Peter for sharing their knowledge of and passion for mathematics so freely. The staff at the Hausdorff Research Institute for Mathematics deserves praise as well for their hospitality and generosity during my stay in Bonn in the summer of 2015, which created an ideal working atmosphere.

Lastly, I want to thank my parents for all their support throughout my studies, and my sister for teaching me to read.

\end{document}